\documentclass{article}

\usepackage{amsmath,amsfonts,amssymb,bbm,graphicx,color,bm,amsthm}
\usepackage{cite}
\newcommand{\openone}{\mathbbm{1}}
\newcommand{\ket}[1]{\left|#1\right\rangle}
\newcommand{\bra}[1]{\left\langle#1\right|}
\newcommand{\tr}{\mathrm{tr}}
\newcommand{\mcS}{\mathcal S}
\newcommand{\mcP}{\mathcal P}

\newcommand{\mc}[1]{\mathcal{#1}}
\newcommand{\cstar}{$\mathrm{C}^*$}

\newcommand{\lin}[1]{{\bm L}(\mathbb C^{#1})}
\newcommand{\MPS}[1]{\big\vert\mathcal M(#1)\big\rangle}

\theoremstyle{plain}
\newtheorem{lemma}{Lemma}[section]
\newtheorem{theorem}[lemma]{Theorem}
\newtheorem{corollary}[lemma]{Corollary}

\newtheorem{definition}[lemma]{Definition}

\theoremstyle{definition}
\newtheorem{observation}[lemma]{Observation}

\numberwithin{equation}{section}  % number equations by section

\title{PEPS as ground states: degeneracy and topology}
\author{Norbert Schuch,$\hspace*{-0.3em}^{1,2}$ 
    Ignacio Cirac,$\hspace*{-0.2em}^{2}$
    and David P\'erez-Garc\'ia$^{3}$}
\date{}

\begin{document}

\maketitle
\vspace*{-2.5em}
\begin{center}
\footnotesize
$^1$ California Institute of Technology, Institute for Quantum
Information,\\
MC 305-16, Pasadena CA 91125, USA\\
$^2$ Max-Planck-Institut f\"ur Quantenoptik,\\ Hans-Kopfermann-Str. 1,
D-85748 Garching, Germany\\
$^3$ Dpto.\ Analisis Matematico and IMI,\\ Universidad Complutense de Madrid,
28040 Madrid, Spain
\end{center}

\vspace{1em}

\begin{abstract}
We introduce a framework for characterizing Matrix Product States (MPS)
and Projected Entangled Pair States (PEPS) in terms of symmetries.  This
allows us to understand how PEPS appear as ground states of
local Hamiltonians with finitely degenerate ground states and to characterize
the ground state subspace.  Subsequently, we apply our framework to show
how the topological properties of these ground states can be explained
solely from the symmetry: We prove that ground states are locally
indistinguishable and can be transformed into each other by acting on a
restricted region, we explain the origin of the topological entropy, and
we discuss how to renormalize these states based on their symmetries.
Finally, we show how the anyonic character of excitations can be
understood as a consequence of the underlying symmetries.  \end{abstract}

\section{Introduction}

What are the entanglement properties of quantum many-body
states which characterize ground states of Hamiltonians with local
interactions? The answer seems to be ``an area law'': the bipartite
entanglement between any region and its complement grows as the area
separating them -- and not as their volume, as is the case for a random
state (see~\cite{cramer:arealaw-review} for a recent review).
Moreover, particular corrections to this scaling law are linked with
critical points (logarithmic corrections) or topological order (additive
corrections).  A rigorous general proof of the area law, however, could up
to now only be given for the case of one-dimensional
systems~\cite{hastings:arealaw}, where an area law has been proven for all
systems with an energy gap above the ground state, whereas the currently
strongest result for two
dimensions~\cite{masanes:arealaw,hastings:mps-entropy-ent} requires a
hypothesis on the
eigenvalue distribution of the Hamiltonian. Surprisingly, there is a
completely general proof in arbitrary dimensions if instead, we consider
the corresponding quantity for thermal
states~\cite{wolf:mutual-info-arealaw}, and similar links to 
topological order persist~\cite{iblisdir:topo-entropy-finite-T}.

The area law can be taken as a guideline for designing classes of quantum
states which allow to faithfully approximate ground states of local
Hamiltonians. There are several of these classes in the literature: Matrix
Product States (MPS)~\cite{ostlund-rommer} and Projected Entangled Pair
States (PEPS)~\cite{frank:2D-dmrg} are most directly motivated by the area
law, but there are other approaches such as MERA (the Multi-Scale
Entanglement Renormalization Ansatz)~\cite{vidal:mera} which e.g.\ is
based on the scale invariance of critical systems; all these classes are
summarized under the name of Tensor Network or Tensor Product States.
Though the main motivation to introduce them was numerical -- they
consitute variational ansatzes over which one minimizes the energy of a
target Hamiltonians and thus obtains an approximate description of the
ground state -- they have turned out to be powerful tools for
characterizing the role of entanglement in quantum many body systems, and
thus helped to improve our understanding of their physics.

In this paper, we are going to present a theoretical framework which
allows us to understand how MPS and PEPS appear as ground states of local
Hamiltonians, and to characterize the properties of their ground state
subspace.  This encompasses previously known results for MPS and
particular instances of PEPS, while simultaneously giving rise to a 
range of new phenomena, in particular topological effects.  Our work is
motivated by the contrast between the rather complete understanding in one
and the rather sparse picture in two dimensions, and we will review
what is known in the following. We will thereby focus on analytical
results, and refer the reader interested in numerical aspects
to~\cite{murg:peps-review}.

\subsection{Matrix Product States}

Matrix Product States (MPS)~\cite{ostlund-rommer} form a family of 
one-dimensional quantum states whose description is inherently local, in
the sense that the degree to which two spins can be correlated is related
to their distance. The total amount of correlations across any cut is
controlled by a parameter called the \emph{bond dimension}, such that
increasing the bond dimension allows to grow the set of states described.
MPS have a long history, which was renewed in 1992 when two apparently
independent papers appeared: In~\cite{fannes:FCS}, Fannes, Nachtergaele,
and Werner generalized the AKLT construction of~\cite{aklt} by introducing
the so-called Finitely Correlated States, which in retrospect can be
interpreted as MPS defined on an infinite chain; in fact, this work layed
the basis for our understanding of MPS and introduced many techniques
which later proved useful in characterizing
MPS~\cite{perez-garcia:mps-reps}.  The other was~\cite{white:DMRG},
where White introduced the Density Matrix Renormalization Group (DMRG)
algorithm, which can now be understood as a variational algorithm over the
set of MPS.  In~\cite{frank:dmrg-mps}, MPS were explained from a quantum
information point of view by distributing ``virtual'' maximally entangled
pairs between adjacent sites which can only be partially accessed by
acting on the physical system.  This entanglement-based perspective has
since then fostered a wide variety of results.

\subsubsection{The complexity of simulating 1D systems}

Motivated by the extreme success of DMRG, people investigated how
hard or easy the problem of approximating the ground state of a 1D local
Hamiltonian (or simply its energy) was. The history of this problem is
full of interesting positive and negative results. A number of them was 
devoted to prove that every ground state of a gapped 1D local Hamiltonian
can be approximated by an
MPS~\cite{osborne:1d-gs-approx,frank:faithfully}; this was finally proven
by Hastings~\cite{hastings:arealaw}, justifying the use of
MPS as the appropriate representation of the state of one-dimensional spin
systems.  Very recently, also in the positive, it was shown that dynamical
programming could be used to find the best approximation to the ground
state of a one-dimensional system within the set of MPS with fixed bond
dimension in a provably efficient
way~\cite{schuch:mps-dynprog,aharonov:mps-dynprog}.  On
the other hand, in the negative it could be shown that finding the ground
state energy of Hamiltonians whose ground states are MPS with a
bond dimension polynomial in the system size is
NP-hard~\cite{schuch:mps-gap-np}; this is based on a previous result of
Aharonov \emph{et~al.}~\cite{aharonov:1d-qma} proving that finding the
ground state energy of 1D Hamiltonians is QMA-complete (the quantum
version of NP-complete).

\subsubsection{Hidden orders, symmetries and entanglement in spin chains}

As we have seen, MPS provide the right description for one-dimensional
quantum spin chains. Therefore, and given their simple structure, one can
employ MPS to improve our understanding of the physics of one-dimensional
systems.  One field in which significant insight could be gained was the
characterization of symmetries in terms of entanglement.  First, the
relation between string order parameters and localizable entanglement was
explained in~\cite{verstraete:div-ent-length,venuti:locent-stringorder}.
In~\cite{david:stringorder-1d} (see also \cite{singh:mps-decomp-sym}),
global symmetries in generic MPS have been characterized, and related to
the
existence of string order parameters, thus explaining many of the
properties of string order, for instance its fragility
\cite{anfuso:stringorder-frag}. This characterization of global symmetries
was generalized to arbitrary MPS in \cite{sanz:mps-syms}, where it was used to
shed light on the Hamiltonian-free nature of the Lieb-Schultz-Mattis
theorem as well as to find new $\mathrm{SU(2)}$--invariant two-body
Hamiltonians with MPS ground states, beyond the AKLT and Majumdar-Ghosh
models.  Other examples of MPS with global symmetries were already
provided in \cite{fannes:FCS,karimipour:aklt-syms,tu:SOn-mps-gs}.
Recently, also reflection symmetry has been investigated, showing how it
provides topological protection of some MPS such as the odd-spin
AKLT model, as opposed to the even-spin case~\cite{pollmann:symprot-1d}.

MPS have also been extremely useful in understanding the scaling of
entanglement in quantum spin chains, where special attention has been devoted
to the case of quantum phase transitions. In \cite{wolf:mps-qpt}, MPS were used
to give examples of phase transitions with unexpected properties, namely
analytic ground state energy and finite entanglement entropy of an
infinite half-chain; the entanglement properties of these examples were
further analyzed in \cite{cozzini:mps-qpt-fidelity,wei:mps-qpt}. In
\cite{orus:mps-qpt-geoent},
MPS theory was used to compute how the geometric entanglement with respect
to large blocks diverges logarithmically with the correlation length near
a critical point, and thus takes the role of an order parameter (see also
\cite{wei:mps-qpt}).

Apart from that, MPS theory has been used to decompose global operations
(such as cloning or the creation of an entangled state) into a sequence of
local operations~\cite{schoen:hen-and-egg,lamata:sequential-operations},
to characterize renormalization group transformations and their fixed
points in 1D \cite{frank:renorm-MPS}, to understand which quantum circuits
can be simulated classically \cite{vidal:simulation-of-comput}, or even to
propose new numerical methods to solve differential equations
\cite{iblisdir:cont-mps} or to compress images
\cite{latorre:imagecompression}. 

But what
is it that makes MPS so useful in deriving all these results?

\subsubsection{The structure of MPS}

The main reason seems to be that MPS, despite being able to faithfully
represent the states of one-dimensional systems, have a simple and
well-understood structure, which makes them quite easy to deal with. For
instance, as shown in \cite{vidal:simulation-of-comput}, they naturally
reflect the Schmidt decomposition at any cut across the chain, which makes
dealing with their
entanglement properties particularly easy.  
If one moreover restricts to the physically
relevant case of translational invariant states, it turns out that one can
fully characterize the set of all translationally invariant MPS by
bringing them into a canonical
form~\cite{fannes:FCS,perez-garcia:mps-reps}; in fact this canonical form
constitutes one of the main ingredients in many of the results mentioned
above. 

In the canonical form, the matrices characterizing the MPS obtain a block
diagonal form, and the properties of the state can be simply read off the
structure of these blocks. In particular, for the case of one block
(termed the \emph{injective} case), it can be shown that the MPS arises as
the unique ground state of a so-called ``parent'' Hamiltonian with local
interactions, which moreover is frustration free. For the
\emph{non-injective} case where one has several blocks, the number of
blocks determines the degeneracy of the parent Hamiltonian, and the ground
state subspace is spanned by the injective MPS described by the individual
blocks~\cite{perez-garcia:mps-reps}.  Beyond that, the injective case has
other nice properties, such as an exponential decay of correlations.  For
the case of an infinite chain, all these properties were proven in
\cite{fannes:FCS}, together with the fact that the parent Hamiltonian of
an injective MPS has an energy gap above the ground state.
The block structure of the canonical form is also useful beyond the
relation of Hamiltonians and ground states, and e.g.\ allows to read off 
the type of the RG fixed point of a given 1D system.

\subsection{PEPS}

Projected Entangled Pair States (PEPS) constitute the natural
generalization of MPS to two and higher dimensions, motivated by the
quantum information perspective on MPS which views them as arising from
virtual entangled pairs between nearest
neighbors~\cite{sierra:2d-dmrg,frank:2D-dmrg}.  Though there has not yet
been a complete formal proof that PEPS approximate
efficiently all ground states of gapped local Hamiltonians, this could be
proven under a (realistic) assumption on the spectral density in the
low-energy regime~\cite{hastings:locally}, showing that PEPS are the
appropriate class to describe a large variety of two-dimensional systems.
However, as compared to MPS, PEPS are much harder to deal with:
For instance, computing expectation values of local observables, which
would be the key ingredient in any variational algorithm such as DMRG, is a
\#P-complete problem, and thus in particular
NP-hard~\cite{schuch:cplx-of-PEPS}.  This poses an obstacle to numerical
methods, and different ideas to overcome this problem  have been proposed
(we refer again to \cite{murg:peps-review} for numerical issues).
Fortunately, the bad news comes with good ones: The increase in complexity
allows to find a much larger variety of different interesting behavior
within PEPS as compared to MPS; for instance, in
\cite{frank:comp-power-of-peps} it is shown that there exist PEPS with a
power-law decay of two-point correlation functions, something which cannot
be achieved for MPS. 

\subsubsection{Many examples}

There have been identified various different classes of PEPS which exhibit rich
properties.  To start with, in~\cite{frank:mbc-peps} it has been shown how
the interpretation of the 2D cluster state as a PEPS can be used to
understand measurement based quantum
computation~\cite{raussendorf:cluster-short} -- a way of performing a
quantum computation solely by measurements -- by viewing it as a way to
carry out the computation as teleportation-based computation on the virtal
maximally entangled states underlying the PEPS description. This motivated
the search for different models for measurement based quantum computation,
and indeed models with very different properties have been subsequently
proposed~\cite{gross:mbqc-prl,gross:mbqc-pra,brennen:aklt-mbqc}.

Another category of examples has been found with regard to topological
models, where it was realized that many topologically ordered states have
a PEPS representation with a small bond dimension: The case of Kitaev's
toric code was observed in \cite{frank:comp-power-of-peps}, and this was
later generalized to all string net models in \cite{vidal:stringnet-peps}.
Yet, despite their ability to \emph{describe} these states, up to now PEPS
did not help much in \emph{understanding} the topological behavior of
these states, which is one of the things we will assess in this work.

\subsubsection{Few general results}

As we have seen, the class of PEPS is rich enough to incorporate states
with a variety of different behaviors. However, given the complexity of
topological systems or of (measurement based) quantum computation, is this
class still simple enough to prove useful as a tool to improve our
understanding of 2D systems? That is, can PEPS help to uncover new effects
and relations in nature, or to give a better understanding  of the
mechanisms behind quantum effects in two dimensions?

Judging from the experience with one-dimensional systems, in order to do
so it would be highly desirable to have an understanding of the structure
of PEPS comparable to the one obtained using the canonical form in one
dimension.  As it turns out, in the case of ``injective'' PEPS, several 1D
results can be transferred; in particular, injective PEPS appear as unique
ground states of their parent
Hamiltonian~\cite{perez-garcia:parent-ham-2d}.  Also,
in~\cite{perez-garcia:inj-peps-syms} it is shown how global symmetries can
be characterized in injective PEPS, which helped to understand the
mechanism behind the 2D version of the Lieb-Schultz-Mattis
theorem~\cite{hastings:liebschultzmattis-2d}, to define an appropriate
analogue for string orders in 2D~\cite{david:stringorder-1d}, and  to
improve the PEPS-based algorithms used to simulate 2D systems with
symmetries~\cite{singh:tns-decomp-symmetry}.  While these results
illustrate that PEPS are a useful tool to understand properties of quantum
states which appear as unique ground states of local Hamiltonians, it is
also true that some of the most interesting physics in two dimensions
takes place in systems which do not have unique ground states, but rather
lowly degenerate ones, such as systems with symmetry broken phases or
states with topological order. 

Yet, in order to be able to fully apply the toolbox of PEPS to the
understanding of these systems, it would be crucial to have a mathematical
characterization of the structure of PEPS, providing a framework similar
to the one which proved so useful for one-dimensional systems: Is there a
canonical form for PEPS which allows to easily determine their properties,
and how can it be found?  How do PEPS appear as ground states of local
``parent'' Hamiltonians, and what is the ground state degeneracy?
What is the structure of the ground state subspace, and how do these
states relate to the PEPS under consideration?

\subsection{Content of the paper}

In this paper, we introduce a framework for characterizing both MPS and
PEPS in terms of symmetries of the underlying tensors. This classification
allows us to rederive all results known for one-dimensional systems, while
it can be equally applied to the characterization of two- and
higher-dimensional PEPS states, answering the aforementioned
questions:
Using the characterization based on symmetries, we prove how PEPS appear as
ground states of local Hamiltonians with finitely degenerate ground
states, and how these states can be obtained as variants of the original
state. Subsequently, we demonstrate the power of our framework by using it
to explain in a simple and coherent way the topological properties of
these ground states: We prove that these states are locally
indistinguishable and can be transformed into each other by acting on a
restricted region, we explain the origin of the topological entropy, and
we discuss how to renormalize these states based on their symmetries. We
also discuss the excitations of these Hamiltonians, and demonstrate how to
understand their anyonic statistics as a consequence of their symmetries.
Thus, the characterization of PEPS in terms of symmetries provides a
powerful framework which allows to explain a large range of their
properties in a coherent and natural way.

The material in this paper is structured as follows: In
Sec.~\ref{sec:MPS-and-PEPS}, we introduce MPS and PEPS and discuss their
basic properties. In Sec.~\ref{sec:MPS-injective}, we review the proof for
the ``injective'' case, in which the MPS appears as a unique ground state
of the parent Hamiltonian.  In Sec.~\ref{sec:mps-noninj}, we show how
general MPS can be classified in terms of symmetries, and use this
description to generalize the results of the preceding section; we also
discuss how our results relate to the ones obtained using the 
original canonical form of~\cite{perez-garcia:mps-reps}.
 In Sec.~\ref{sec:PEPS-noninj}, we generalize the
symmetry-based classification of MPS to the case of PEPS, where we derive
the parent Hamiltonian and characterize the structure of its ground state
space. In Sec.~\ref{sec:iso-PEPS}, we consider a more restricted class of
PEPS with symmetries, for which we derive a variety of results concerning the
structure of the ground state space and the parent Hamiltonian, such as
topological entropy, local indistinguishability, renormalization
transformations, or the fact that the parent Hamiltonian commutes. We
conclude the section by characterizing the excitations of the system and
explaining how their anyonic statistics emerges from the symmetries of the
PEPS. Finally, in Sec.~\ref{sec:examples}, we discuss examples which
illustrates the applicability of our classification, before we conclude in
Sec.~\ref{sec:conclusion}.

\section{MPS and PEPS\label{sec:MPS-and-PEPS}}

We start by defining Matrix Product States, which describe 
the state of a one-dimensional chain of $d$-level systems of length $L$.
\begin{definition}
A state $\ket\psi\in(\mathbb C^d)^{\otimes L}$
is called a \emph{Matrix Product State (MPS)} if it
can be written as
\begin{equation}
\label{eq:mps-def}
\ket{{\mathcal{M}}(A)}=\sum_{i_1,\dots,i_L}\tr[A^{i_1}\cdots A^{i_L}]\ket{i_1,\dots,i_L}\ .
\end{equation}
Here, $A^{i}\in \lin{D}$ (the space of $D\times D$ matrices over $\mathbb
C$), and $D$ is called the
\emph{bond dimension}.
\end{definition}
Note that we restrict to the translational invariant setting with only one
tensor $A$, whereas a general MPS can be defined with a different tensor
at each site. However, for each MPS describing a translationally invariant
state a translationally invariant description (\ref{eq:mps-def}) can be
constructed; note, however, that the bond dimension can increase with $L$.

In a more graphical representation, $A$ can be written as a three-index
tensor,
\[
A\equiv\sum_{i\alpha\beta}(A^i)_{\alpha\beta} \ket{i}_p\ket{\alpha}_v\bra{\beta}_v=
\includegraphics[height=2.5em]{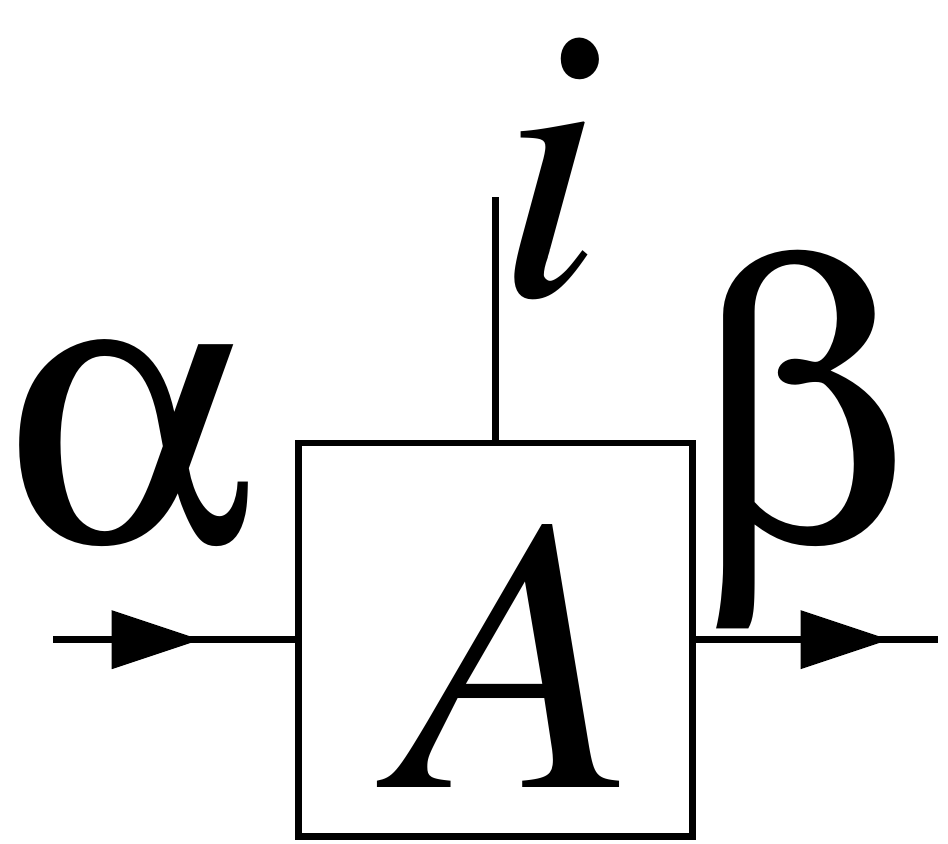}\ .
\]
Here, $p$ refers to the ``physical'' index [characterizing the actual
state in \eqref{eq:mps-def}], and $v$ to the ``virtual'' index, which is
only used to construct the MPS~\eqref{eq:mps-def} and does not appear in
the final state. For clarity, we assign arrows to the virtual indices,
pointing from bras $\bra{\cdot}$ to kets $\ket{\cdot}$.  Connecting the
legs of two tensors (called a ``bond'') denotes contraction,
\[
\sum_\beta (A^i)_{\alpha\beta}(A^j)_{\beta\gamma}
\ket{i,j}_p\ket{\alpha}_v\bra{\gamma}_v
=
    \includegraphics[height=2.5em]{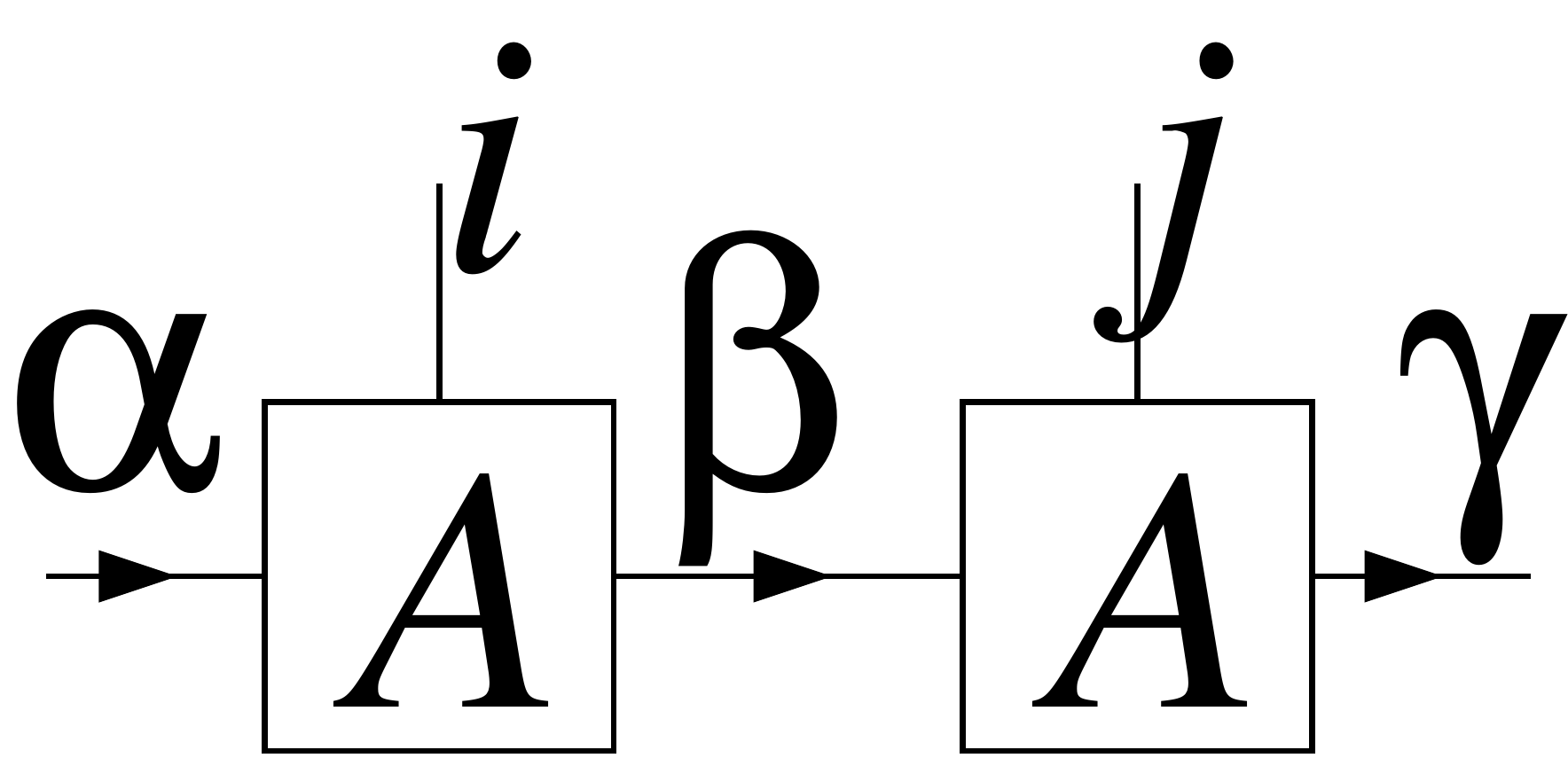}\ .
\]
Thus, we can write the MPS (\ref{eq:mps-def}) as
\[
\ket{\psi_A}=
\raisebox{-.8em}{
\includegraphics[height=3em]{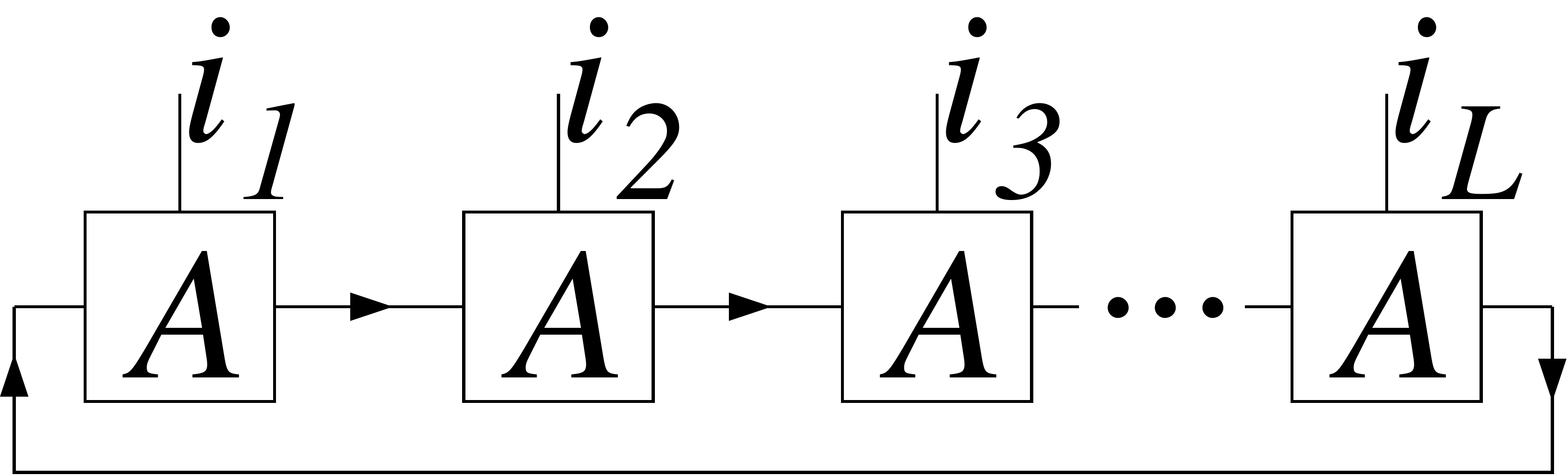}
}\ .
\]

Having this picture in mind, we can immediately define a two-dimensional
generalization:

\begin{definition}
A state $\ket{\psi}\in (\mathbb C^d)^{\otimes (L\times L)}$ is called a
\emph{Projected Entangled Pairs State (PEPS)} of \emph{bond dimension} $D$
if it can be written as
\begin{equation}
    \label{eq:peps-def}
\ket{\psi_A}=\raisebox{-3em}{\includegraphics[height=7em]{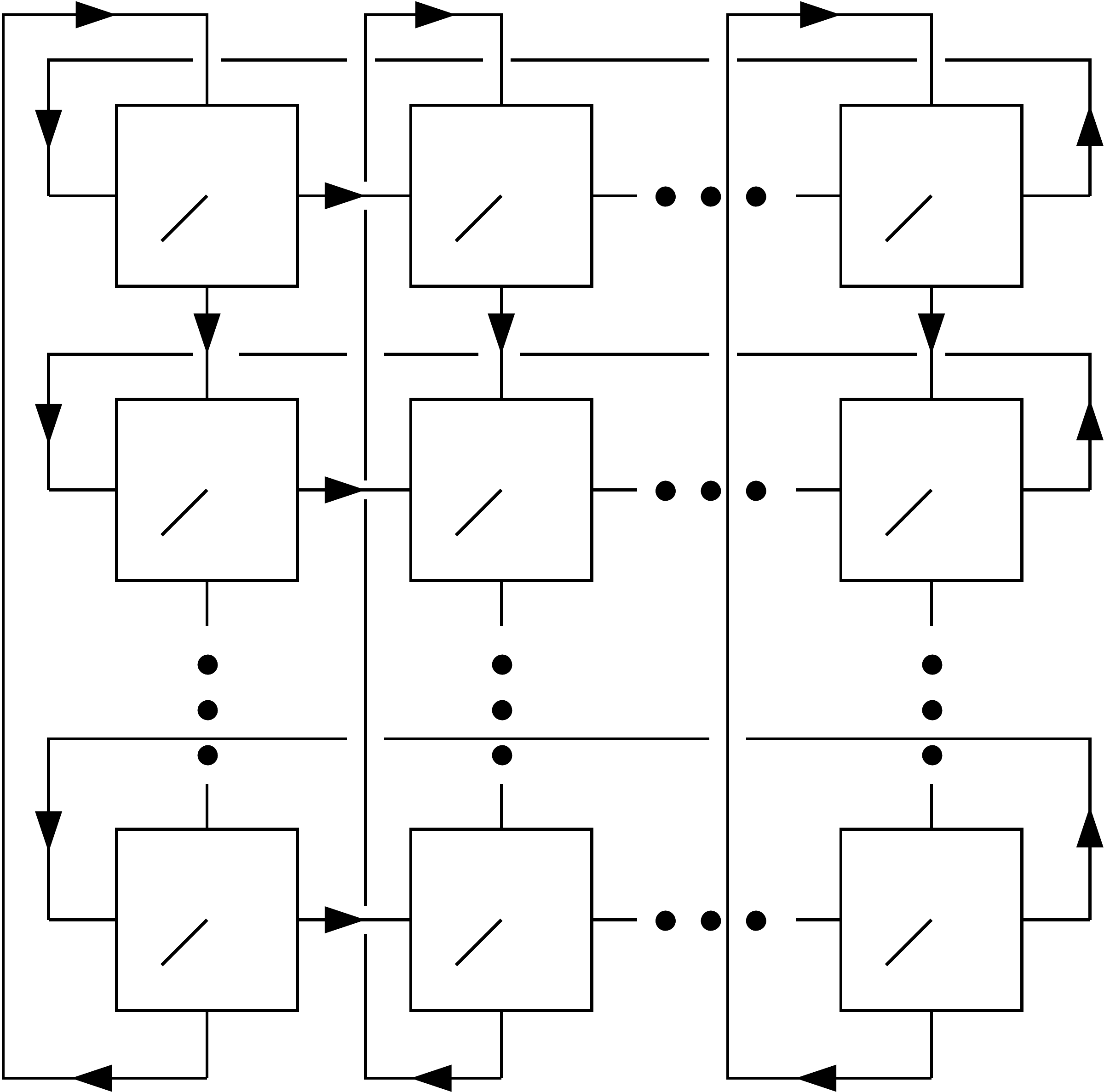}}
\end{equation}
with $A^{i}_{\alpha\beta\gamma\delta}$, $i=1,\dots,d$,
$\alpha,\beta,\gamma,\delta=1,\dots,D$ a five-index tensor, where the
index denoted $i$ corresponds to the physical system and ``goes out of the
paper''.
\end{definition}
Note that in the definition \eqref{eq:peps-def} we have implicitly defined
the direction of the arrows, and thus the assignments of kets and bras in
$A$:
\[
A \equiv \sum_{iltrb} A^i_{ltrb}\ket{i}_p\ket{l,t}_v\bra{r,b}_v\equiv
\raisebox{-1.5em}{\includegraphics[height=3.5em]{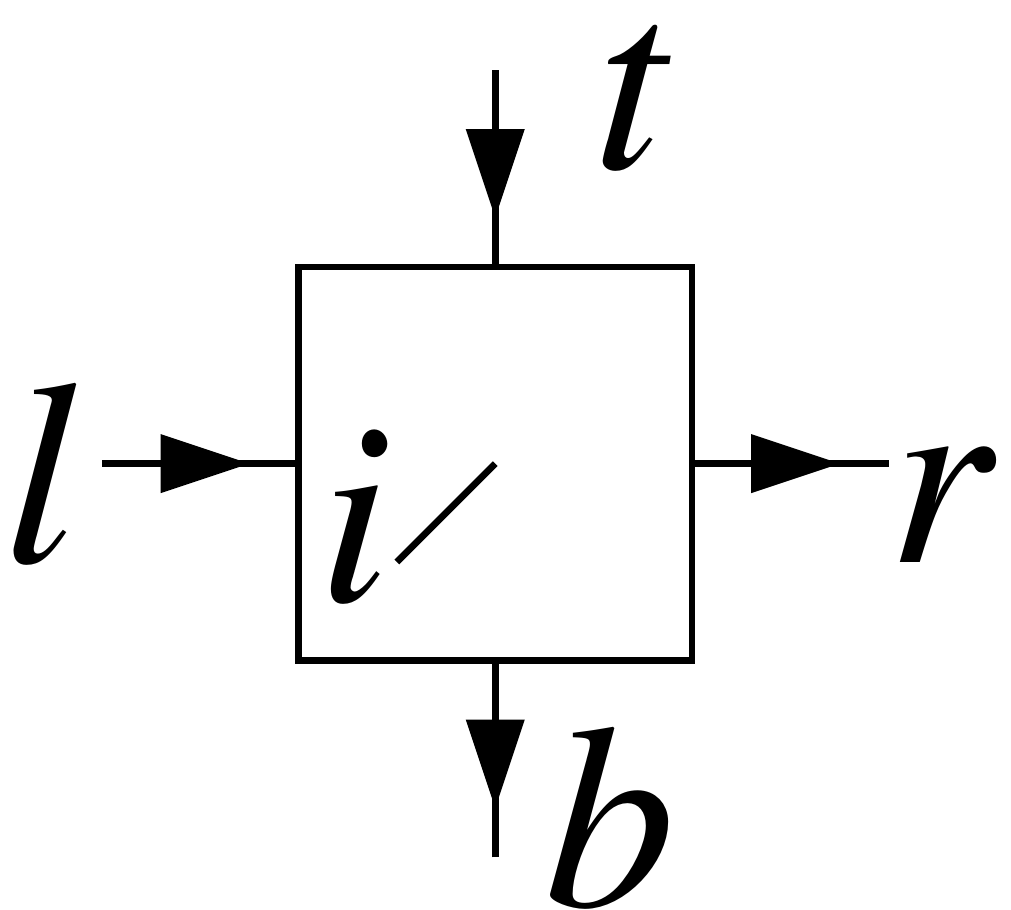}}\ .
\]

Let us now introduce two simplifications to the MPS tensors, which will
allow us to characterize all MPS based on their non-local properties.

First, we define the ``projector'' -- the map which creates the
physical system from the virtual layer.
\begin{definition}
For an MPS given by a tensor $A$,
\begin{equation}
\label{eq:P-of-B}
\mathcal P(A):=\sum_{i\alpha\beta}
A^i_{\alpha\beta}\ket{i}_p\bra{\alpha,\beta}_v\ ,
\end{equation}
is the map which maps the virtual system to the physical one. The
definition extends directly to PEPS.
\end{definition}
The importance of the map $\mathcal P(A)$ lies in the fact that it tells
us how virtual and the physical space of the MPS are connected. On the one
hand, it tells us which physical states can be created by acting on the
virtual degrees of freedom, and on the other hand, it tells us which virtual
configurations -- which in turn enforce physical configurations on the
surrounding sites -- can be realized by acting on the physical system.
Thus, studying the properties of $\mc P(A)$ will allow us to infer
properties of the MPS.

Let us now see which simplifications we can make in the analysis of $\mc
P(A)$, given that we are only interested in the non-local properties of
the MPS (PEPS).  Define $\mc R(A):=\mathrm{rg}\,\mc P(A)$, and
$\mc D(A):=(\mathrm{ker}\,\mc P(A))^\bot=\mathrm{span}(\{A^i\}_i)$.
Then, $\mc P(A)$ can be inverted on $\mc D(A)$ and $\mc R(A)$, respectively:
\[
\exists \mc P(A)^{-1}:\
    \mc P(A)^{-1}\mc P(A)=\openone\big\vert_{\mc D(A)}\ \wedge\
    \mc P(A)\mc P(A)^{-1}=\openone\big\vert_{\mc R(A)}\ .
\]
Note that $\mc R(A)$ characterizes the local support of the MPS:
Projecting sites $2,\dots,L$ on any basis state $i_2,\dots,i_L$ leaves us
with $\sum_{i_1}\tr[A^{i_1}X]\ket{i_1}$ (where $X=A^{i_2}\cdots A^{i_L}$),
and thus, the single-site reduced operator is supported on $\mc R(A)$.
Thus,  we can restrict the MPS to the subspace $\mc R(A)^{\otimes L}$,
i.e., each local system can be mapped to a $\mathrm{dim}{\mc
R}(A)$--dimensional system by a local isometry, i.e., without changing any
non-local properties.  This implies that we can w.l.o.g.\ restrict our
analysis to MPS with the following property.

\begin{observation}
    \label{obs:surjective}
Any MPS/PEPS $\ket{{\mathcal{M}}(A)}$ can be characterized (up to local isometries)
by a tensor $A$ for which $\mc{P}(A)$ has a right inverse (denoted by
$A^{-1}$ in the diagram) such that:
\begin{equation}
    \label{eq:inj:rightinv-always}
\raisebox{-1.5em}{\includegraphics[height=4em]{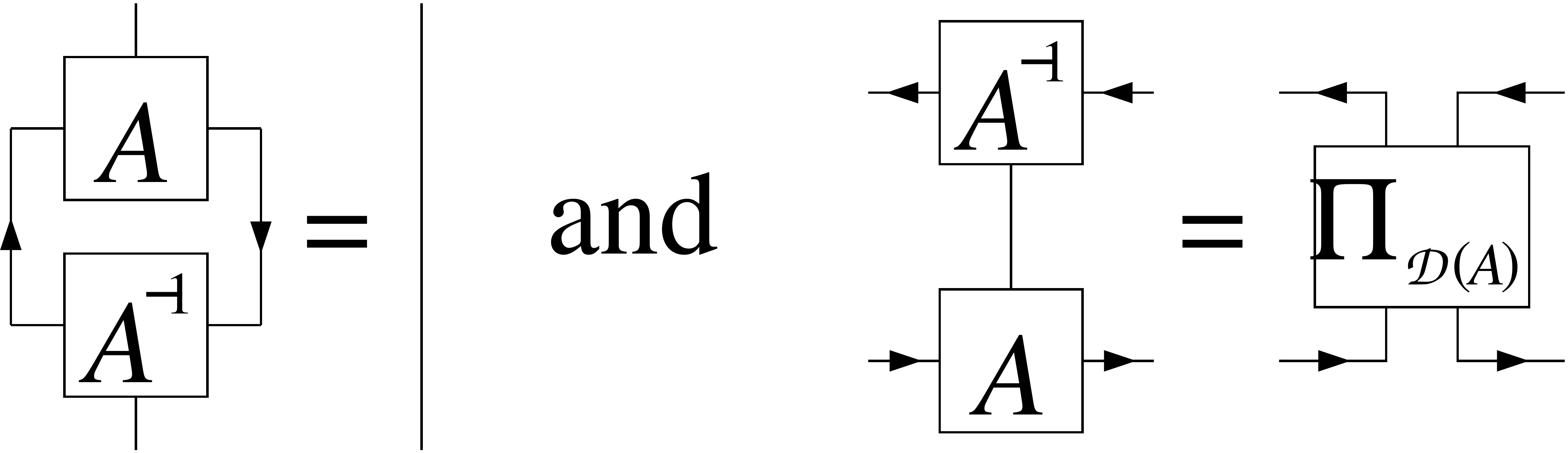}}\ ,
\end{equation}
where $\Pi_{\mc D(A)}\equiv\openone\big\vert_{\mc D(A)}$ is the 
orthogonal projector on
$D(A)$.  W.l.o.g., we will assume a description with this property from
now on.
\end{observation}

The purpose of this paper is to characterize MPS and PEPS by looking at
the structure of the subspace $\mc D(A)$, and especially at its
symmetries. As it turns out, it is sufficient to consider two cases: First,
$\mc D(A)=\mathrm{span}(\{A^i\}_i)= {\bm L}(\mathbb C^D)$, the space of
all linear operators on $\mathbb C^D$, and second, the case where $\mc
D(A)=\mathrm{span}(\{A^i\}_i)$ is an arbitrary \cstar-algebra, i.e., a
linear space of matrices closed under multiplication and hermitian
conjugation.

In order to see why we only need to consider these two cases, take an MPS
with tensor $A$, and group its sites into super-blocks of $k$ sites each.
This results in a new MPS with tensor
\begin{equation}
    \label{eq:block-A-to-B}
\includegraphics[height=2.6em]{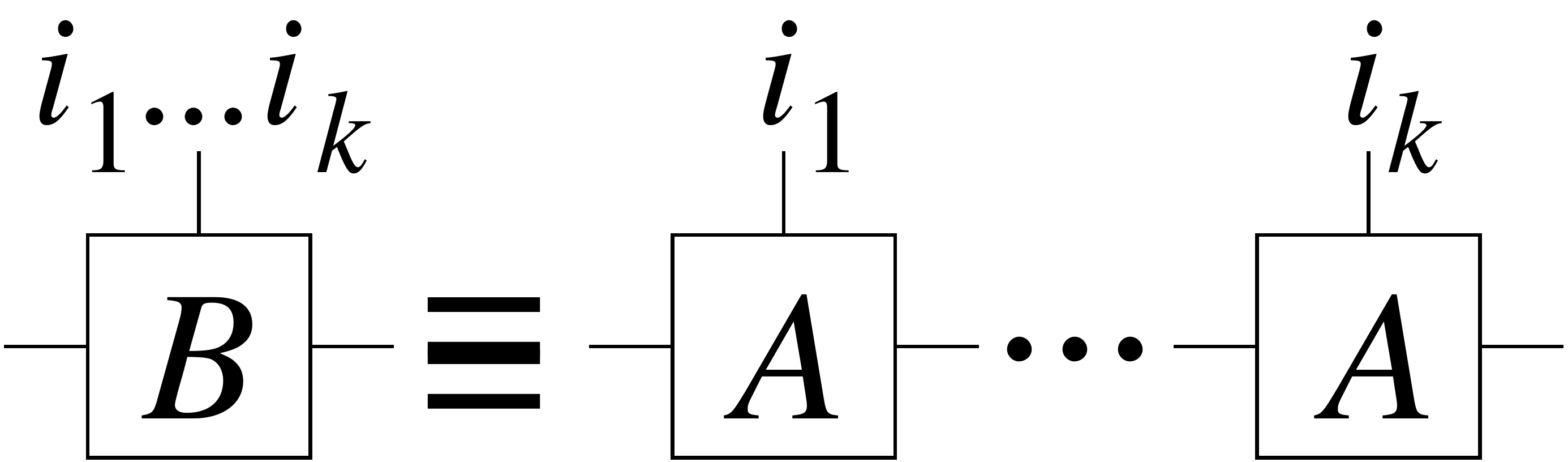}\ .
\end{equation}
The map $\mc P(B)$ for this new tensor goes from a $D^2$ to a
$d^k$--dimensional space:
 Thus, for $k>2\log_d D$, we have that $D^2<d^k$, which means
that the map $\mc P(B)$ will typically be injective, and thus
$D(A)={\bm L}(\mathbb C^D)$.\footnote{See~\cite{perez-garcia:mps-reps} for
a discussion of how to understand ``typical'' in this context.} These MPS
are known as \emph{injective}, and they appear as unique ground states of
their associated parent Hamiltonians.

The second case -- $D(A)$ being a \cstar-algebra -- arises e.g.\ if the
$A^i$ are block diagonal matrices, but within each block span the whole
space of linear operators. As shown in~\cite{perez-garcia:mps-reps}, this
does in fact cover the case of a general MPS, as any MPS can be brought
into this form.  While these states are no longer unique ground states of
local Hamiltonians, their parent Hamiltonians have a finite ground state
degeneracy, and the ground states are described by the individual blocks
of the $A_i$.

In the following, we will first review the situation of injective MPS, and
show how to prove that they are unique ground states of local
Hamiltonians.  We will then turn towards the case where $\mc D(A)$ is a
general \cstar-algebra, which we translate into a condition on the
symmetry of $A$ under unitaries. While in one dimension this reproduces
the results previously derived using the block structure of the matrices
$A^i$~\cite{fannes:FCS,perez-garcia:mps-reps}, it will enable us to
generalize these results to the two-dimensional scenario, where we will
find states exhibiting topological order and anyonic excitations, all of
which can be understand purely in terms of symmetries.

\section{MPS: the injective case\label{sec:MPS-injective}}

\subsection{Definition and basic properties}

We start by analyzing the \emph{injective} case in which
$D(A)=\mathrm{span}(\{A^i\}_i)={\bm L}(\mathbb C^D)$.  According to
Observation~\ref{obs:surjective}, we can choose $A$ such that $\mathcal
P(A)$ is invertible. This leads us to the following formal definition of
injective.

\begin{definition}[Injectivity]
    \label{def:injective-mps}
    A tensor $A$ is called \emph{injective} if $\mathcal P(A)$
    has a left inverse
    \begin{equation}
        \label{eq:inj-linv}
        \raisebox{-1.2em}{\includegraphics[height=3em]{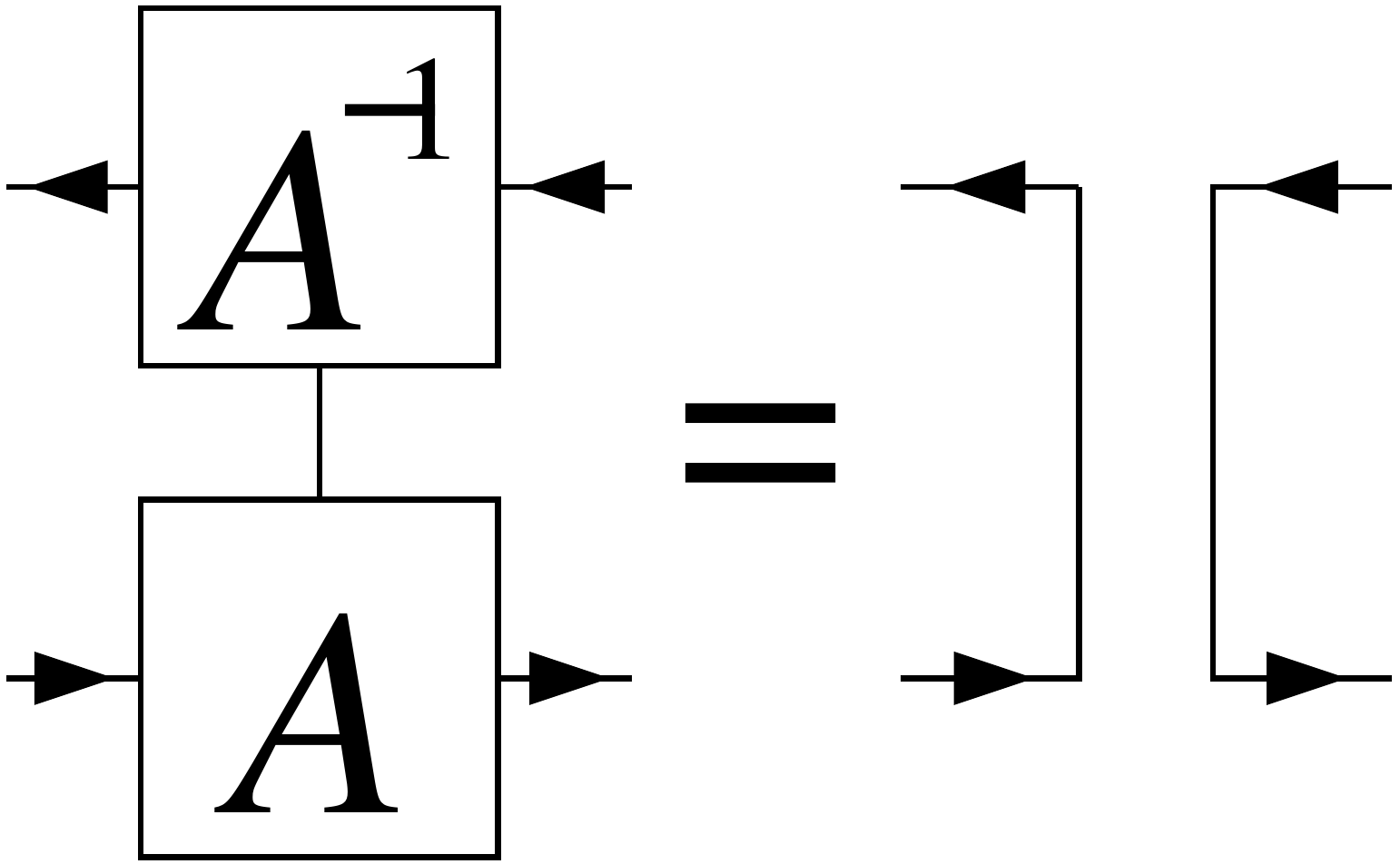}}\
            .
    \end{equation}
    (The corresponding MPS $\ket{{\mathcal{M}}(A)}$ will also be termed injective.)
    Intuitively, injectivity means that we can achieve any action on the
    virtual indices by acting on the physical spins.
\end{definition}

\begin{lemma}[Stability under concatenation]
    \label{lemma:inj-stable}
    The injectivity property \eqref{eq:inj-linv} is stable under
    concatenation of tensors: If $A$ and $B$ are injective, then the
    tensor obtained by concatenating $A$ and $B$ is also injective, since
    \[
    \includegraphics[height=3em]{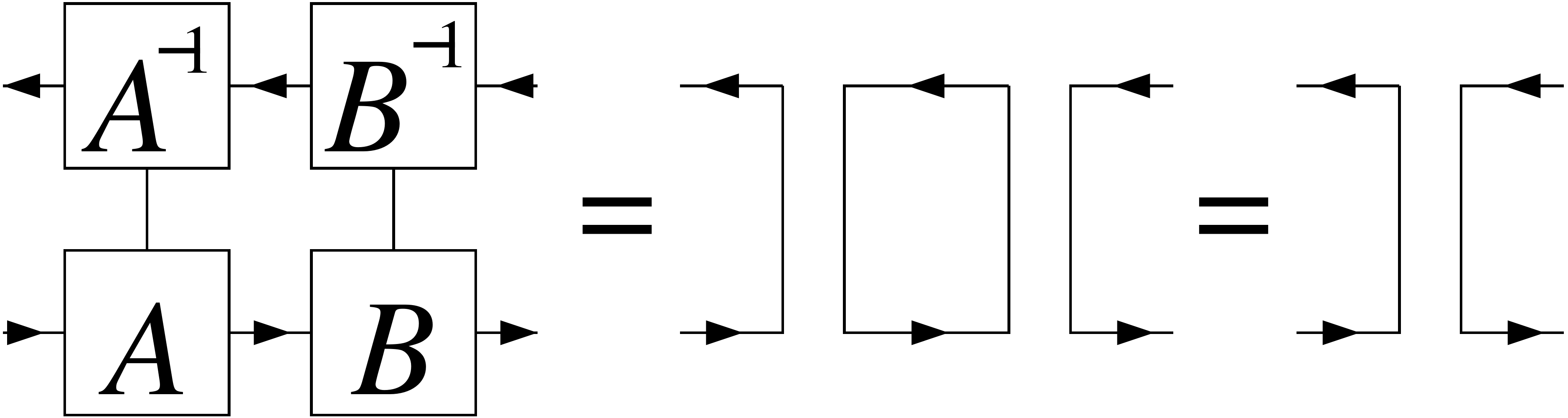}\quad\raisebox{1em}{.}
    \]
\end{lemma}
Note that we generally omit normalization constants in the diagrams (the
contraction of the loop is $\tr\,\openone=D$).

\subsection{Parent Hamiltonians}

Let us now see how injective MPS give rise to \emph{parent Hamiltonians},
to which they are unique ground states. To this end, note that the
two-particle reduced operator of the MPS
$\ket{{\mathcal{M}}(A)}$
is given by
\[
\rho^{[2]}(A)=\sum_{i_3,\dots,i_L}\psi^{[2]}_{i_3,\dots,i_L}(A)
\]
with $\psi^{[2]}_{i_3,\dots,i_L}(A)$ the projector onto the
two-particle state obtained
by projecting sites $3,\dots,L$ to the basis state $i_3,\dots,i_L$,
\[
\ket{\psi^{[2]}_{i_3,\dots,i_L}(A)}=
\sum_{i_1,i_2}\tr[A^{i_1}A^{i_2}X^{i_3,\dots,i_L}]\ket{i_1,i_2}\
\]
with $X^{i_3,\dots,i_L}=A^{i_3}\cdots A^{i_L}\in D(A)$. Thus, $\rho^{[2]}(A)$ is
supported on the subspace
\begin{equation}
\mathcal S_2=\left\{ \sum_i\tr[A^iA^jX]\ket{i,j}
    \middle\vert X\in \lin{D}\right\} =
    \left\{
    \raisebox{-1.7em}{\,\includegraphics[height=4em]{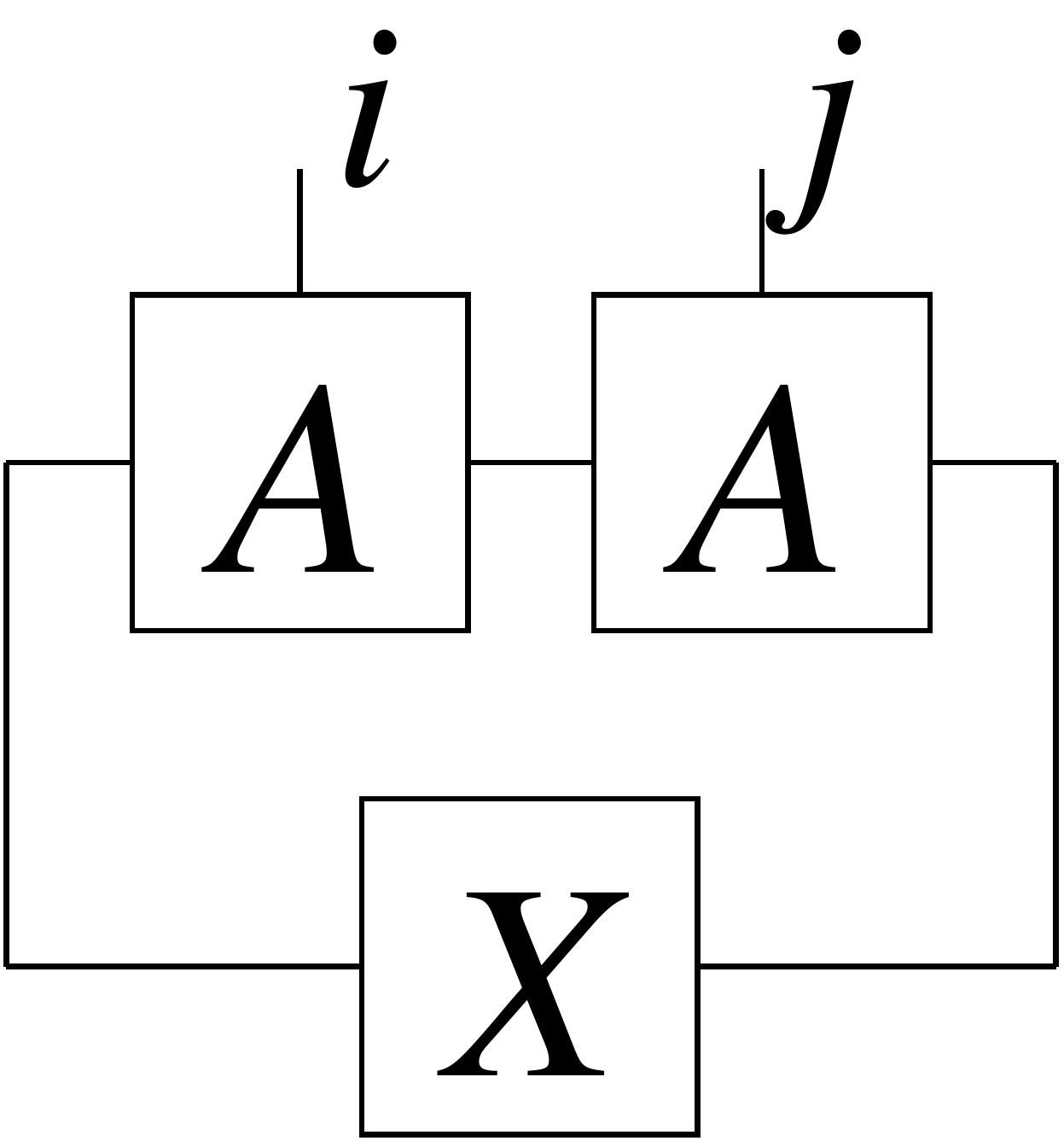}}\
    \,\middle\vert
    X\in \lin{D}
    \right\}\ .
\label{eq:trBBX}
\end{equation}
Moreover, the injectivity of $A$ implies that $X^{i_3,\dots,i_L}$ spans
the space of all $D\times D$ matrices, and thus, $\rho^{[2]}$ has actually
full rank on $\mathcal{S}_2$.  Analogously to \eqref{eq:trBBX}, one can define
a sequence of subspaces
\[
\mathcal S_k=\Big\{\sum_{i_1,\dots,i_k}
    \tr[A^{i_1}\cdots A^{i_k}X]\ket{i_1,\dots,i_k}
    \Big\vert X\in\lin D\Big\}
\]
which by the same arguments exactly support the $k$-body reduced operator
$\rho^{[k]}(A)$.

The idea for obtaining a parent Hamiltonian is now as follows: Define a
two-body Hamiltonian
which has $\mc S_2$ as its ground state subspace, and let
the parent Hamiltonian be the sum of these local terms. The proof consists of
two parts: First,
we show that for a chain of length $k$ with open boundaries, the
ground state subspace is $\mc S_k$ (i.e., optimal), and second, when
closing the boundaries, the only state remaining is $\ket{{\mathcal{M}}(A)}$.
This is formalized in the following two theorems.

\begin{theorem}[Intersection property]
\label{thm:inj:intersection}
Let $A$ and $B$ be injective tensors. Then,
\begin{equation}
    \left\{\raisebox{-1.7em}{\,\includegraphics[height=4em]{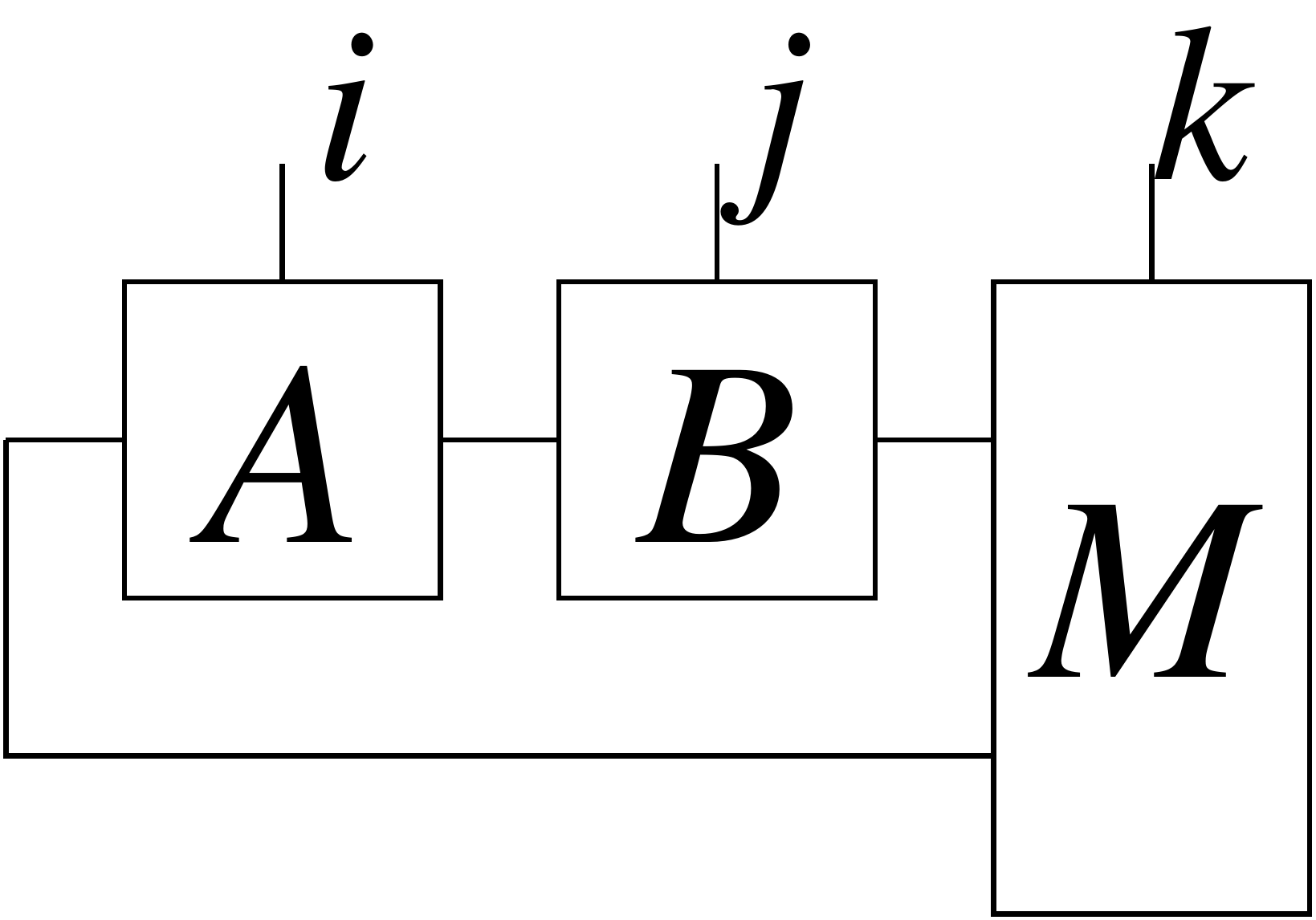}}\ \,\middle\vert \ M \right\}
    \cap
    \left\{\raisebox{-1.7em}{\,\includegraphics[height=4em]{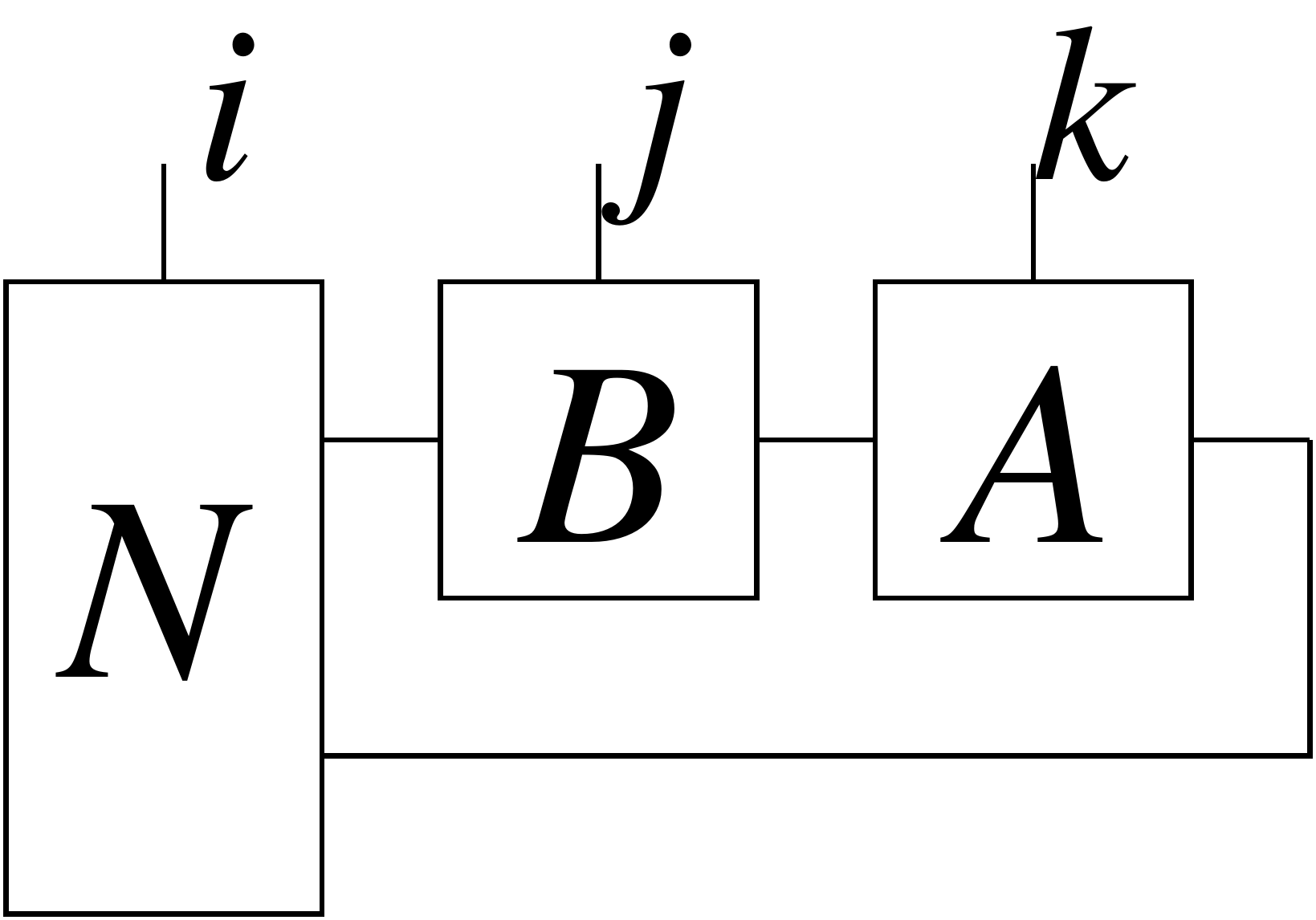}}\ \,\middle\vert \ N \right\}=
    \left\{\raisebox{-1.7em}{\,\includegraphics[height=4em]{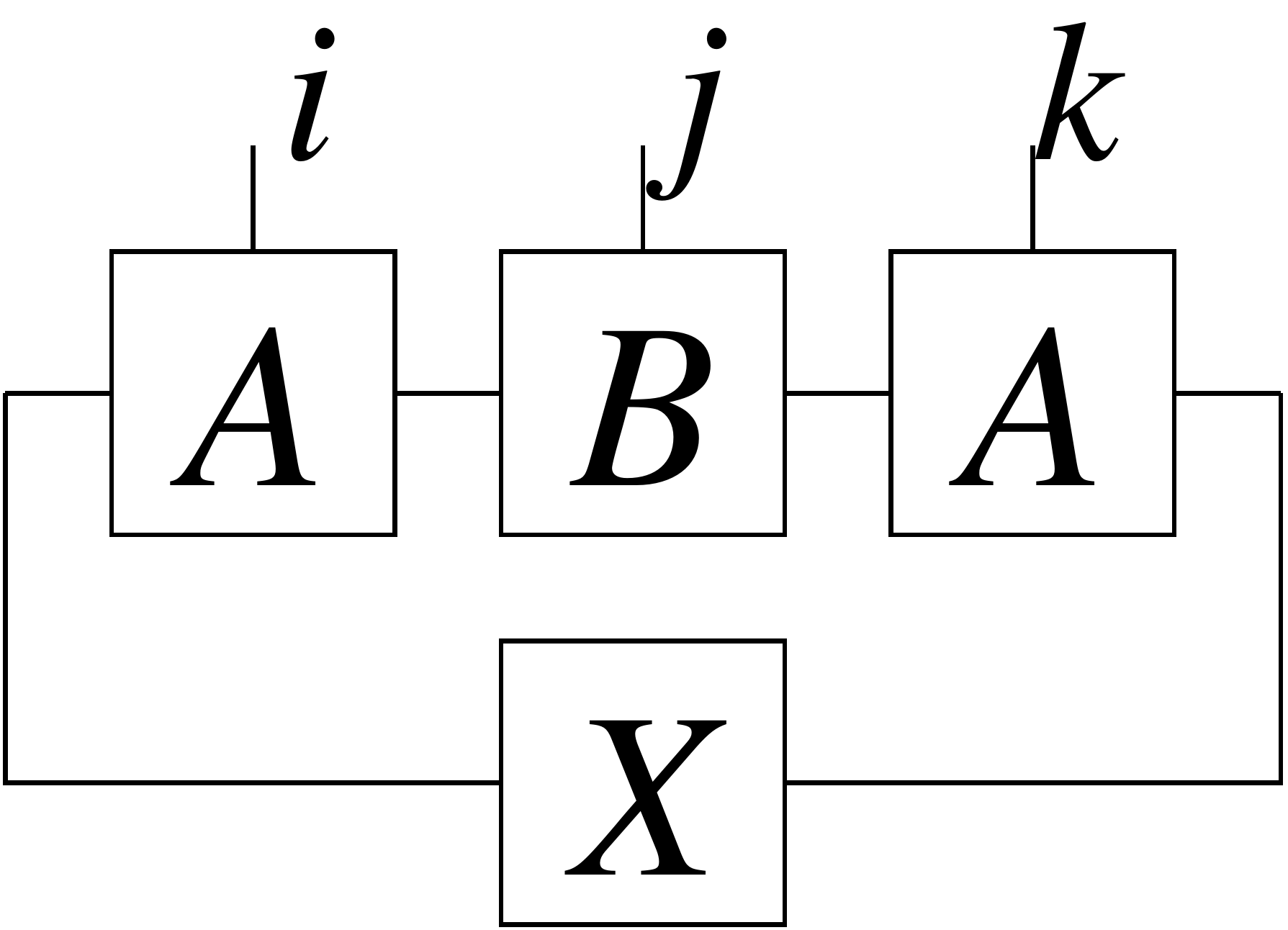}}\
    \,\middle\vert \ X \right\}\ ,
\label{eq:inj:intersection}
\end{equation}
with $M,N,X\in\lin D$.
\end{theorem}

\begin{proof}
It is clear that the right side is contained in
the intersection, since for any $X$ we can choose
\begin{equation}
\raisebox{-1.3em}{\includegraphics[height=3.2em]{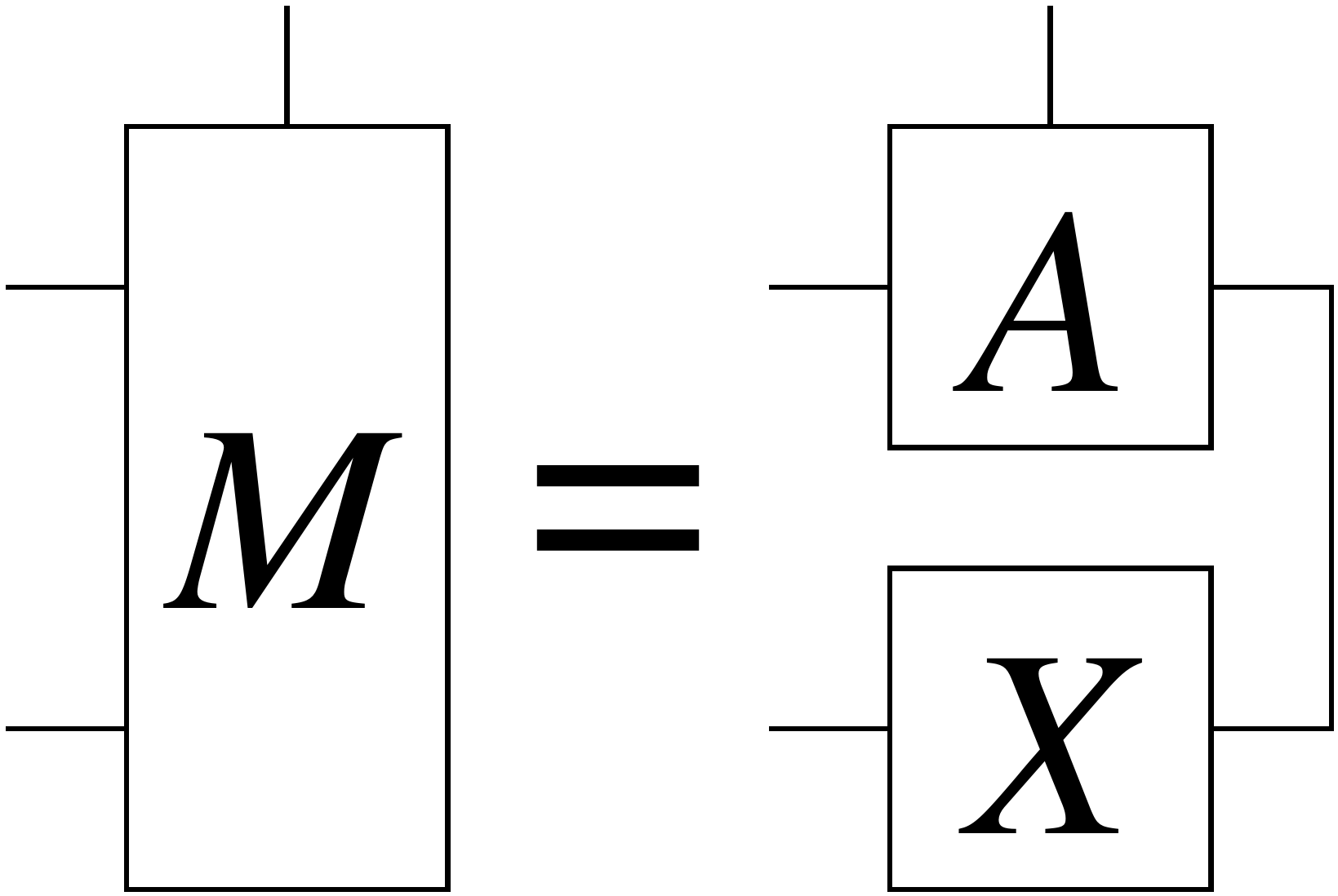}}\mbox{\quad and \quad}
\raisebox{-1.3em}{\includegraphics[height=3.2em]{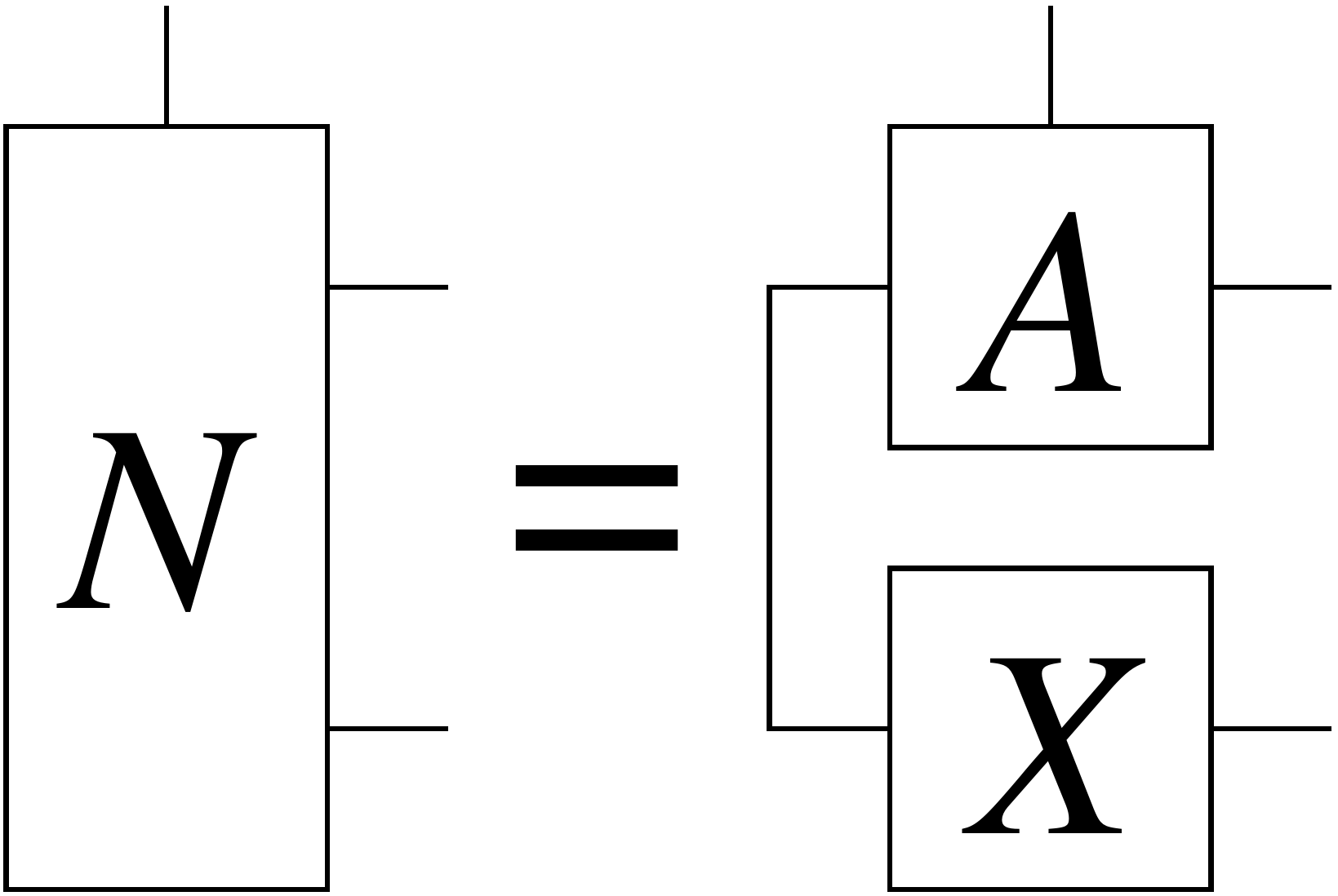}}\ .
\label{eq:inj:form-of-NandM}
\end{equation}
Conversely, for any $\ket{\psi}$ in the l.h.s.\ of
(\ref{eq:inj:intersection}), there exist $M$, $N$ such that
\[
    \ket{\psi}=
    \raisebox{-1.7em}{\,\includegraphics[height=4em]{figs1/trBBM}} =
    \raisebox{-1.7em}{\,\includegraphics[height=4em]{figs1/trNBB}}\ \ .
\]
Applying the left inverse of $\mc P(A)$ and $\mc P(B)$ to the $i$ and $j$
index, we find that
\[
\includegraphics[height=4.5em]{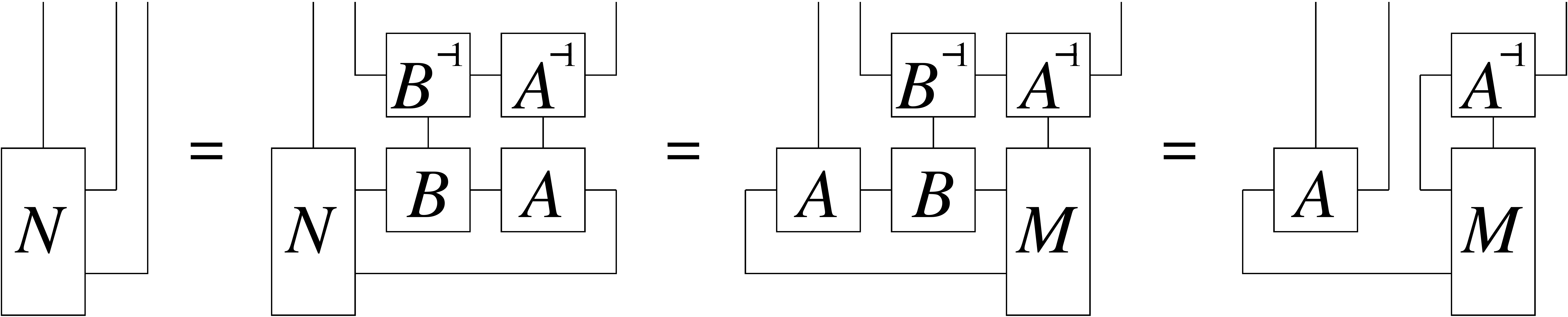}\ ,
\]
i.e., $N$ is of the form (\ref{eq:inj:form-of-NandM}), and thus
$\ket{\psi}$ is contained in the r.h.s. of (\ref{eq:inj:intersection}).
\end{proof}

\begin{theorem}[Closure property] For injective $A$ and $B$,
\label{thm:inj:closure}
\begin{equation}
    \left\{\raisebox{-1.7em}{\,\includegraphics[height=4em]{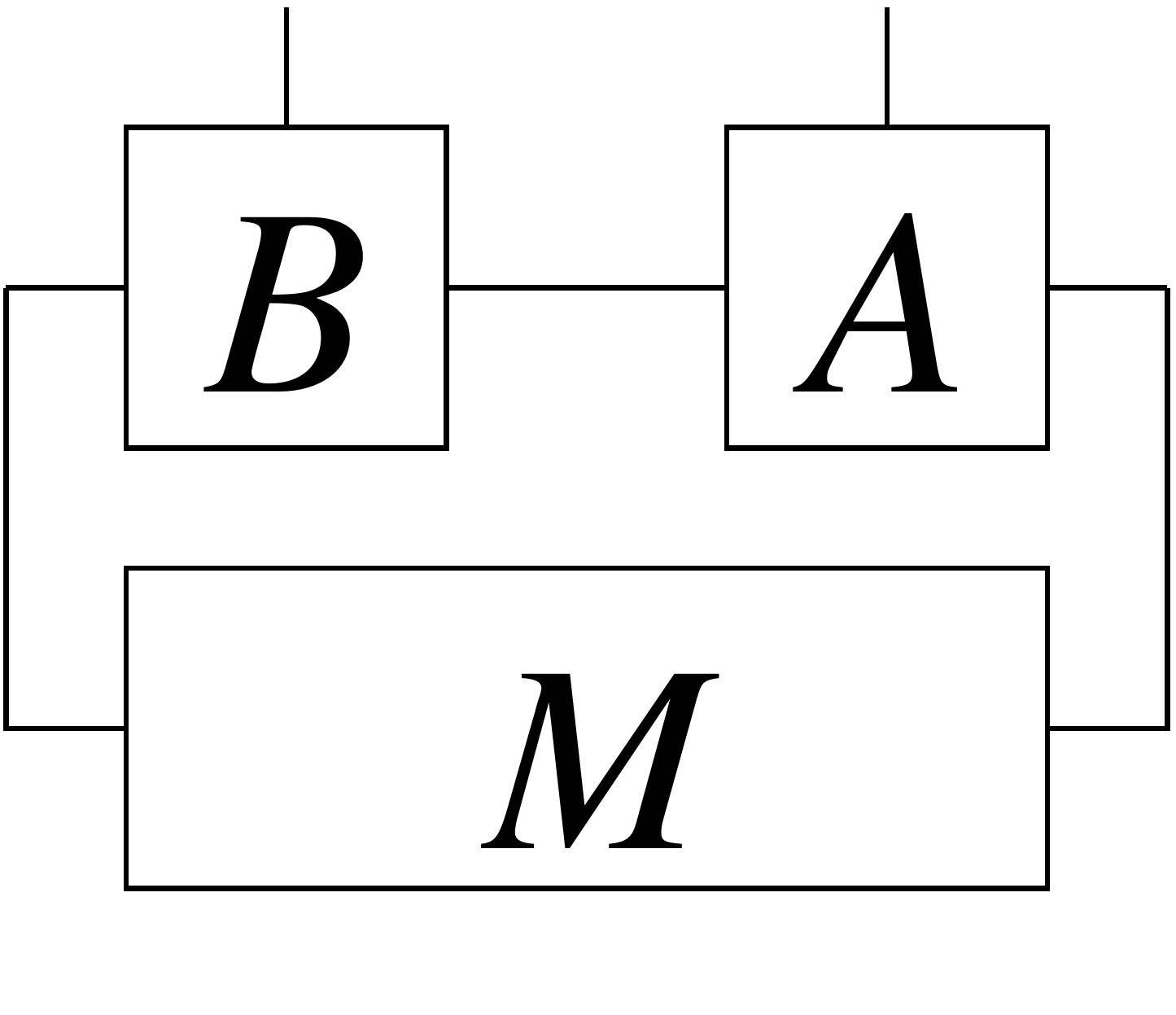}}\ \,\middle\vert \ M \right\}
    \cap
    \left\{\raisebox{-1.7em}{\,\includegraphics[height=4em]{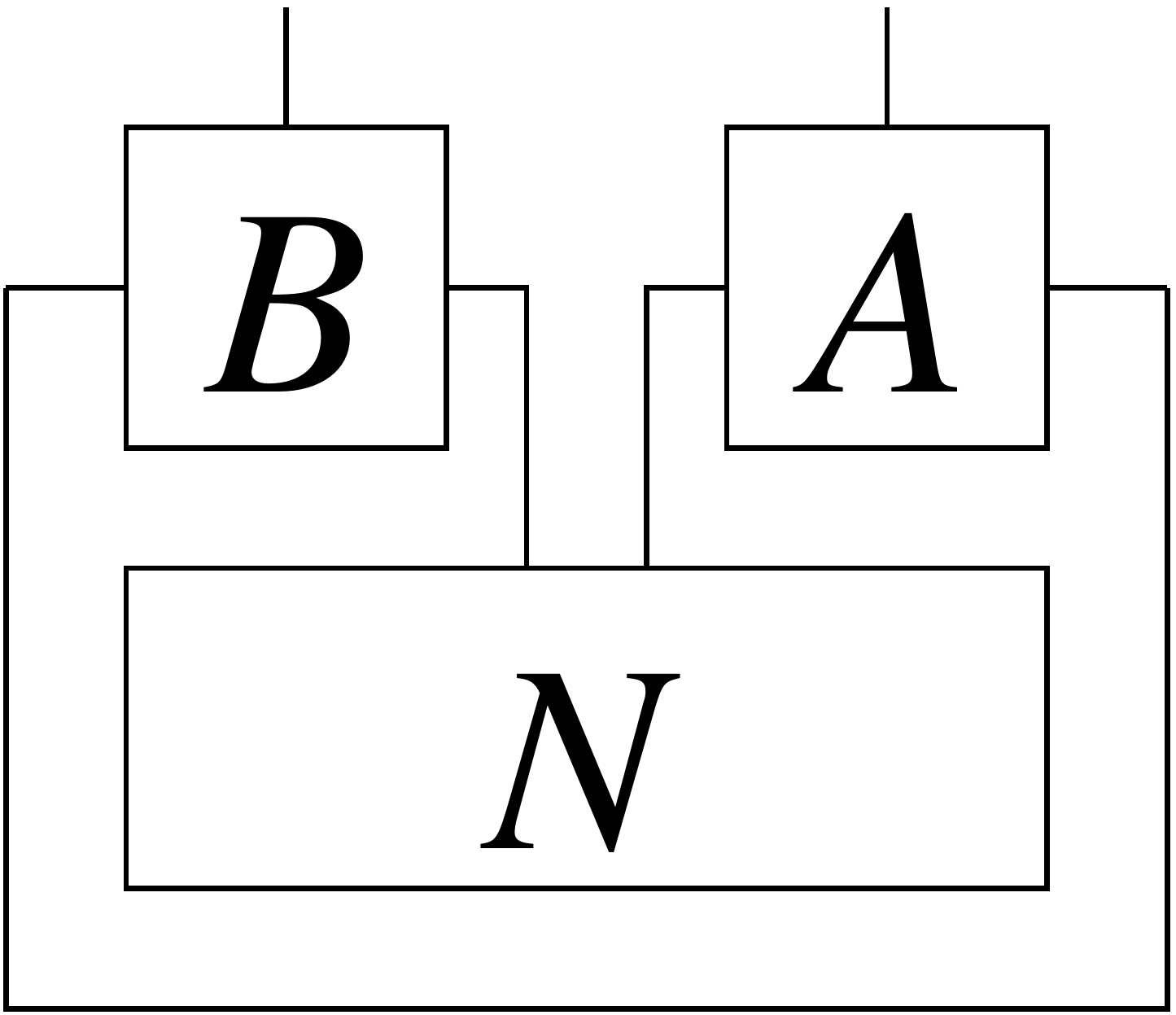}}\ \,\middle\vert \ N \right\}=
    \raisebox{-.9em}{\,\includegraphics[height=2.5em]{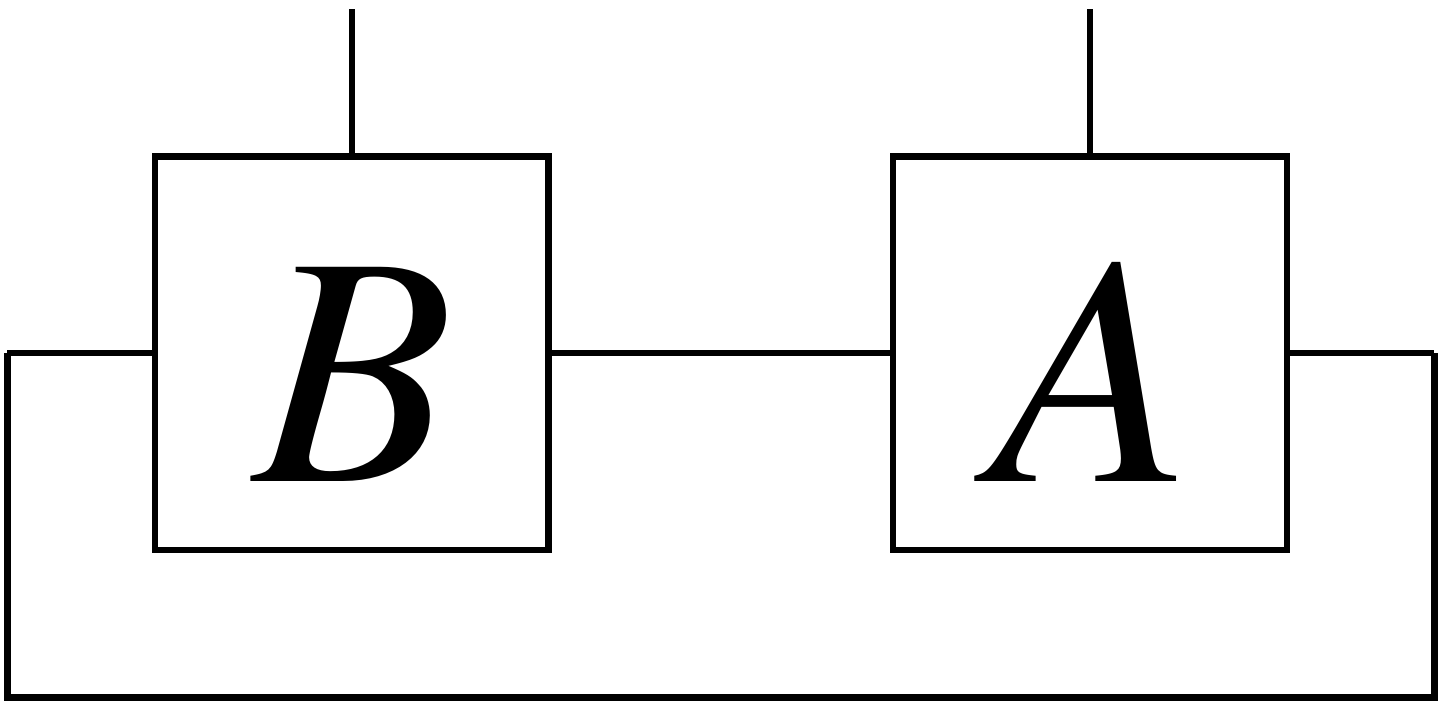}}
    \ .
\label{eq:inj:closure}
\end{equation}
\end{theorem}

\begin{proof}
As in the proof of Theorem~\ref{thm:inj:intersection}, the r.h.s.\ is
trivially contained in the intersection by choosing $N=M=\openone$.
Conversely, by taking an arbitrary element in the intersection and applying
the left inverses of $\mathcal P(A)$ and $\mathcal P(B)$, we find that
\[
\includegraphics[height=5.5em]{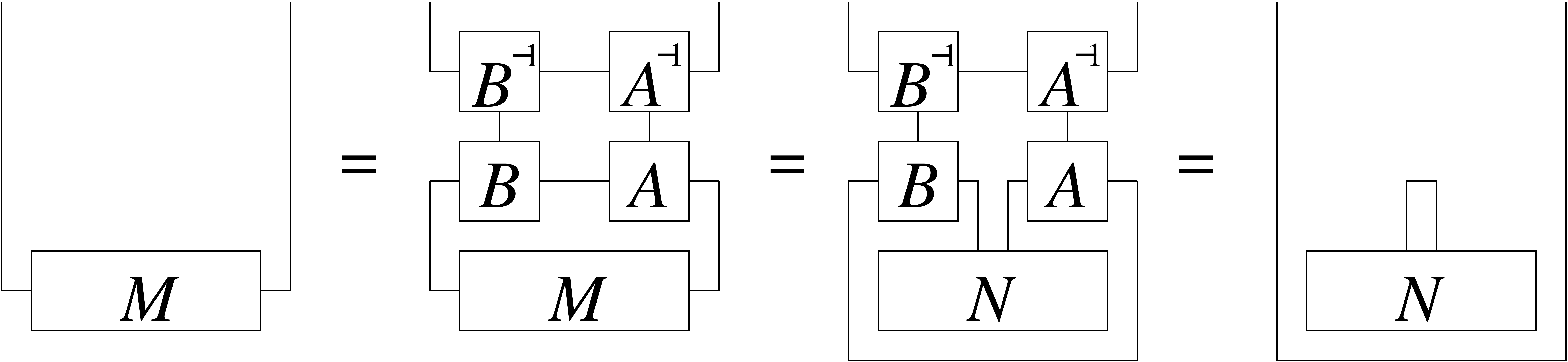}
\]
which proves (\ref{eq:inj:closure}).
\end{proof}

Let us now put Theorems~\ref{thm:inj:intersection}
and~\ref{thm:inj:closure} together to show that the MPS $\ket{{\mathcal{M}}(A)}$
arises as the unique frustration free ground state of a local Hamiltonian.

\begin{theorem}[Parent Hamiltonians]
    \label{thm:inj:parent-ham}
    Let $A$ be injective, and $\mc S_2$ as in Eq.~\eqref{eq:trBBX}.
    Define
\[
h_i=\openone_d^{\otimes(i-1)}\otimes (1-\Pi_{\mathcal S_2})
    \otimes \openone_d^{L-i-2}
\]
as the orthogonal projector on the subspace orthogonal to $\mathcal S_2$
on sites $i$ and $i+1$ (modulo $L$). Then,
\[
H_\mathrm{par}=\sum_{i=1}^L h_i
\]
has $\ket{{\mathcal{M}}(A)}$ as its unique and frustration free ground state.
\end{theorem}
\begin{proof}
For $k<L$, define $\bm j[k-1]:=(i_2,\dots,i_{k-1})$ and
$B^{\bm j[k-1]}:=
A^{i_2}\cdots A^{i_{k-1}}$. Rewriting
\[
\mc S_{k-1}\otimes\mathbb C^d=
\Big\{\sum_{i_1,\bm j[k-1],i_k}\hspace{-1em}
    \tr[A^{i_1} B^{\bm j[k-1]}M^{i_{k}}]
    \ket{i_1,\bm j[k-1],i_{k}} \Big\vert
    M\in \mathbb{C}^d\otimes\lin D \Big\}
\]
(and similarly for $\mathbb C^d\otimes \mc S_{k-1}$ and $S_{k}$),
Theorem~\ref{thm:inj:intersection} implies that
$
\mc S_{k-1}\otimes\mathbb C^d\cap\mathbb C^d\otimes\mc S_{k-1}=\mc S_{k}
$,
and thus by induction
\begin{equation}
\label{eq:inj:S-L-as-intersect}
\mathcal S_L = \mc S_2\otimes (\mathbb C^d)^{\otimes (k-2)} \cap
    \mathbb C^d \otimes \mc S_2\otimes (\mathbb C^d)^{\otimes (k-3)} \cap  \
    \cdots \ \cap (\mathbb C^d)^{\otimes (k-2)}\otimes \mathcal S_2\ ,
\end{equation}
i.e., the subspace supporting the length $L$ chain is given by the
intersection of the two-body supports $\mc S_2$.

Now let $H_\mathrm{left}=h_1+\tfrac12h_2+\dots+\tfrac12h_{L-1}$, and
$H_\mathrm{right}=\tfrac12h_2+\dots+\tfrac12h_{L-1}+h_L$. As the $h_i$ are
projectors, the null space of $H_\mathrm{left}$ is given by the intersection
\eqref{eq:inj:S-L-as-intersect}, i.e., by
\[
\mc S_\mathrm{left}=\mc S_L =
    \Big\{\sum_{i_1,\bm j[k]}\tr[A^{i_1} B^{\bm j[k]}M]
    \ket{i_1,\bm j[k]} \Big\vert
    M\in \lin D \Big\}\ .
\]
Correspondingly, the null space of $H_\mathrm{right}$ is
\[
\mc S_\mathrm{right} =
    \Big\{\sum_{i_1,\bm j[k]}\tr[A^{i_1} N B^{\bm j[k]}]
    \ket{i_1,\bm j[k]} \Big\vert
    N\in \lin D \Big\}\ ,
\]
and thus by the closure property, Theorem~\ref{thm:inj:closure},
$\ket{{\mathcal{M}}(A)}=\mc S_\mathrm{left}\cap \mc S_\mathrm{right}$
is the unique zero-energy (i.e., frustration free) ground
state of $H_\mathrm{par}=H_\mathrm{left}+H_\mathrm{right}$.
\end{proof}

Note that all these results hold equally for injective
PEPS~\cite{perez-garcia:parent-ham-2d}. We will
discuss the case of PEPS in detail for the non-injective scenario, which
includes the injective one as a special case.  Note also that the results
of this section do not rely on the translational invariance of the PEPS --
the central Theorems~\ref{thm:inj:intersection} and~\ref{thm:inj:closure}
hold for any pair $A$, $B$ of injective tensors.

\section{MPS: the $G$-injective case \label{sec:mps-noninj} }

\subsection{Definition and basic properties}

In the following, we will consider the case where $A$ is not injective,
but where nevertheless
\begin{equation}
\label{eq:noninj:alg-A}
\mathcal D \equiv \mathcal D(A) =
    \Big\{\sum_{i} \lambda_i A^i \Big| \lambda_i\in\mathbb C \Big\}
\end{equation}
forms a \cstar--algebra.

In the following, we will characterize the structure of $\mathcal D$,
and thus of $A$, in terms of symmetries.
Using Observation~\ref{obs:surjective} -- that $\mathcal P(A)$ has a left
inverse on $\mathcal D$ -- together with an appropriate characterization
of $\mathcal D$ will allow us to base proofs on this left inverse, similar
to the injective case.

\begin{theorem}
\label{thm:noninj:mcA-via-Ug}
For any \cstar-algebra $\mathcal D\subset \lin D$
there exists a finite group $G$ and a unitary
representation $g\mapsto U_g$ such that $\mathcal D$ is the commutant of
$U_g$, i.e.,
\begin{equation}
\label{eq:noninj:mcA-via-Ug}
\mathcal D = \big\{ X\in\lin D
    \big\vert [X,U_g]=0\ \forall\,g\in G\big\} \ .
\end{equation}
\end{theorem}
\begin{proof}
Any \cstar-algebra $\mc{D}$ can be decomposed as
\begin{equation}
    \label{eq:stdform:cstar-dec}
\mathcal D \cong \bigoplus_{i=1}^I \openone_{d_i}\otimes \lin {m_i},
\end{equation}
where $\cong$ denotes unitary equivalence,
$A\cong B:\Leftrightarrow \exists\,W,\,WW^\dagger=\openone\colon
A = WBW^\dagger$.  Now choose a finite group $G$ and a representation $U_g$
with irreducible representations $D^i$ of dimensions $d_i$ and
multiplicity $m_i$, respectively,\footnote{
    This can be achieved, e.g., by choosing $I$ finite groups $G_i$ which
    have a $d_i$-dimensional irreducible representation $\tilde D_i$, and
    letting $G=G_1\times\dots\times G_I$, and $D^i(g_1,\dots,g_I)=\tilde
    D_i(g_i)$.  In particular, the choice of $G$ and $U_g$ is not unique.
    }
\[
U_g \cong \bigoplus_{i=1}^I D^i(g)\otimes \openone_{m_i}\ .
\]
Now Schur's lemma implies Eq.~(\ref{eq:noninj:mcA-via-Ug}).
\end{proof}

Note that conversely, any unitary representation $U_g$
has a \cstar-algebra as its commutant \eqref{eq:noninj:mcA-via-Ug}, since
$U_g^\dagger = U_{g^{-1}}$. Thus, the characterization in terms of
unitaries is equivalent to the characterization in terms of $\mathcal D$,
the span of the $A^i$.  This motivates the following definition.

\begin{definition}
\label{def:Ug-inj}
Let $g\mapsto U_g$ be a unitary representation of a finite group $G$.  We
say that an MPS tensor $A$ is \emph{$G$--injective} if\\
i) $\forall i,g:\ U_g A^i U_g^\dagger = A^i$, and \\
ii) the map $\mc P(A)$ [Eq.~\ref{eq:P-of-B}]
has a left inverse on the subspace
\begin{equation}
\label{eq:mpssym:symspace}
\mathcal S=\{X|[X,U_g]=0\}
\end{equation}
of
$U_g$--invariant matrices,
\begin{equation}
\label{eq:noninj-linv-eq}
\mcP^{-1}(A)\mcP(A) = \left.\openone\right|_\mcS\ .
\end{equation}
\end{definition}

In graphical notation, we will denote unitary representations as circles
labelled $U_g$, or simply $g$ when unambigous. Note that the arrow now
point towards the ket on which to apply $U_g$.  Condition i) of
Definition~\ref{def:Ug-inj} then reads
\[
\includegraphics[height=1.6em]{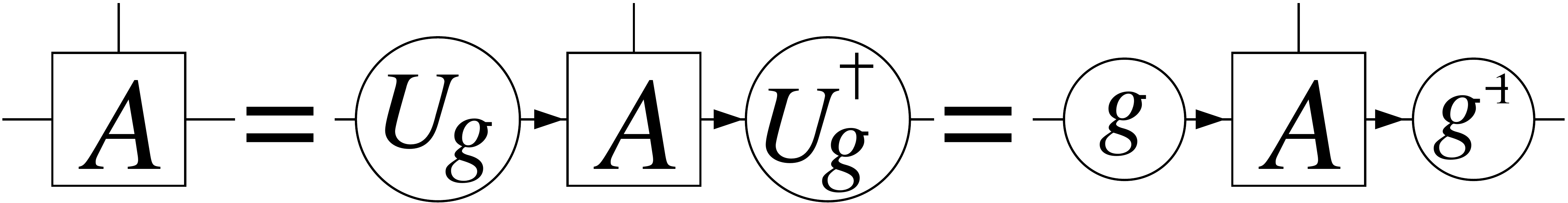}\ .
\]
Note that the group $G$ is more important than its representation $U_g$
(though we will require certain properties at some point), and in fact, a
different representation can be attached to each bond.

\begin{lemma}
The orthogonal projector $\bm\sigma$ onto the $U_g$--invariant subspace
\eqref{eq:mpssym:symspace} is given by
\[
\bm\sigma(X) = \frac{1}{|G|}\sum_g U_g X U_g^\dagger\ .
\]
Here, $|G|$ is the cardinality of the group $G$.
Thus, condition \eqref{eq:noninj-linv-eq} corresponds to
\[
\raisebox{-1em}{
\includegraphics[height=3em]{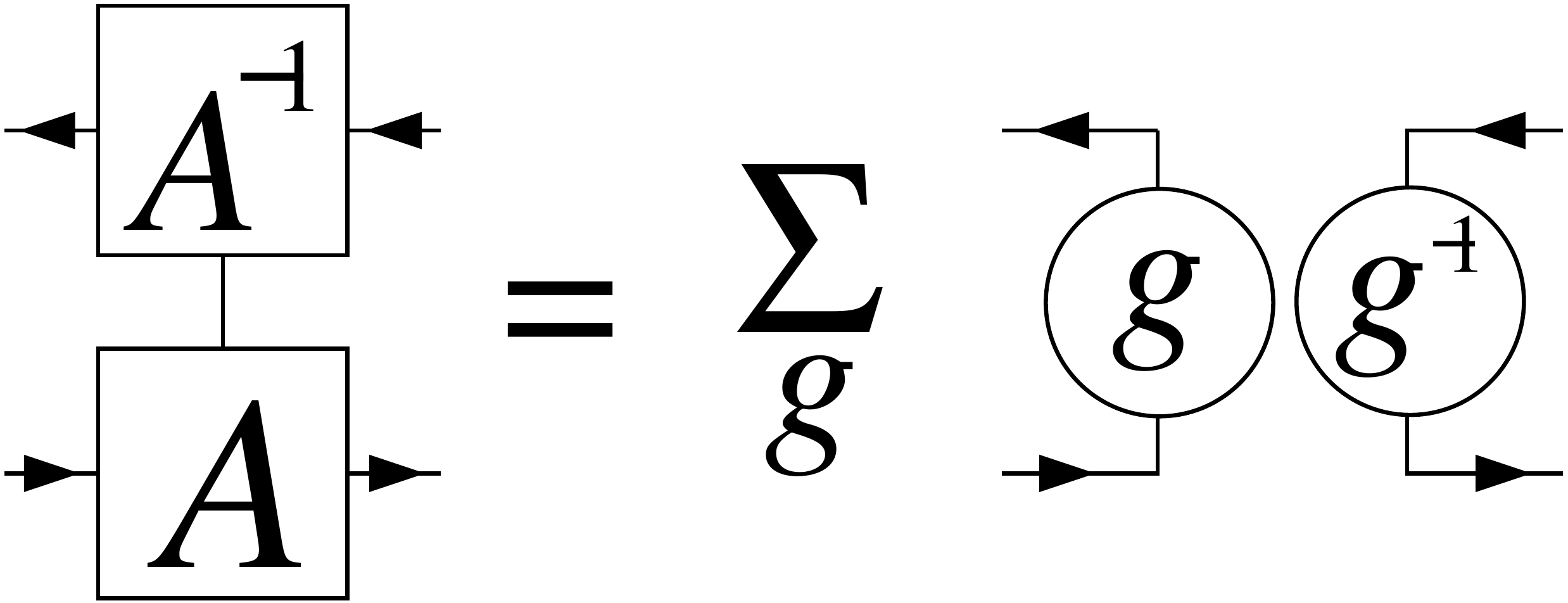}
}\ .
\]
(We generally omit normalization in diagrams.)
\end{lemma}

\begin{proof}
Since
\begin{align*}
\bm\sigma^2(X)&=\frac{1}{|G|^2}\sum_{gh} U_hU_gXU_g^\dagger U_h^\dagger \\
&= \frac{1}{|G|^2} \sum_{hg} U_{hg}XU_{hg}^\dagger =
\frac{1}{|G|} \sum_k U_k X U_k^\dagger = \bm\sigma(X)\ ,
\end{align*}
$\bm\sigma$ is a projection. As it leaves $\mcS$ invariant,
\[
U_h\bm\sigma(X) U_h^\dagger =
\frac{1}{|G|}\sum_g U_{hg}X U_{hg}^\dagger = \bm\sigma(X)
\]
(i.e., the image is contained in $\mcS$),
and it is hermitian, it is the orthogonal projector on $\mcS$.
\end{proof}

We will now show the analogue of Lemma~\ref{lemma:inj-stable}:
$G$--injectivity is stable under concatenation of tensors. To this end,
we will use the following identity.

\begin{lemma}
\label{lemma:1d-ghh-anyrep}
For any unitary representation $U_g$ of a finite group,
\[
\raisebox{-1.05em}{
\includegraphics[height=2.8em]{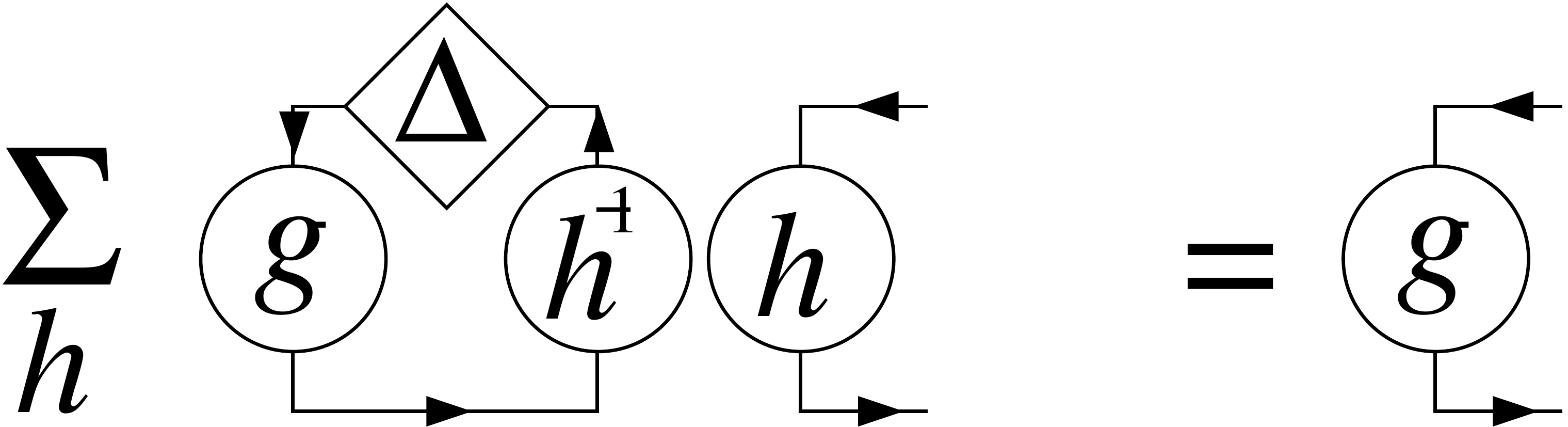}
}
\]
with
\begin{equation}
        \label{eq:noninj:deltadef}
\Delta \cong \frac{1}{|G|}
    \bigoplus_i \frac{d_i}{m_i}\openone_{d_i} \otimes \openone_{m_i}\ ,
\end{equation}
where $(d_i,m_i)$ are dimensions and multiplicites of the
irreducible representations $D^i$ of $U_g$, and $\Delta$ is diagonal in 
the basis
in which $U_g\cong\bigoplus_i D^i(g)\otimes \openone_{m_i}$.
In a formula,
\begin{equation}
        \label{eq:noninj:tr-ghdeltaUh-is-Ug}
\sum_h \tr[U_h^\dagger \Delta U_g]U_h=
\sum_h \tr[U_{gh^{-1}}\Delta]U_h=U_g\ .
\end{equation}
\end{lemma}
\begin{proof} The group orthogonality theorem implies
\begin{align}
        \nonumber
\sum_{g}\tr[D^i(g^{-1})]D^j(g)
&= \sum_{g,m,k,l}\bar D^i_{mm}(g) D^j_{kl}(g)\ket{k}\bra{l}\\
        \label{eq:noninj:sum-CHIi-Di-is-PROJi}
&
=\frac{|G|}{d_i}\sum_{m,k,l}\delta_{i,j}\delta_{m,k}\delta_{m,l}\ket{k}\bra{l}\\
        \nonumber
&    =\frac{|G|}{d_i}\delta_{i,j}\openone_{d_i}\ .
\end{align}
 Thus,
\begin{align*}
\sum_h\tr[U_{gh^{-1}}\Delta]U_h
& = \sum_k \tr[U_{k^{-1}}\Delta]U_k U_g \\
& \cong \sum_k
    \left[\sum_i m_i \tr[D^i(k^{-1}) \tfrac{d_i}{m_i|G|}]\right]
    \bigoplus_j\, \big[D^j(k)D^j(g)\big] \otimes \openone_{m_i} \\
& = \bigoplus_j D^j(g)\otimes\openone_{m_i}\\
& \cong U_g\ .
\end{align*}
\end{proof}

For a restricted class of representations, we can make a stronger
statement.
\begin{definition}[Semi-regular representations] A unitary representation
$g\mapsto U_g$ of a finite group $G$ is called \emph{semi-regular} if
$U_g$ contains \emph{all} irreducible representations of $G$.
\end{definition}

\begin{lemma}[Linear independence of semi-regular representations]
        \label{lemma:noninj:semireg-trace-ug-delta}
Let $g\mapsto U_g$ be a semi-regular representation of a group $G$. Then,
\[
\tr[U_g^\dagger U_h\Delta]=\delta_{g,h}
\]
with $\Delta$ as of \eqref{eq:noninj:deltadef}.
Note that for the regular representation, $\Delta\propto\openone$.
\end{lemma}
\begin{proof}
This follows as
$|G|\tr[U_k\Delta]=|G|\sum_{i}m_i\tr[\tfrac{d_i}{m_i|G|}D^i(k)]$
is the character of the regular representation, which is $|G|\,\delta_{k,1}$.
\end{proof}

We are now ready to prove the analogue of Lemma~\ref{lemma:inj-stable} for
the $G$--injective case.

\begin{lemma}[Stability under concatenation]
Let $A$ and $B$ be $G$--injective tensors. Then $C^{ij}=A^iB^j$ is also
$G$--injective with left inverse
\begin{equation}
\label{eq:noninj:linv}
\raisebox{-0.8em}{
\includegraphics[height=1.8em]{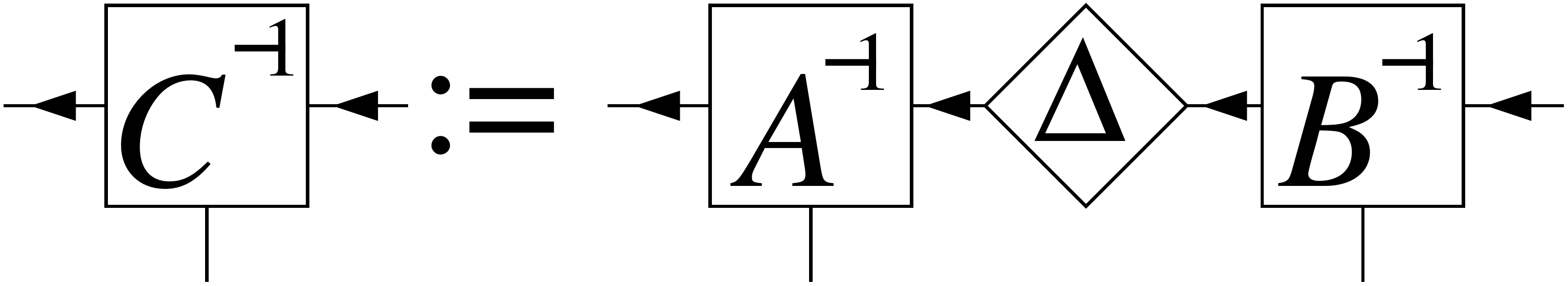}
}\ .
\end{equation}
Here, $\Delta$ is defined as in Lemma~\ref{lemma:1d-ghh-anyrep}.
\end{lemma}
\begin{proof}
$G$--invariance of $C$ follows from
\begin{equation}
\label{eq:mpssym:gsym-grow}
\raisebox{-0.4em}{
\includegraphics[height=1.8em]{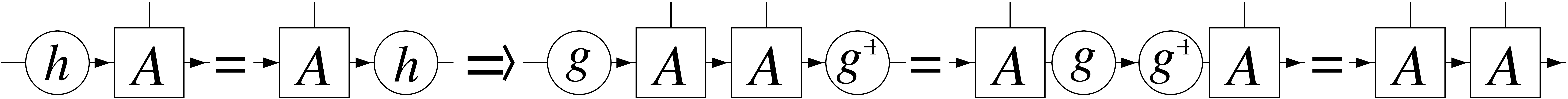}
}\ .
\end{equation}
Moreover, \eqref{eq:noninj:linv} is the left inverse of $\mathcal P(C)$ on
the $U_g$--invariant subspace since
\[
\raisebox{-1.2em}{
\includegraphics[height=3.15em]{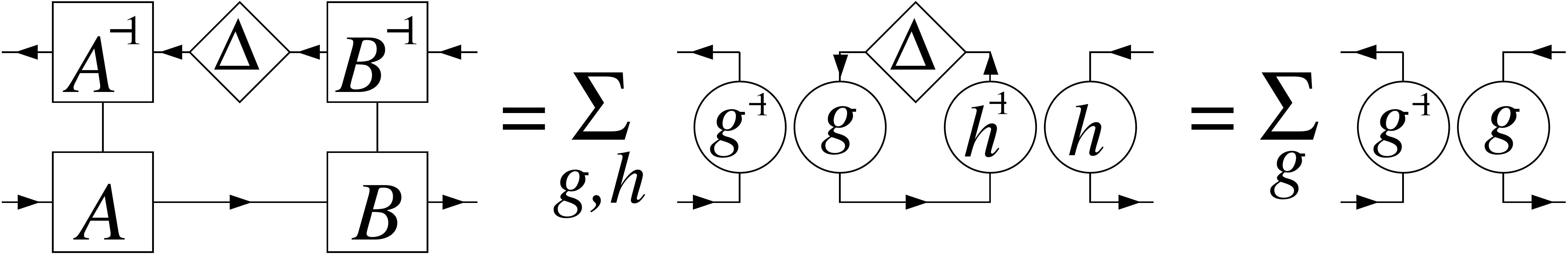}
}
\]
from Lemma~\ref{lemma:1d-ghh-anyrep}.
\end{proof}

\subsection{Parent Hamiltonians}

Let us now proceed to the relation of $G$--injective MPS and parent
Hamiltonians. The construction is exactly analogous to the injective case:
We define
\[
\mathcal S_k=\Big\{\sum_{i_1,\dots,i_k}
    \tr[A^{i_1}\cdots A^{i_k}X]\ket{i_1,\dots,i_k}
    \Big\vert X\in\lin D\Big\}\ ,
\]
and prove the analogues of Theorem~\ref{thm:inj:intersection}
(Intersection Property) and Theorem~\ref{thm:inj:closure} (Closure
Property), but now for $G$--injective MPS. While the intersection
property will be the same as in the injective case, the closure will give
rise to a subspace whose dimension equals the number of
conjugacy classes of $G$, for a properly chosen group $G$.

\begin{theorem}[Intersection property]
\label{thm:noninj:intersection}
Let $A$, $B$ be $G$--injective. Then,
\begin{equation}
\label{eq:1dsym:intersect}
    \left\{\raisebox{-1.2em}{\,\includegraphics[height=3em]{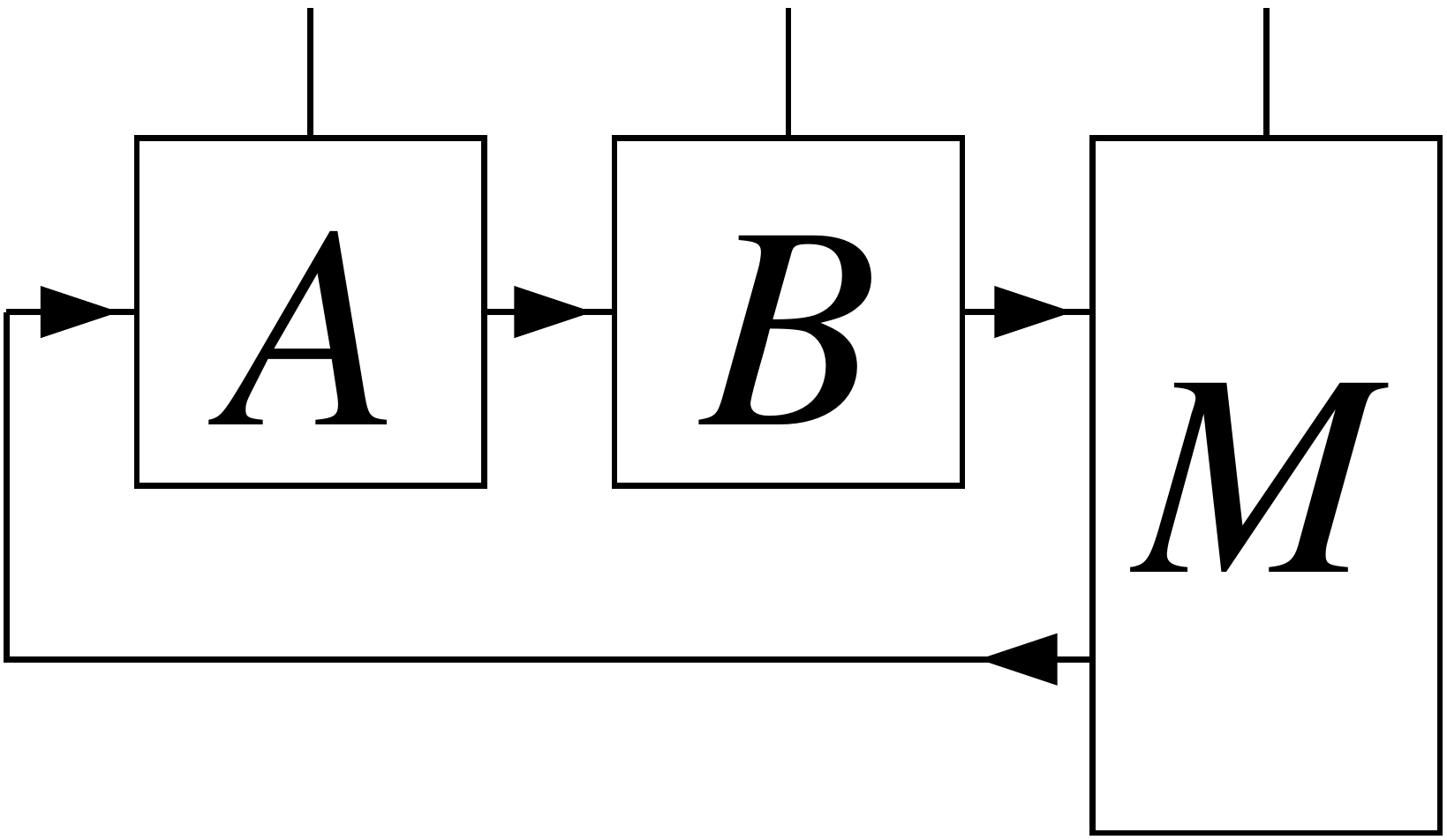}}\ \,\middle\vert \ M \right\}
    \cap
    \left\{\raisebox{-1.2em}{\,\includegraphics[height=3em]{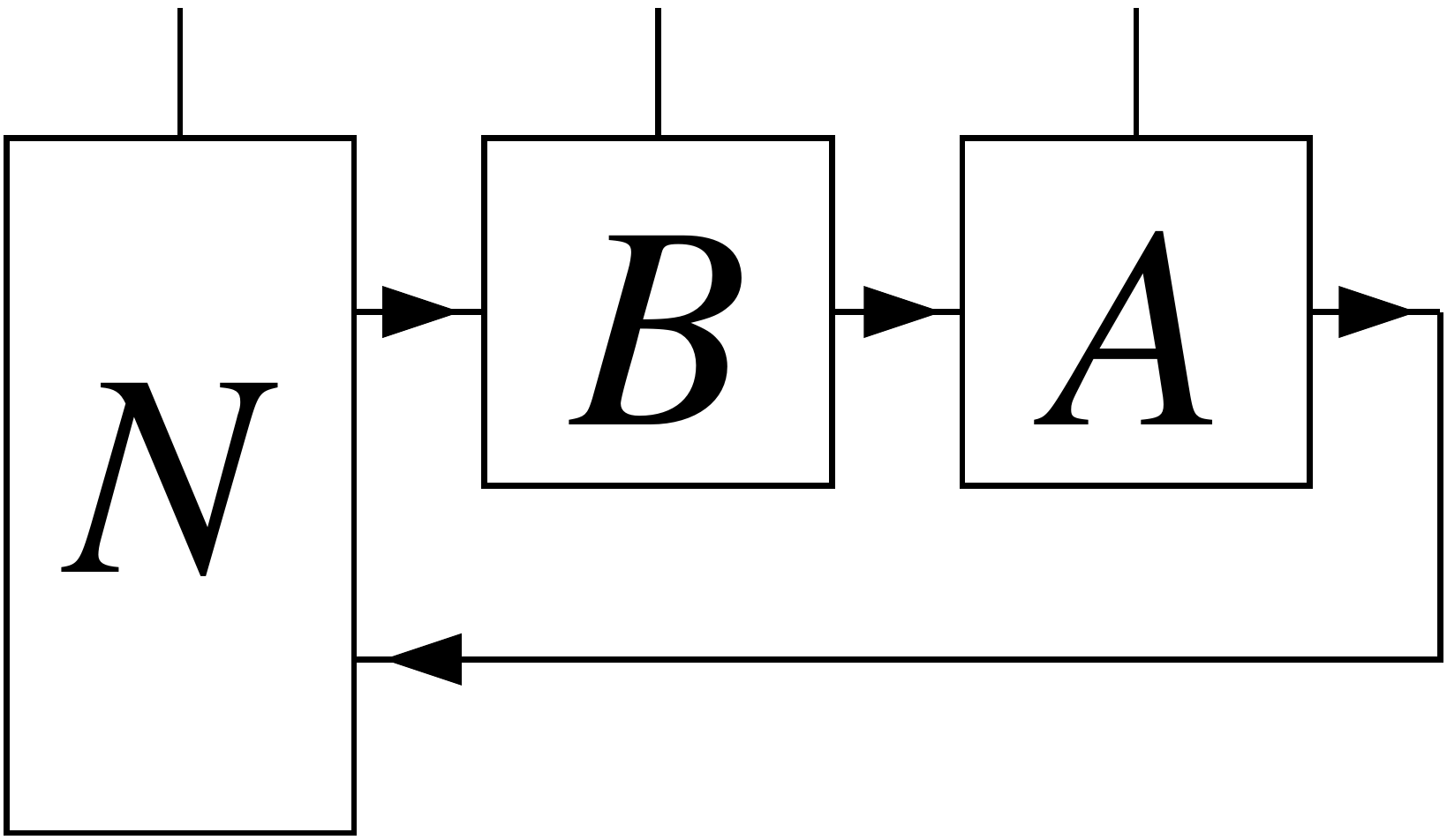}}\ \,\middle\vert \ N \right\}=
    \left\{\raisebox{-1.2em}{\,\includegraphics[height=3em]{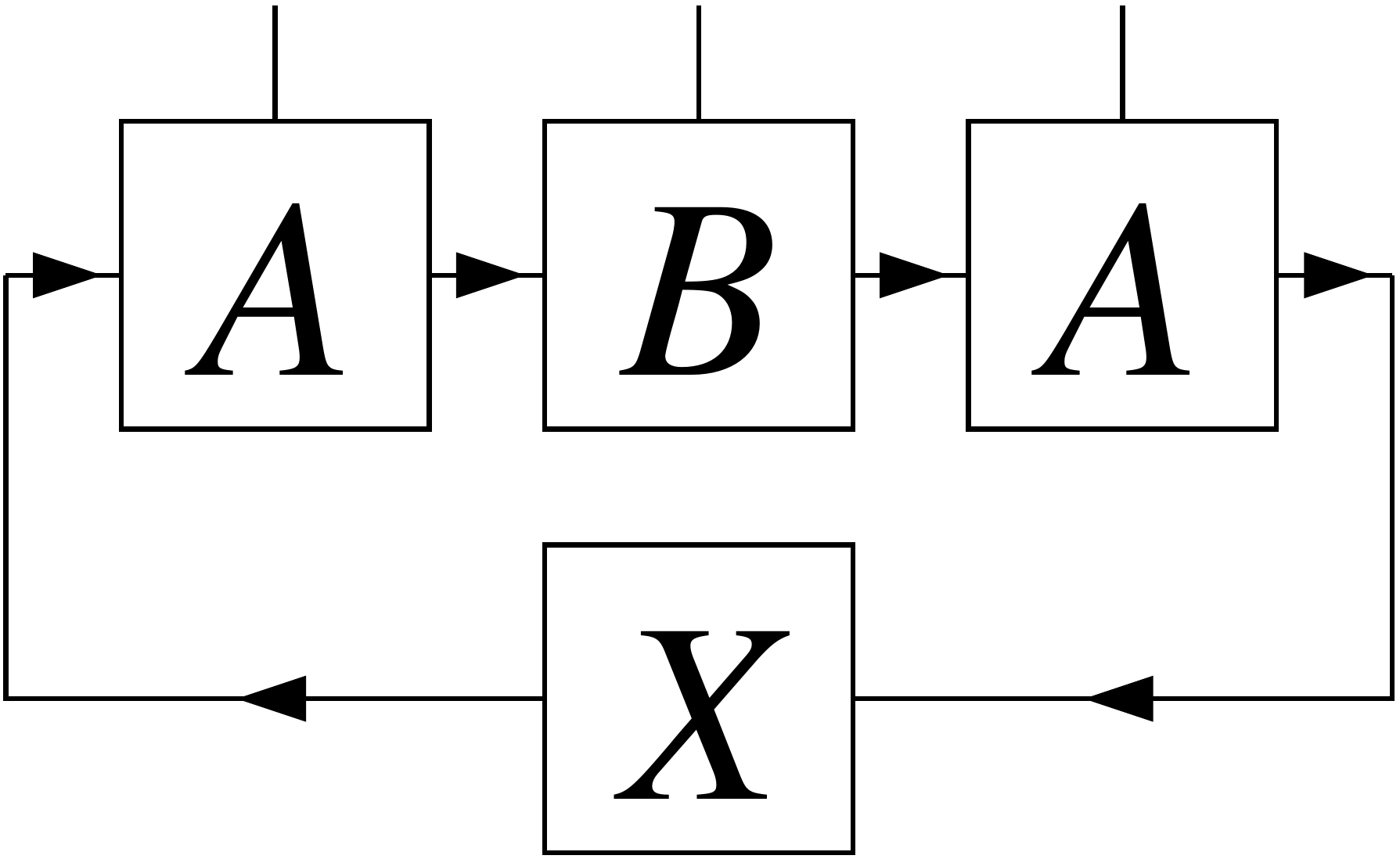}}\
    \,\middle\vert \ X \right\}\ ,
\end{equation}
\end{theorem}
\begin{proof}
Using the $G$--invariance of $A$ and $B$, we can infer that we can
restrict $M$, $N$, and $X$ to also be $G$ invariant in the virtual
indices. For instance, for any $M$,
\[
\raisebox{-1em}{\includegraphics[height=3em]{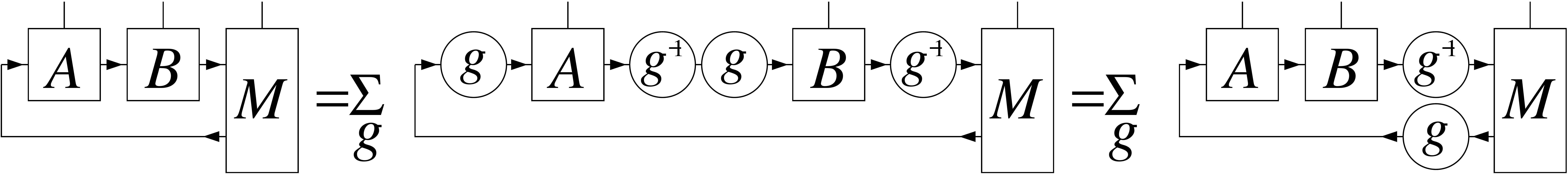}}\ ,
\]
this is, we can replace $M$ by
\[
\raisebox{-1em}{\includegraphics[height=3em]{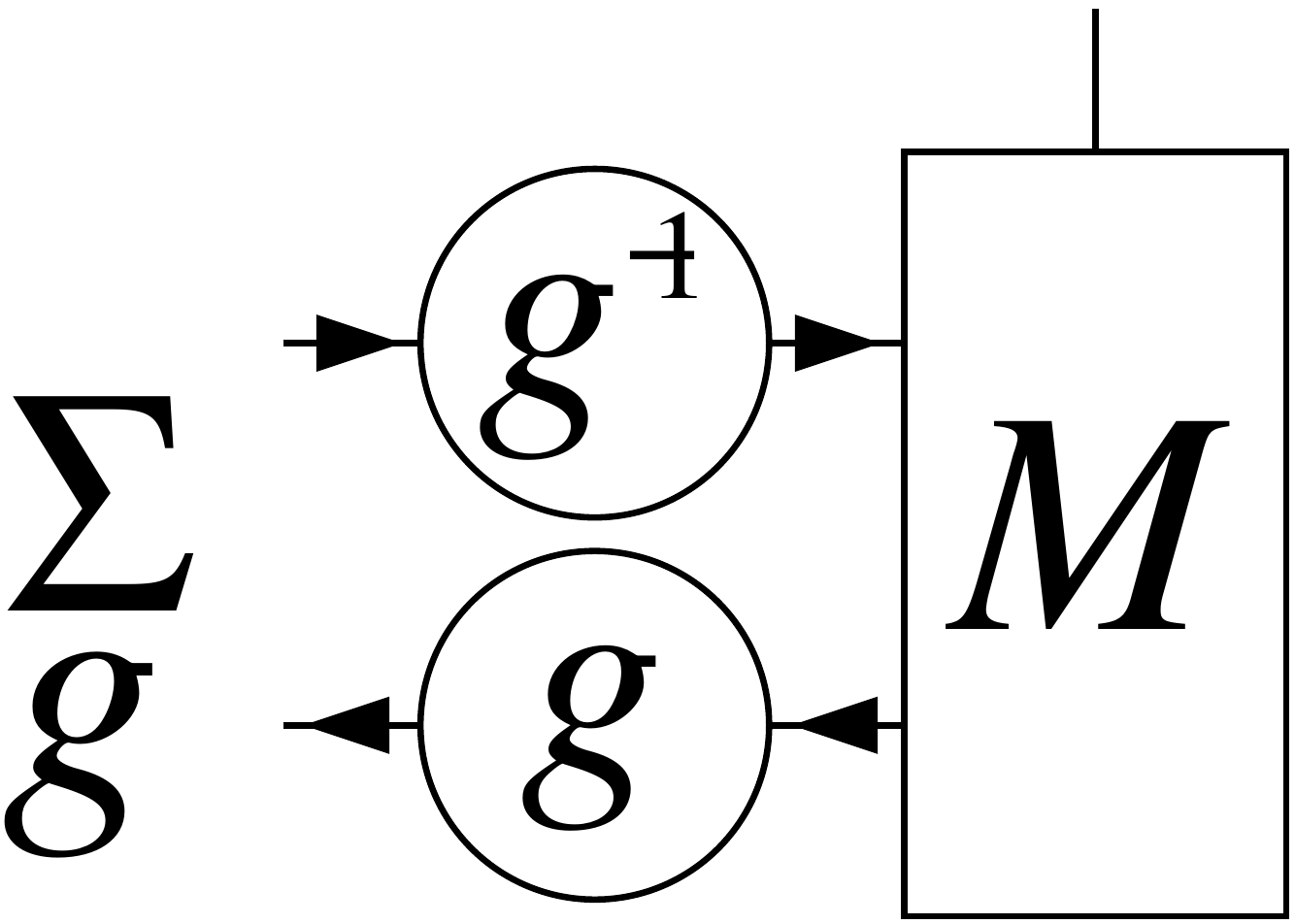}}\ ;
\]
the same holds true for $N$ and $X$.
We will always assume symmetrized tensors from now on.

We are now ready to prove \eqref{eq:1dsym:intersect}.  First, the r.h.s.
is clearly contained in the l.h.s., by choosing $M^i=A^iX$, $N^i=XA^i$.
On the other hand, each element in the intersection can be simultaneously
characterized by a pair $N$, $M$ of tensors, and by applying the inverse
maps we find that
\[
\raisebox{-1.0em}{\includegraphics[height=10em]{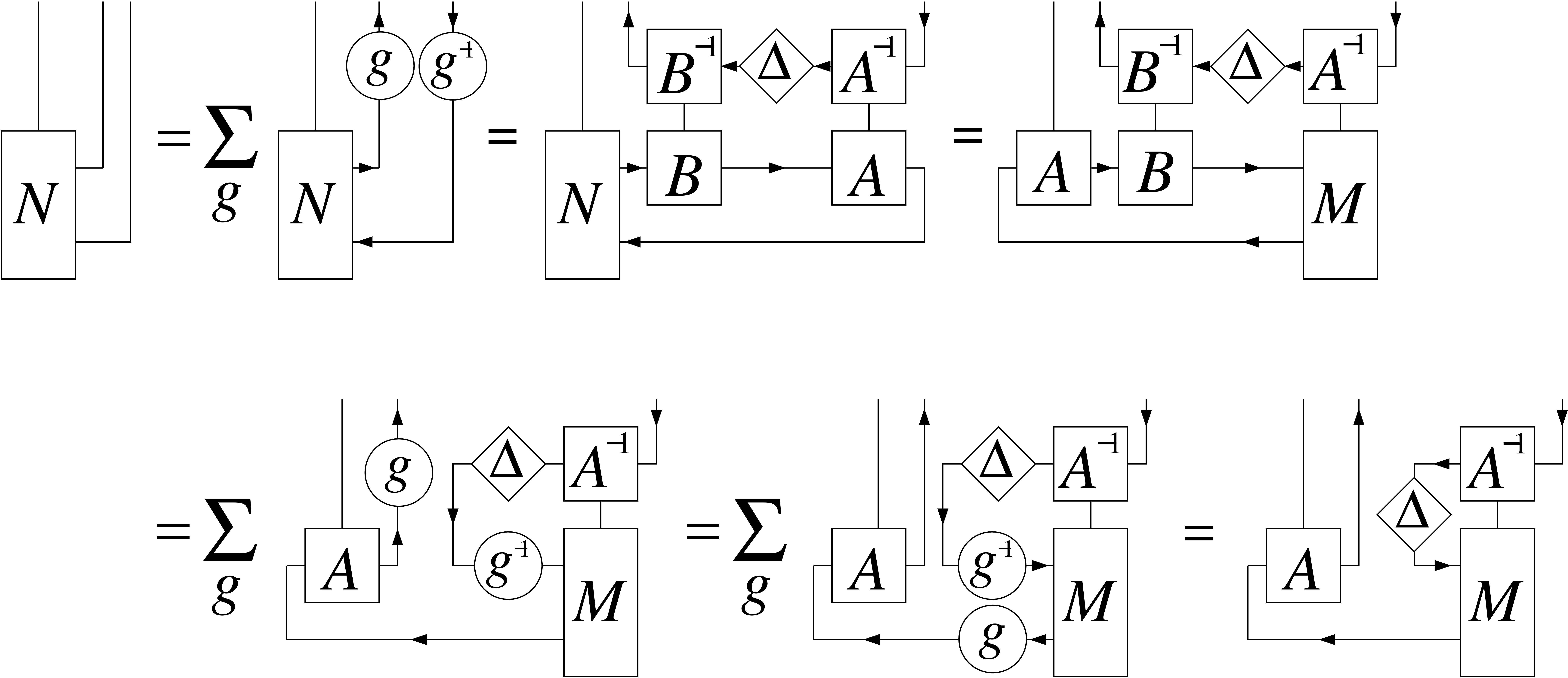}}
\quad .
\]
This shows that $M$ (and equally $N$) are of the form required by
Eq.~(\ref{eq:1dsym:intersect}), and thus proves the theorem.
\end{proof}

\begin{theorem}[Closure property] For $G$--injective $A$ and $B$,
\label{thm:noninj:closure}
\begin{equation}
    \left\{\raisebox{-1.7em}{\,\includegraphics[height=4em]{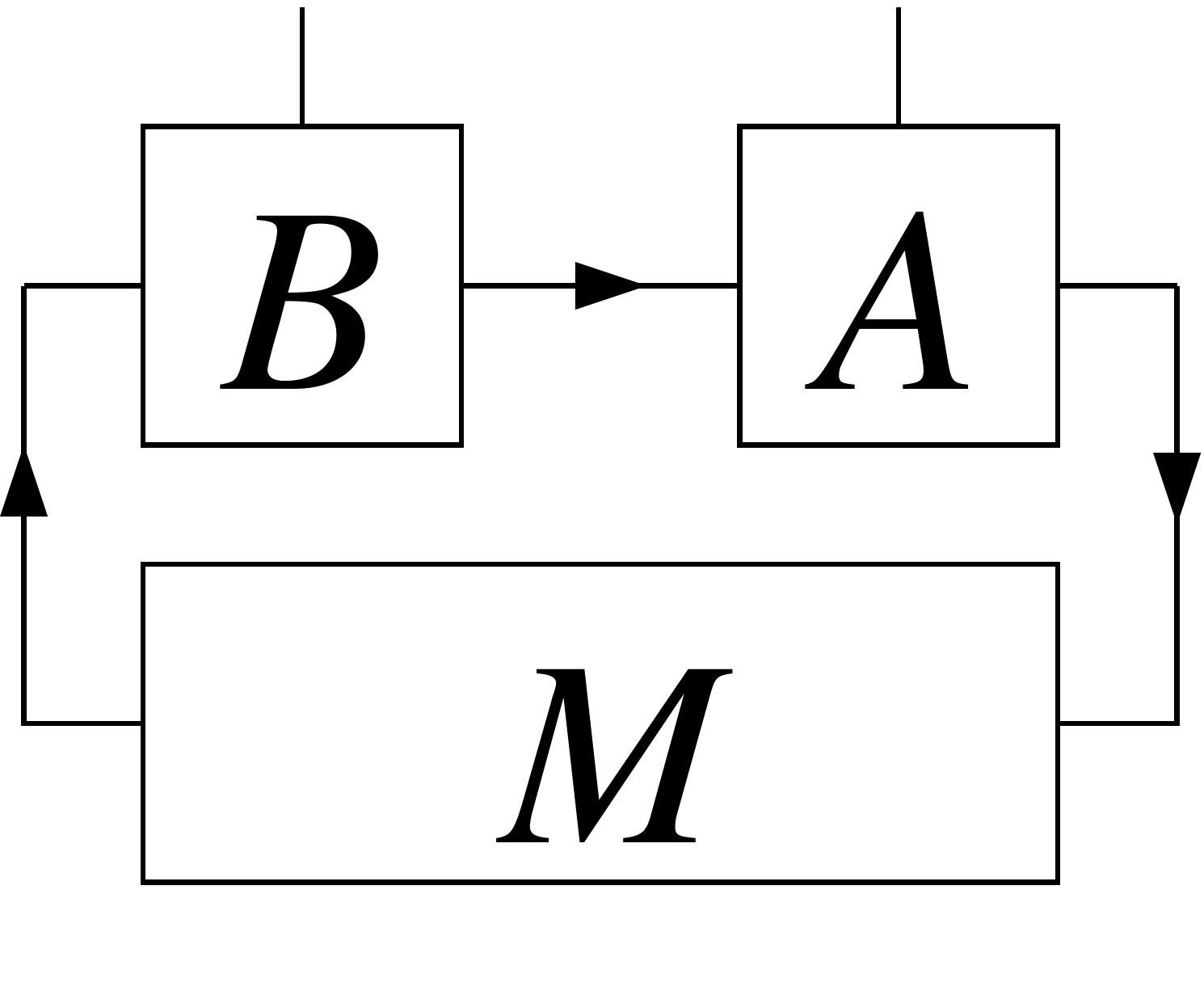}}\ \,\middle\vert \ M \right\}
    \cap
    \left\{\raisebox{-1.7em}{\,\includegraphics[height=4em]{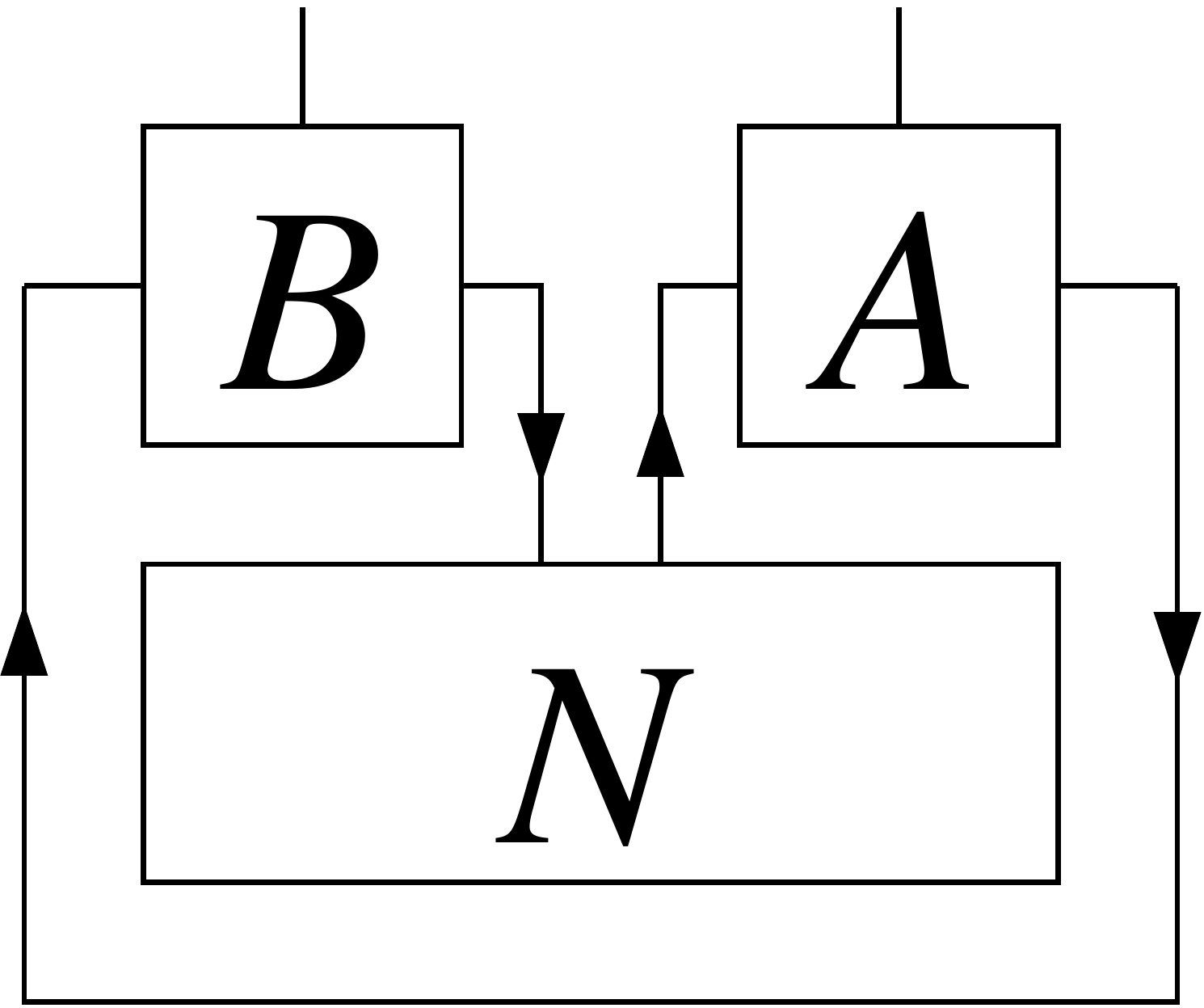}}\ \,\middle\vert \ N \right\}=
    \left\{\raisebox{-1.5em}{\,\includegraphics[height=3em]{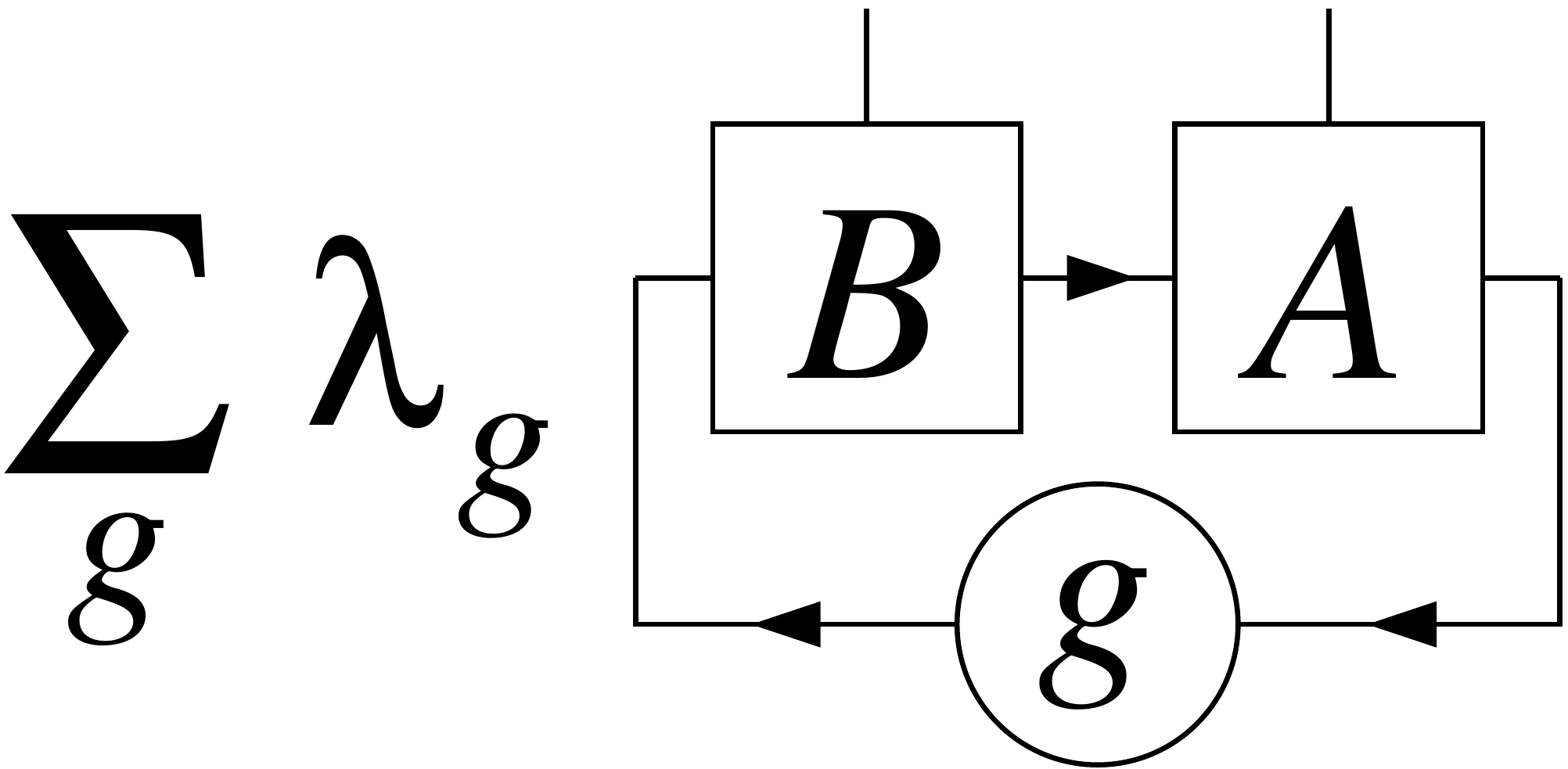}}\
    \,\middle\vert \ \lambda_g\right\}
    \ .
\label{eq:1dsym:closure}
\end{equation}
\end{theorem}

Before proving the theorem, let us give an intuition why the closure can
be done using any $U_g$ (and nothing else). To this end, regard $B$ as the
$(L-1)$-fold blocking of $A$'s. The closures $U_g$ are exactly the operators
which commute with $A$, and therefore, they can be moved to any position in the chain.
Thus, no local block of $A$'s needs to hold the closing $U_g$,
i.e., the state looks the same locally independent of the closure.

\begin{proof}
We may again assume that $M$ and $N$ are $G$--invariant.  It is
clear that the r.h.s.\ is contained in the l.h.s., by moving $U_g$
to the relevant link and setting $M=U_g$ or $N=U_g$, respectively.
On the other
hand, any element in the intersection can be written using some $M$ and
$N$, for which it holds that
\[
\raisebox{-2.0em}{
\includegraphics[height=5.5em]{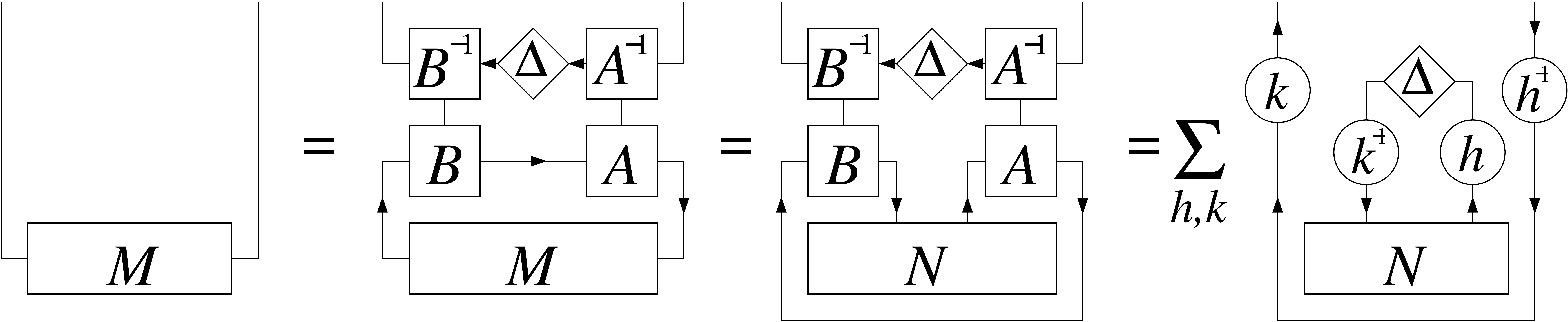}
}\quad .
\]
(In the first step, we have used that $M$ is $G$-invariant.) Thus, any
element of the intersection
is of the form of the r.h.s.\ of \eqref{eq:1dsym:closure}, with
$\lambda_g=\tr[U_{g^{-1}}N\Delta]$, $g\equiv kh^{-1}$.
\end{proof}

Let us now formally define translational invariant MPS with an operator in
the closure, as they appear in the above theorem.
\begin{definition} For an MPS tensor $A$ and $K\in\lin{D}$, we define
\begin{equation}
    \label{eq:noninj:mps-a-Ugclosed-def}
\ket{{\mathcal{M}}(A|K)}:=\sum_{i_1,\dots,i_L}
\tr[A^{i_1}\cdots A^{i_L}K]\ket{i_1,\dots,i_L}
\end{equation}
to be the MPS given by $A$ with closure $K$ \footnote{Note that one can
use a different closure to reduce dramatically the bond dimension. One
example is the W-state, whose bond dimension grows to infinity with L with
the standard closure \cite{perez-garcia:mps-reps,sanz:wielandt}, but has
bond dimension 2 with a different closure.}.
\end{definition}

\begin{theorem}[Parent Hamiltonians]
    \label{thm:noninj:parentham}
    Let $A$ be $G$--injective, $\mc S_2$ as in \eqref{eq:trBBX}, and let
\[
h_i=\openone_d^{\otimes(i-1)}\otimes (1-\Pi_{\mathcal S_2})
    \otimes \openone_d^{L-i-2}\ .
\]
Then,
\[
H_\mathrm{par}=\sum_{i=1}^L h_i
\]
has a subspace of frustration free ground states spanned by the MPS
$\MPS{A|U_g}$ with $U_g$-closed boundaries.
\end{theorem}
\begin{proof}
The proof is exactly the same as for Theorem~\ref{thm:inj:parent-ham},
except that is is now based on Theorems~\ref{thm:noninj:intersection}
and~\ref{thm:noninj:closure}; the latter
leading to the degeneracy of the ground state subspace.
\end{proof}

\begin{theorem}[Structure of ground state subspace]
Let $D^1,\dots,D^I$ be the irreducible representations of $U_g$, of
dimension $d_i$. Then, the ground state subspace of
Theorem~\ref{thm:noninj:parentham} is $I$-fold degenerate, and it is
spanned by the MPS $\MPS{A|\Pi_i}$, where
\begin{equation}
    \label{eq:noninj:all-indep-closures}
        \Pi_i = \frac{d_i}{|G|}\sum_g \tr[D^i(g^{-1})] U_g
\end{equation}
is the projector onto the subspace supporting the irreducible
representation $D^i$ in $U_g$ [proven in
\eqref{eq:noninj:sum-CHIi-Di-is-PROJi}].

Moreover, if $U_g$ is a semi-regular representation, $I$ is equal to the
number of conjugacy classes of $G$, and the subspace is spanned by the
linearly independent states $\MPS{A|U_g}$, where for each conjugacy class
$g$, one representative is chosen.
\end{theorem}
\noindent
Note that the first part of the theorem corresponds to the known form of
the different ground states, corresponding to the block structure of the
$A^i$'s~\cite{fannes:FCS,perez-garcia:mps-reps}, while the second part is
the new symmetry-based classification of ground states which can be
extended to the two-dimensional scenario.
\begin{proof}
First, the $G$--invariance of the $A^i$ implies that
\[
\MPS{A|U_g}=\MPS{A|U_hU_gU_h^\dagger}=\MPS{A|U_{hgh^{-1}}}\ ,
\]
i.e., all
closures from the same conjugacy class are equivalent. Let
$D^1,\dots,D^J$, $J\ge I$, be all irreducible representations of $G$.
Since the characters $\chi_i(g)=\tr[D^i(g)]$ form an orthonormal set for
the space of class functions (i.e.\ the functions which are constant over
conjugacy classes), the ground state space is equally spanned by
\[
\sum_g \tr[D^i(g^{-1})]\MPS{A|U_g} =
\MPS{A|\sum_g\tr[D^i(g^{-1})]U_g} \propto
\MPS{A|\Pi_i}\ ;
\]
Note that $\Pi_i=0$ for $i>I$ (i.e.\ if $U_g$ does not contain $D^i$).

Conversely, linear independence of $\MPS{A|\Pi_i}$, $i=1,\dots,I$  can be
seen as follows: For any class function $\lambda_g$, and with $C$ the
$L$-fold blocking of $A$, we have that
\begin{equation}
\raisebox{-1.0em}{\includegraphics[height=5em]{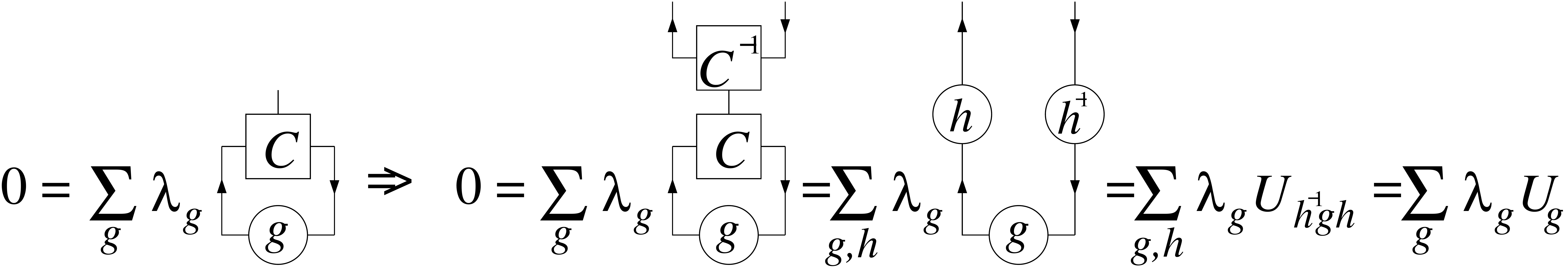}}\ .
\label{eq:noninj:Ug-lin-indep-diag}
\end{equation}
Thus,
\begin{align*}
0 &=\sum_i \mu_i \MPS{A|\Pi_i}
=\sum_{i,g} \mu_i\tfrac{d_i}{|G|}\tr[D^i(g^{-1})] \MPS{A|U_g} \\
\stackrel{\eqref{eq:noninj:Ug-lin-indep-diag}}{\Rightarrow}
\quad
0 &= \sum_{i,g} \mu_i \tfrac{d_i}{|G|}\tr[D^i(g^{-1})] U_g = \sum_{i} \mu_i \Pi_i \\
\Rightarrow \quad 0 &=\mu_i \quad\forall\,i\ ,
\end{align*}
which proves linear independence. (Note that we had to use that $D^i$ is
contained in $U_g$, otherwise $\Pi_i\equiv 0$.)

The second statement follows from the fact that $\MPS{A|U_g}$ is constant
on conjugacy classes, and that the number of conjugacy classes equals the
number of irreducible representations of $G$.
\end{proof}

\section{Two dimensions: $G$--injective PEPS\label{sec:PEPS-noninj}}

\subsection{Definition and basic properties}

Having understood the one-dimensional case of $G$--injective MPS, let us
now turn towards two dimensions. We will introduce some new conventions
for the diagrams (as in principle, we need a third dimension), which we
will explain right after the definition of $G$--injectivity.

\begin{definition}
    \label{def:2d-Ug-inj}
Let $g\mapsto U_g$ be a semi-regular representation of a finite group $G$.
We call a PEPS tensor $A$ $G$--injective if\\
i) It is invariant under $U_g$ on the virtual level,
\begin{equation}
    \label{eq:2d-ug-sym}
    \raisebox{-2.0em}{
        \includegraphics[height=5em]{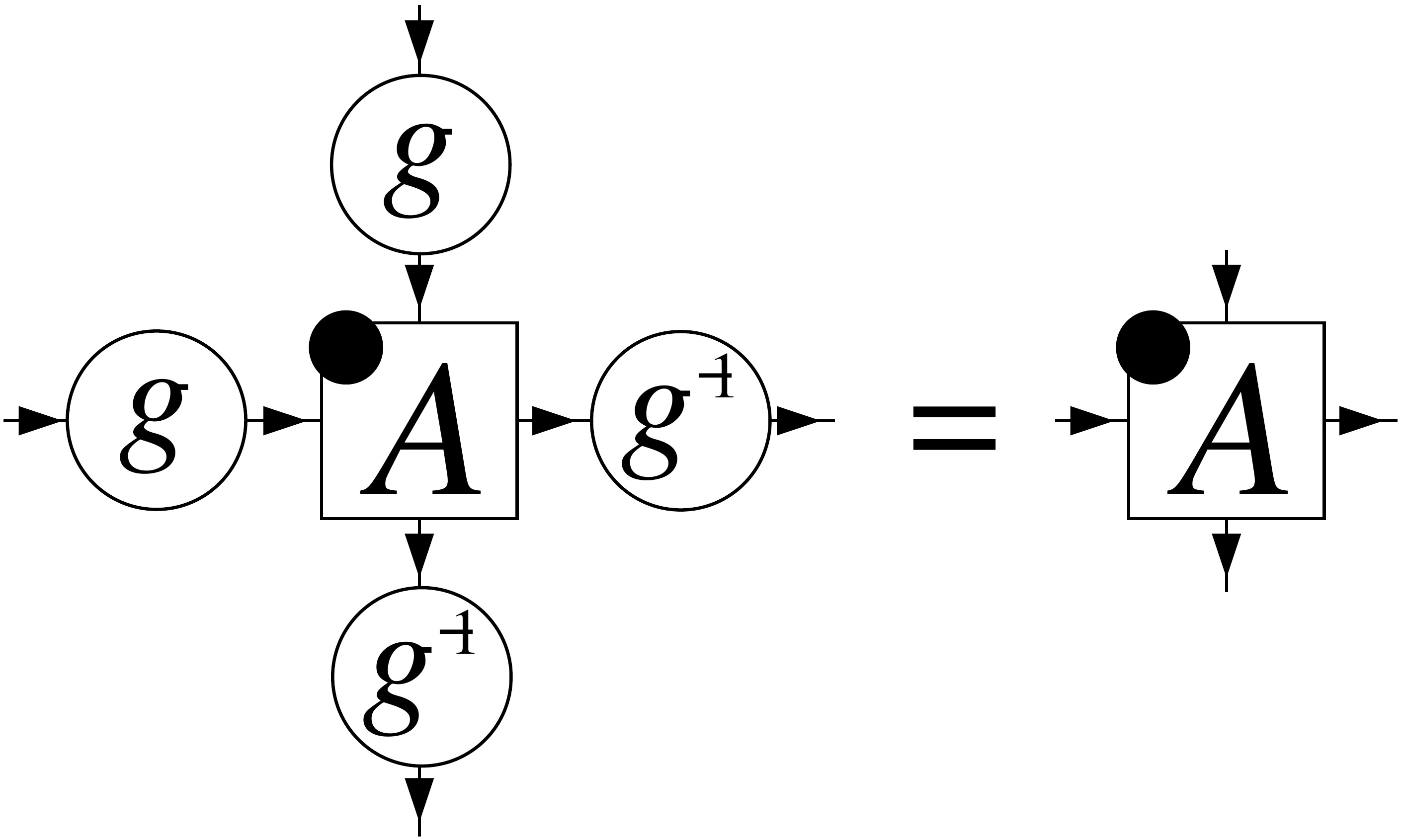}
        }
\end{equation}
ii) There exists a left inverse to $\mathcal P(A)$ such that $\mc
P(A)^{-1}P(A)=\Pi_{\mc U}$, the projector on the $U_g$-invariant subspace:
\begin{equation}
    \label{eq:2d-linv}
\raisebox{-2.0em}{\includegraphics[height=5em]{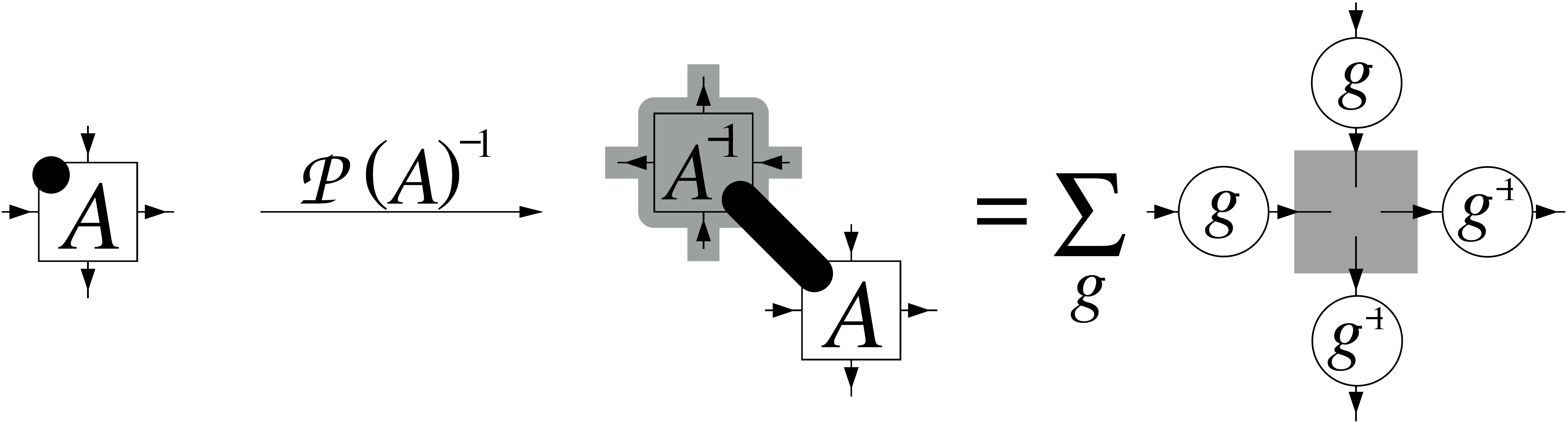}}
\end{equation}
\end{definition}

Let us briefly explain the differences in notation: We will try to avoid
three-dimensional plots as far as possible. PEPS tensors are generally
depicted ``from the top'': The four legs in \eqref{eq:2d-ug-sym} are the
virtual indices, and a black dot denotes the physical index (for the
inverse, it is in the lower right instead of the upper left corner). As we
use the let inverse to make the virtual subspace accessible via
the physical indices, we will depict only the situation
\emph{after} applying the left inverse whenever possible, as depicted on
the very right of \eqref{eq:2d-linv}. In order to distinguish the ``original''
virtual level and the one after the application of the left-inverse, we
shade the latter gray. (This corresponds to the lower and upper layer,
respectively, in the 1D case.)

Note that there is no need to choose the same representation $U_g$ for the
horizontal and vertical direction. In fact, as mentioned earlier
representations are assigned to \emph{links}, and every link can carry its
own representation -- all that matters is that the two tensors acting on a
link act with the same representation.  It is this possibility of changing
the representations which enables us to prove that $G$--injectivity is is
preserved under blocking.

\begin{lemma}[Stability under concatenation]
Let $A$ and $B$ be $G$--injective tensors.  Then,
\begin{equation}
    \label{eq:2d:c-is-ab}
        \raisebox{-0.7em}{
        \includegraphics[height=1.8em]{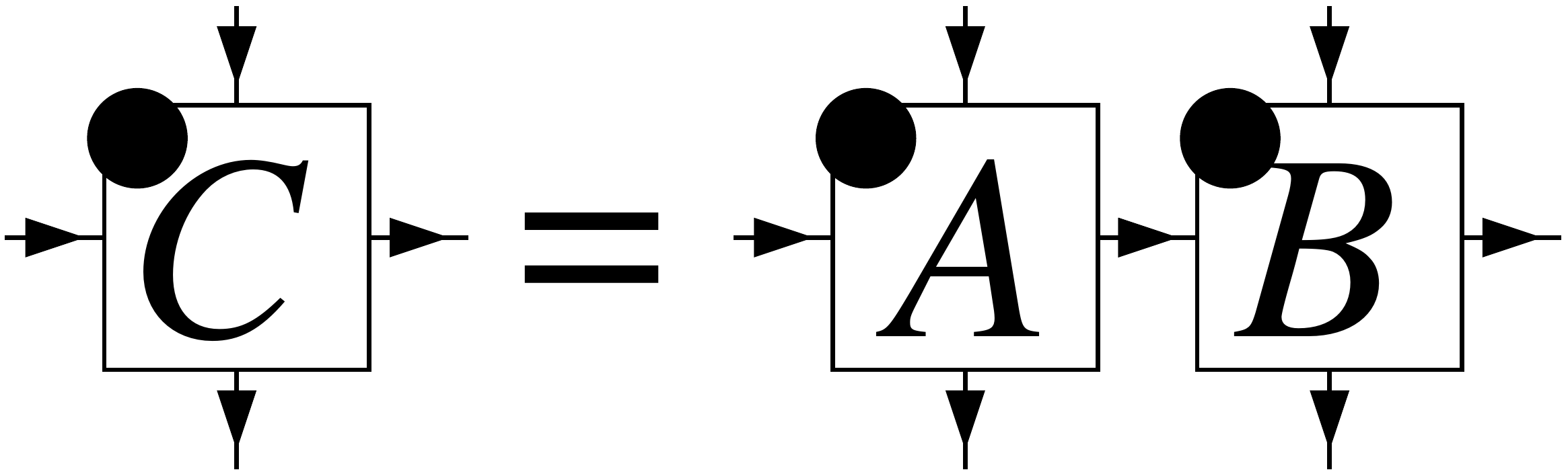}
        }
\end{equation}
(with blocked up, down, and physical indices)
is also $G$--injective with left inverse
\begin{equation}
    \label{eq:2d:ab-inv}
    \raisebox{-0.9em}{\includegraphics[height=2.0em]{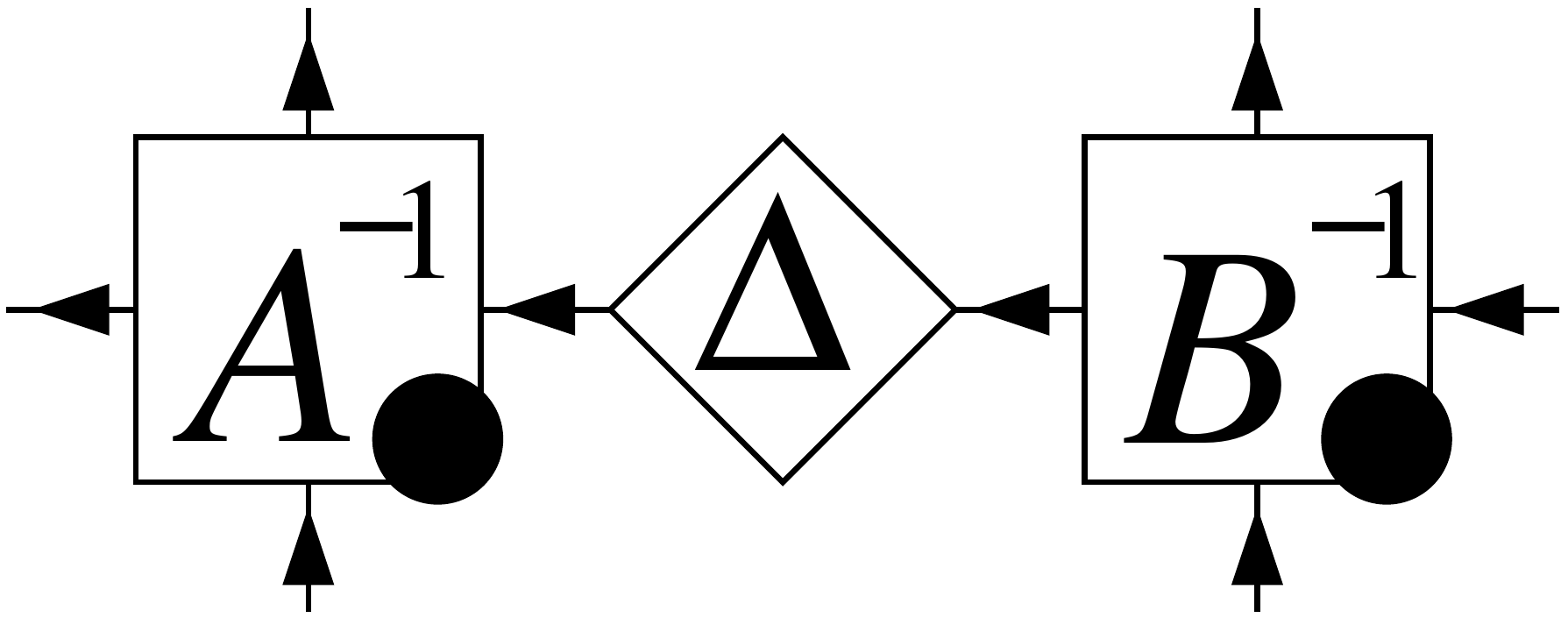}}
\end{equation}
\end{lemma}
\begin{proof}
First, note that for any two semi-regular representations $U_g$ and $V_g$,
$W_g=U_g\otimes V_g$ is again a semi-regular representation. Then, it is
clear that $C$ is also $G$--invariant, as the action of the $U_g$ on the
inner link cancels.  That \eqref{eq:2d:ab-inv} is
left-inverse to $C$ follows using $\tr[U_g\Delta]\propto\delta_{g,1}$,
Lemma~\ref{lemma:noninj:semireg-trace-ug-delta}:
\begin{equation}
    \label{eq:2d:linv-application}
\raisebox{-2em}{
\includegraphics[height=5.0em]{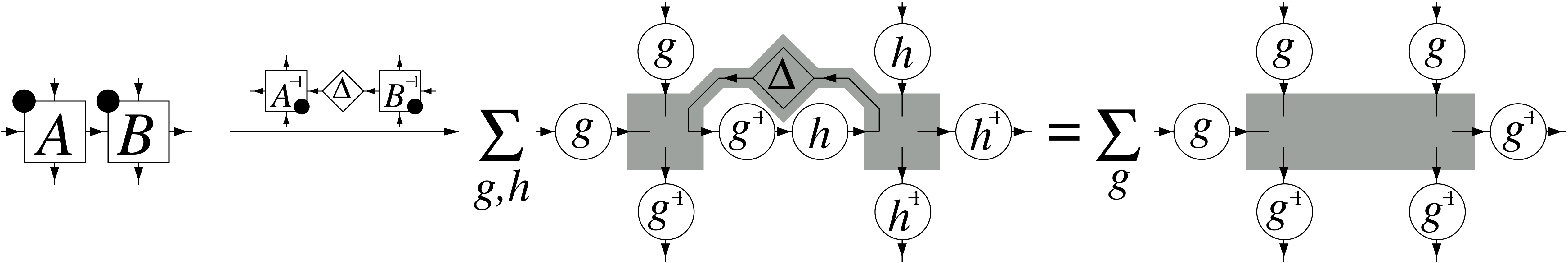}
}\ .
\end{equation}
\end{proof}
\begin{observation}
Note that $G$--injectivity is also preserved when contracting legs of an
already connected block, e.g.\ the up leg of $A$ with the down leg of $B$
in \eqref{eq:2d:c-is-ab}: The resulting tensor is clearly again
$G$--invariant, and the left-inverse is obtained by contracting the
corresponding legs of \eqref{eq:2d:ab-inv} with any any operator with
nonzero trace (e.g., $\Delta$): The group elements attached to the two
legs cancel out, as they belong to the same tensor.
\end{observation}

\subsection{Parent Hamiltonians}

The idea to construct parent Hamiltonians is essentially the same as in
one dimension. We define the local Hamiltonian as $1$ minus the
projector on the span of a $2\times 2$ block (the smallest block which allows
for an overlapping tiling of the lattice), and study how the ground state
subspace behaves when growing the block. For simplicity, we first grow the
$2\times 2$ block in one direction until we reach a
$2\times L$ lattice, and then in the other direction until we have the
full $L\times L$ lattice with open boundaries. Finally, we study what
happens when we close the boundaries. While the growing will work
essentially exactly as in 1D, closing the boundaries will give a richer
structure.

\begin{theorem}[Intersection property]
    \label{thm:2d:intersection}
Let $A$, $B$ be $G$--injective. Then,
\[
\left\{
    \vphantom{\raisebox{-1.2em}{\rule{3em}{0em}}}
    \smash{\underbrace{
    \raisebox{-1.2em}{\includegraphics[height=3em]{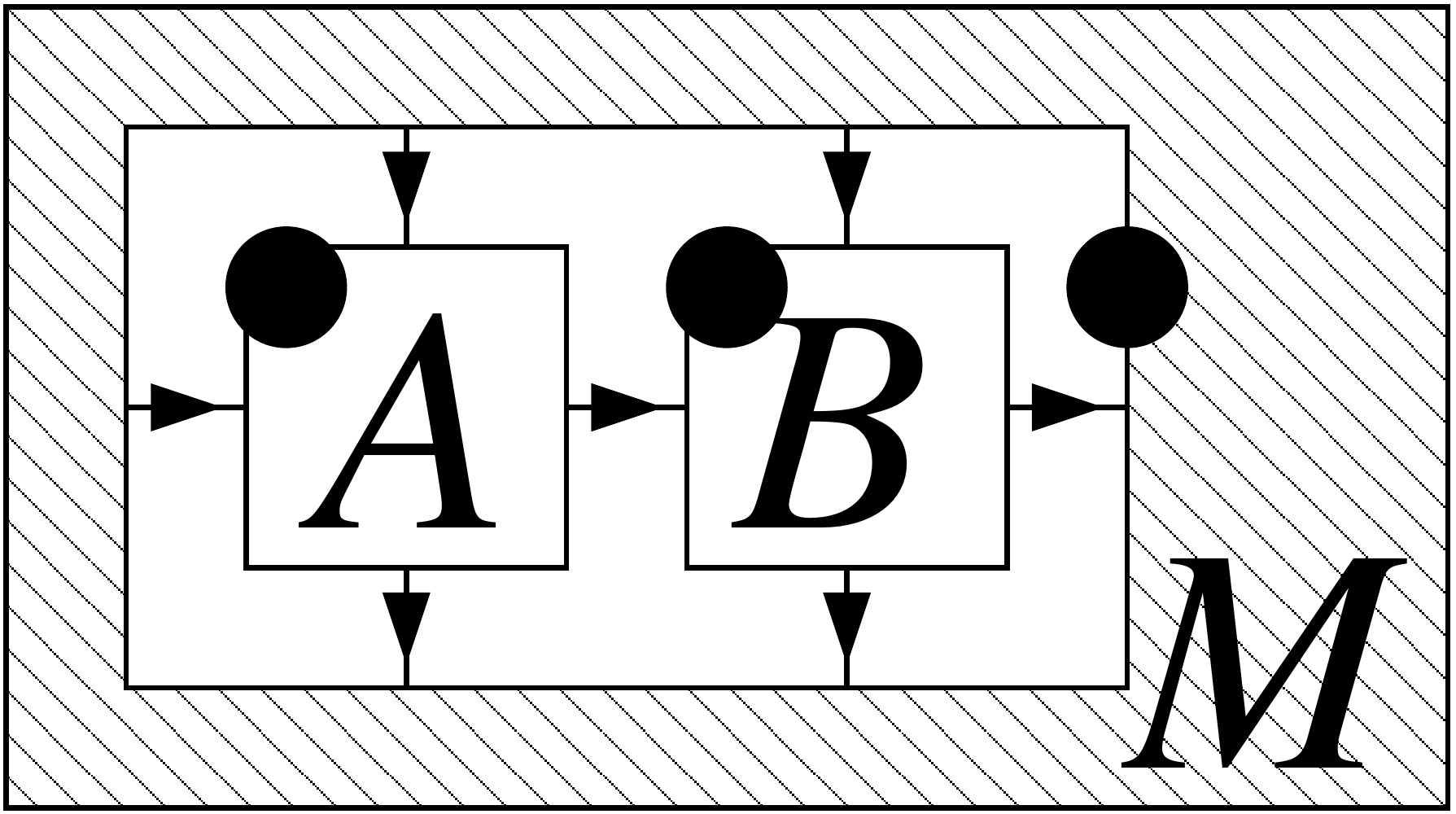}\rule[-0.2em]{0.2em}{0em}}
    }_{\displaystyle=:\ket{\alpha(M)}}}\
    \middle\vert M \right\}
\cap
\left\{
    \vphantom{\raisebox{-1.2em}{\rule{3em}{0em}}}
    \smash{\underbrace{
    \raisebox{-1.2em}{\includegraphics[height=3em]{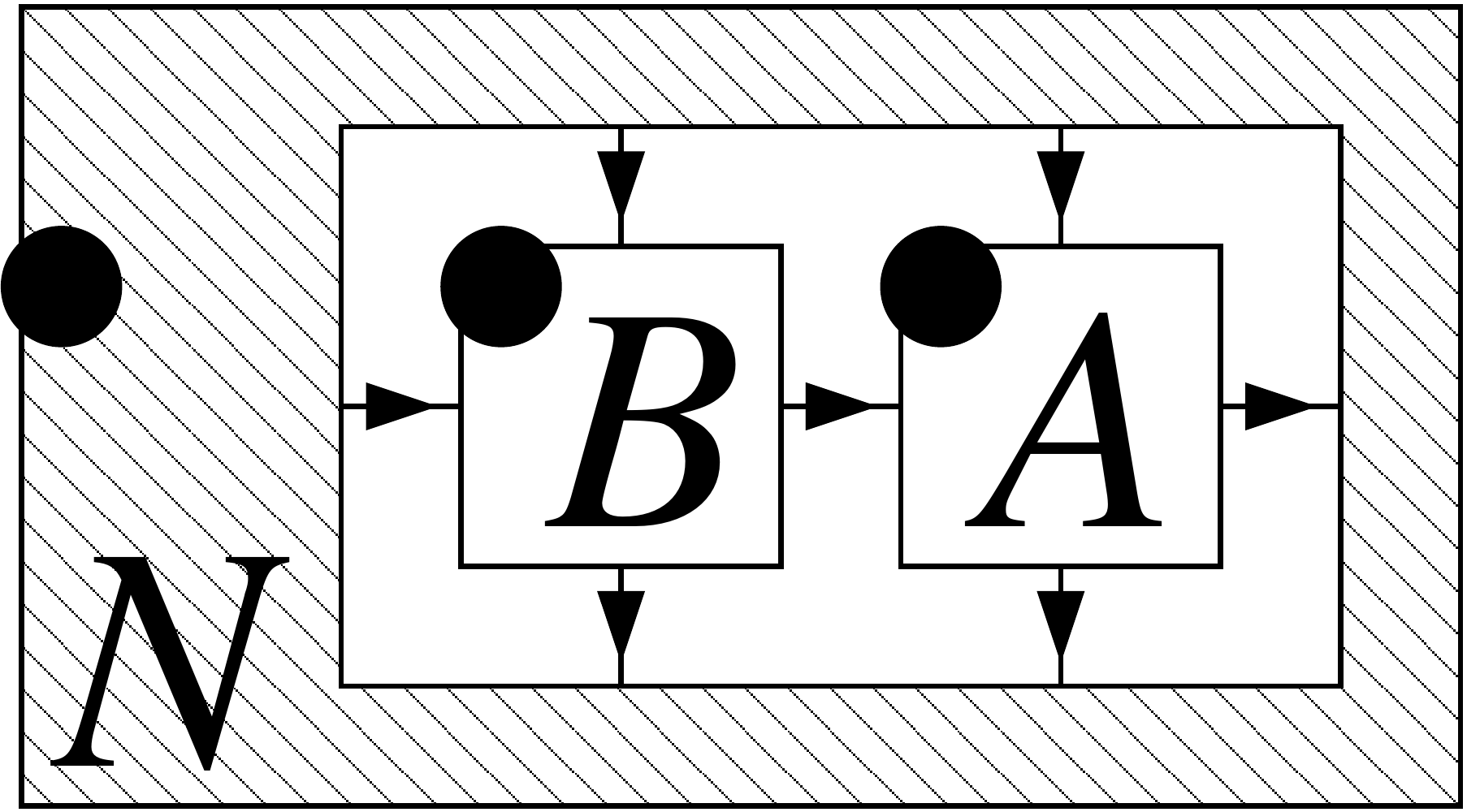}\rule[-0.2em]{0.2em}{0em}}
    }_{\displaystyle=:\ket{\beta(N)}}}\
    \middle\vert N \right\}
=
\left\{
    \vphantom{\raisebox{-1.2em}{\rule{3em}{0em}}}
    \smash{\underbrace{
    \raisebox{-1.2em}{\includegraphics[height=3em]{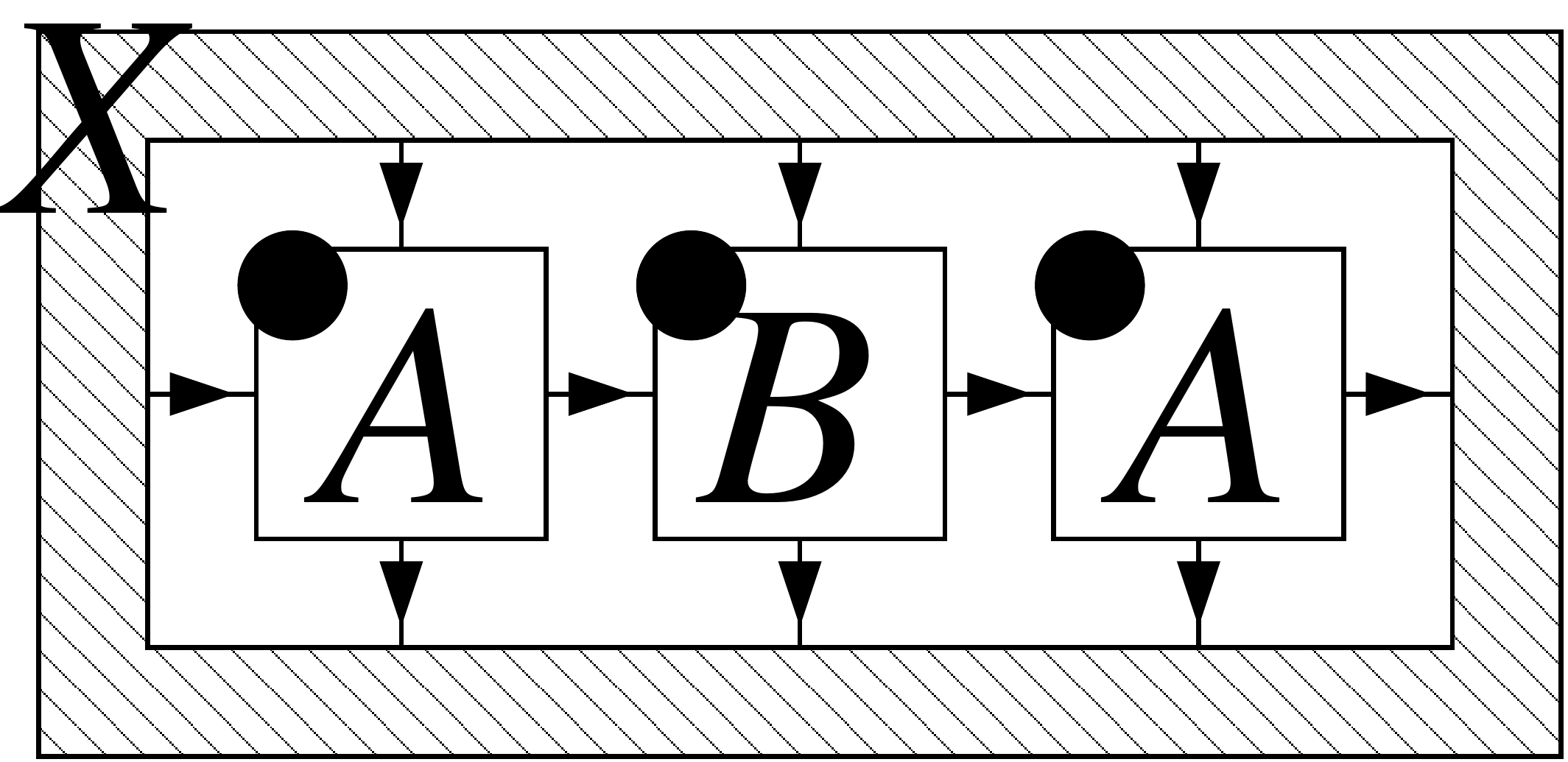}\rule[-0.2em]{0.2em}{0em}}
    }_{\displaystyle=:\ket{\zeta(X)}}}\
    \middle\vert X \right\}\ .
\vphantom{\underbrace{
    \raisebox{-1.2em}{\includegraphics[height=3em]{figs3/grow-N}}\
    }_{\displaystyle=:\ket{\beta(M)}}
}
\]
\end{theorem}
%\noindent
Here, we have chosen to shade the inside of the ``boundary condition''
tensors $M$, $N$, and $X$.
\begin{proof}
The proof is exactly analogous to the one-dimensional case,
Theorem~\ref{thm:noninj:intersection}. The r.h.s.\ is contained in the
l.h.s., as any $\ket{\zeta(X)}$ can be written as both $\ket{\alpha(M)}$
and $\ket{\beta(N)}$. Conversely, any element in the intersection can be
written as $\ket{\alpha(M)}=\ket{\beta(N)}$ for some $M$ and $N$; as in
the one-dimensional case, we can assume both to be $G$--invariant.
To recover $M$, we apply the left inverse \eqref{eq:2d:ab-inv} to the left
two physical modes of $\ket{\alpha(M)}=\ket{\beta(N)}$ and obtain
\[
\raisebox{-1.4em}{
\includegraphics[height=6.5em]{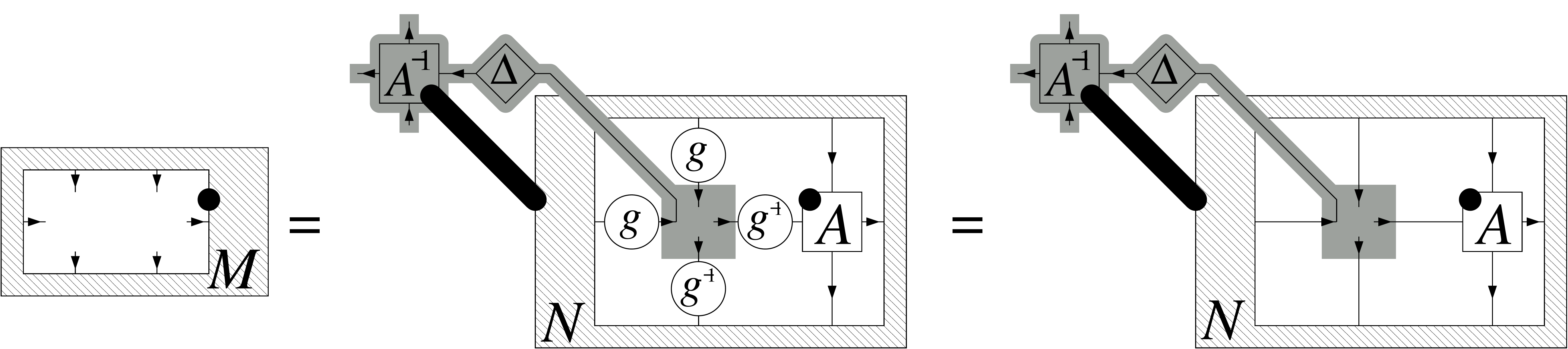}}\quad ,
\]
which implies that the state is of the form $\ket{\zeta(X)}$.
\end{proof}

\begin{theorem}[Closure property]
    \label{thm:2d:closure}
For $G$--injective $A$, $B$, $C$, and $D$,
\begin{align}
\left\{
    \vphantom{\raisebox{-2.2em}{\rule{5em}{0em}}}
    \smash{\underbrace{
    \raisebox{-2.2em}{\includegraphics[height=5em]{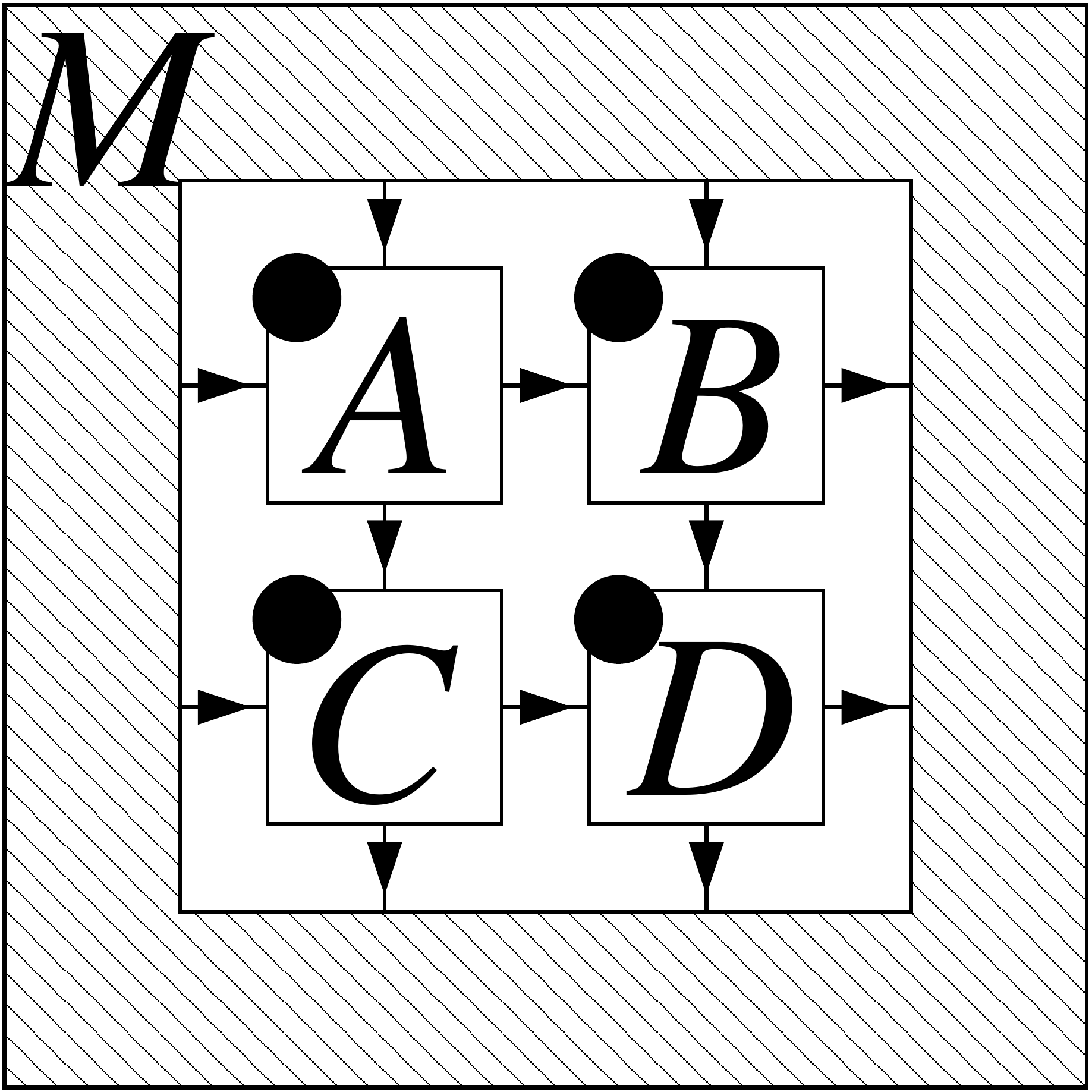}\rule[-0.2em]{0.2em}{0em}}
    }_{\displaystyle=:\ket{\alpha(M)}}}\
    \middle\vert M \right\}
&\cap
\left\{
    \vphantom{\raisebox{-2.2em}{\rule{5em}{0em}}}
    \smash{\underbrace{
    \raisebox{-2.2em}{\includegraphics[height=5em]{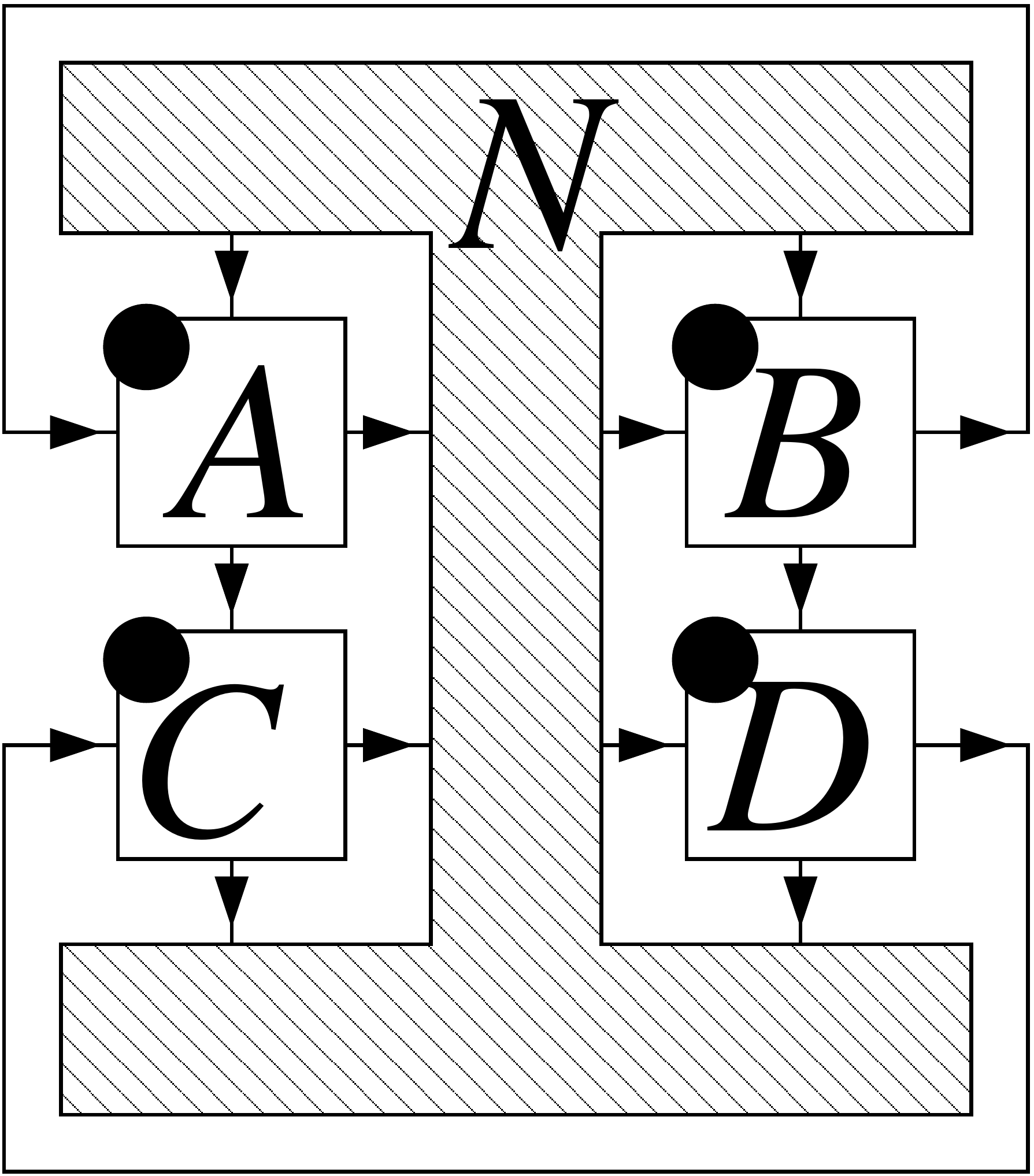}\rule[-0.2em]{0.2em}{0em}}
    }_{\displaystyle=:\ket{\beta(N)}}}\
    \middle\vert N \right\}
\cap
\left\{
    \vphantom{\raisebox{-2.2em}{\rule{5em}{0em}}}
    \smash{\underbrace{
    \raisebox{-2.2em}{\includegraphics[height=5em]{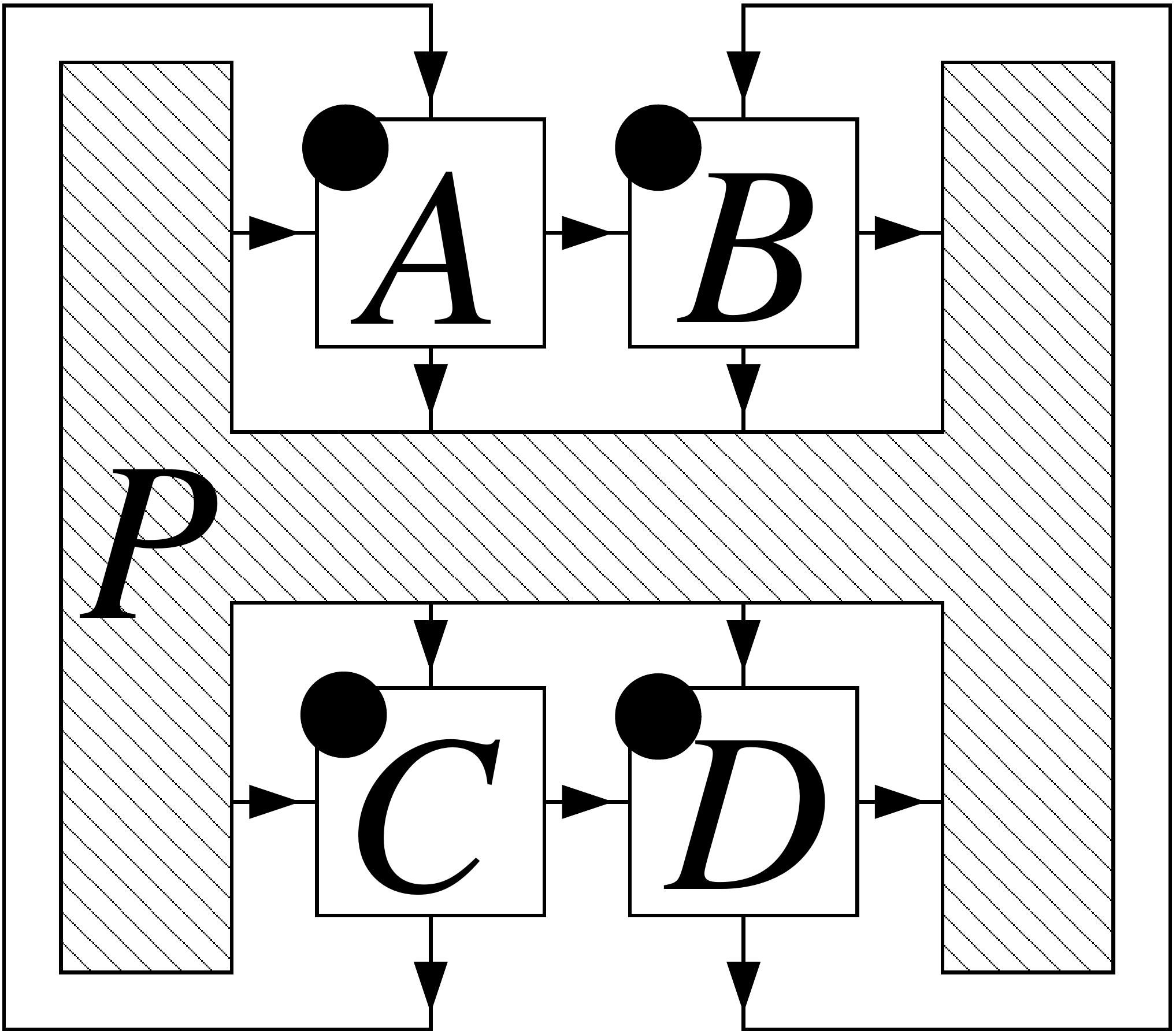}\rule[-0.2em]{0.2em}{0em}}
    }_{\displaystyle=:\ket{\gamma(P)}}}\
    \middle\vert P \right\}
\cap
\vphantom{\underbrace{
    \raisebox{-2.2em}{\includegraphics[height=2em]{figs3/close-out-in}}\
    }_{\displaystyle=:\ket{\beta(M)}}
}
\label{eq:2d:closure-intersection}
\\[1em] \nonumber
 & \cap
\left\{
    \vphantom{\raisebox{-2.2em}{\rule{5em}{0em}}}
    \smash{\underbrace{
    \raisebox{-2.2em}{\includegraphics[height=5em]{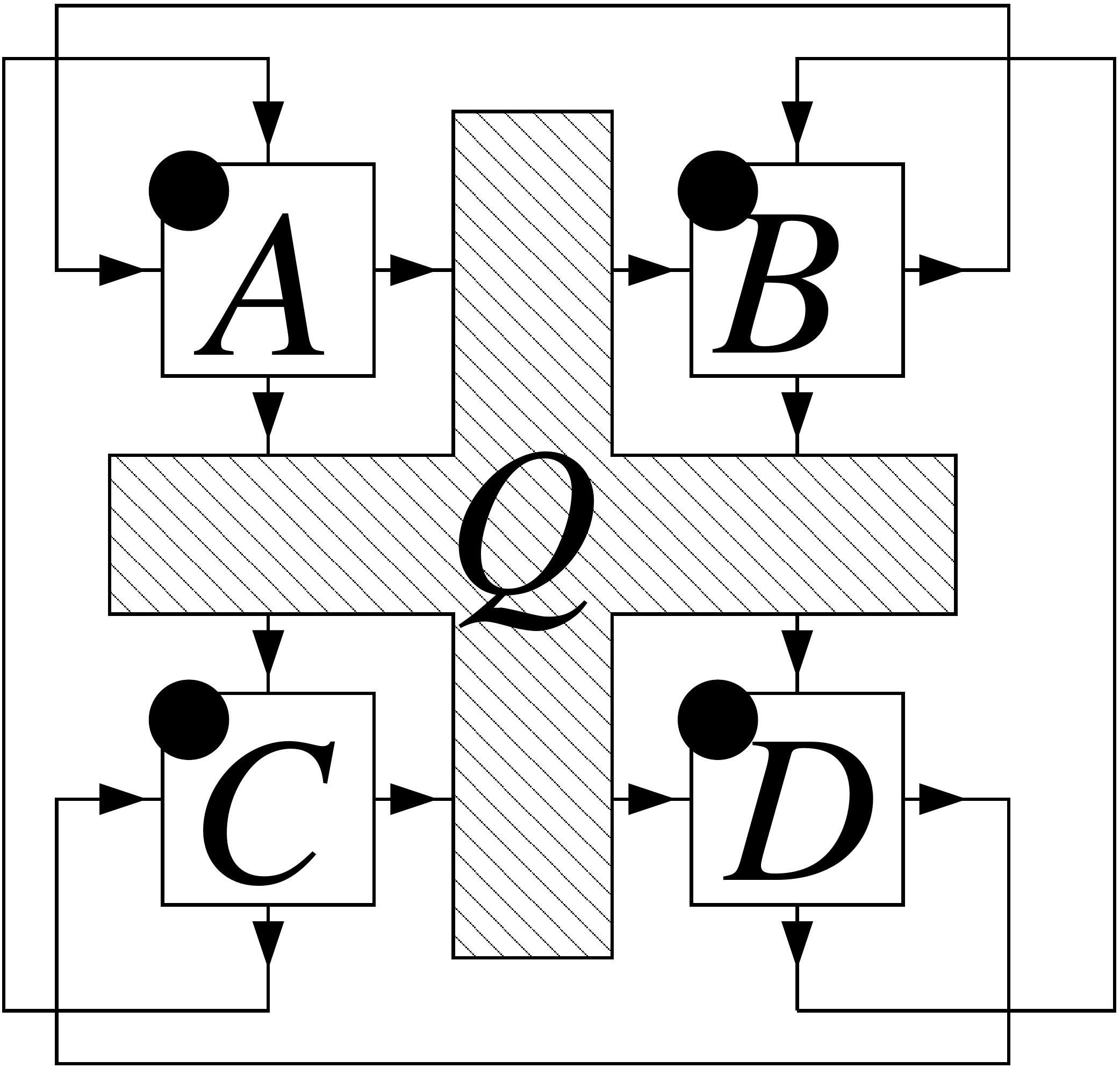}\rule[-0.2em]{0.2em}{0em}}
    }_{\displaystyle=:\ket{\delta(Q)}}}\
    \middle\vert Q \right\}
=
\left\{
    \vphantom{\raisebox{-2.2em}{\rule{5em}{0em}}}
    \smash{\underbrace{
    \raisebox{-2.2em}{\includegraphics[height=5em]{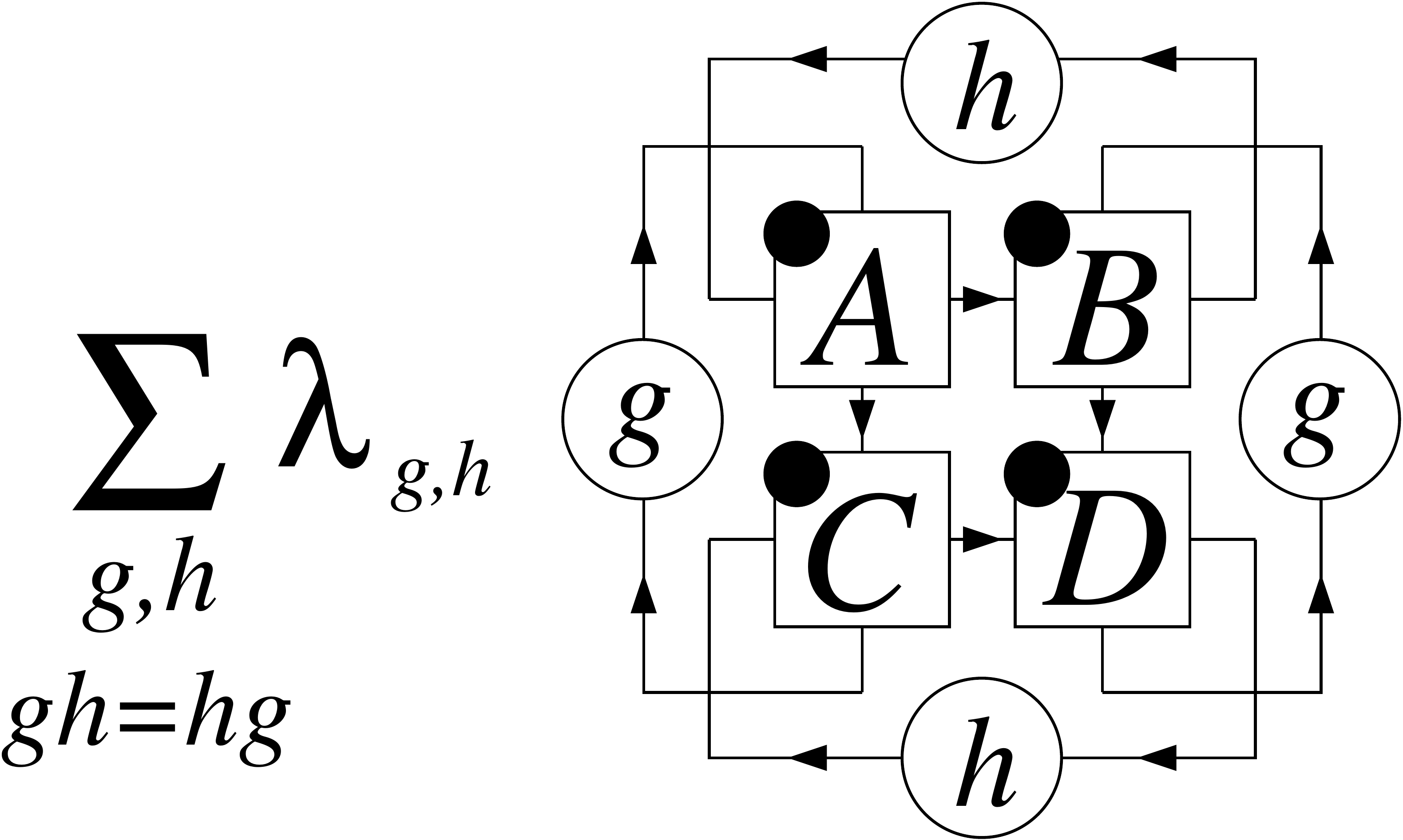}\rule[-0.2em]{0.2em}{0em}}
    }_{\displaystyle=:\ket{\zeta}}}\
    \middle\vert \lambda_{g,h} \right\}\ .
\vphantom{\underbrace{
    \raisebox{-2.2em}{\includegraphics[height=2em]{figs3/close-out-in}}\
    }_{\displaystyle=:\ket{\beta(M)}}
}
\end{align}
Here, the sum runs over all pairs $(g,h)\in G\times G$ such that $gh=hg$.
\end{theorem}
Note that if $A$, $B$, $C$, and $D$ arise from blocking the original
tensor, this corresponds to a closure with two unitaries $U_g^{\otimes L}$
and $U_h^{\otimes L}$ at the horizontal and vertical closure,
respectively.
\begin{proof}
The proof again follows closely the proof for one dimension. First, it is
clear that the r.h.s.\ in contained in the intersection by choosing $M$,
$N$, $P$, and $Q$ appropriately. To show that every element in the
intersection is of the form $\ket{\zeta}$, we consider an element of the
intersecion, which can be written as $\ket{\alpha(M)}=\ket{\delta(Q)}$
with boundaries $M$ and $Q$, respectively. As before, we can assume the
boundary conditions to be $G$--invariant. To recover $M$, we apply the
left-inverse
\begin{equation}
    \label{eq:2d:closure-inv}
\raisebox{-2.4em}{
\includegraphics[height=5.5em]{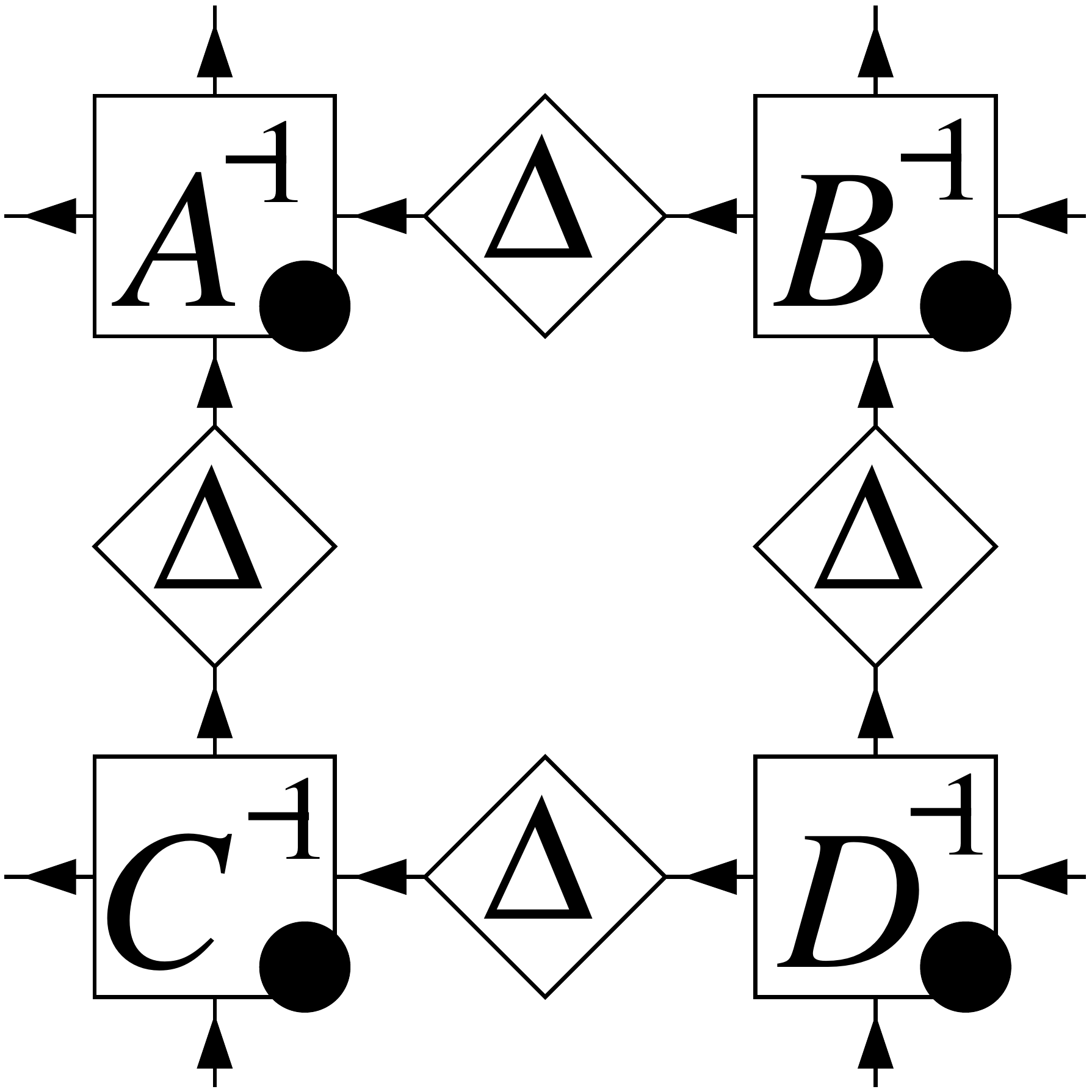}
}
\end{equation}
to $\ket{\alpha(M)}=\ket{\delta(Q)}$, and obtain
\begin{equation}
    \label{eq:2d:close-in-in}
\raisebox{-4em}{\includegraphics[height=9.5em]{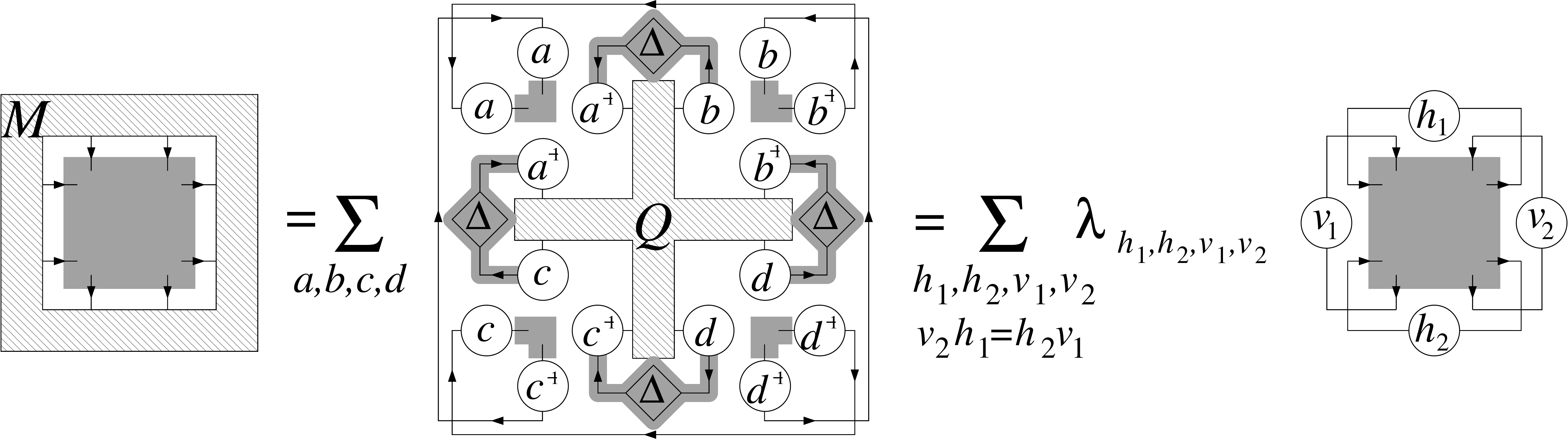}}\ ,
\end{equation}
where $h_1=b^{-1}a$, $h_2=d^{-1}c$, $v_1=c^{-1}a$, and $v_2=d^{-1}b$,
which results in the constraint $v_2h_1=h_2v_1$ on the sum.

Applying the inverse \eqref{eq:2d:closure-inv} also to
$\ket{\alpha(M)}=\ket{\beta(N)}$, we obtain
\[
\includegraphics[height=7.5em]{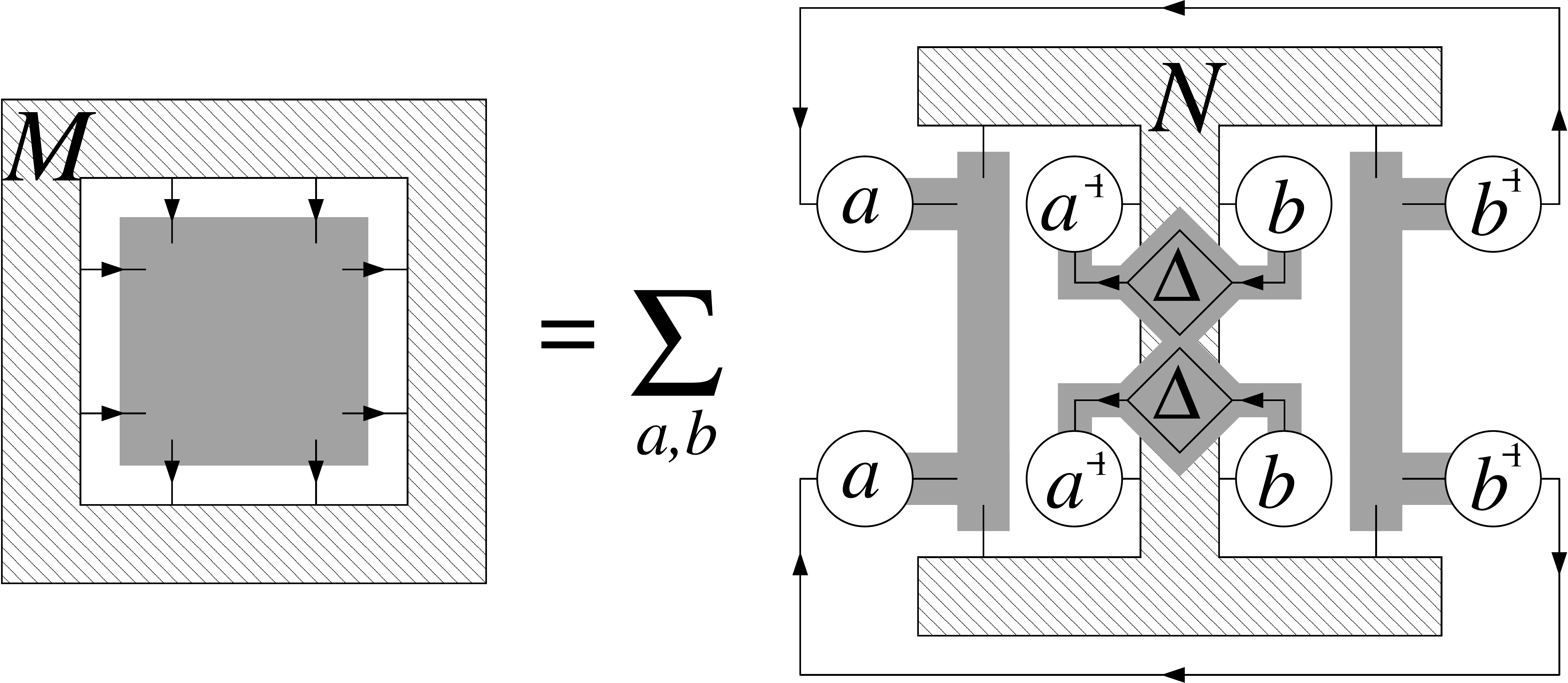}
\]
[note that the left and right block of \eqref{eq:2d:closure-inv} are the
inverses for the corresponding blocks in $\ket{\beta(N)}$], which proves
that $h_1=h_2=:h$ in \eqref{eq:2d:close-in-in}. Finally, the same for
$\ket{\alpha(M)}=\ket{\gamma(P)}$ yields $v_1=v_2=:g$, and the constraint
$v_2h_1=h_2v_1$ gives $gh=hg$.
\end{proof}

\begin{definition}
For a $G$--injective PEPS tensor $A$, and $g,h\in G$, define
\begin{equation}
    \label{eq:2d:peps-with-ug-uh}
\MPS{A|(g,h)} :=
\raisebox{-2.65em}{\includegraphics[height=6em]{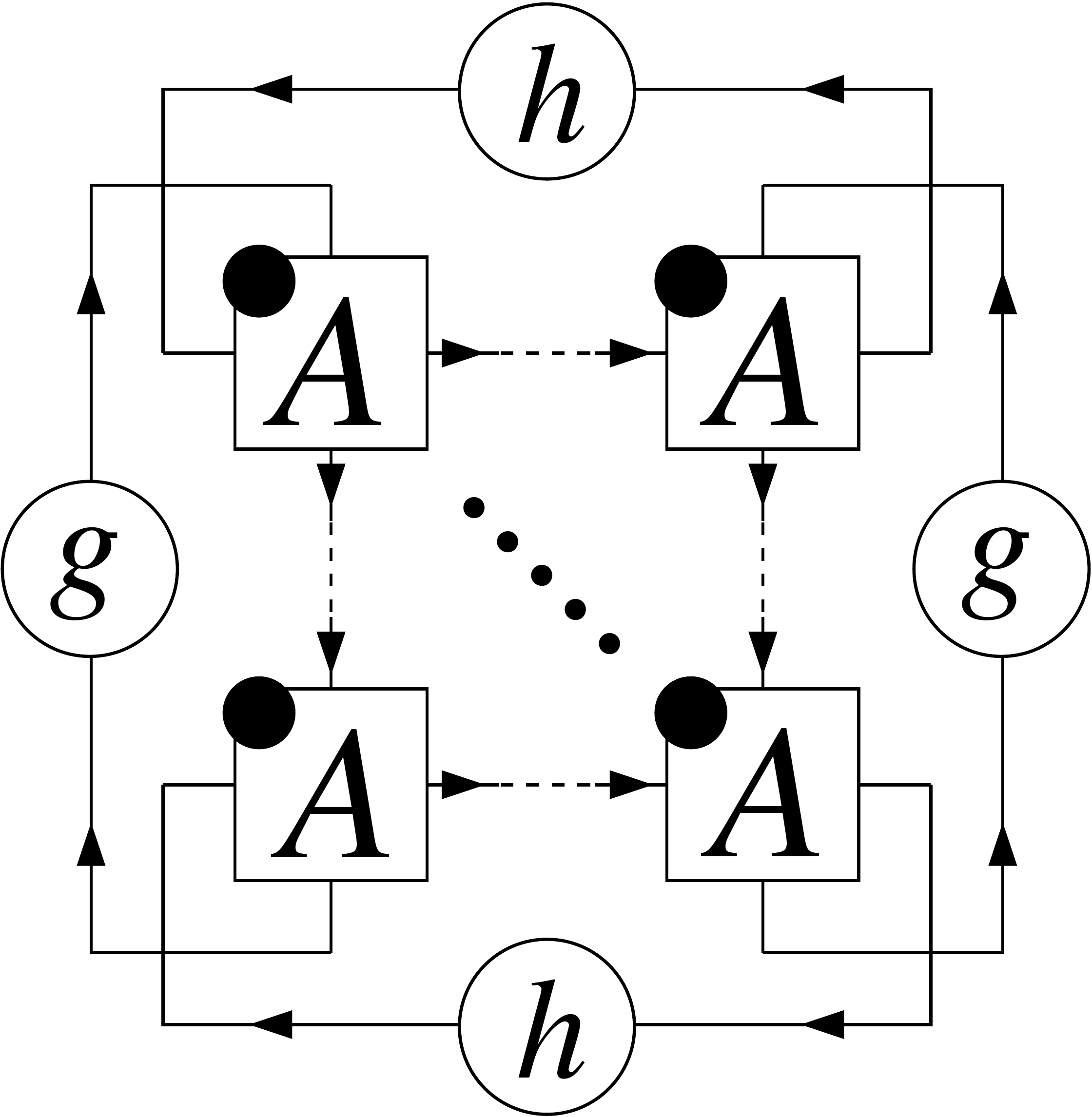}}
\end{equation}
to be the $L\times L$ PEPS built from the tensor $A$, with closures $g$
and $h$. (The representation of $g$ is a property of the bond and omitted
for brevity.)
\end{definition}

\begin{theorem}[Parent Hamiltonian]
    \label{thm:2d:parent-ham}
Let $A$ be $G$--injective,
\[
\mc S_{2\times 2}=\left\{
    \raisebox{-2.2em}{\includegraphics[height=5em]{figs3/close-out-out}}\
    \middle\vert M \right\}\ ,
\]
and define
\begin{equation}
\label{eq:2d:parentham-localterm}
h_{i,j}=(1-\Pi_{\mc S_{2\times 2}})\otimes \openone
\end{equation}
where the projector $\Pi_{\mc S_{2\times 2}}$ acts on the square
$\{i,i+1\}\times\{j,j+1\}$. Then,
\begin{equation}
\label{eq:2d:parentham}
H_\mathrm{par}=\sum_{i,j=1}^L h_{i,j}
\end{equation}
has a subspace of frustration free ground states spanned by the PEPS
$\MPS{A|(g,h)}$ with $(g,h)$-closed boundaries, where $gh=hg$.
\end{theorem}
\begin{proof}
The proof goes exactly as in one dimension: We divide the Hamiltonian in
four blocks with open boundary conditions (one which does not cover the
boundary, two which cover the horizontal or vertical boundary,
respectively, and one which covers both boundaries). For each of the
blocks, by virtue of the intersection property
(Theorem~\ref{thm:2d:intersection}) the ground state subspace is given by
one of the sets in~\eqref{eq:2d:closure-intersection}, and thus,
Theorem~\ref{thm:2d:closure} yields the desired result.
\end{proof}

\begin{definition}[Pair-conjugacy classes]
    \label{def:2d:pair-cc}
For $(g,h)\in G\times G$, $(g',h')\in G\times G$, define the equivalence
relation
\[
(g,h)\sim (g',h')\quad :\Leftrightarrow
    \exists x\in G:\ (g,h)=(xg'x^{-1},xh'x^{-1})\ .
\]
It divides $G\times G$ into disjoint equivalence classes
\[
C[(g,h)] =\{ (g',h')|(g,h)\sim   (g',h')\}
\]
which we call \emph{pair-conjugacy classes}.
To each pair-conjugacy class $C$, we define
\[
\MPS{A|C}:=\MPS{A|(g,h)}\quad (g,h)\in C\ .
\]
\end{definition}
Note that $\MPS{A|C}$ is well-defined, since the $g$-invariance of $A$
shows that $\MPS{A|(g,h)}=\MPS{A|(g',h')}$ for $(g,h)\sim(g',h')$.

\begin{theorem}[Structure of the ground state subspace]
    \label{thm:2d:gs-struct}
Let $C_1\,\dot\cup\,\cdots\,\dot\cup\, C_K\subset G\times G$ be the
pair-conjucacy classes of $G\times G$ for which $gh=hg$ for $(g,h)\in C_k$.
Then, the ground state subspace of
Theorem~\ref{thm:2d:parent-ham} is $K$-fold degenerate and spanned by the
states $\MPS{A|C_k}$.
\end{theorem}
\begin{proof}
The proof resembles the one-dimensional case:  As observed in
Definition~\ref{def:2d:pair-cc}, closures from the same pair-conjugacy class
describe the same state. It remains to show that the states corresponding
to different pair-conjugacy classes are linearly independent.
Let
\begin{equation}
    \label{eq:2d:linindep-init}
0=\sum\lambda_k\MPS{A|C_k}=\sum_k\lambda_k\MPS{A|(g_k,h_k)}\ ,
\end{equation}
with
$(g_k,h_k)\in C_k$ representatives of $C_k$. Denoting the $L\times L$
block of $A$'s by $C$, and applying the left inverse of $C$, we have that
\[
\raisebox{-2em}{\includegraphics[height=5em]{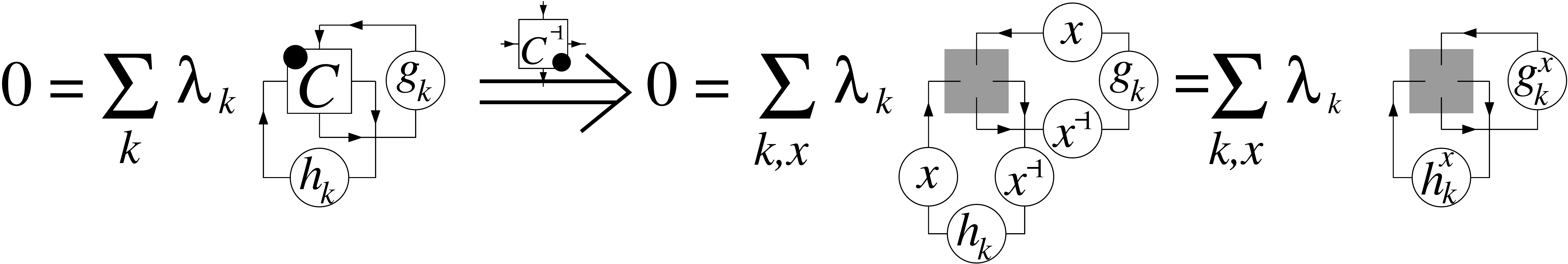}}\ ,
\]
with $(g_k^x,h_k^x)\in C_k$ other representatives of $C_k$.  This is,
\eqref{eq:2d:linindep-init} implies that $\sum\lambda_k U_{h_k^x}\otimes
U_{g_k^x}=0$, and since semi-regular representations are linearly
independent (Lemma~\ref{lemma:noninj:semireg-trace-ug-delta}), this gives
$\lambda_k=0\ \forall k$.
\end{proof}

It is straightforward to see that the number of different closures, i.e.,
the dimension of the ground state subspace, is equal to the number of
particle types of the quantum double model of $G$~\cite{kitaev:toriccode}:
A pair-conjugacy class $C\equiv C[(h,k)]$, $hk=kh$,  can be characterized
by specifying i)  the conjugacy class $C[h]$ of $h$, and ii) the conjugacy
class of the normalizer $N[h]=\{k:hk=kh\}$ of $h$ which contains $k$.
Since the number of conjugacy classes of $N[h]$ is equal to the number of
irreducible representations, this is equal to the number of particle types
in the $G$ quantum double~\cite{preskill:lecturenotes}.

Let us now give an intuitive explanation why the closures should be done using two
closed loops of identical unitaries $U_g^{\otimes N}$ and $U_h^{\otimes N}$,
and why one should require $gh=hg$. The key observation here (which will
be essential for the introduction of topological excitations later on) is
that strings of $U_g$'s can be deformed freely due to the $G$--invariance
of $A$.  To this end, consider a string of $U_g$'s in the lattice,
\[
\raisebox{-2.2em}{
\includegraphics[height=7.5em]{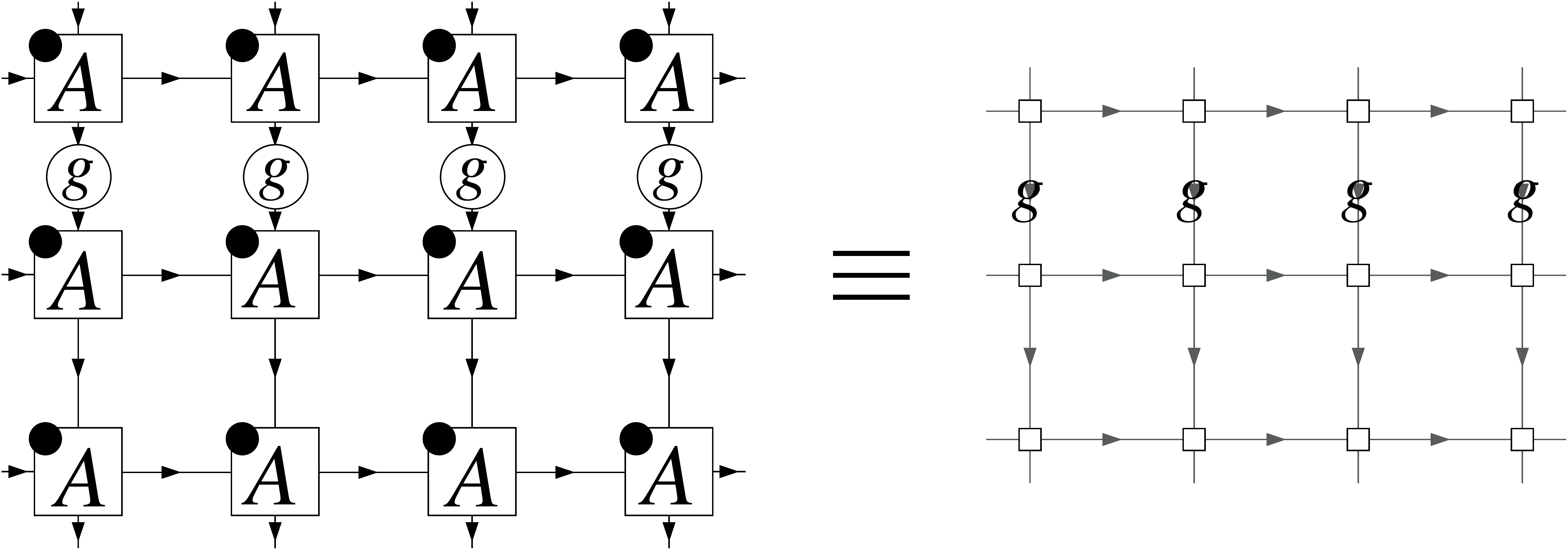}
} \ ,
\]
where on the right, we have introduced a new simplified notation which
allows for more compact diagrams in situations where we care only about the
pattern of $U_g$'s on the virtual level.
Using the $G$--invariance of the tensors, we can deform the string
continuously: e.g.,
\begin{equation}
        \label{eq:2d:move-strings}
\raisebox{-2.3em}{
\includegraphics[height=6em]{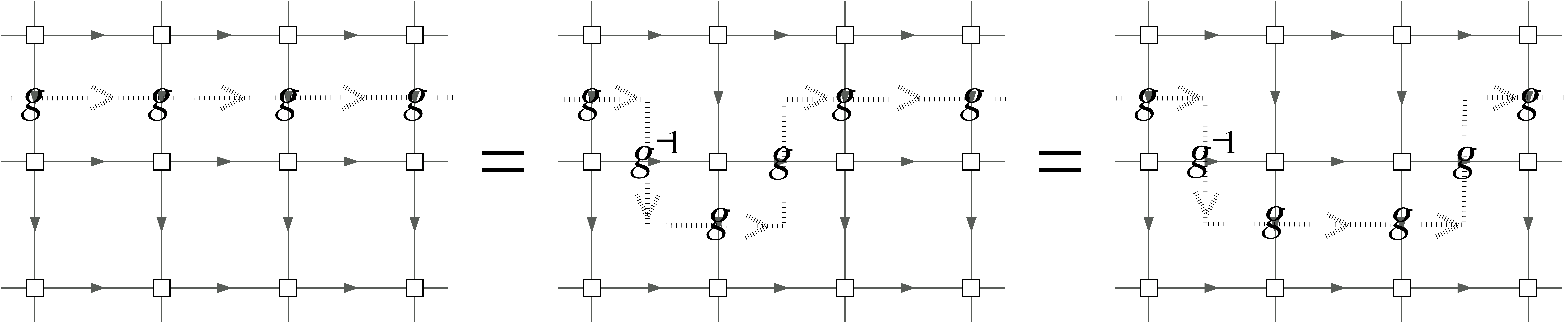}
}\ ,
\end{equation}
and so further.  Note that whether $g$ or $g^{-1}$ has to be used depends
on whether the (oriented) string of $U_g$'s crosses the (oriented) edges
from the right or from the left.

The possibility to arbitrarily deform strings of $g$'s and $h$'s, together
with the fact that $g$ and $h$ commute, i.e., the that the two strings
don't interact, implies that we can move the loops anywhere we want on
the torus. Thus, the closure
cannot be detected locally, and the $\MPS{A|(g,h)}$ are all ground states
of the Hamiltonian~\eqref{eq:2d:parentham}. (Note, however, that this does
not imply that the reduced density operators are the same, but only that
they all live within the same subspace $\mc S_2$. Indeed, in one dimension
the reduced operators can look different; e.g., the GHZ state MPS
$A^i=\ket{i}\bra{i}$, $i=0,1$ gives a Hamiltonian for which all
$\alpha\ket{0,\dots,0}+\beta\ket{1,\dots,1}$ are ground states.

\section{Isometric PEPS\label{sec:iso-PEPS}}

\subsection{Definition an basic properties}

We will now consider a subclass of $G$--injective PEPS, namely those for
which $\mc P(A)$ is an isometry (mapping the space of $G$--invariant
tensors unitarily to the range of $\mc P(A)$). Using
Observation~\ref{obs:surjective} which says that we can w.l.o.g.\ restrict
the physical system to the range of $\mc P(A)$, we obtain that a
$\mc P(A)$ is isometric if $\mc P(A)^{-1}=\mc P(A^\dagger)$, where the
dagger is with respect to the virtual levels:
\[
A=\sum_{i\alpha\beta\gamma\delta}A^i_{\alpha\beta\gamma\delta}
\ket{i}_p\ket{\alpha,\beta}\bra{\gamma,\delta}\quad
\rightarrow\quad
A^\dagger:= \sum_{i\alpha\beta\gamma\delta}{\bar{A}}^i_{\alpha\beta\gamma\delta}
\ket{i}_p\ket{\gamma,\delta}\bra{\alpha,\beta}\ .
\]
(Loosely speaking, the reason to use the dagger instead of the conjugation
is that the edges for $\mc P(A)$ and $\mc P(A)^{-1}$ have opposite
orientation.)

For clarity, we will illustrate proofs for the one-dimensional case where
appropriate.

\begin{definition}
    \label{def:iso:isopeps}
A $G$--injective PEPS is called $G$--isometric if $U_g\equiv L_g$ is the
left-regular representation, and if $\mc P(A)^{-1}=\mc P(A^\dagger$), or
equivalently, if $\mc P(A)$ restricted to its domain and range is unitary.
\end{definition}

Here, the left-regular representation $L_g:\mathbb C^{|G|}\rightarrow \mathbb
C^{|G|}$ acts as $L_g\ket{h}=\ket{gh}$.
The following lemma explains why we reqire the restriction to the
regular representation.

\begin{lemma}[Stability of isometry under concatenation]
        \label{lemma:iso:iso-stable-under-concat}
$G$--isometry of tensors is stable under concatenation.
\end{lemma}

\begin{proof}
$U_g=L_g  \Rightarrow \Delta=\tfrac{1}{|G|}\openone$,
and thus the left-inverse of the concatenated tensor in
\eqref{eq:noninj:linv} is $(B^j)^\dagger(A^i)^\dagger=(A^iB^j)^\dagger$ --
note that the ordering of operators (as indicated by the arrows) is
reversed for the layer with the inverse $P(A)^{-1}$.
\end{proof}

Note that we do \emph{not} block indices at this stage, as this would
change the representation. We will show how this can be done in
the section on renormalization.

\subsection{Equivalence of virtual and physical system}

The importance of $G$--isometric PEPS lies in the fact that the way the
physical subspace and the virtual subspace (restricted to the
$G$--invariant subspace!) are connected corresponds to a unitary, i.e., a
physical operation. This implies that any unitary on the virtual level
which leaves the $G$--invariant subspace invariant can be implemented by
means of a unitary acting on the physical system.
\begin{lemma}
        \label{lemma:iso:sym-virt-can-be-done-on-phys}
Let $A$ be a $G$--isometric PEPS. Then the unitary transformations $V$
which can be implemented on the virtual level by unitarily acting on the
physical system are exactly those which commute with the symmetry, i.e.,
those for which
\begin{equation}
    \label{eq:iso:V-comm-sym}
    \raisebox{-.5em}{\includegraphics[height=2.5em]{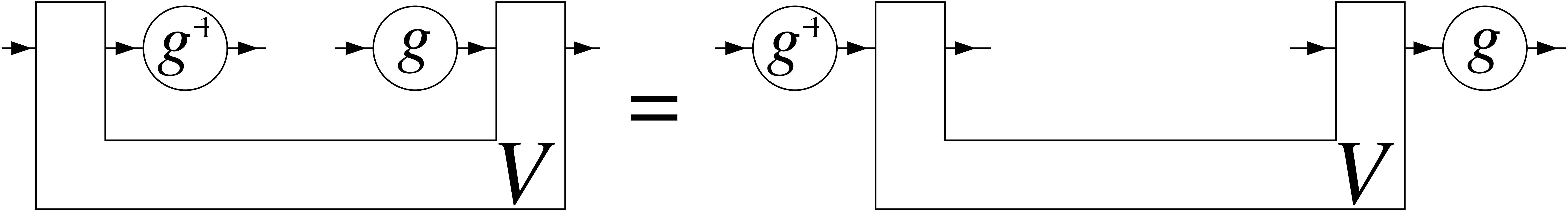}}\ .
\end{equation}
\end{lemma}
\begin{proof}
Let $V$ be a unitary which satisfies \eqref{eq:iso:V-comm-sym}. Then,
$V$ acts unitarily on the $G$--invariant subspace. Thus, the physical
operation
\[
\includegraphics[height=4.8em]{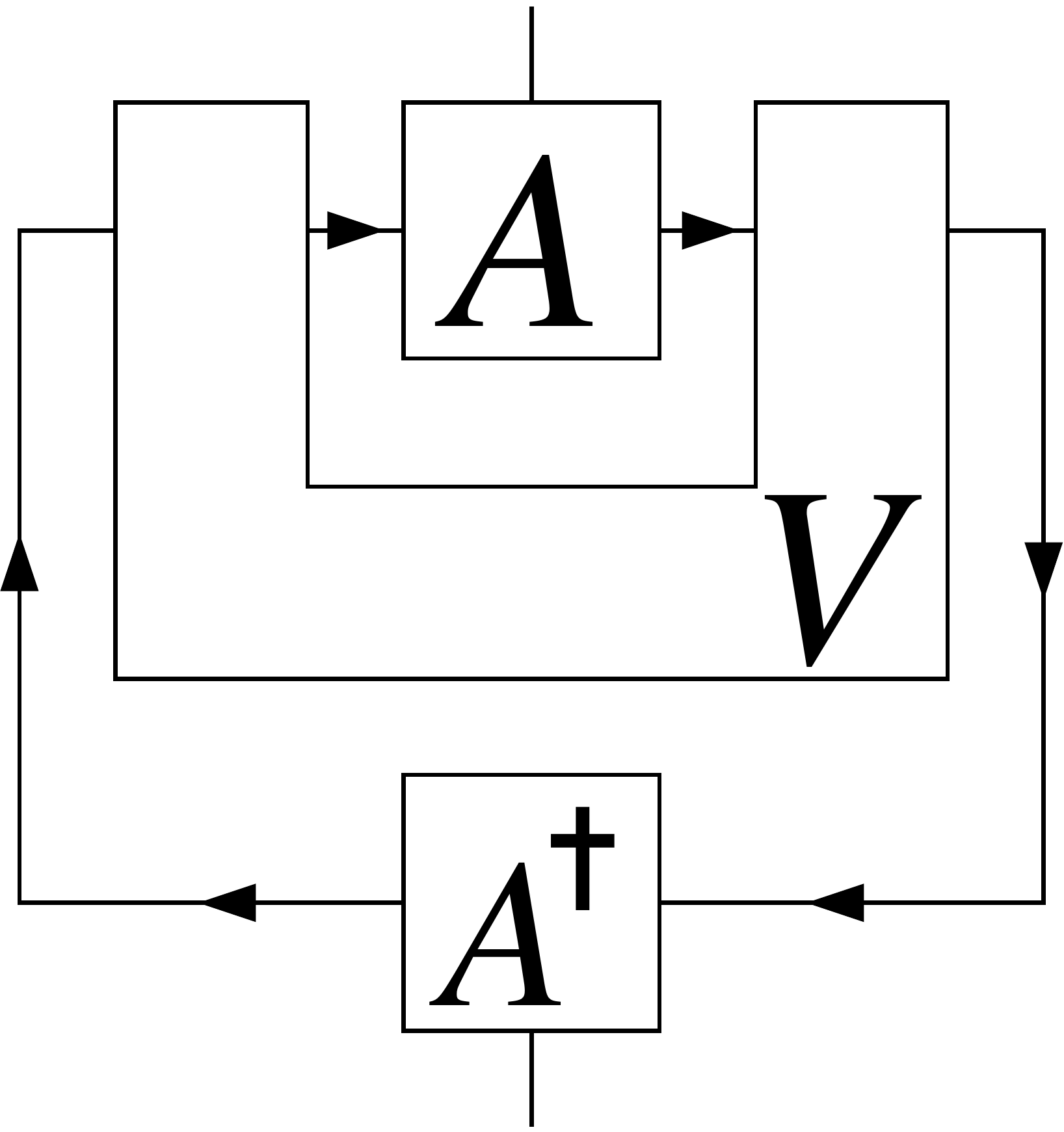}
\]
is unitary, and
\[
\raisebox{-2.2em}{\includegraphics[height=6em]{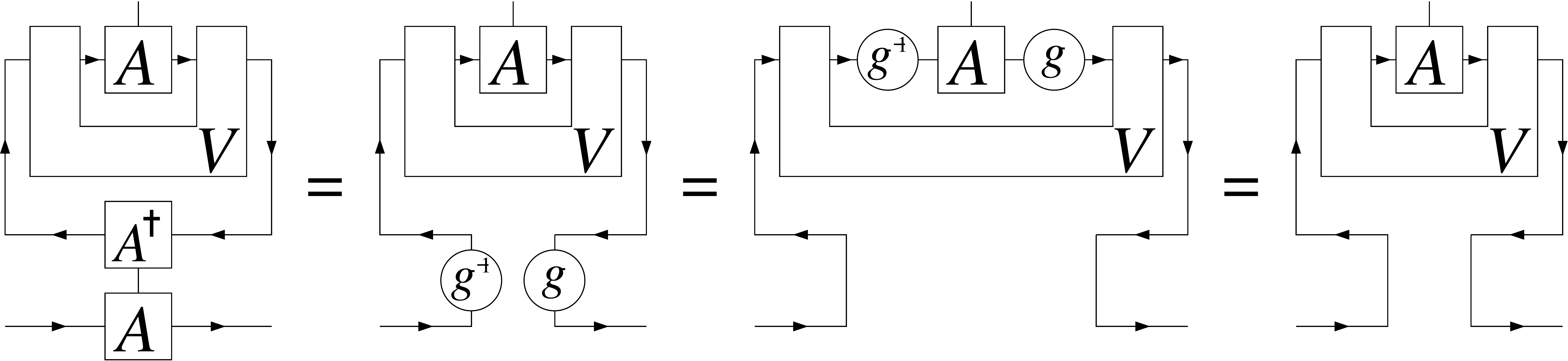}}\
,
\]
i.e., it acts as $V$ on the virtual indices.
Conversely, any unitary operation $U$ on the physical level results in a
$G$--invariant unitary $V=\mc P(A)^{-1}U\mc P(A)$ on the virtual level by
virtue of
\[
\raisebox{-1.8em}{\includegraphics[height=6em]{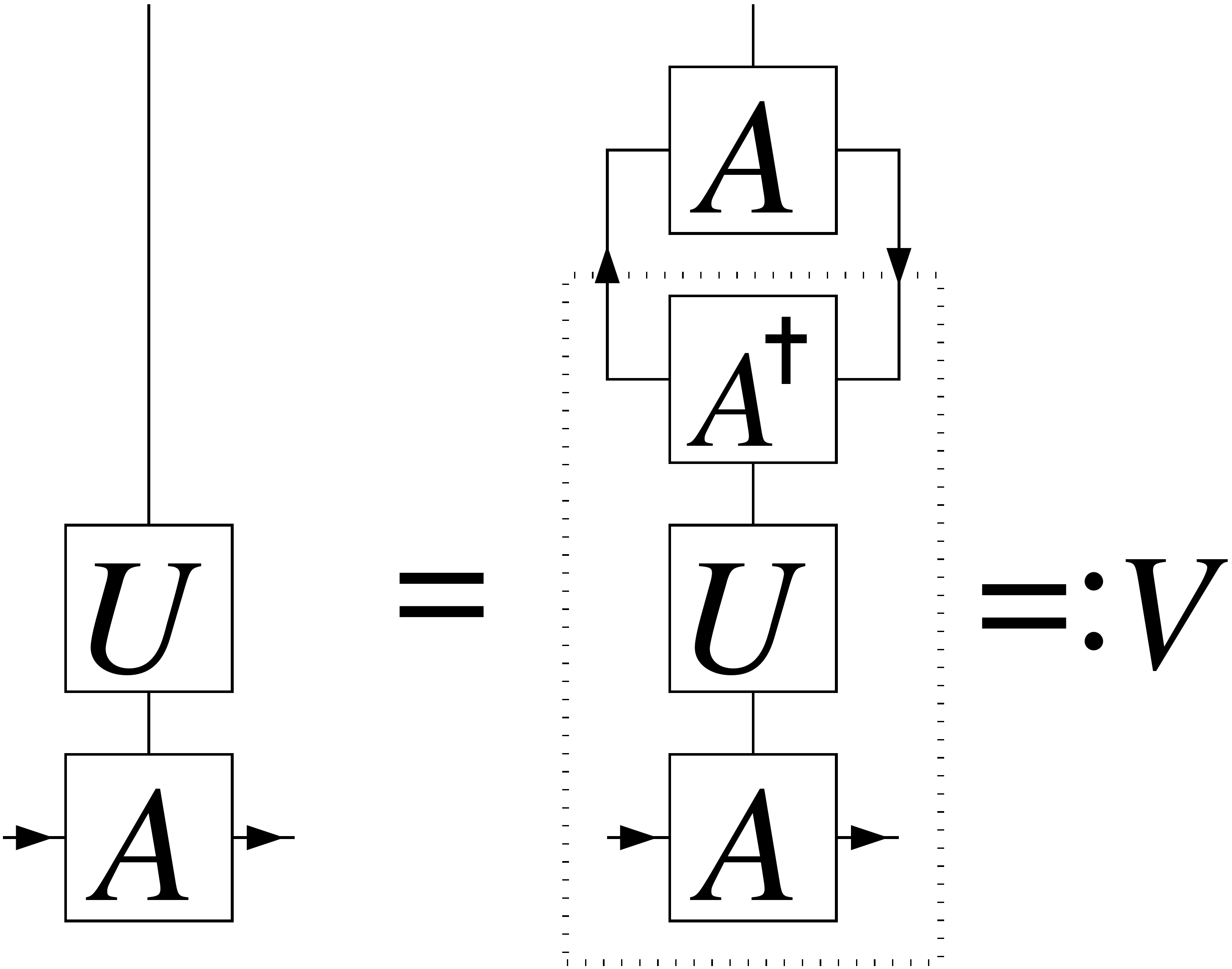}}\ ,
\]
where we have used \eqref{eq:inj:rightinv-always}.
\end{proof}

A particularly appealing perspective on isometric PEPS is the following.
\begin{observation}
    \label{obs:iso:accessible-virt}
Every $G$--isometric PEPS tensor $A$ is isomorphic to
\begin{equation}
        \label{eq:iso:iso-accessbonds}
\raisebox{-1.9em}{\includegraphics[height=5em]{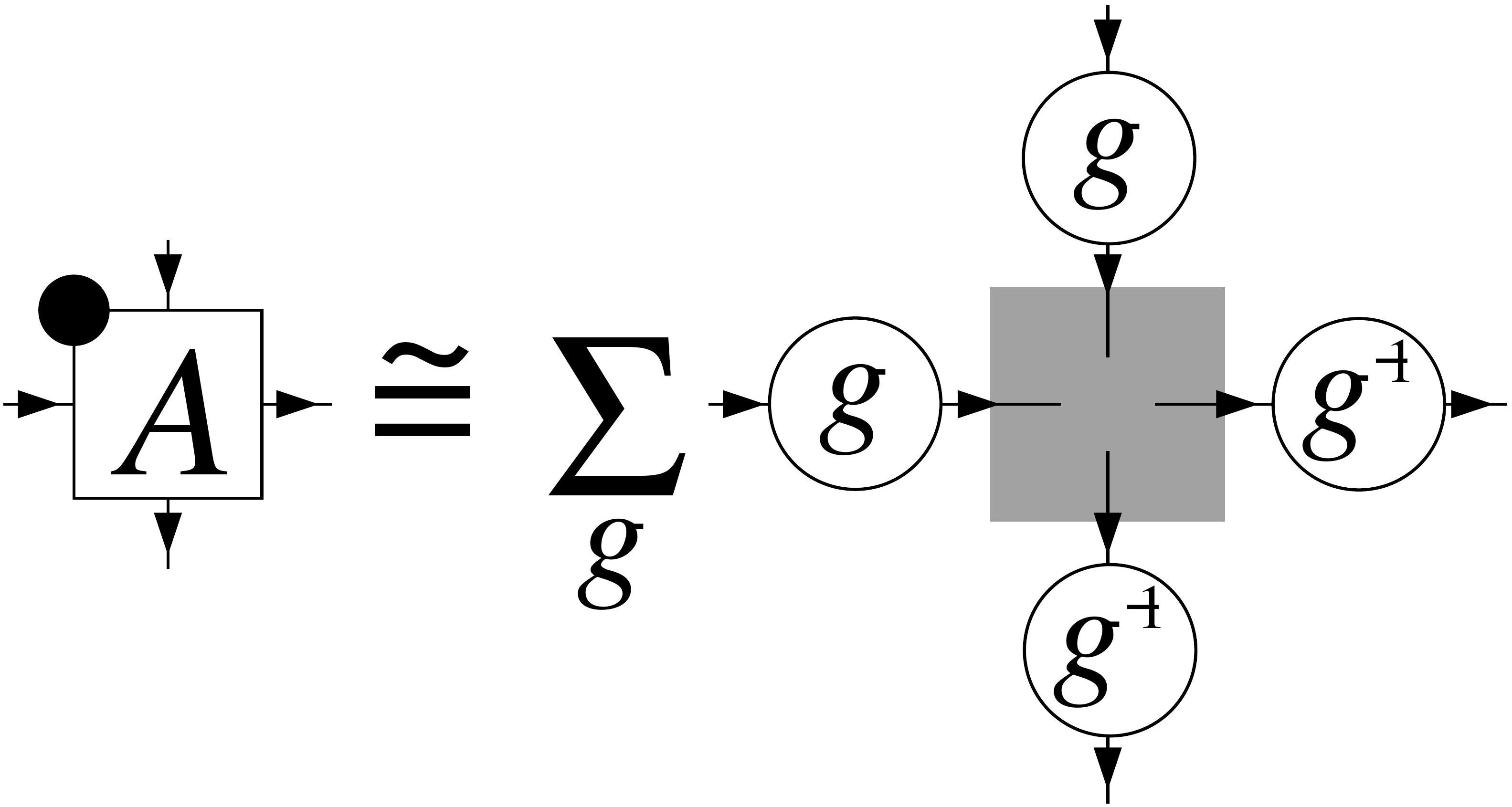}}
\quad ,
\end{equation}
by acting unitarily on the physical system.\footnote{Note that this
decomposes the physical level into four systems, represented by kets {\em
or} bras depending on the directions of the arrows.}  Moreover, this
property is stable under concatenation. E.g.,
\begin{equation}
        \label{eq:iso:iso-accessbonds-pair}
\raisebox{-1.9em}{\includegraphics[height=5em]{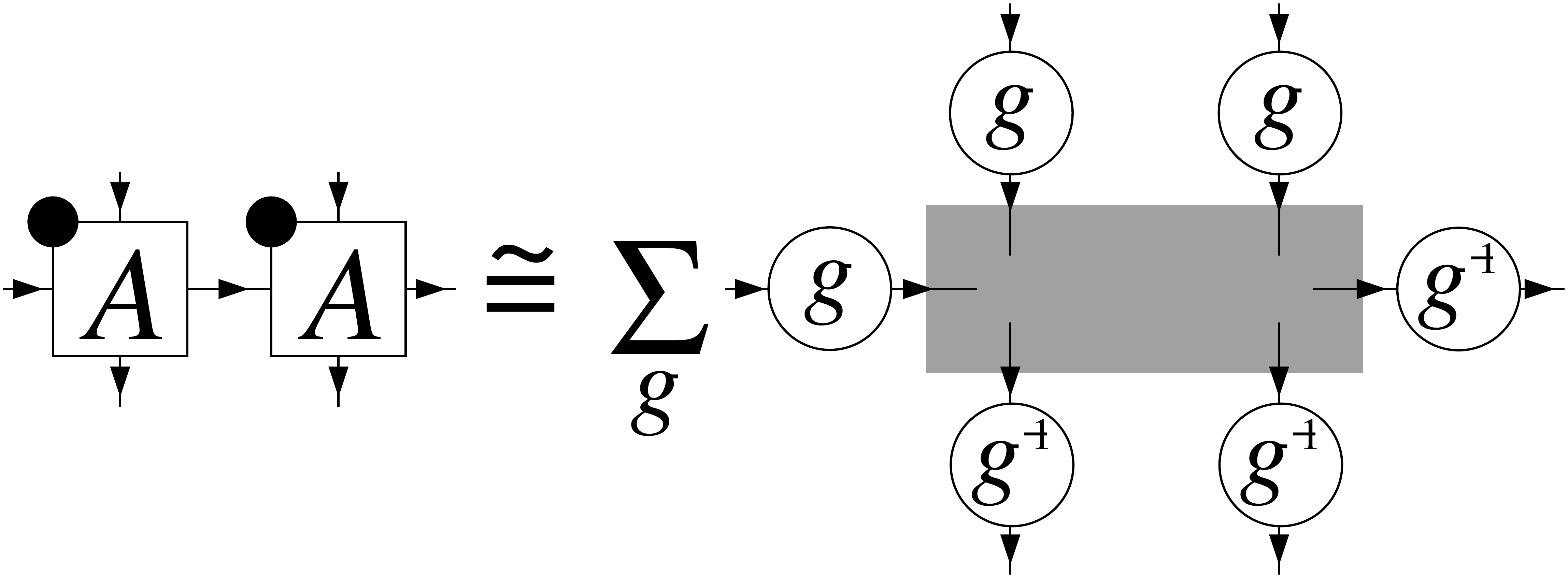}}
\quad ,
\end{equation}
up to discarding unneeded local degrees of freedom (i.e., restricting the
physical system to the subspace actually used, cf.~Observation~\ref{obs:surjective}).
\end{observation}
Eq.~\eqref{eq:iso:iso-accessbonds} follows directly from applying the left
inverse~\eqref{eq:2d-linv}, using that it is an isometry. While stability
under concatenation can be inferred from
Lemma~\ref{lemma:iso:iso-stable-under-concat} (``Isometry is stable under
concatenation''), it is instructive to see how
\eqref{eq:iso:iso-accessbonds-pair} can be obtained from two tensors of
the form \eqref{eq:iso:iso-accessbonds} in a reversible way. Concatenating
two tensors of the form \eqref{eq:iso:iso-accessbonds}, we face a state of
the form
\[
\raisebox{-2.5em}{\includegraphics[height=6.2em]{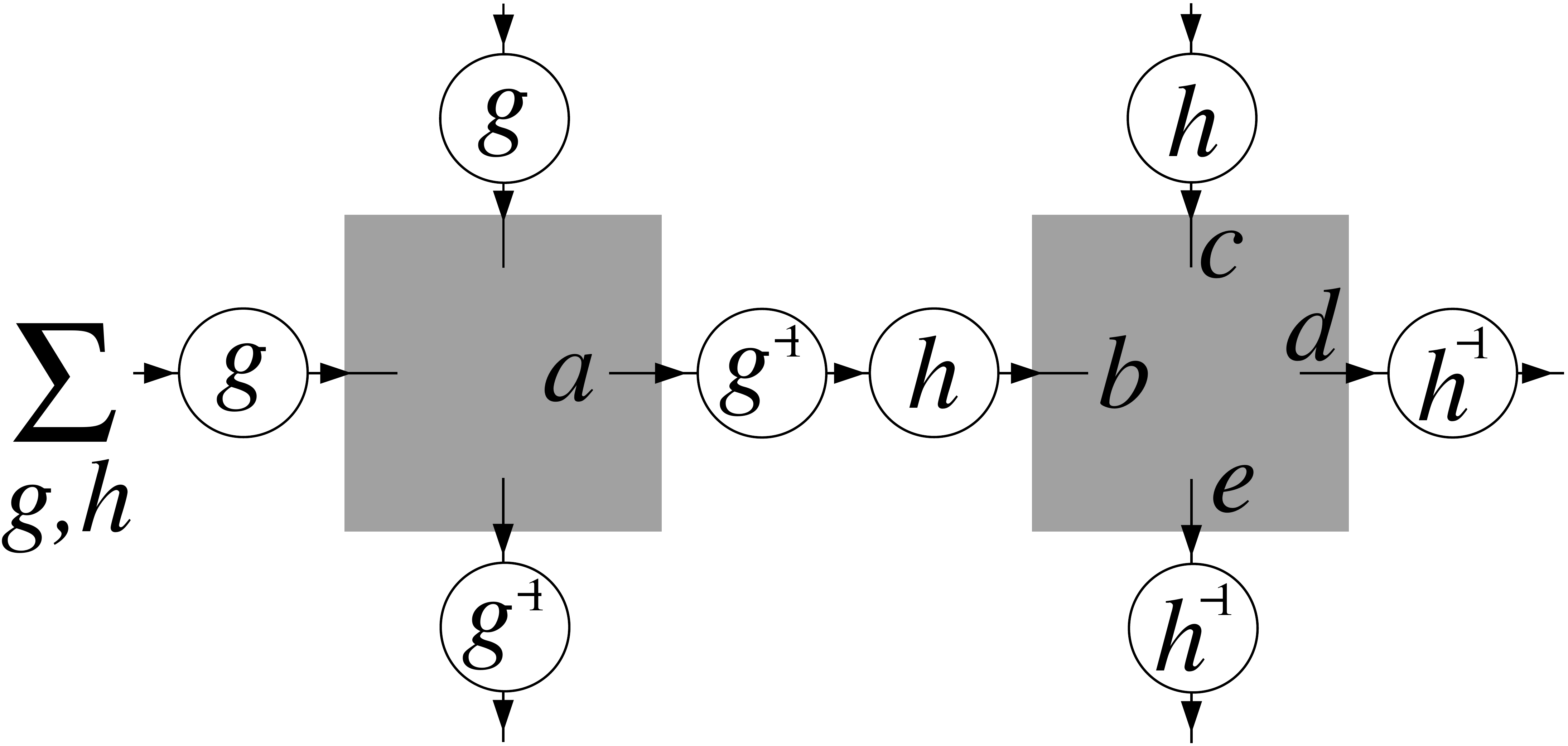}}
\quad .
\]
Then, we measure the two sites labeled $a$ and $b$ in the standard
basis, and thus infer $g^{-1}h=ab^{-1}$.
(We will use $a$ $b$, etc.\ for both the site and its state.)
Now, apply $U_{g^{-1}h}$ to sites $c$, $d$, and $e$.
This maps $U_h\mapsto U_hU_{g^{-1}h}^\dagger=U_g$ (on site $c$) or
$U_{h^{-1}}\mapsto U_{g^{-1}h}U_{h^{-1}}=U_{g^{-1}}$ (on sites $d$ and
$e$, as they are bras), and thus yields a tensor of the form
\eqref{eq:iso:iso-accessbonds-pair}.

Do we loose reversibility in this process, given that we measure? No: One
way to see this is to replace the measurement followed by $U_{g^{-1}h}$ by a
unitary controlled by the state of $a$ and $b$, which keeps the process
reversible.  In particular, the state of $a$ and $b$ is separable from the
rest after this operation (it is the only state which still depends
on $h$). This, however, implies that we could equally well discard (i.e.,
measure) it, and rebuild it if we need to reverse the mapping.

\subsection{Blocking and renormalization\label{sec:RG}}

An aspect which has to be taken care of separately for $G$--isometric PEPS
is blocking. Recall that the proofs for the 2D case involved taking a
block of $G$--injective tensors $A$ with representation $U_g$
and regarding it as a new $G$--injective $4$-index tensor $B$ with representation
$V_g=U_g\otimes\cdots\otimes U_g$. This was possible since a tensor
product of semi-regular representations is again semi-regular. In
contrast, the regular case needs some special attention.

The main observation to deal with blocking in the regular case is to
observe that for the left-regular representation $L_g$,
\[
(L_g)^{\otimes N}\otimes (L_g^\dagger)^{\otimes M} \cong L_g \otimes
\openone^{\otimes (N+M-1)}\ .
\]
(This follows since both representations have the same character,
which means they are isomorphic.)
This implies that any number of bonds with a regular representation can be
mapped -- on the virtual level -- to only one bond carrying the
$L_g$--symmetry, whereas all other bonds have no symmetry.

Let us now first show that this isomorphism can be implemented by a
physical operation, before discussing what this implies for the structure
of the PEPS. For simplicity, we restrict to the case of two bonds,
$L_g\otimes L_g \cong L_g\otimes\openone$. The isomorphism is implemented
by the map
\[
T=\sum_{a,b}\ket{a,ab}\bra{a,b}\ :\quad T^\dagger (L_g\otimes
L_g)T=L_g\otimes \openone\ .
\]
Let us now consider two tensors $A$ and $B$ adjacent to the $L_g\otimes
L_g$ bond. By virtue of Observation \ref{obs:iso:accessible-virt}, and
grouping all other indices, they are isomorphic to
\[
\raisebox{-1em}{
\includegraphics[height=3em]{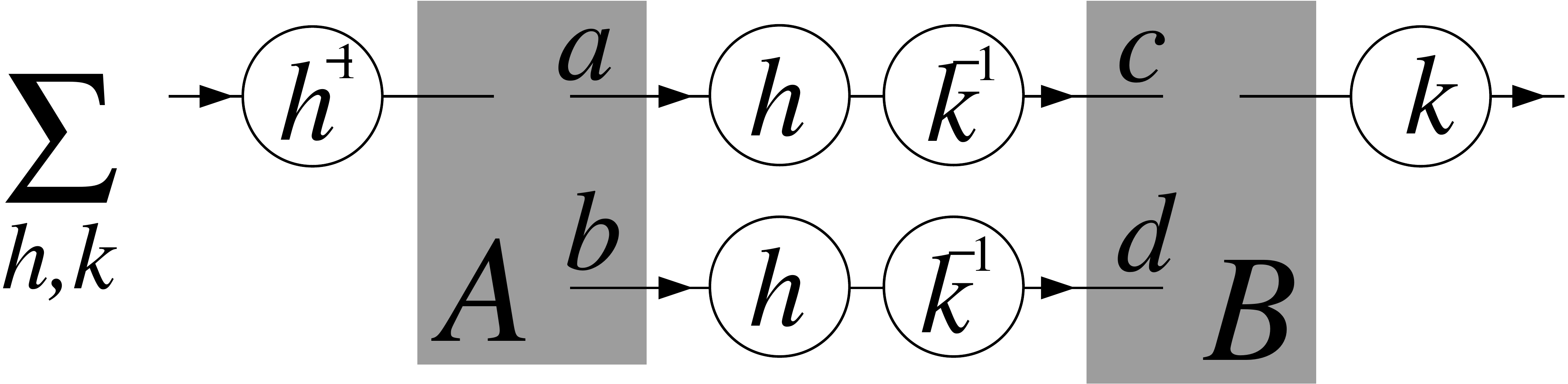}
}\ .
\]
Now, we  apply the unitary $T$ ($T^\dagger$) to the pair of
indices labelled $a$, $b$ ($c$, $d$): This physical transformation changes
the network to tensors $A'$, $B'$ with symmetry $L_g\otimes\openone$:
\[
\raisebox{-1em}{
\includegraphics[height=3em]{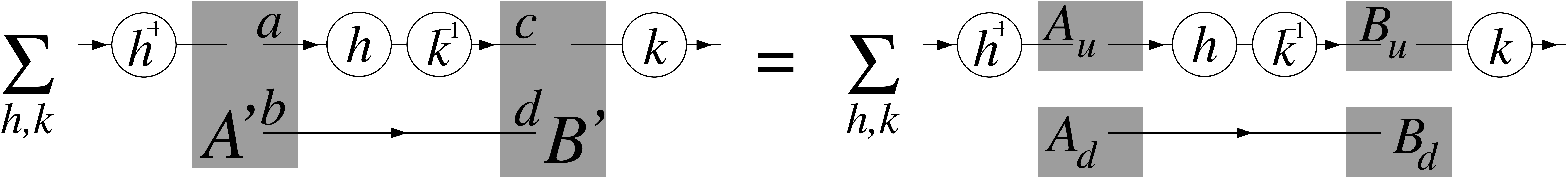}
}\ .
\]
As illustrated on the r.h.s., this means that $A'$ and $B'$ are actually a
tensor product of two independent tensors: $A\cong A_u\otimes A_d$ and
$B\cong B_u\otimes B_d$. Whereas the upper tensors have an
$L_g$--symmetry, the lower tensors have no symmetry, and form a maximally
entangled pair between adjacent sites. Thus, these degrees of freedom do
not contribute to the non-local properties of the PEPS, and can thus be
discarded.

\begin{observation}[Blocking indices preserves $G$--isometry]
Blocking indices of $G$--isometric tensors gives a new $G$--isometric
tensor, up to local isomorphisms and discarding maximally entangled pairs
between adjacent sites.
\end{observation}

Let us now show that this procedure establishes a renormalization scheme
for PEPS. The idea of renormalization is to study non-local properties of
quantum many-body states by grouping several sites into one, identifying
and discarding ``irrelevant'' (i.e., local) degrees of freedom, and
iterating until convergence is reached.

Consider now a $G$--isometric PEPS with tensor $A$ (and symmetry $L_g$),
and consider a block of $2\times 2$ tensors. The blocked tensor has then
symmetry $L_g\otimes L_g$, which by virtue of the above procedure
can be mapped to a $L_g\otimes \openone$ symmetry by local operations:
\[
\includegraphics[height=5.5em]{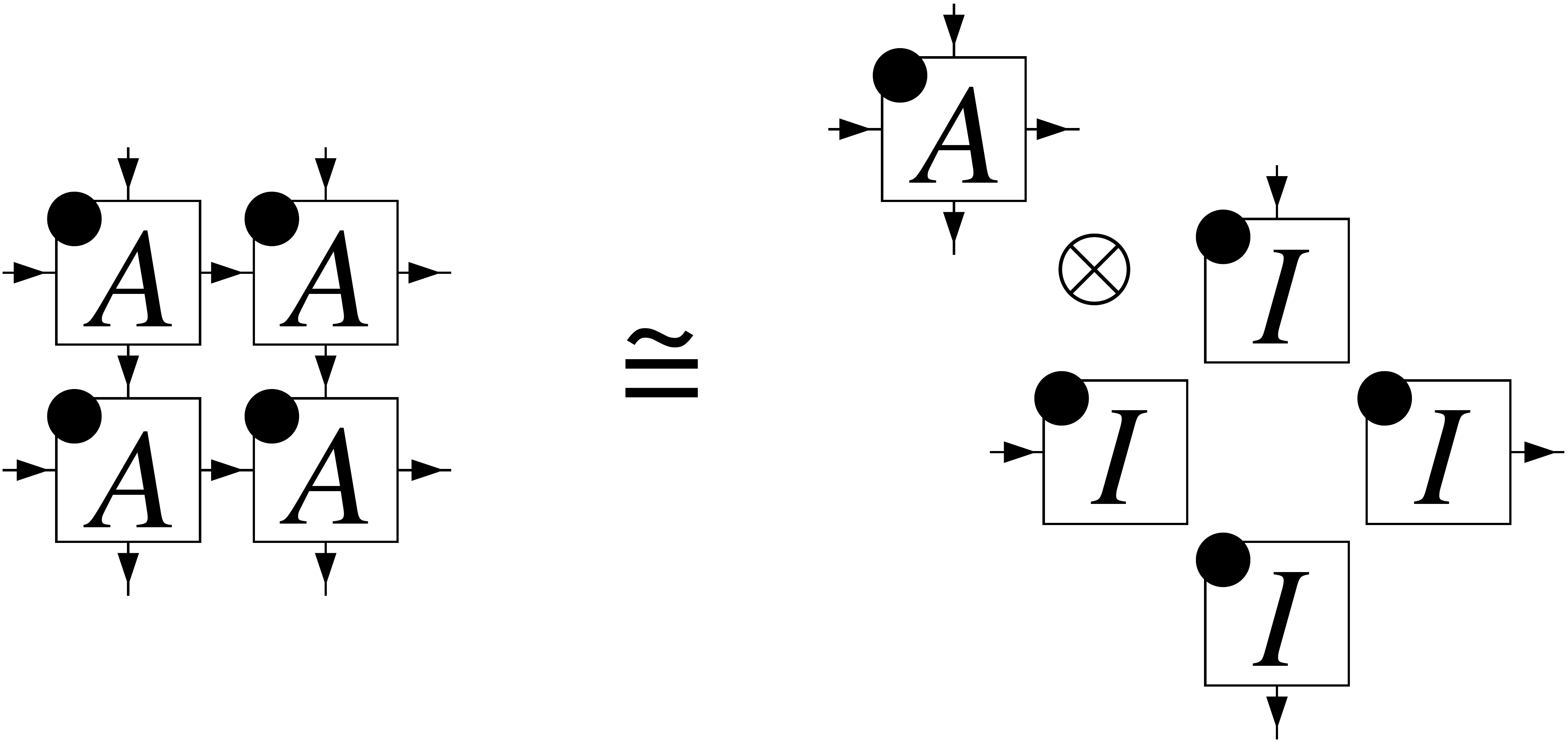}
\]
Here, the $I$ tensors are isometric tensors with no symmetry, i.e., they
map the virtual system one-to-one to the physical system.  Thus, by local
unitaries on $2\times 2$ blocks, the $G$--invariant PEPS with tensor $A$
has been transformed into a tensor product of the following states: First,
the same PEPS described by the tensor $A$, but on a coarse-grained lattice, and
second, the state given by the identity tensors $I$ -- a tensor product of
maximally entangled states between neighboring blocks, which can be
discarded if only looking for non-local properties. This results in the
following observation.

\begin{observation}
$G$--isometric PEPS are renormalization group fixed points.\footnote{
Since the regular representation is the only faithful one verifying
$U_g\otimes U_g\cong U_g\otimes \openone$, $G$--isometric PEPS seem to be
the unique fixed points of this type of renormalization procedure.
}  
\end{observation}

A similar procedure should also apply to the case of isometric tensors
with a non-regular representation $U_g$ of the symmetry, or even to
non-isometric PEPS: (Approximate)
decompositions of the tensorized symmetry $U_g^{\otimes k}$ will yield
(approximate) RG schemes for the PEPS. Another application of this RG
scheme will be illustrated in the examples section: It can be used to get
rid of extra symmetries in the original tensors which are of local nature
and vanish after a step of blocking.

\subsection{Structure of the ground state subspace}

We will now study the properties of PEPS with a $(g,h)$-closure,
Eq.~\eqref{eq:2d:peps-with-ug-uh}, which appear as ground states of the 2D
parent Hamiltonian~\eqref{eq:2d:parentham}, for the case of $G$--isometric
PEPS. We want to understand three things: First, is it possible to map
between arbitrary states in
the ground state subspace by acting only on a restricted region? Second,
are different ground states locally distinguishable? Third, what is the
entropy of a continguous block of spins?

All these questions can be addressed by considering the following
scenario with $G$--isometric $B$ and $D$:
\begin{equation}
        \label{eq:iso:L-shape-scenario}
\raisebox{-2.2em}{
    \includegraphics[height=5.5em]{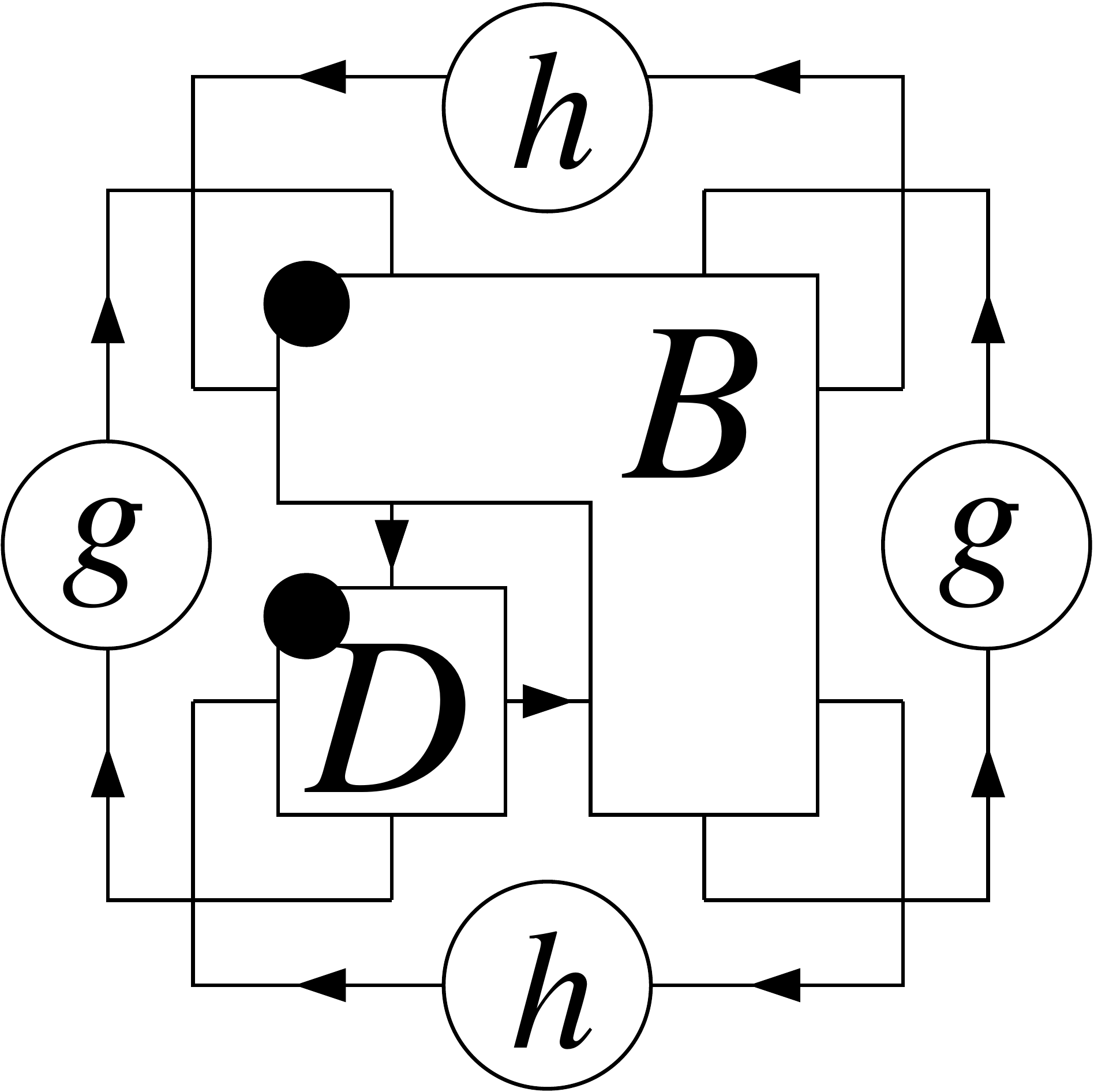}
}\ ,
\end{equation}
which by observation \ref{obs:iso:accessible-virt} is locally (where local
refers to the regions $B$ and $D$) isomorphic to
\begin{equation}
        \label{eq:iso:L-shape-scenario-accessible}
\raisebox{-5em}{
    \includegraphics[height=11em]{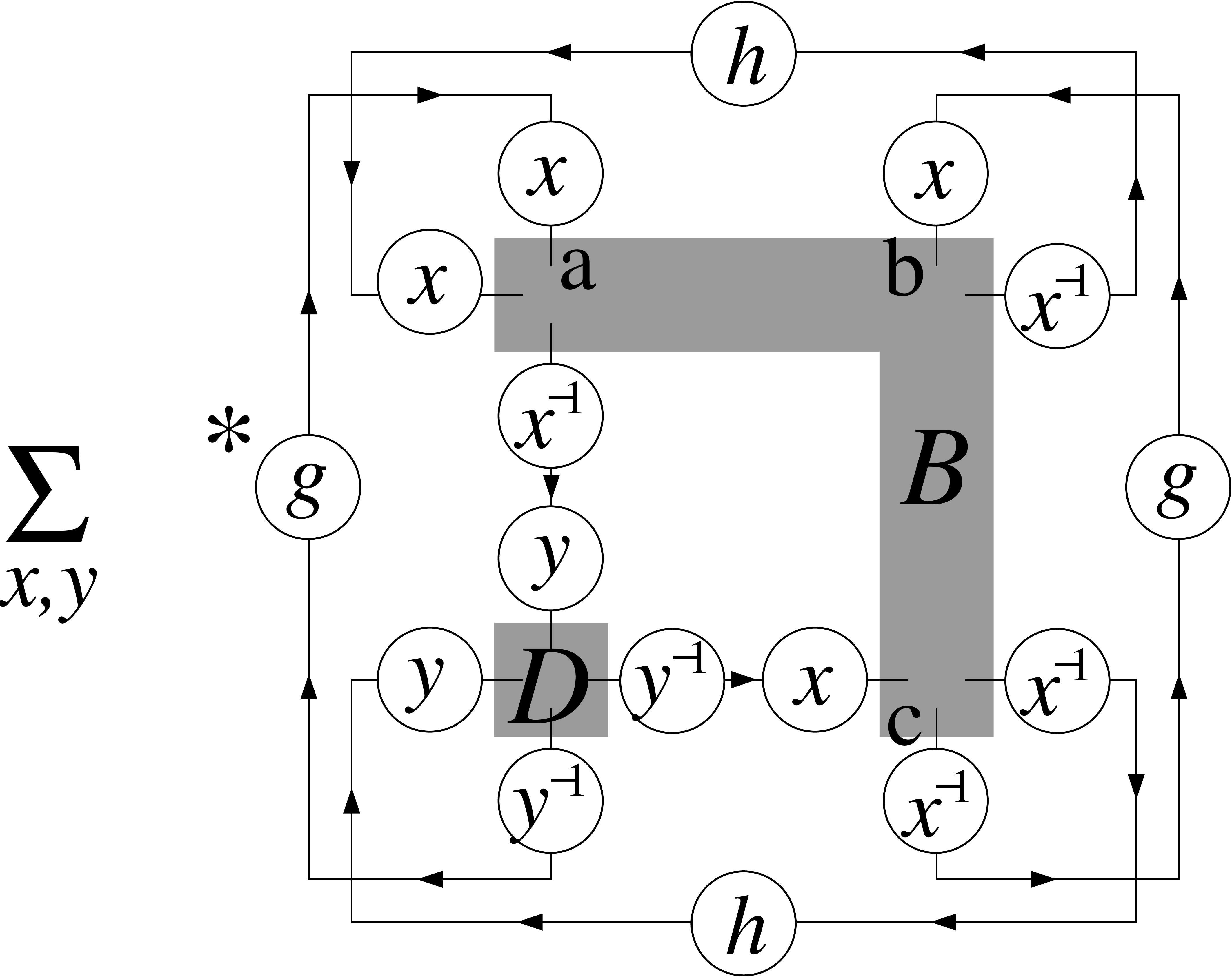}
}\qquad .
\end{equation}
When dividing a large lattice in two blocks as above, we need to
include an RG step in which we separate out $\partial D-4$ maximally
entangled states between $B$ and $D$; here, the boundary $\partial D$ is
the number of bonds crossing the boundary of $D$ before the RG step.

We will now devise a unitary transformation acting only on the $B$ region,
which will decouple $h$ and $g$ from $D$. To this end, consider the sites
marked $a$, $b$, and $c$ above. Since all group transformations are left
regular, we have $c=x^{-1}gxb$, and thus $cb^{-1}=x^{-1}gx$:
Having access to $b$ and $c$ actually gives us access to $x^{-1}gx$. Thus,
a unitary transform
\begin{equation}
        \label{eq:iso:tweezer-duplicate-operation}
\ket{c}\bra{b}\bra{a}\mapsto \ket{c}\bra{b}\bra{cb^{-1}a}=
\ket{x^{-1}g^{-1}xa}\ket{b}\bra{c}
\end{equation}
will remove the left $g$ in \eqref{eq:iso:L-shape-scenario-accessible}
(the one marked *) from the boundary. A
corresponding transformation allows us to remove the lower $h$,
leaving us with
\begin{equation}
    \label{eq:iso:block-disent-decomposition}
\raisebox{-5em}{
\includegraphics[height=11em]{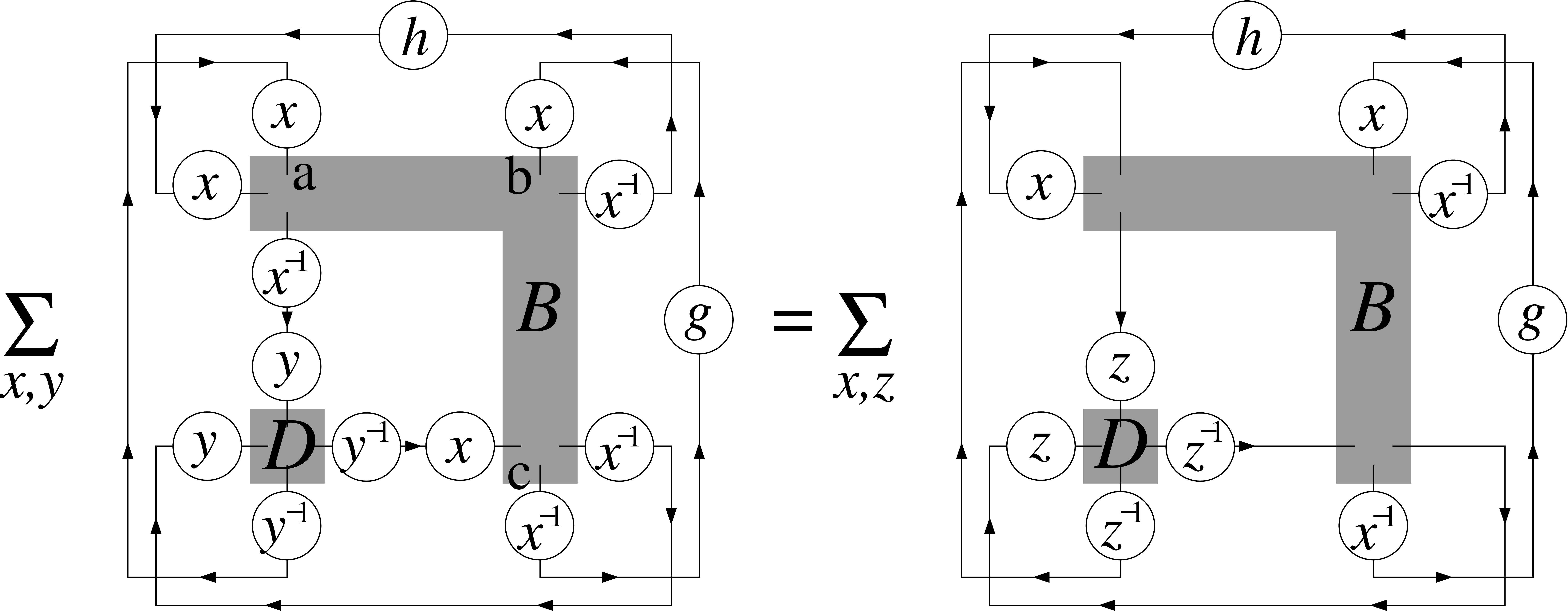}
}\quad ,
\end{equation}
where in the second step we have substituted $z:=x^{-1}y$. Now, however,
$g$, $h$, and $x$ only act on degrees of freedom
which are completely inside the
region $B$, while the degrees of freedom of $D$ are equal to the
corresponding ones of $B$, up to a constant shift $z$.
In formulas, the state above can be written (up to isomorphism) as
\[
\sum_z \ket{\mu(z)}^{\otimes 4}_{BD}\otimes\ket{\zeta(g,h)}_B\ ,
\]
with
\begin{align*}
\ket{\mu(z)}_{BD} & =  \sum_{b} \ket{zb}_B\ket{b}_D \ ,  \\
\ket{\zeta(g,h)}_{B} &= \sum_{x,p,q}
\ket{x^{-1}gxp}_B\ket{p}_B\ket{x^{-1}hxq}_B\ket{q}_B\ .
\end{align*}
Note that additionally, there are $\partial D-4$ maximally entangled
states between $B$ and $D$ which have been separated in the RG step.

Equation~\eqref{eq:iso:block-disent-decomposition} has several immediate
consequences.

\begin{theorem}[Local equivalence of ground states]
For $G$--isometric PEPS, it is possible to unitarily transform between any
two states in the ground state subspace $\sum_C\lambda_{C}\MPS{M|C}$ of
Theorem~\ref{thm:2d:gs-struct} by acting only on two stripes (each
of width one) wrapping around the torus.
\end{theorem}
\begin{proof}
By fixing two such stripes,  the PEPS is partitioned as in
\eqref{eq:iso:L-shape-scenario}, with the stripes labelled $B$ and the
(topologically trivial) rest $D$. Then, decomposition
\eqref{eq:iso:block-disent-decomposition} shows that the states
$\MPS{M|(g,h)}$ spanning the ground state manifold differ only within $B$,
allowing to unitarily transform between any two states
by acting only on $B$.
\end{proof}

\begin{corollary}[Local indistinguishability of ground states]
For $G$--isometric PEPS, the states in the ground state subspace of
Theorem~\ref{thm:2d:gs-struct} cannot be distinguished by local
operations, i.e., those which act only on a topologically trivial region.
\end{corollary}
\begin{proof}
This follows immediately by using that for any topologically trivial
region $D$, it is possible to transform between any two states
in the ground state subspace  by acting with a unitary on its complement,
$B$.
\end{proof}

\begin{theorem}[Topological entropy of $G$--isometric states]
For a $G$--isometric PEPS, the reduced density operator
$\rho_D$ of any topologically trivial block $D$  has
rank $|G|^{\partial D-1}$
and a flat spectrum. Here, the length $\partial D$ of the boundary
of $D$  is defined as the number of bonds crossing the boundary.
Thus, both the von Neumann entropy $S$ and all R\'enyi entropies
$S_\alpha$ of $\rho_D$ are
\[
S(\rho_D)=S_\alpha(\rho_D)=\log|G|(\partial D)-\log |G|\ ,
\]
where $\log |G|$ is the topological correction to the area
law~\cite{kitaev:topological-entropy,levin:topological-entropy,flammia:stringnet-renyientropy}.
\end{theorem}
\begin{proof}
We can again use the partitioning \eqref{eq:iso:L-shape-scenario}:
We can then infer from \eqref{eq:iso:block-disent-decomposition} that
(after unitary transformations on $B$ and $D$, and discarding the
unentangled part of $B$), the state is of the form $\ket{\mu(z)}^{\otimes
4}$, with
\[
\ket{\mu(z)}=\sum_b\ket{zb}_B\ket{b}_D\ .
\]
Additionally, we also have $(\partial D-4)$ maximally entangled states
$\ket{\mu(1)}$ from the initial RG step.
It is now possible to use one copy of $\ket{\mu(z)}$ as a reference frame
which allows to remove the reference to $z$ from the other copies: Take
two copies of $\ket{\mu(z)}$, and apply the operation
$\ket{p,q}\mapsto \ket{p,p^{-1}q}$ both on $B$ and $C$. Then,
\[
\ket{\mu(z)}\ket{\mu(z)}=\sum_{bc}\ket{zb,zc}_B\ket{b,c}_D
\mapsto
\sum_{bc}\ket{zb,b^{-1}c}_B\ket{b,b^{-1}c}_D =
\ket{\mu(z)}\ket{\mu(1)}\ .
\]
Thus, all but one $\ket{\mu(z)}$ can be mapped locally to a maximally
entangled state of Schmidt rank $|G|$, while the remaining ``reference
frame'' copy -- the only $z$-dependent contribution to the state --
turns into $\sum_z\ket{\mu(z)}$ which is separable. Thus, $B$ and $D$
share a total of $(\partial D-1)$ maximally entangled states, which yields
a flat Schmidt spectrum of rank $|G|^{\partial D-1}$.\footnote{
    Note that we could have used the same argument without the initial RG
    step: In that case, we would have arrived at $\partial D$ copies
    of $\ket{\mu(z)}$ which could have been converted to $(\partial D-1)$
    maximally entangled states the same way.}
\end{proof}

\begin{corollary}[Topological R\'enyi entropy of $G$-injective states]
Let $A$ be $G$--injective, with $U_g=L_g$ the regular representation.
Then, any reduced density operator $\rho_D$ of any topologically trivial
block $D$  has rank $|G|^{\partial D-1}$, and thus its zero R\'enyi
entropy is
\[
S_0(\rho_D)=\log|G|(\partial D)-\log |G|\ .
\]
\end{corollary}
\begin{proof}
This follows immediately from the preceding theorem, since any
$G$--injec\-tive $A$ with $U_g=L_g$ can be transformed to a $G$--isometric
one by a local (non-unitary) transformation on the physical system. (This
can be seen by doing a singular value decomposition of $\mc P(A)=UDV$:
$UD$ is the desired linear operation, which by
Observation~\ref{obs:surjective} has full rank.)
\end{proof}

\subsection{Commuting parent Hamiltonians}

Let us now show that the parent Hamiltonians of $G$--isometric PEPS are
special: They are sums of commuting local terms. We will demonstrate the
proof for MPS, but the generalization to PEPS is straightforward.

\begin{lemma} For $G$--isometric MPS, the local terms $h_i$
[Eq.~\eqref{eq:2d:parentham-localterm}] in the parent Hamiltonian
-- which project onto the complement of
$\mathcal S_2=\left\{\sum_{ij}\tr[A^iA^jX]\ket{ij}\middle\vert X\right\}$ --
are of the form
\begin{equation}
\label{eq:iso:ham-proj-from-A}
\raisebox{-1.2em}{
\includegraphics[height=3.5em]{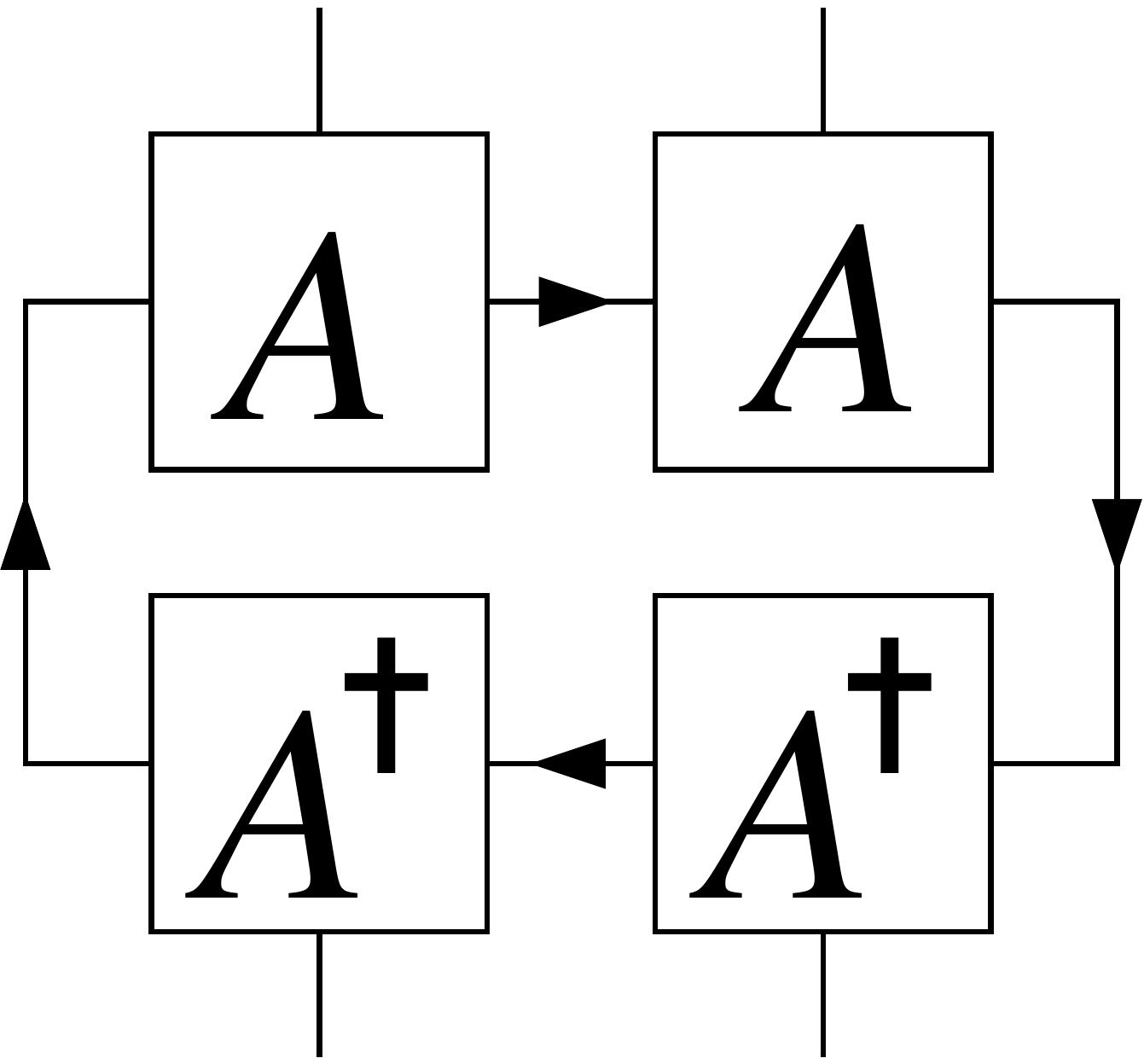}
}\ .
\end{equation}
\end{lemma}
\begin{proof}
First, $h_i$ is a projector, since
\[
\raisebox{-3em}
{\includegraphics[height=7em]{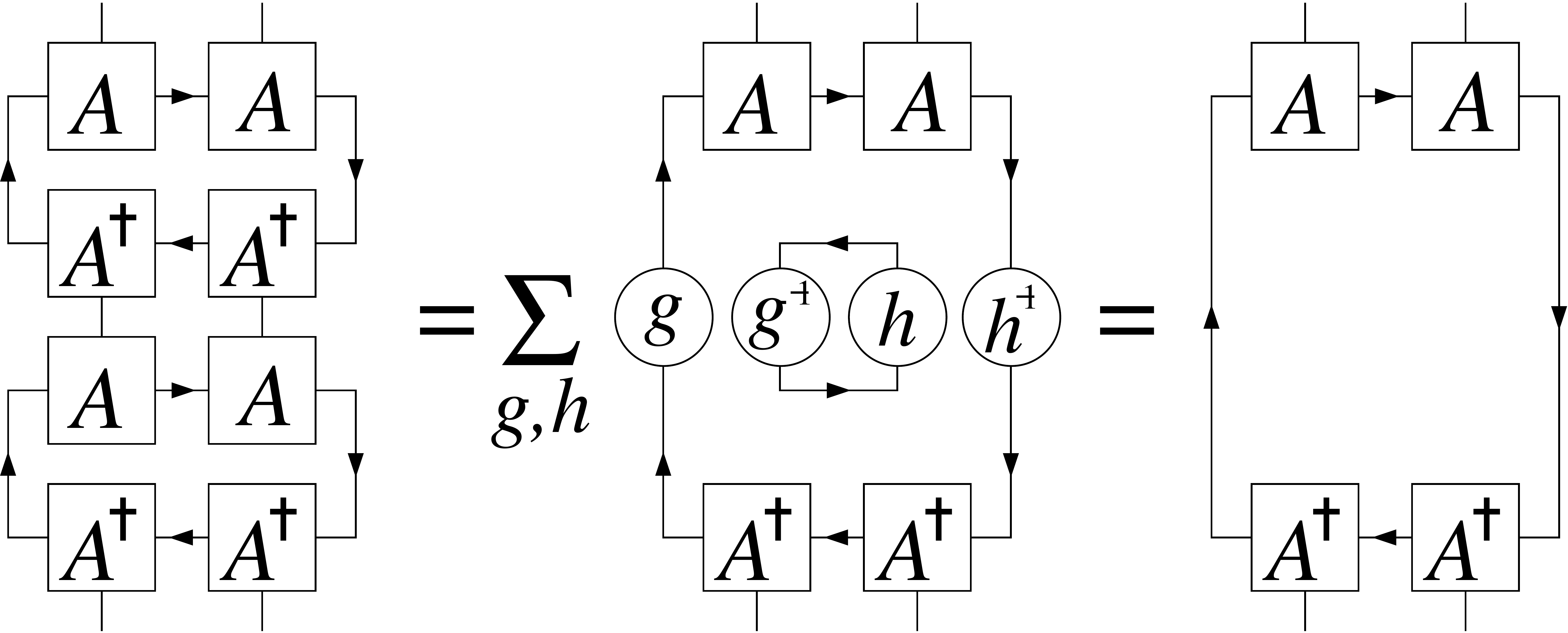}
} \ .
\]
Second, its range is clearly contained in $\mc S_2$, and third, any state of
the form $\sum_{ij}\tr[A^iA^jX]$ is preserved by $h_i$:
\[
\raisebox{-2.7em}
{\includegraphics[height=6.7em]{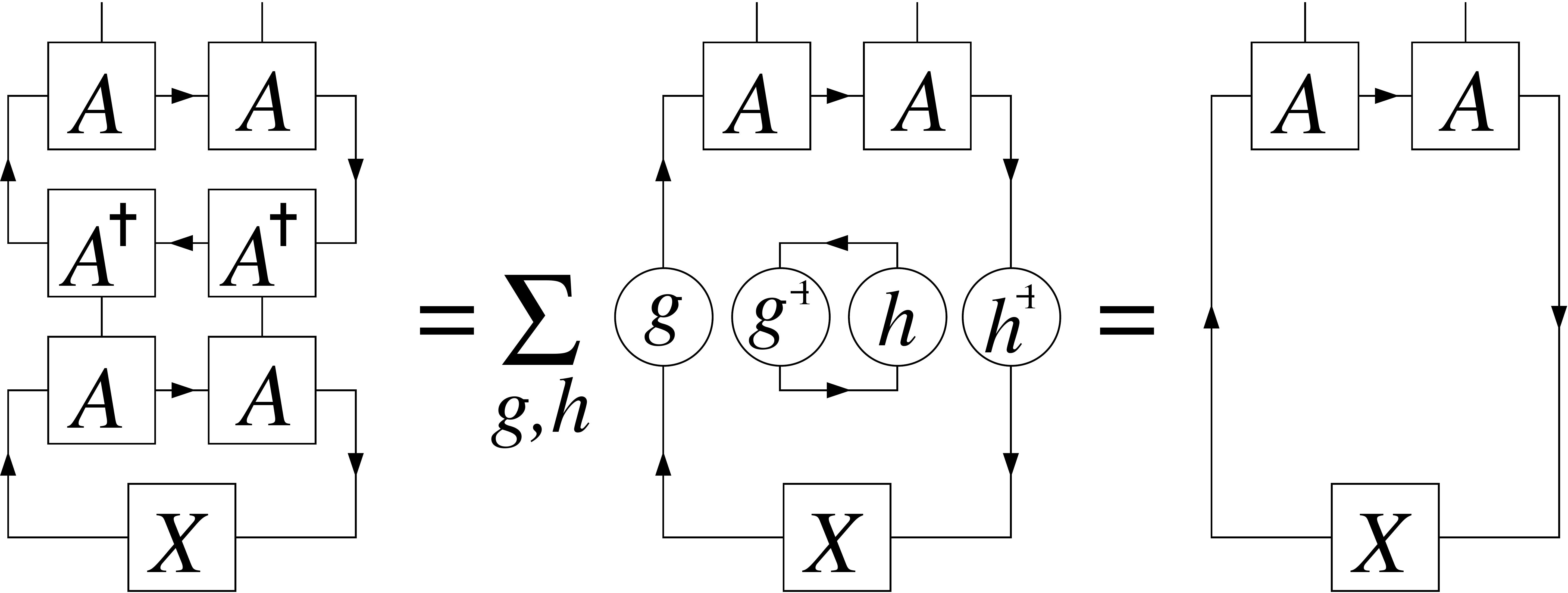}
} \ .
\]
Finally, $h_i$ is self-adjoint, which proves the lemma.
\end{proof}
Note that the corresponding operator \eqref{eq:iso:ham-proj-from-A}
defined for non-isometric PEPS (with $A^{-1}$ instead of $A^\dagger$)
fails the last condition: while it is a projector onto $\mc S_2$, it is not
self-adjoint, i.e., not a Hamiltonian.

\begin{theorem}[Commuting Parent Hamiltonians]
For $G$--isometric PEPS, the terms $h_i$ of the parent Hamiltonian
of Theorem~\ref{thm:2d:parent-ham}, cf.~\eqref{eq:iso:ham-proj-from-A},
commute.
\end{theorem}
\begin{proof}
We start from the identity
\[
\raisebox{-1.5em}{\includegraphics[height=3.8em]{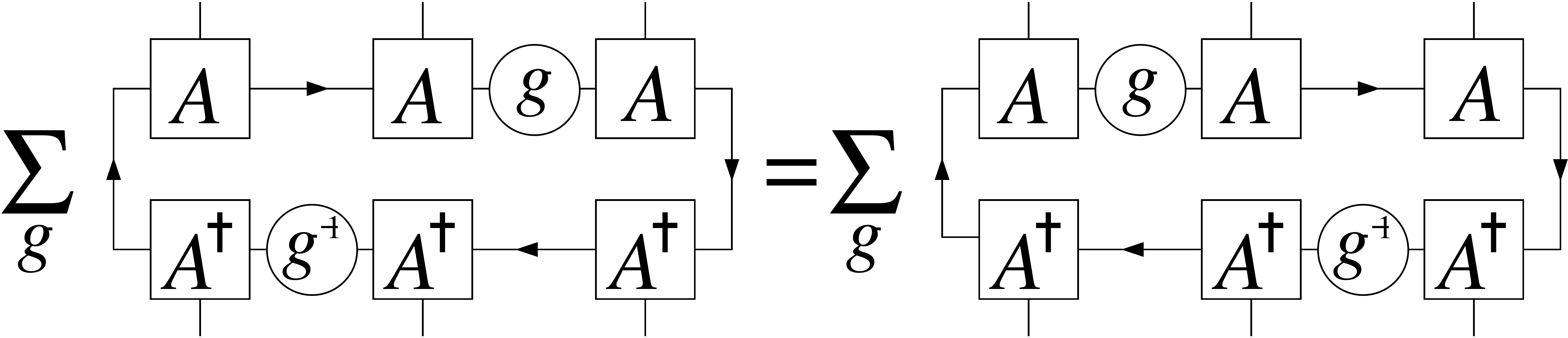}
}\ .
\]
(Note that this the projector onto the joint kernel of $h_1$ and $h_2$,
which is equal $h_1h_2$ for commuting $h_1$, $h_2$!)
The l.h.s.\ can be transformed to
\[
\raisebox{-2.5em}{
\includegraphics[height=6.2em]{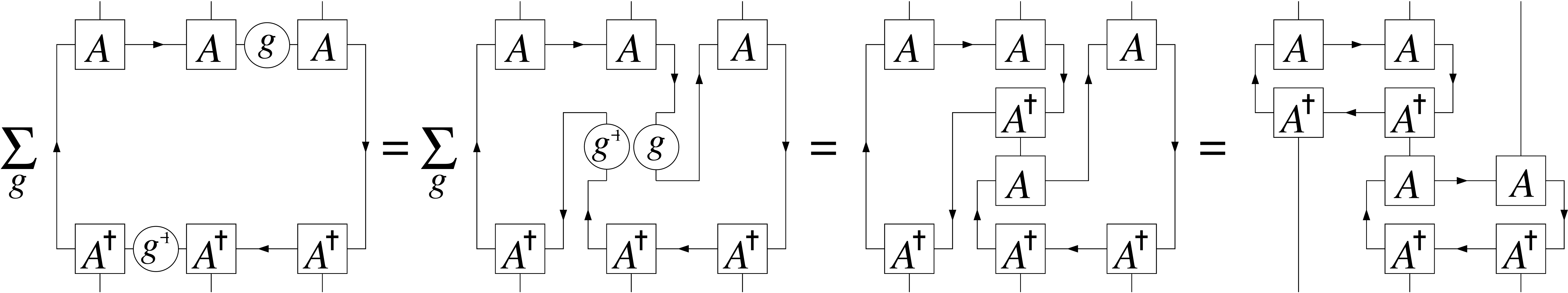}
}\ ,
\]
i.e., $h_1h_2$, while the corresponding transformation on the r.h.s.\
yields $h_2h_1$, proving that $h_1h_2=h_2h_1$.  The same
proof applies to two dimensions.
\end{proof}

\subsection{Anyons}

In the following, we show how to understand anyonic excitations of
$G$--isometric PEPS using the symmetry of the tensors. Intuitively,
excitations should be formed by (open) strings of $U_g$'s, similar to the
(closed) strings distinguishing the different ground states: Since such
strings can be continuously deformed, they cannot be detectable anywhere
locally, except for the endpoints. As we will show, these strings give
rise to one type of particles: fluxes, but there will also be a second type of
complementary particles: charges. 

Note that in the following, we consider the bulk of a PEPS formed by
$G$--isometric tensors $A$, and assume the boundaries to be far away: The
results generally hold independent of the form of the boundaries.

\subsubsection{Magnetic fluxes}

\begin{definition}[Fluxons]
    \label{def:anyons:fluxes}
Consider a PEPS formed of $G$--isometric tensors $A$, let $g\in G$, and
let $U_g=L_g$ be the left-regular representation.  A pair of \emph{magnetic
fluxes} or \emph{fluxons} in the state $(g,g^{-1})$ is defined
by placing a string of $U_g$ and $U_{g^{-1}}$ on the virtual level as
follows:
\begin{equation}
        \label{eq:anyons:fluxon-def}
\raisebox{-3.5em}{
\includegraphics[height=8em]{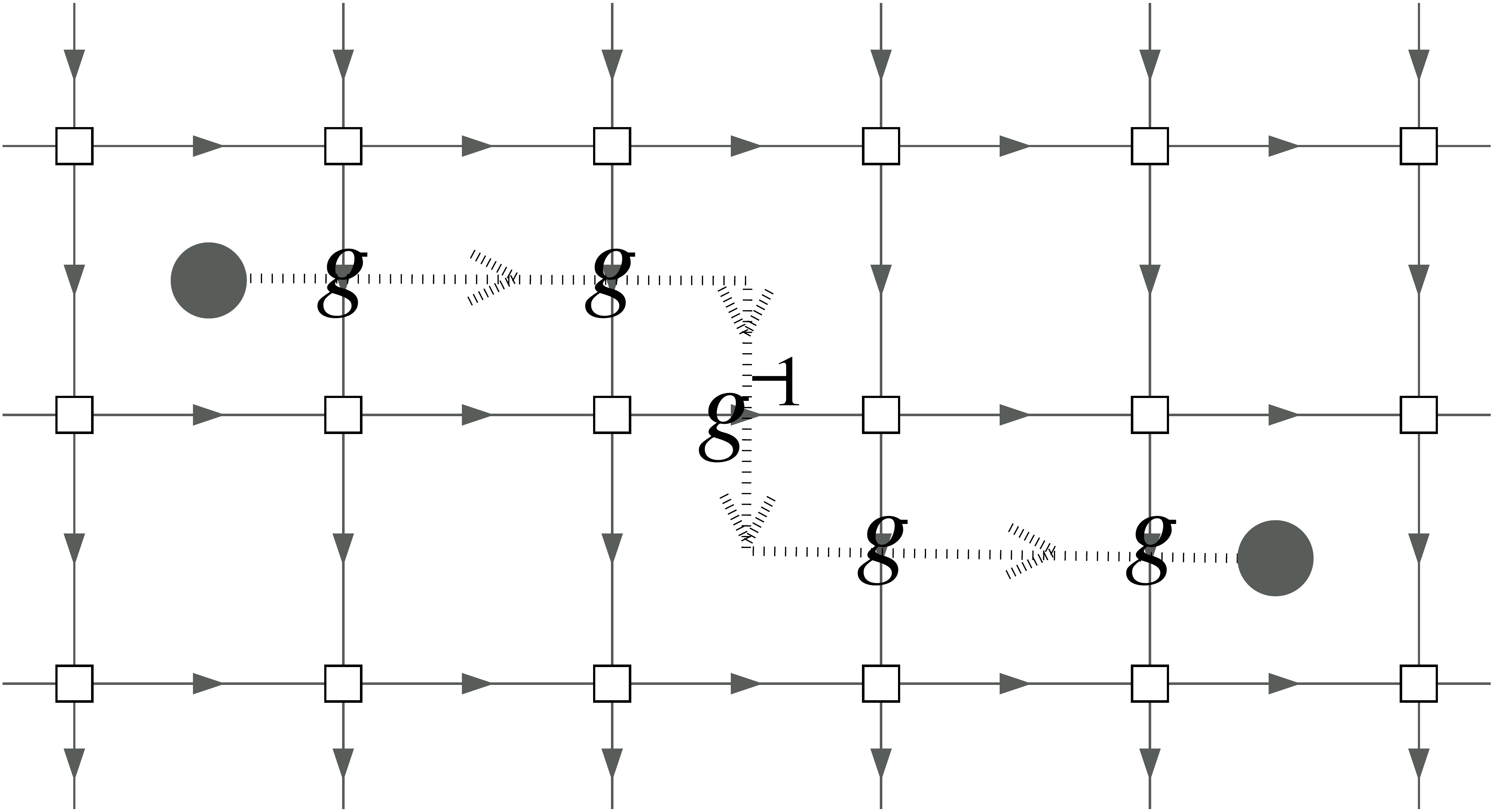}
}\ .
\end{equation}
The fluxons are attached to the two marked plaquettes forming the
endpoints of the string; they are characterized by the fact that there is
an odd number of $g$ and $g^{-1}$ on the adjacent edges.
Whether $g$ or $g^{-1}$ is used depends on the orientation of the string
relative to the bonds; we will generally talk of a ``string of $g$'s''
meaning both $g$ and $g^{-1}$. Also, we assign the state $g$ to the
starting point and $g^{-1}$ to the endpoint.  The \emph{particle type} of
the two fluxons is given by the conjugacy classes $C[g]$ and $C[g^{-1}]$,
while the exact states $g$ and $g^{-1}$ determine the internal state of
the fluxons.  We will show in the following why these choices make sense.
\end{definition}

\begin{lemma}[Deformations of the string]
The endpoint plaquettes of a string are fixed. Except for that, the string
connecting a pair of fluxons can be deformed at will without acting on the
system. This implies that it cannot be observed anywhere except for the
endpoints.
\end{lemma}
\begin{proof}
The string can be defomed using the $G$--invariance of the tensors,
cf.~\eqref{eq:2d:move-strings}, as long as the endpoints are not involved.
On the other hand, the $G$--invariance cannot change the parity of $g$'s
around a plaquette, which implies that the endpoints cannot be moved.
(However, the string can be deformed to reach the endpoint from any side.)
Since according to the following theorem, endpoints can be measured, there
also cannot be any other way to move them away without acting on the
physical system.
\end{proof}

\begin{theorem} [Detection of fluxons]
        \label{thm:anyons:detect-fluxons}
The particle type $C[g]$ of a fluxon can be detected by a measurement on
the plaquette supporting it.
\end{theorem}
\begin{proof}
As we can arbitrarily deform the string of $g$'s ending at a given
plaquette, it suffices to consider the following situation:
\begin{equation}
        \label{eq:anyons:flux-detection-scenario}
\raisebox{-1.2em}{
\includegraphics[height=4em]{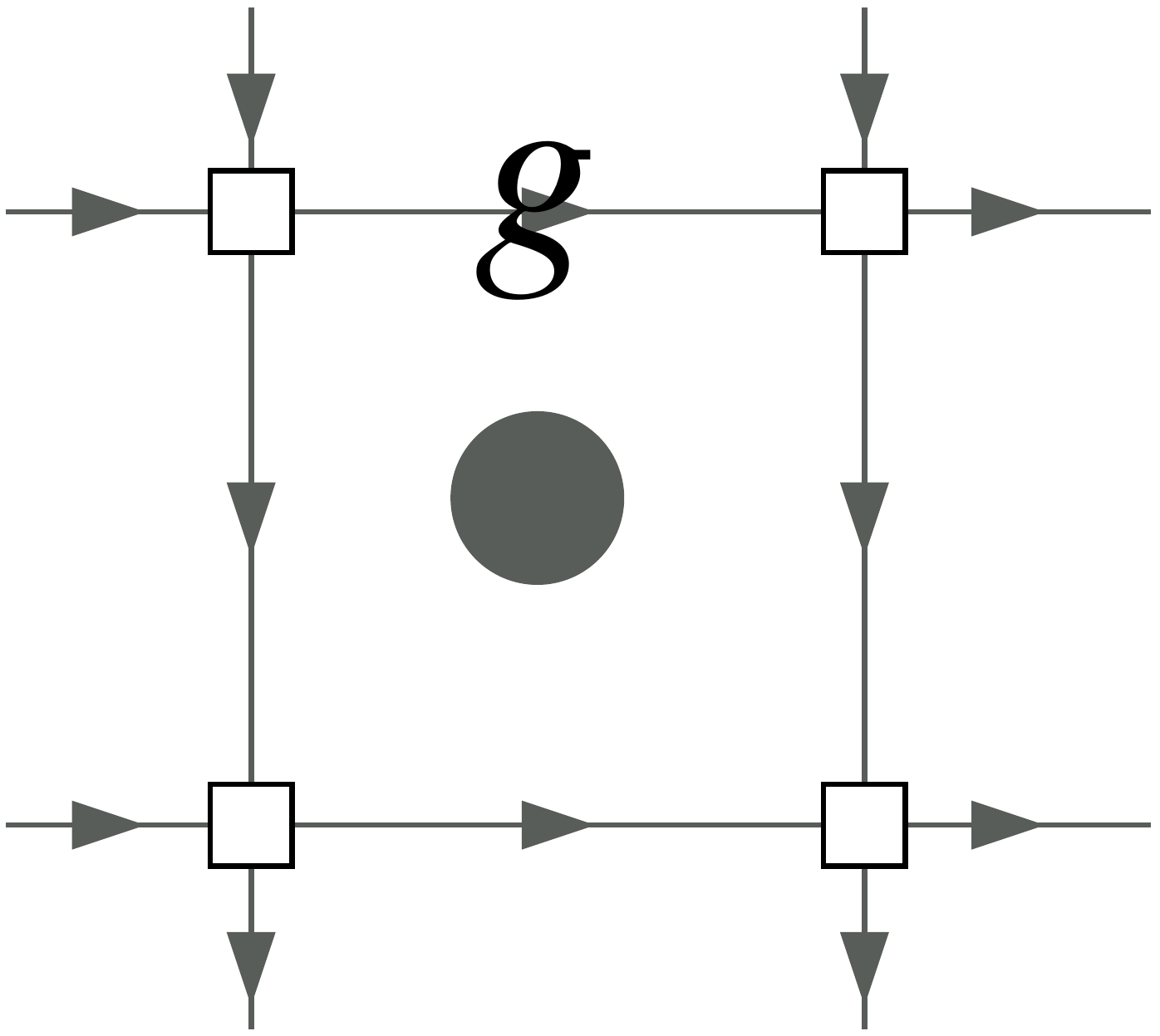}
}\ .
\end{equation}
We will use the four sites of the plaquette
to build a ``tweezer'' which allows us to
``grab'' $g$, up to a twirl. To this end, block the four sites to a
U--shaped tensor. From Observation~\ref{obs:iso:accessible-virt}, we know
that there is a reversible operation which transforms the blocked tensor
to
\begin{equation}
        \label{eq:anyons:flux-detect-tweezer-explicit}
\raisebox{-4em}{
\includegraphics[height=9em]{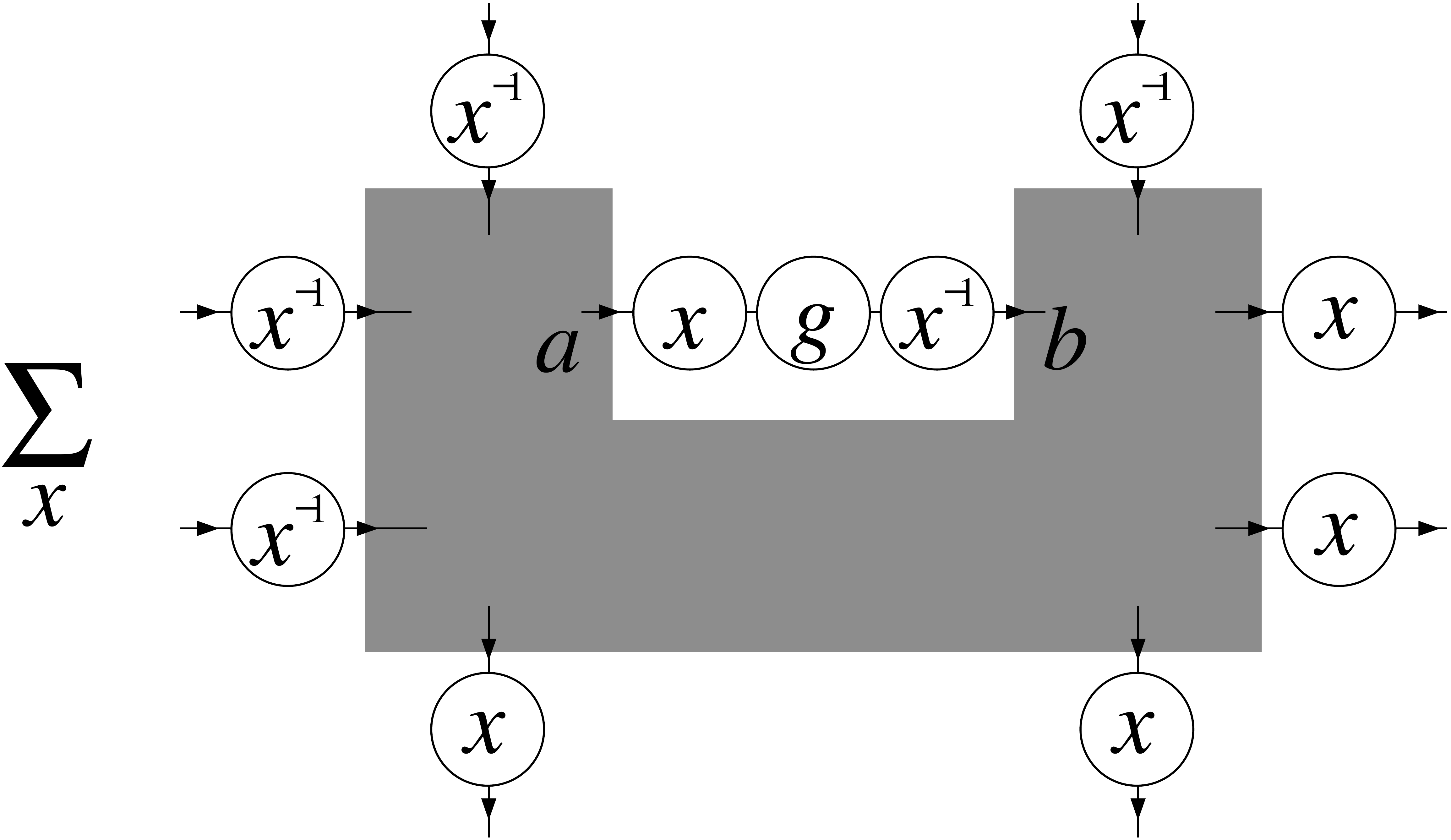}
}\quad .
\end{equation}
By measuring $ab^{-1}=xgx^{-1}$, we can now infer the particle type $C[g]$
of the fluxon.  Since the value of $x$ cannot be determined, this is all
the information we can gain about a single fluxon.

Looking at \eqref{eq:anyons:flux-detect-tweezer-explicit} and the way our
measurement of $C[g]$ works, one can see that we actually do not use any
of the outgoing indices, i.e., those which connect to the virtual level
and thus to neighboring tensors. Therefore, we can simplify notation by
omitting these indices together with the sum over $x$ and analyze the
action of the detection tweezer on
\[
\raisebox{-1.0em}{
\includegraphics[height=2.8em]{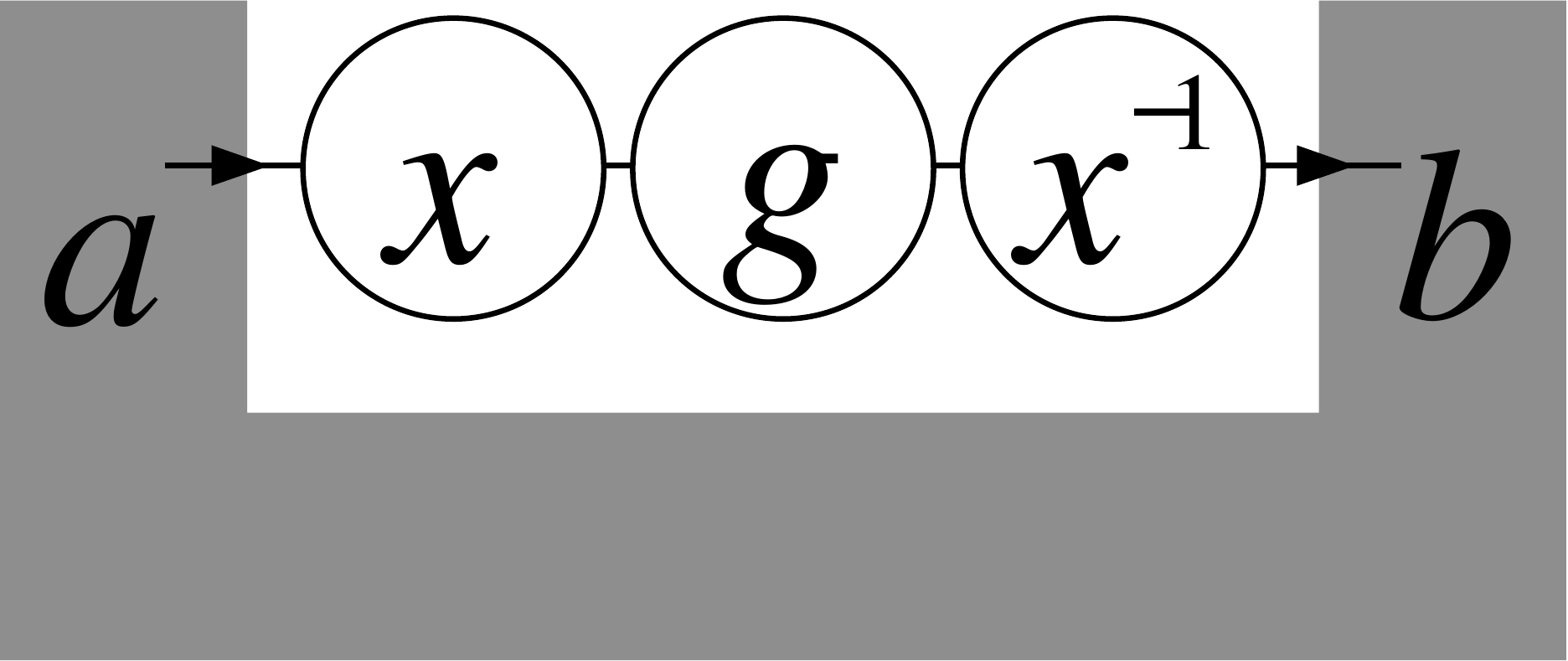}
}\quad ,
\]
where $x$ is unknown. We will make use of this simplified notation from
now on in the analysis of subsequent constructions for detecting, moving,
and creating anyons.

Note that the reason we built a tweezer and not just acted on the site
left and right of $g$ in \eqref{eq:anyons:flux-detection-scenario} is that
we needed to \emph{synchronize} the twirl; otherwise, we could have only
determined $xgy^{-1}$, which is completely random.  Let us also note that
the tweezer can open in any direction, as the string of $g$'s can be
attached to any egde of the plaquette.
\end{proof}

\begin{theorem}[Moving of fluxons]
        \label{thm:anyons:move-fluxons}
Fluxons can be moved by local unitaries, without knowledge of their state.
\end{theorem}
\begin{proof}
We give the construction for the scenario where we want to move a string
to the right, i.e., achieve the following transformation
\[
\raisebox{-2em}{
\includegraphics[height=5.5em]{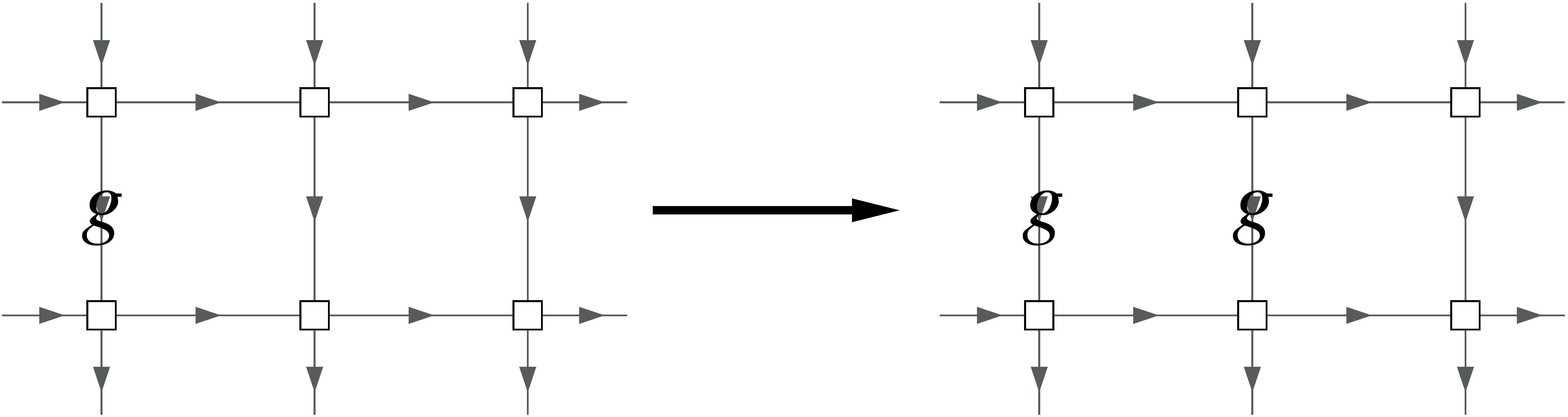}
}
\]
by means of local unitaries. To this end, we will again build a tweezer
(opening towards the left) which will allow us to access the two vertical
edges on the left.  Again,
using Observation~\ref{obs:iso:accessible-virt}, and neglecting
unneeded bonds, the above transformation is equivalent to the
transformation
\[
\raisebox{-2.8em}{
\includegraphics[height=6.5em]{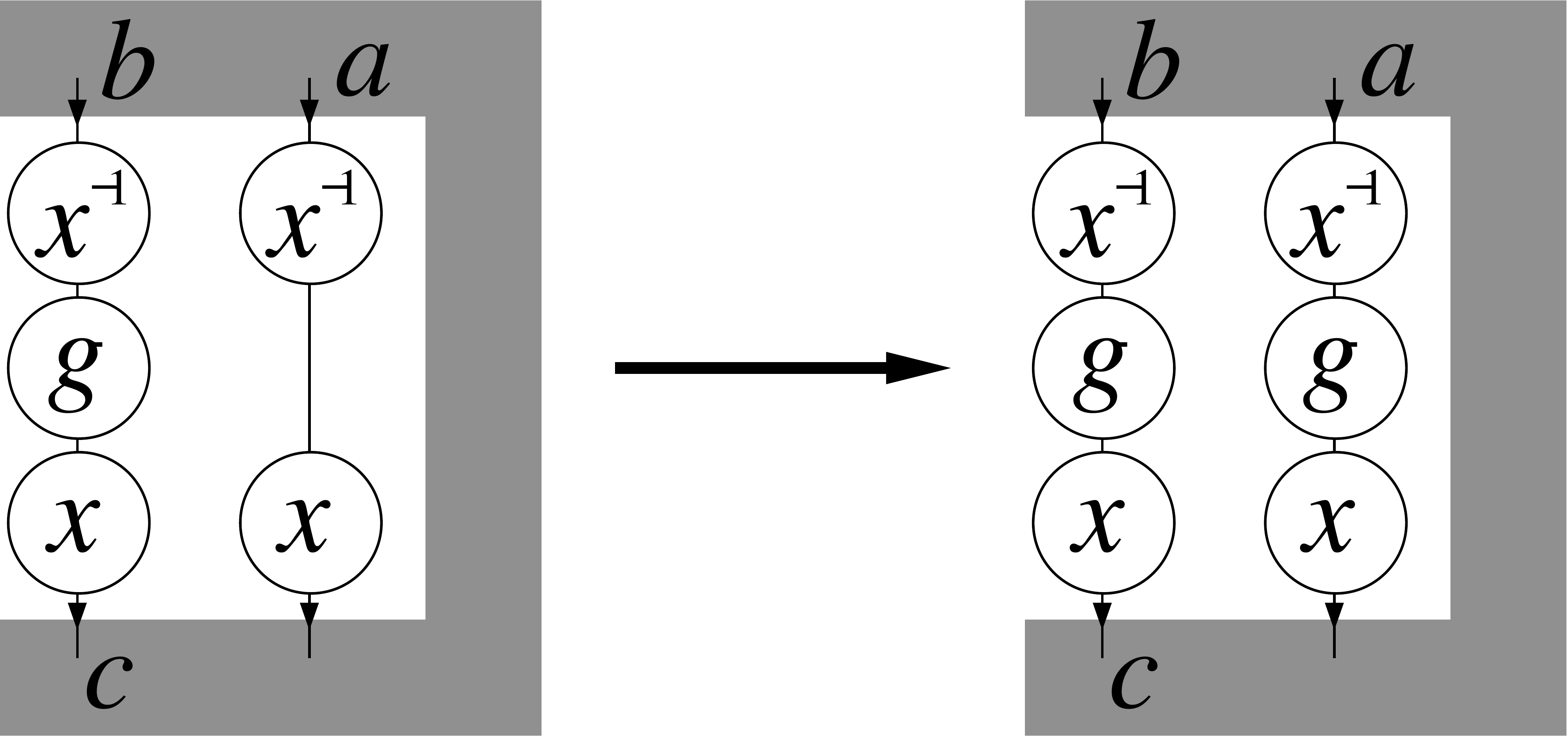}
}
\]
for arbitrary $x$.
This can be accomplished using the same transformation already used in
\eqref{eq:iso:tweezer-duplicate-operation} for this purpose:
$\ket{a}\ket{b}\bra{c}\mapsto\ket{bc^{-1}a}\ket{b}\bra{c}$.
(Note that each of the kets and bras describes the state of a
\emph{separate} quantum system: $\ket{c}\bra{a}\bra{b}$ thus describes a
pure state of the three systems $a$, $b$, and $c$.  Whether we use kets or
bras is determined by the arrows associated to the links.)
\end{proof}

\begin{theorem}[Creation of fluxons]
For any $g\in G$, the fluxon pair
\begin{equation}
        \label{eq:anyons:chargeless-fluxon}
\sum_{z} \ket{\smash{(zgz^{-1},zg^{-1}z^{-1})}}
\end{equation}
can be created deterministically by local operations.\footnote{
        The reason that we can only create the equal weight superposition
        is that all other superpositions carry non-zero charge, which is a
        conserved quantity, cf.~\cite{preskill:lecturenotes}.
}
Note there is one such state per conjugacy class $C[g]$.
\end{theorem}
\begin{proof}
We want to implement the transformation
\[
\raisebox{-3.0em}{
\includegraphics[height=7.2em]{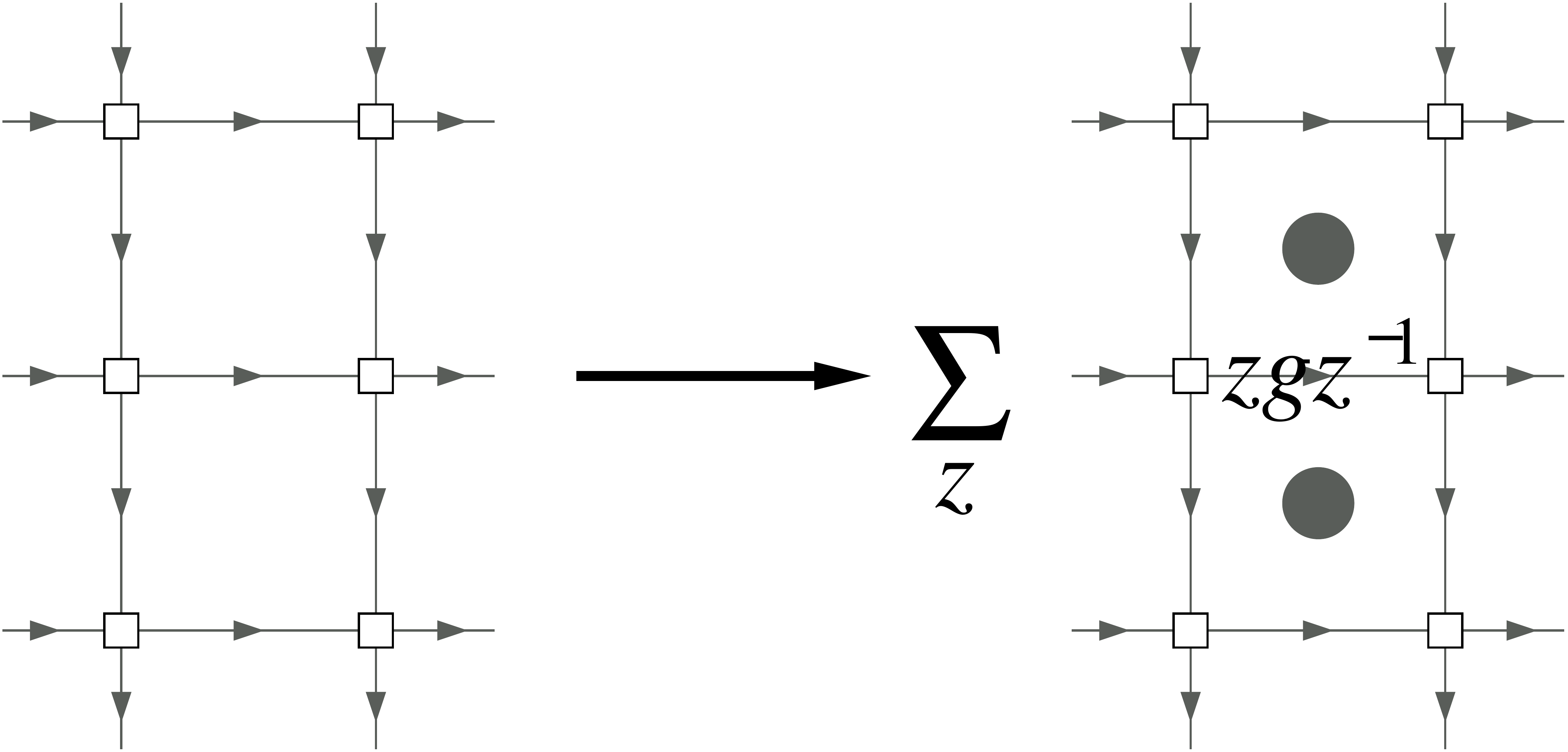}
}\quad .
\]
To this end, we build a tweezer on the lower plaquette. Neglecting
unneeded degrees of freedom, the transformation we want to achieve is
equivalent to
\begin{equation}
        \label{eq:anyons:make-chargeless-fluxon}
\raisebox{-1em}{
\includegraphics[height=2.8em]{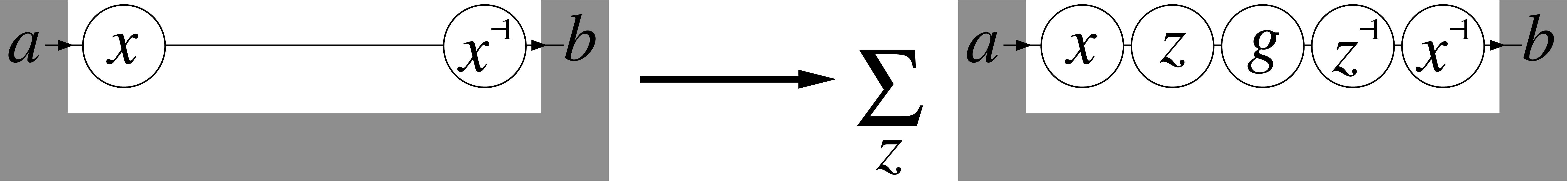}
\raisebox{-0.5em}{\rule{0em}{0em}}}
\end{equation}
for some unknown $x$.  This can be accomplished by replacing the state
$\sum_r\ket{r}_a\bra{r}_b$ with $\sum_{s,z} \ket{zgz^{-1}s}_a\bra{s}_b$,
which can be done unitarily if wanted.
\end{proof}

\subsubsection{Braiding of fluxons}

In the following, we investigate what happens to fluxons when they are
braided around each other. The general scenario is as follows: Consider
two pairs of fluxons, and move one fluxon of one pair around one of the
other
one:
\[
\includegraphics[height=5em]{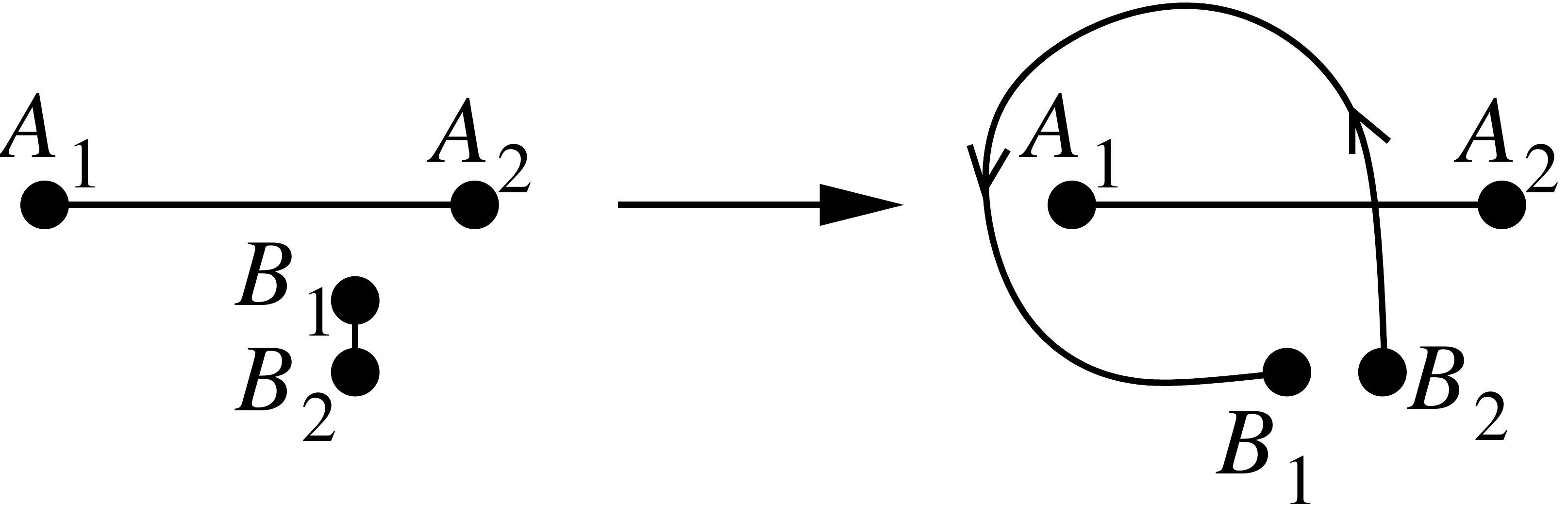}
\]

In principle, it does not make sense to ask how the anyon $B_1$ is
affected when it crosses the line between $A_1$ and $A_2$, since this line
can be deformed at will.  However, it does make sense to think like that
if the pair $B_1$ and $B_2$ is brought together after braiding to measure
the joint flux, since then $B_1$ eventually \emph{has} to cross this
line. We will discuss how to measure the joint flux later on.

Instead of investigating what happens when we move the fluxon $B_1$ across the
string connecting $A_1$ and $A_2$, we rather keep $B_1$ fixed and move the
string across it -- this is easier to analyze, since it takes place purely
on the virtual level:
\begin{equation}
        \label{eq:anyons:fluxon-braiding-lazy}
\raisebox{-2em}{
\includegraphics[height=5em]{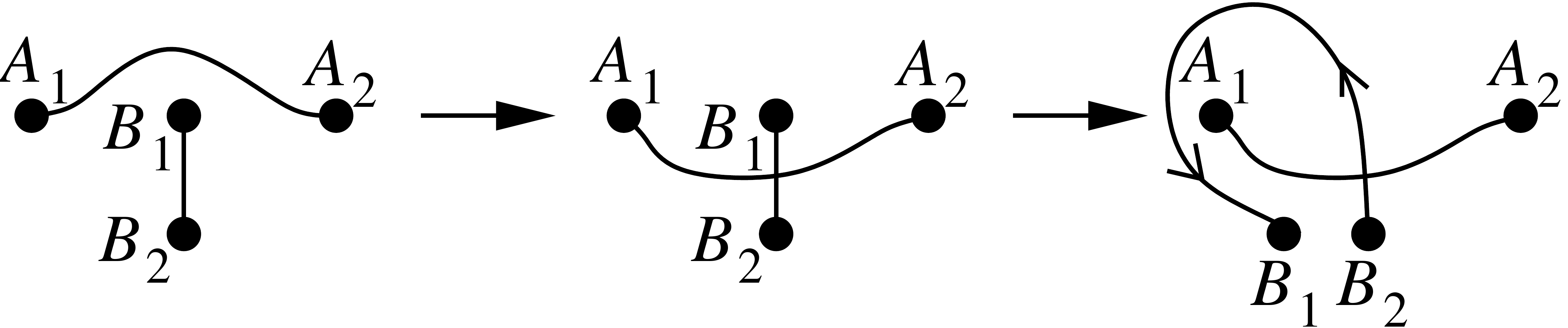}
}\ .
\end{equation}
Let us now consider this scenario -- moving the $A_1$--$A_2$ string
accross the anyon $B_1$ -- in the PEPS representation on the virtual
level. To this end, let $A_1$ and $A_2$ be in the state
$(g,g^{-1})$, and $B_1$ and $B_2$ in state $(h,h^{-1})$.
Then, using the $G$--invariance of the tensors,
\[
\raisebox{-3.5em}{
\includegraphics[height=7.5em]{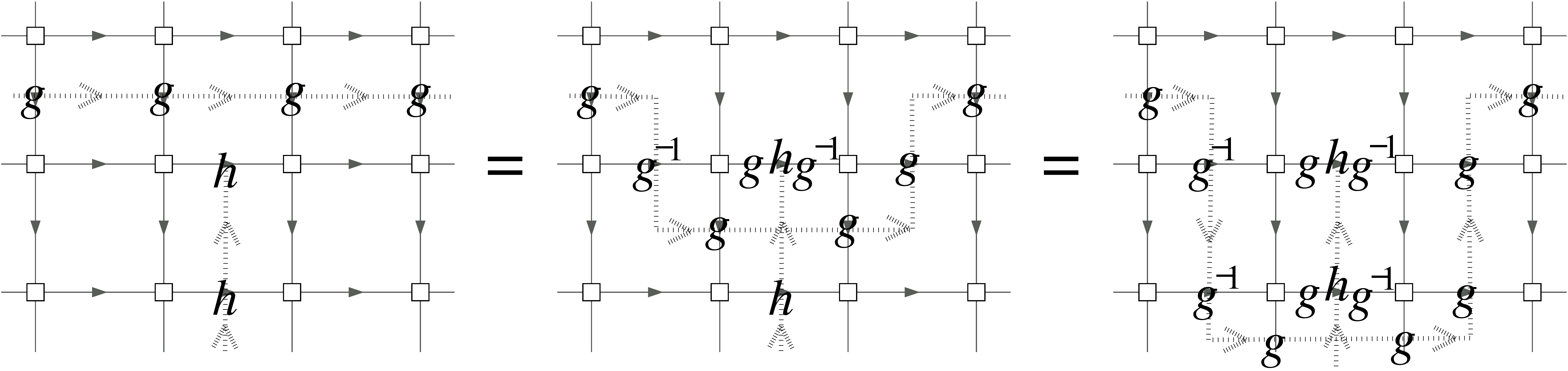}
}\ :
\]
i.e., the $B_1$ end of the string is conjugated by $g$, and is now in the
state $ghg^{-1}$.  Note that this does not change the conjugacy class,
i.e., the particle type. This is crucial since the
crossing of $B_1$ and the $A_1$--$A_2$ line cannot be assigned a well-defined
position, and thus the change of the state of $B_2$ cannot be detected by
measuring $B_2$ alone.

However, when we bring $B_1$ and $B_2$ together, it is possible to measure
the particle type $C[ghg^{-1}h^{-1}]$ of their joint flux, and thus infer
information about the state $g$ of $A_1$. For two fluxes in states
$a$ and $b$, their joint state is given by $ab$; this is
consistent with how the two fluxes together affect another flux when
braiding. Note that actually merging the two fluxes will require a
measurement, since this is an irreversible process.

Let us now describe how to measure the joint flux. To this end, consider
two fluxes $a$ and $b$ sitting on adjacent plaquettes -- this
is as close as we can bring them using the ``fluxon moving tweezer'' of
Theorem~\ref{thm:anyons:move-fluxons} -- and route the virtual strings as
follows:
\[
\raisebox{-1.7em}{
    \includegraphics[height=4.5em]{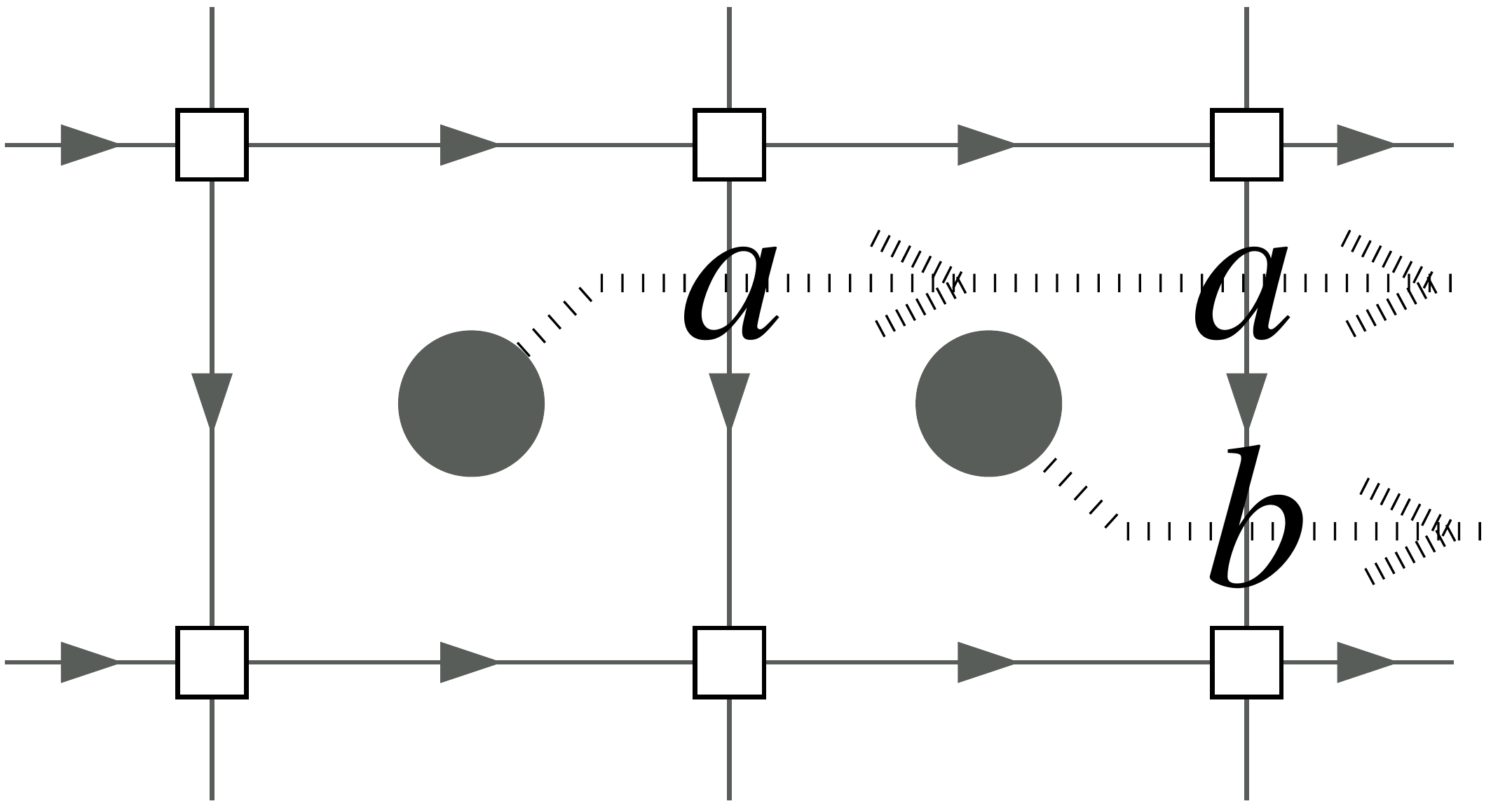}
}\ .
\]
Now, we can use basically the same tweezer construction used for measuring
single fluxons in Theorem~\ref{thm:anyons:detect-fluxons} to measure the
$C[ab]$: To this end, we build a tweezer covering both plaquettes, and
connected at the left side, and measure $C[ab]$ across the right vertex.
In fact, this construction can be extended to measure the joint flux of
any number of anyons enclosed in the tweezer.

Note that the braiding~\eqref{eq:anyons:fluxon-braiding-lazy} does not
only change the state of $B_1$, but also that of $A_1$ if it brought back
to interfere with $A_2$. The reason is that even after repairing the $B$
particles, the loop of $h$'s remains there and affects any anyon crossing
it. Note that while the way $A_1$ is affected when moving it out of the
loop depends on whether it is moved above or below the path it came from,
this difference cannot be observed when repairing it with $A_2$, and
vanishes when $A_1$ and $A_2$ are aligned in a fixed way.

\subsubsection{Electric charges}

Let us now define electric charges. Different from fluxes which always
come in pairs, it is possible to define (though not create) electric
monopoles.

\begin{definition}[Electric charges]
An electric charge (``chargeon'') with charge $c$, where $c$ labels an
irreducible representation $D^c$ of $G$, is given by a defect
\begin{equation}
\label{eq:anyons:chargeon-def}
\raisebox{-3.5em}{
\includegraphics[height=8em]{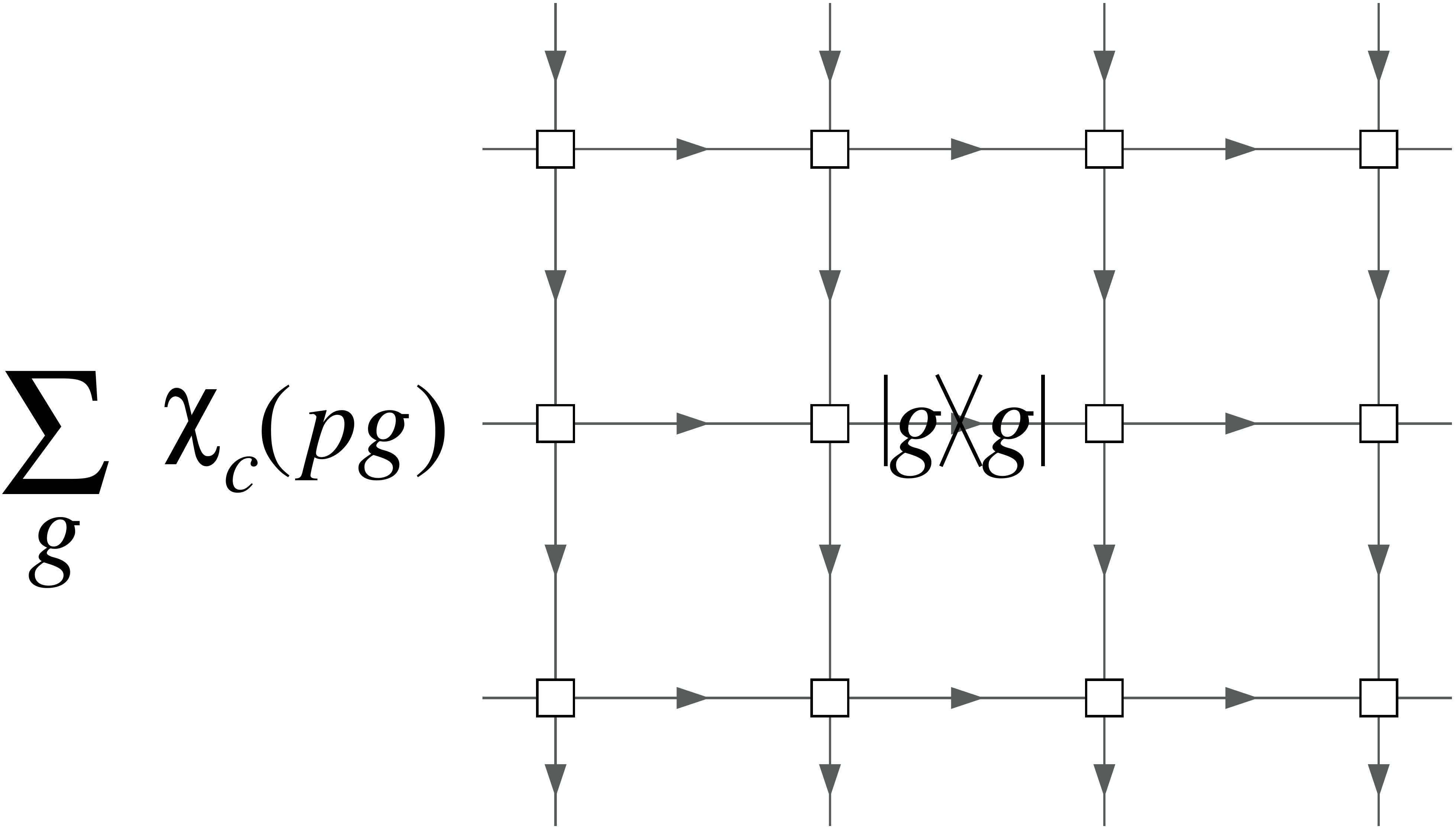}
}
\end{equation}
which is attached to an edge of the lattice, i.e., instead of the
identity, the edge acts as $\sum_g\chi_g(pg)\ket{g}\!\!\bra{g}$.
Here, $\chi_c(\cdot)=\tr[D^c(\cdot)]$ is the character of the irreducible
representation $c$, and $p\in G$ characterizes the internal state of the
chargeon.
\end{definition}

\begin{theorem}[Detection of chargeons]
        \label{thm:anyons:detect-chargeons}
The charge $c$ of a chargeon can be detected by measuring across the two
vertices adjacent to the edge supporting the charge.
\end{theorem}
Note that we will later show that the internal degree of freedom $p$ is
changed by strings connecting fluxons, which rules out that we can measure
$p$ itself.
\begin{proof}
We will perform a joint measurement on the vertices left and right of the
chargeon in \eqref{eq:anyons:chargeon-def}.  From
Observation~\ref{obs:iso:accessible-virt}, the state of these two
sites is isomorphic (up to unneeded links) to
\begin{equation}
\label{eq:anyons:measure-chargeon}
\raisebox{-1.2em}{
\includegraphics[height=2.3em]{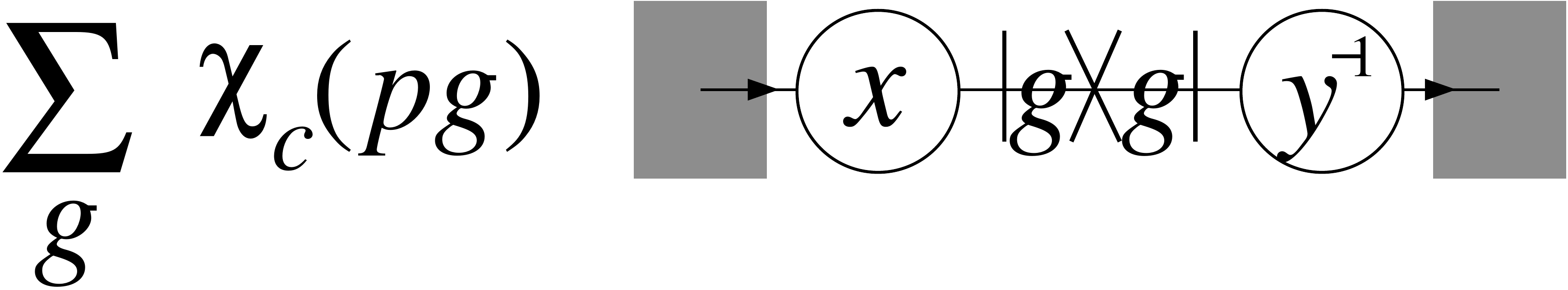}
\raisebox{-0.2em}{\rule{0em}{0em}}}\ .
\end{equation}
i.e., we have access to the state
\[
\sum_{g}\chi_c(pg)\ket{xg}\bra{yg}
=\sum_k\chi_c(px^{-1}k)\ket{k}\bra{zk}
\]
for some $x$ and $y$, with $k:=xg$ and $z=yx^{-1}$. As a consequence of the group orthogonality relations,
$\ket{v_h^c}=\sum_k\chi_c(hk)\ket{k}$ are orthogonal vectors for different
irreducible representations $c$, $\langle v_h^c\vert
v_{h'}^{c'}\rangle=f(h,h')\delta_{c,c'}$.  Therefore, the charge
$c$ (as well as $z$) can be determined by measuring the two sites in
\eqref{eq:anyons:measure-chargeon}.  Note that
on the other hand, no information about $p$ can be learned, since $x$ is
unkown.
\end{proof}

\begin{theorem}[Moving of chargeons]
Chargeons can be moved by local operations.
\end{theorem}
\begin{proof}
Since chargeons are localized perturbations on the virtual level (unlike
fluxons which arise from strings), moving then is as simple as swapping:
In order to implement the transformation
\begin{equation}
\label{eq:anyons:chargeon-move-setting}
\raisebox{-2.2em}{
    \includegraphics[height=5.5em]{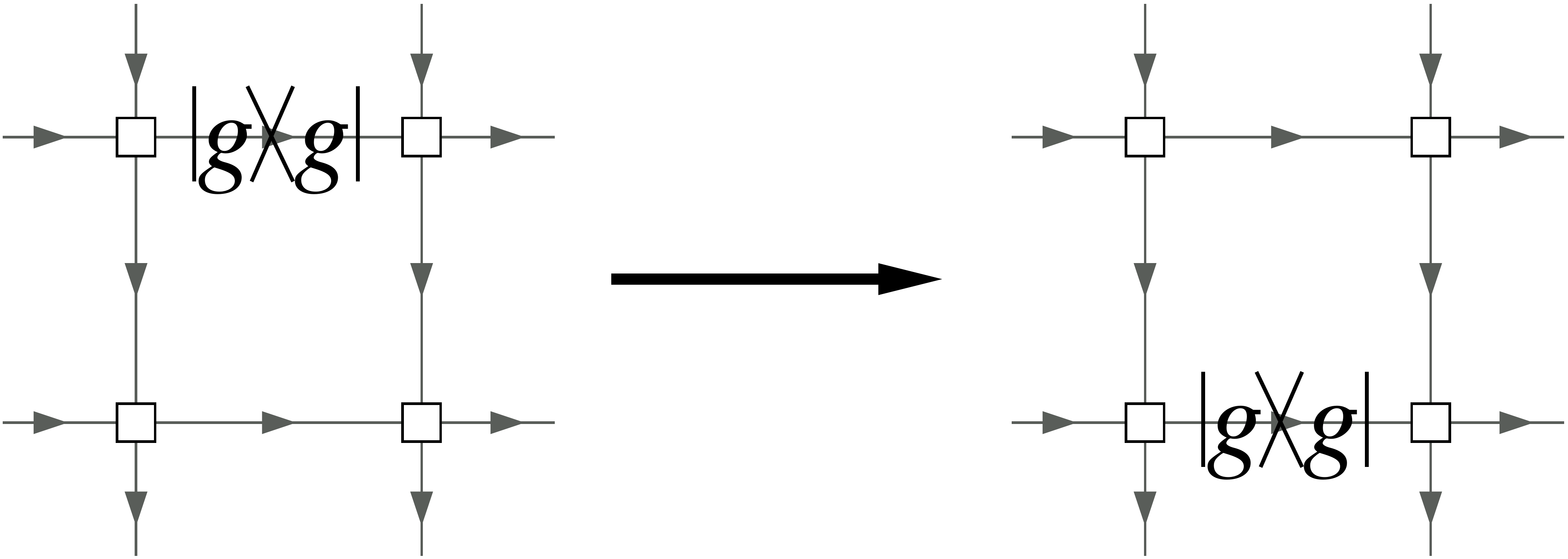}
}\quad ,
\end{equation}
we will need to block the left two and right two spins independently.  The
transformation is then locally equivalent (up to unneeded degrees of
freedom) to
\[
\raisebox{-1.3em}{
    \includegraphics[height=3.5em]{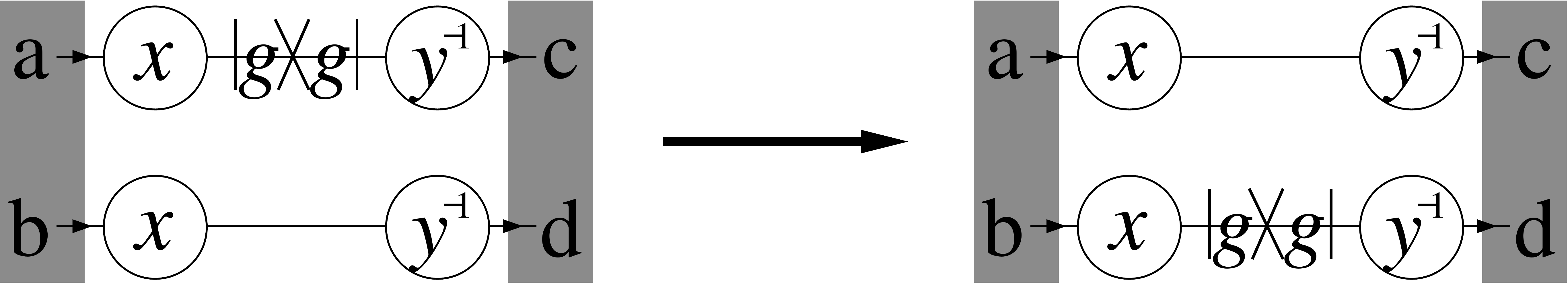}
}\ ,
\]
and thus, by swapping a with b and c with d, the chargeon is
moved as in~\eqref{eq:anyons:chargeon-move-setting}. Note that a similar
construction can also be used to move a fluxon around the corner.
\end{proof}

\begin{theorem}[Creation of chargeons]
        \label{thm:anyons:create-chargeons}
For any irreducible representation $c$ and any $p\in G$, the chargeon pair
described by
\begin{equation}
        \label{eq:anyons:chargeon-pair-state}
\Pi_{c,p}=\sum_{g,h} \chi_c(ph^{-1}g)\ket{g}\bra{g}_u\otimes \ket{h}\bra{h}_d
\end{equation}
can be created by local operations.
\end{theorem}
Note that this gives a particle with charge $c$ in  position $u$ and a
particle of charge $\bar{c}$ in position $d$, where $\bar c$ labels the
complex conjugate of representation $c$.
\begin{proof}
First, observe that the pair~\eqref{eq:anyons:chargeon-pair-state} is invariant
under conjugation by $x\in G$:
\begin{equation}
        \label{eq:anyons:Ux-invariance-of-charge-pair}
\begin{aligned}
(U_x\otimes U_x)\Pi_{c,p}(U_x^\dagger\otimes U_x^\dagger)
&=
\sum_{g,h} \chi_c(ph^{-1}g)\ket{xg}\bra{xg}\otimes \ket{xh}\bra{xh} \\
&=
\sum_{\tilde g,\tilde h} \chi_c(p\tilde h^{-1}xx^{-1}\tilde g)
\ket{\tilde g}\bra{\tilde g}\otimes \ket{\smash{\tilde h}}\bra{\smash{\tilde h} }
=
\Pi_{c,p}\ .
\end{aligned}
\end{equation}
To create the state \eqref{eq:anyons:chargeon-pair-state}, we need to
implement the transformation
\[
\raisebox{-2.5em}{
    \includegraphics[height=7em]{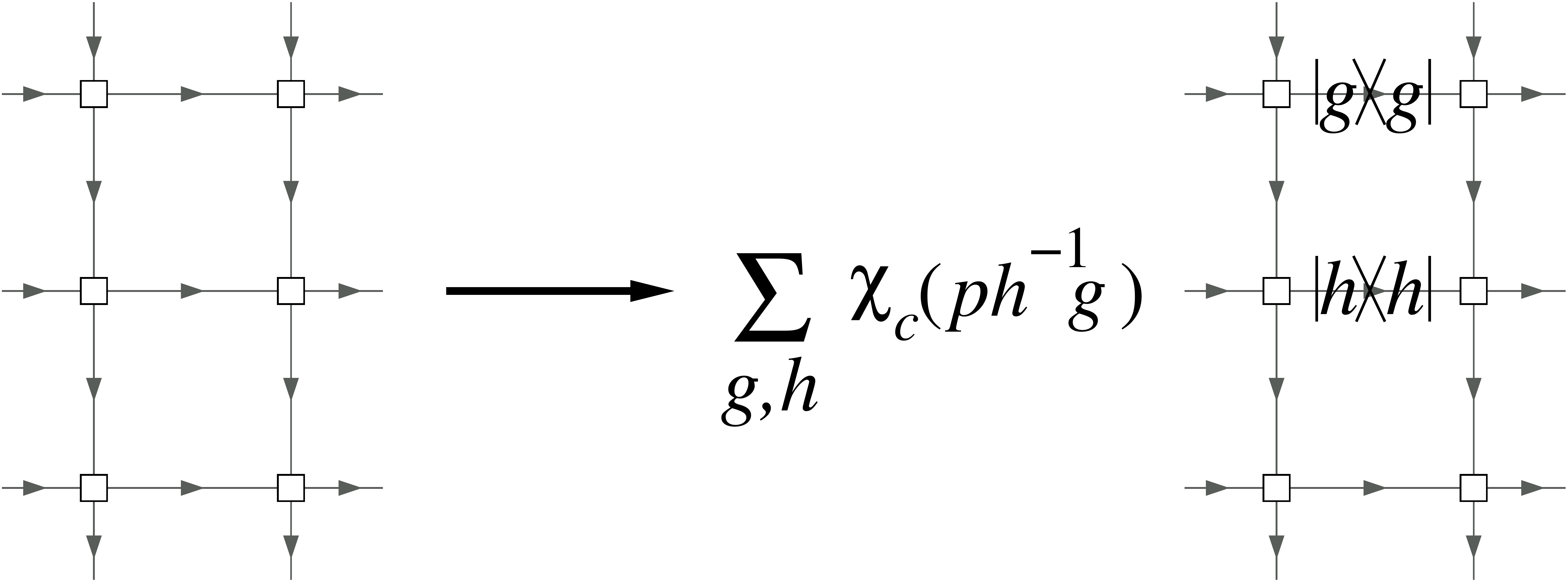}
}\quad .
\]
To this end, we use a U-shaped tweezer encompassing all six sites, which
leaves us with the task of implementing the transformation
\begin{equation}
        \label{eq:anyons:create-chargepair-accessible}
\raisebox{-2.1em}{
    \includegraphics[height=4.8em]{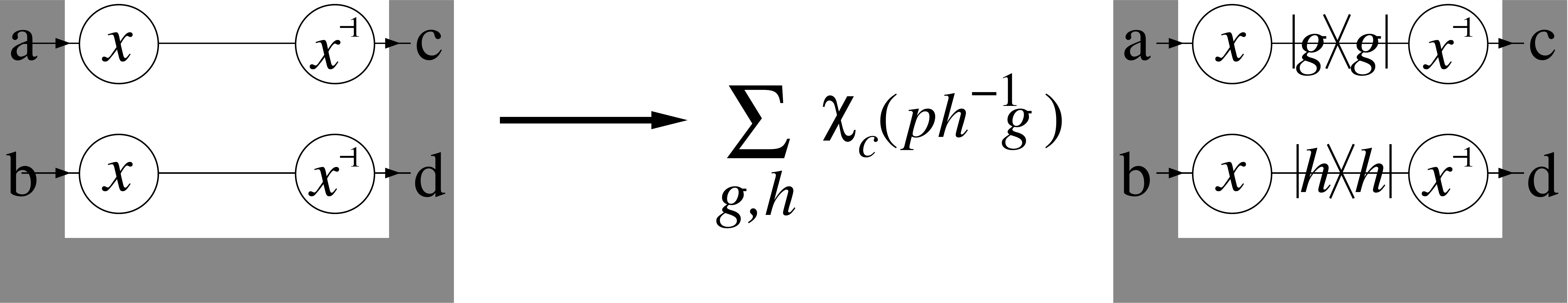}
}
\end{equation}
for some unkown $x$. Due to the $U_x$--invariance of $\Pi_{c,p}$
[Eq.~\eqref{eq:anyons:Ux-invariance-of-charge-pair}], neither side depends
on $x$, and the transformation can be implemented (unitarily) by mapping
the initial state of the four sites a,b,c,d to $\Pi_{c,p}$.
\end{proof}

\subsubsection{Braiding of charges with fluxes}

Let us finally see what happens when we braid chargeons with other
particles. Since chargeons are localized defects in the lattice, they do
not affect other particles which are braided around them. On the other
hand, they are affected if braided around a fluxon. Consider again
the braiding~\eqref{eq:anyons:fluxon-braiding-lazy}, now with
$A_1$--$A_2$ a pair of fluxons in state $(k,k^{-1})$, and $B_1$--$B_2$
a pair of chargeons with charges $c$ and $\bar c$, e.g.\ as
in~\eqref{eq:anyons:chargeon-pair-state}.  Again, we will choose to move
the string connecting the fluxons $A_1$ and $A_2$ over the chargeon $B_1$,
rather than moving $B_1$ across the string.  Analyzing the situation for a
single projection $\ket{g}\bra{g}$, we find that
\[
\raisebox{-3em}{
\includegraphics[height=6.5em]{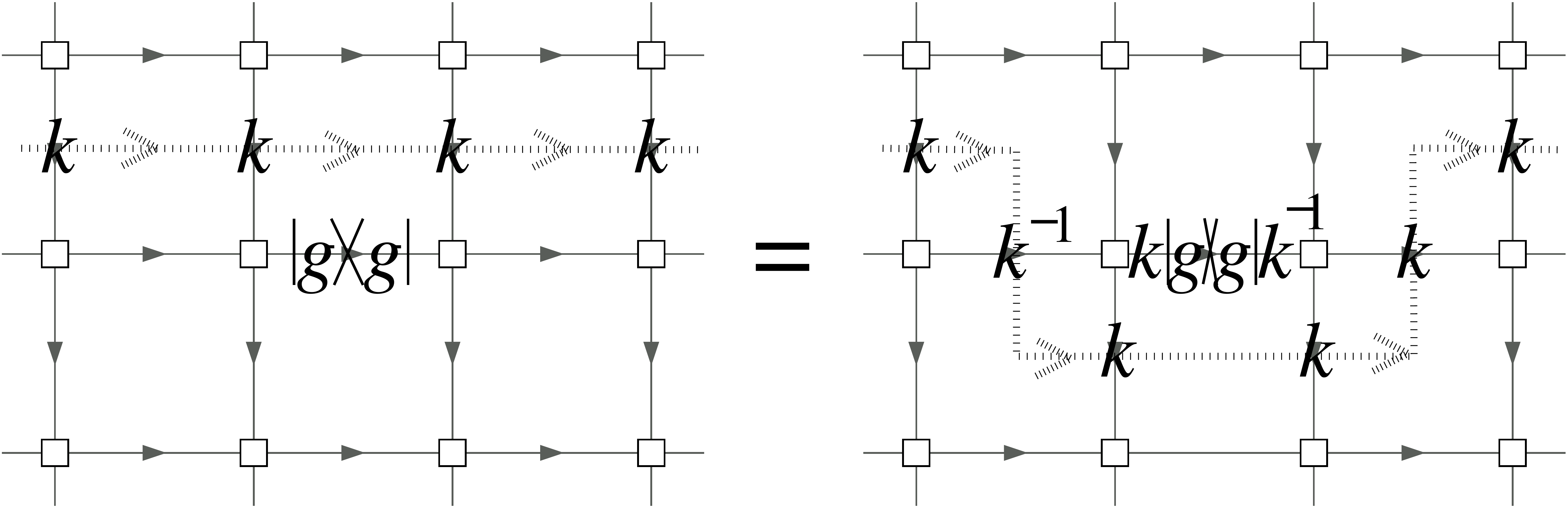}
}\ .
\]
This is, $\ket{g}\bra{g}\mapsto U_k\ket{g}\bra{g} U_k^\dagger=
\ket{kg}\bra{kg}$, and thus
\[
\sum_g \chi_c(pg)\ket{g}\bra{g}\mapsto
\sum_{\tilde g} \chi_c(pk^{-1}\tilde g)\ket{\tilde g}\bra{\tilde g}\ :
\]
braiding a chargeon around a fluxon of flux $k$ leaves its type $c$
invariant, but multiplies the parameter $p$ characterizing its internal
state by $k^{-1}$.

To interferometrically detect this change, one can e.g.\ start from a pair
of charges in state $\Pi_{c,p}$~\eqref{eq:anyons:chargeon-pair-state},
braid them around the flux $k$, and bring them back together. Then one can
reason as follows. We have seen that, by acting with a unitary $W$ on six
sites, the original state (before the braiding) decomposes as
$\sum_x\ket{\Pi_{c,p}}_{abcd}\ket{\psi_x}$ where
$|G|\ket{\Pi_{c,p}}_{abcd}$ is exactly the right hand side of
(\ref{eq:anyons:create-chargepair-accessible}) which is indeed independent
of $x$. The group orthogonality relations assure that
$\ket{\Pi_{c,p}}_{abcd}$ is normalized, which in turn gives
$\sum_{x,y}\langle \psi_x|\psi_y\rangle=1$. Now, after braiding and
applying $W$ we are left with $\sum_x
\ket{\Pi_{c,p,k}^x}_{abcd}\ket{\psi_x}$, where
$|G|\ket{\Pi_{c,p,k}^x}_{abcd}$ is the right hand side of
(\ref{eq:anyons:create-chargepair-accessible}) but with the coefficients
given by $\chi_c(ph^{-1}k^{-1}g)$.  Then, we can implement in abcd (after
having applied $W$) the measurement given by
$\ket{\Pi_{c,p}}\bra{\Pi_{c,p}}_{abcd}$ and its orthogonal complement,
which outputs that the state is unchanged with probability $$\left|
\frac{\chi_c(k)}{|c|} \right|^2\;$$ where $|c|$ denotes the dimension of
the representation. To see that, it is enough to use the group
orthogonality relations once more.

\section{Examples\label{sec:examples}}

In the following, we present examples of $G$--injective PEPS. We start by
discussing Kitaev's code state~\cite{kitaev:toriccode}, including an
explanation of how to construct the original PEPS and to obtain the
$G$--isometric PEPS by an RG step. We then consider the generalization of
the model to arbitrary finite groups~\cite{kitaev:toriccode}, the
so-called \emph{double models}, and show that they correspond to the set
of $G$--isometric PEPS. Finally, we show that by using semi-regular
representations, it is possible to construct PEPS which are locally
equivalent to the double models, yet can be realized with a lower bond
dimension.

\subsection{Kitaev's code state}

\begin{figure}
\begin{center}
\includegraphics[height=8em]{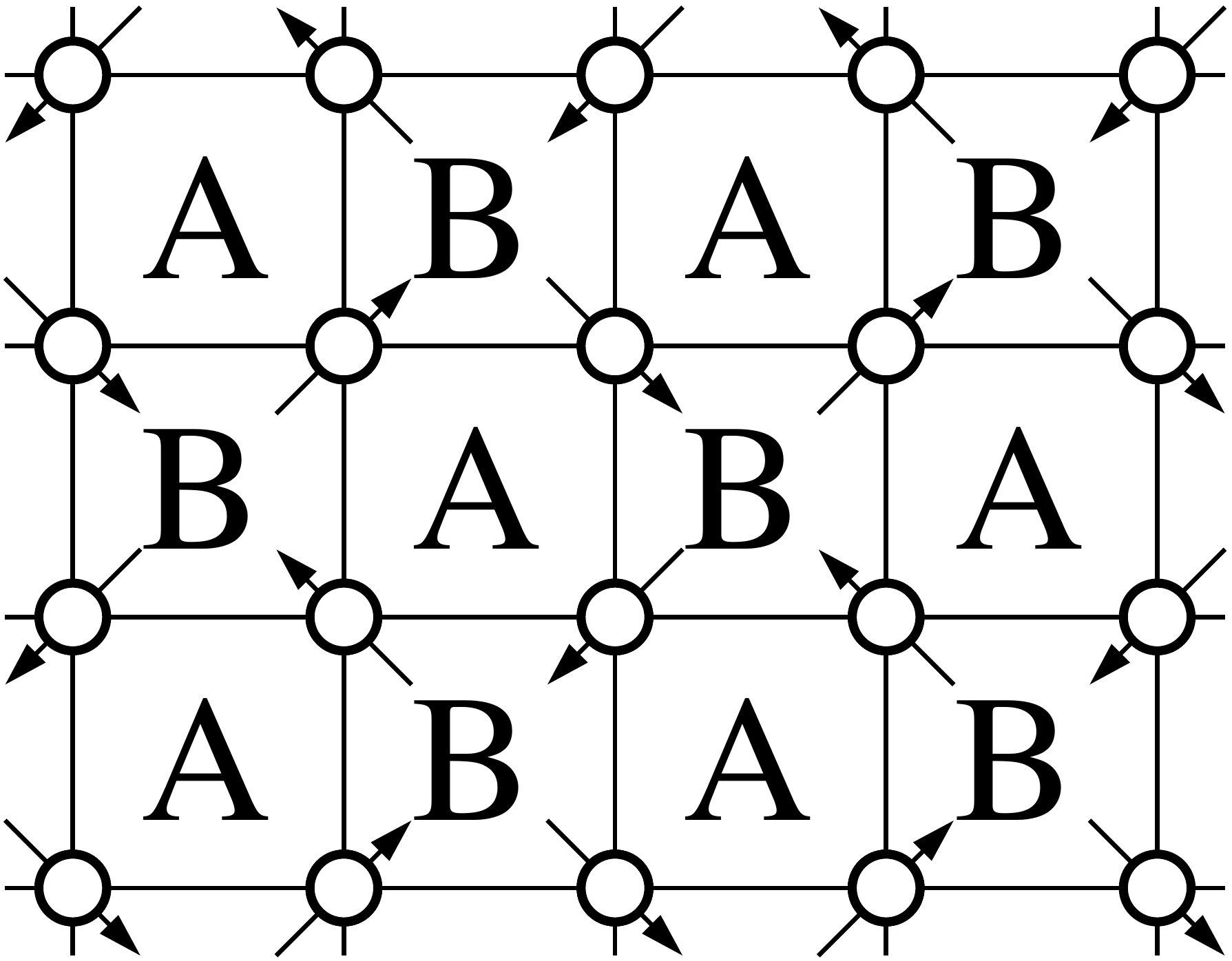}
\end{center}
\caption{
\label{fig:ex:kitaev-lattice}
The lattice for Kitaev's code state and the double models.}
\end{figure}

Kitaev's code state can be defined on a square lattice with qubits
$\{\ket0,\ket1\}$ attached to its vertices, and two types of plaquettes, A
and B, see Fig.~\ref{fig:ex:kitaev-lattice}. (The arrows will be used in
the non-abelian case). The Hamiltonian consists of local terms associated
to plaquettes, each acting on the four surrounding qubits: For each A
plaquette $p_A$, define
\begin{equation}
        \label{eq:ex:A-stabilizer}
        h^{p_A}_A:=\tfrac12(1-Z^{\otimes 4})
\end{equation}
acting on the four sites adjacent to the plaquette $p_A$; this is,
$h^{p_A}_A$ is the projector whose ground state subspace is formed by the
configurations $\ket{g,h,k,l}$ with even parity around the plaquette. For each B
plaquette $p_B$, define
\begin{equation}
        \label{eq:ex:B-stabilizer}
\begin{aligned}
h^{p_B}_B &:= \tfrac12(1-X^{\otimes 4}) \\
&=\sum_{s} \ket{g+s,h+s,k+s,l+s}\bra{g,h,k,l}\ ,
\end{aligned}
\end{equation}
where addition is modulo $2$; this is, $h^{p_B}_B$ is a projector whose
ground state subspace is invariant under flipping all spins adjacent to
$p_B$. Now let
\[
H_A:=\sum_{p_A} h^{p_A}_A\ ,\mbox{ and }
H_B:=\sum_{p_B} h^{p_B}_B\ .
\]
Then, Kitaev's code state Hamiltonian~\cite{kitaev:toriccode} is given by
$H=H_A+H_B$.

Let us now explain how to construct a PEPS representation of a ground
state of $H$~\cite{frank:comp-power-of-peps}. To this end, observe that
all terms $h^{p_A}_A$, $h^{p_B}_B$ in the Hamiltonian commute and are
projections: Thus, the ground state subspace is given by the product
\[
\prod_{p_B}\tfrac12(1-h^{p_B}_B)\prod_{p_A}\tfrac12(1-h^{p_A}_A)\ ,
\]
and a product of local operators can always be written as a
PEPS.

A particularly appealing way to construct the PEPS representation is
the following. We start from the state $\ket{0,\dots,0}$, which has even
parity around each plaquette and is thus a ground state of $H_A$. By
sequentially applying all projections $\tfrac12(1-h^{p_B}_B)$ to the initial
$\ket{0,\dots,0}$ state, we thus end up with a ground state of $H$.\footnote{
    Note we have to make sure that $\ket{0,\dots,0}$ has non-zero
    overlap with the ground state subspace of $H_B$. This is true as it
    has non-zero overlap with
    $(\ket{0}+\ket{1})\otimes\cdots\otimes(\ket{0}+\ket{1})$ which is in
    the kernel of $H_B$.
}

The projector $\tfrac12(1-h^{p_B}_B)$ can be expressed as a tensor network
of the form
\[
\includegraphics[height=4em]{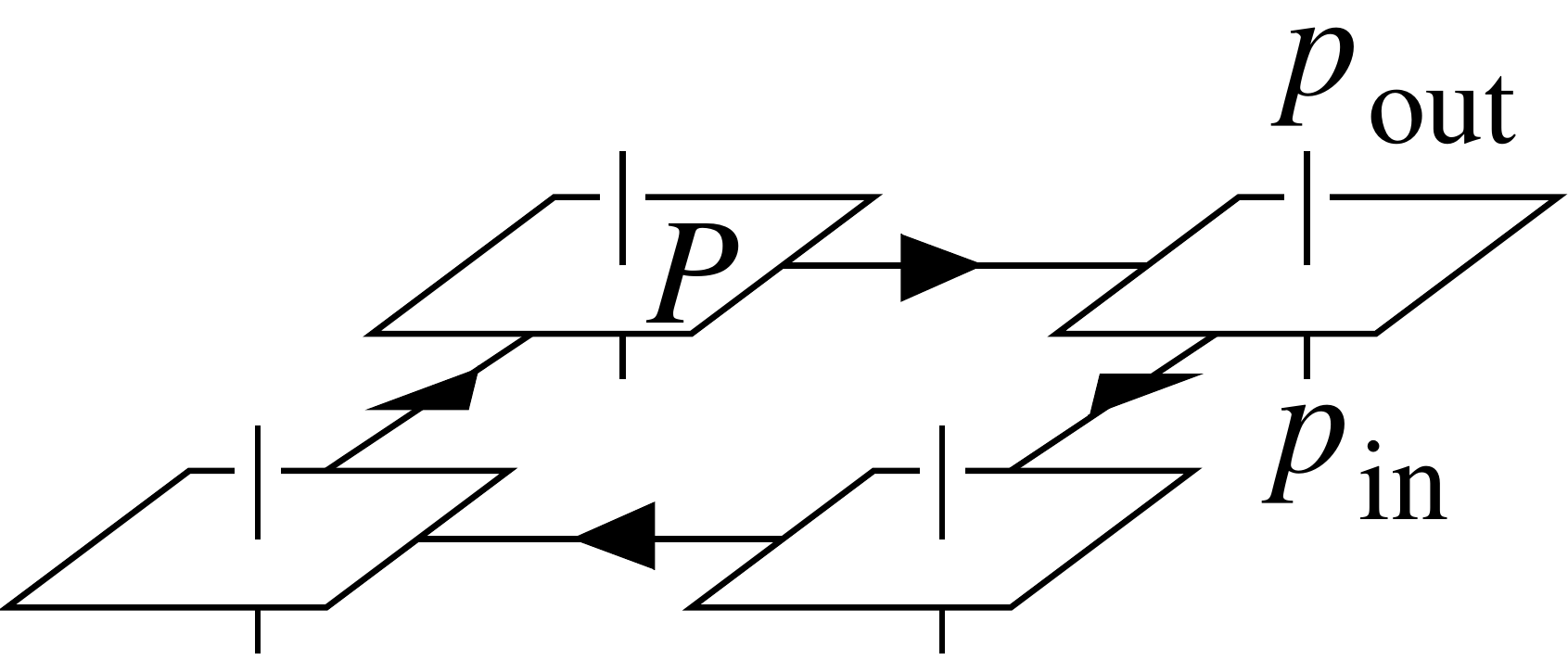}
\]
which takes a layer of physical inputs $p_\mathrm{in}$ at the bottom, and
maps them to a new layer of physical outputs $p_\mathrm{out}$ at the top.
The four PEPS tensors above are all equal,
\begin{equation}
        \label{eq:ex:kitaev-elemproj}
P=\sum_{s,k}\ket{k+s}_{p_\mathrm{out}}\bra{k}_{p_\mathrm{in}}
\otimes\ket{s}_v\bra{s}_{v}\ ,
\end{equation}
where $v$ denotes the virtual indices. The natural way to think of this
construction is to assign the value $s$ to the plaquette rather than to
the bonds: When contracting the virtual indices around the loop, there
only remains a single sum over $s$, which implements the projection
$\tfrac12(1-h_B^{p_B})$ by simulateously flipping all four physical spins
conditionally on $s$.

The B plaquette projections can be divided into two disjoint layers, the
PEPS representation of each of which is given
by~\eqref{eq:ex:kitaev-elemproj}.  When combining the layers, at each site
two $P$ tensors \eqref{eq:ex:kitaev-elemproj} are stacked atop of each
other, with the virtual indices pointing in opposite directions. In order
to obtain the PEPS description of the ground state itself, we have to apply
these two layers of $P$ to the initial state $\ket0$.  This yields the
tensor (cf.~\cite{frank:comp-power-of-peps})
\begin{equation}
\label{eq:ex:kitaev-tens}
\begin{aligned}
T =&  \sum_{r,s}\ket{r+s}_p\ket{r,s}_{v}\bra{r,s}_{v}\\
 = &\ket{0}\!\Big[\!\ket{0,0}_v\bra{0,0}_v+
                \ket{1,1}_v\bra{1,1}_v\!\Big]
    +  \ket{1}\!\Big[\!\ket{0,1}_v\bra{0,1}_v+\ket{1,0}_v\bra{1,0}_v\!\Big]
\end{aligned}
\end{equation}
Here, the $r$'s and $s$'s on the virtual level are placed such
that they are associated with one B plaquette each,
\[
\raisebox{-1.5em}{
\includegraphics[height=3.5em]{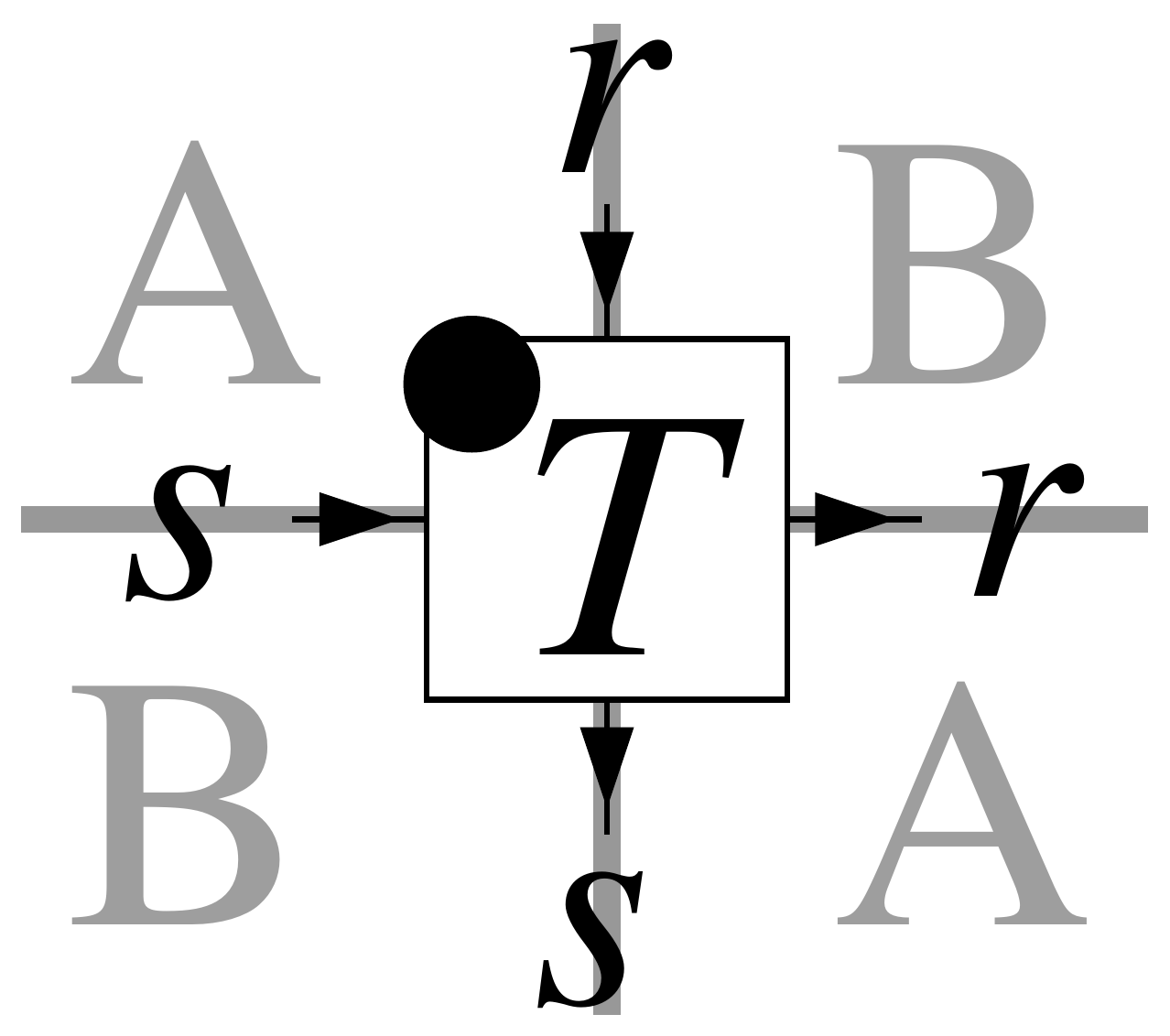}
}\ ,
\]
and correspondingly rotated for the other type of vertices.

It is again convenient to think of the virtual indices as assigned to B
plaquettes, rather than to bonds. Then, the value of each qubit in the
lattice is given by the difference $r-s\equiv r+s$ between the virtual
degrees of freedom $r$ and $s$ assigned to the two adjacent B plaquettes.
Alternatively, one can think of the state as arising from
$\ket{0,\dots,0}$ -- an eigenstate of $H_A$ -- which is symmetrized by
coherently flipping the spins adjacent to any B plaquette, controlled by its
virtual degree of freedom.

From \eqref{eq:ex:kitaev-tens}, it is straightforward to see that $T$ is
characterized by the following symmetries:
\[
\raisebox{-1.8em}{
\includegraphics[height=4.5em]{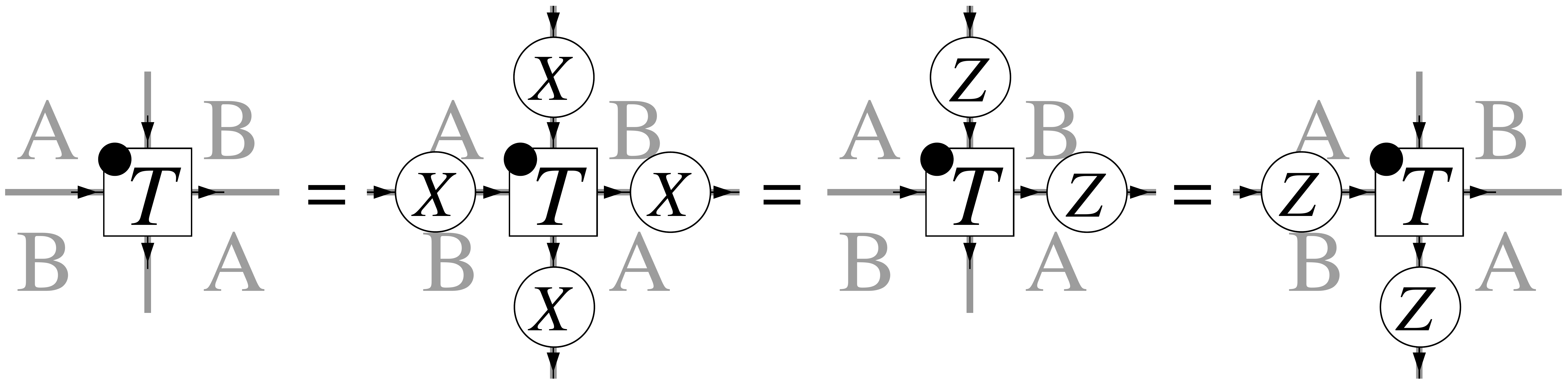}
}\ .
\]
Thus, $T$ has the desired $\mathbb Z_2$--invariance with respect to
$X^{\otimes 4}$, but it is lacking $\mathbb Z_2$--injectivity since it has
additional symmetries with respect to $Z^{\otimes 2}$. In order to get rid
of these symmetries, we consider a $2\times2$ block around an A
plaquette, which has symmetries
\[
\includegraphics[height=5.8em]{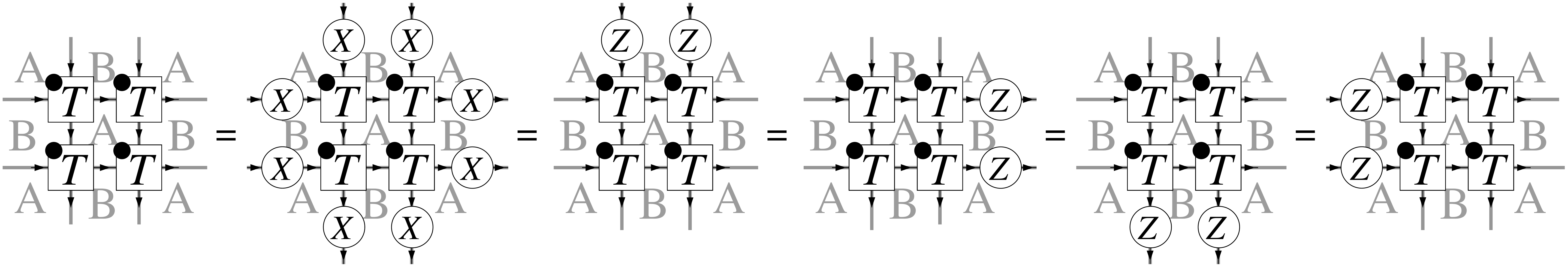}
\]
The $X^{\otimes 8}$ symmetry is the topologically interesting one, whereas
the $Z^{\otimes 2}$ symmetries, after blocking, only act between adjacent
tensors: They should therefore be local in nature. Moreover, since
$X\otimes X$ and $Z\otimes
Z$ commute, there has to be an isometry which separates their action:
This is accomplished by the CNOT gate
$\ket{a}\ket{b}\mapsto\ket{a}\ket{a+b}$ which we have already
used for the RG flow in Sec.~\ref{sec:RG}, and which maps
\begin{align*}
X\otimes X &\mapsto X\otimes \openone \\
Z\otimes Z &\mapsto \openone\otimes Z\ .
\end{align*}
By applying these CNOTs, the $2\times2$-block of $T$'s is thus transformed
to\footnote{
        Note that since $\mathcal P(T)$ is isometric on its support, the
        RG step can be implemented by a physical unitary and a relabelling
        of bonds.
}
\[
\raisebox{-2.5em}{
\includegraphics[height=6em]{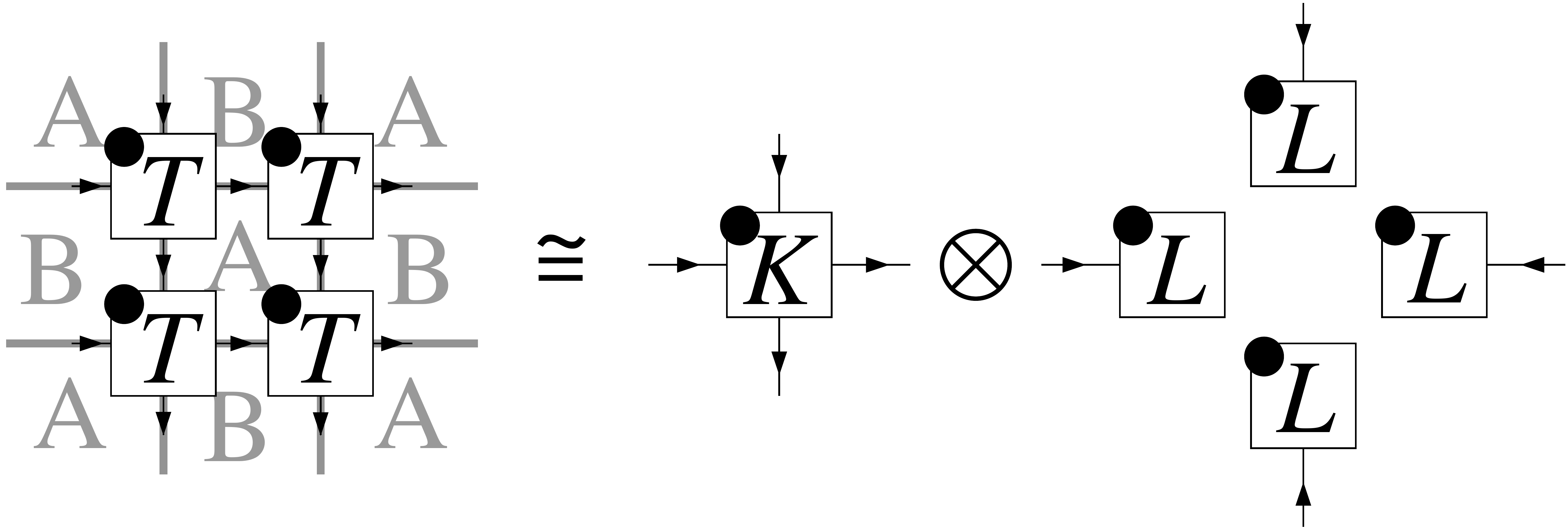}
}\ ,
\]
where the tensors $K$ and $L$ have exactly the symmetries
\[
\raisebox{-2em}{
\includegraphics[height=5em]{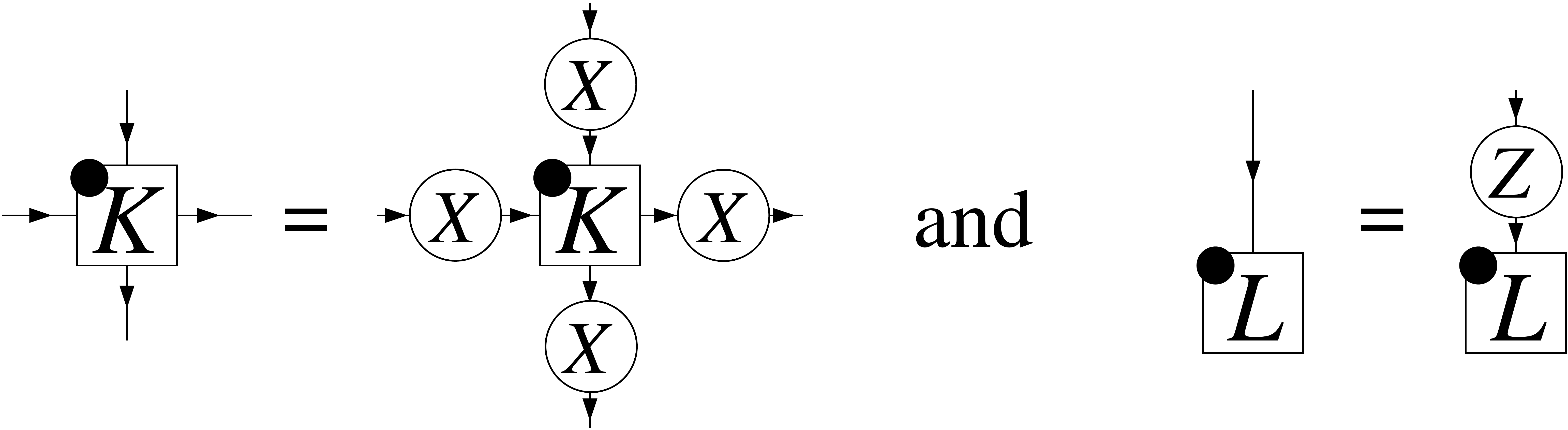}
}\ ,
\]
respectively.
Thus,
the tensor $K$ is the desired $\mathbb Z_2$--isometric tensor
describing the topological code state, while the $L$ are tensors with
only \emph{one} virtual qubit index
with $\mathbb Z_2$ symmetry: They therefore give rise to product states
between adjacent sites.  Note that this also explains why the
original PEPS representation of the code state needed double as many bonds
as actually required to obtain the topological entropy: The original
tensor $T$ has extra local symmetries.

Let us now try to see how the blocked and renormalized tensors $K$ actually
look like. As compared to the original lattice, the tensors sit on top of
(i.e., contain) an A plaquette, and each bond corresponds to a B
plaquette. The other half of the A plaquettes forms the plaquettes between
the renormalized tensors:
\[
\raisebox{-2.3em}{
\includegraphics[height=6em]{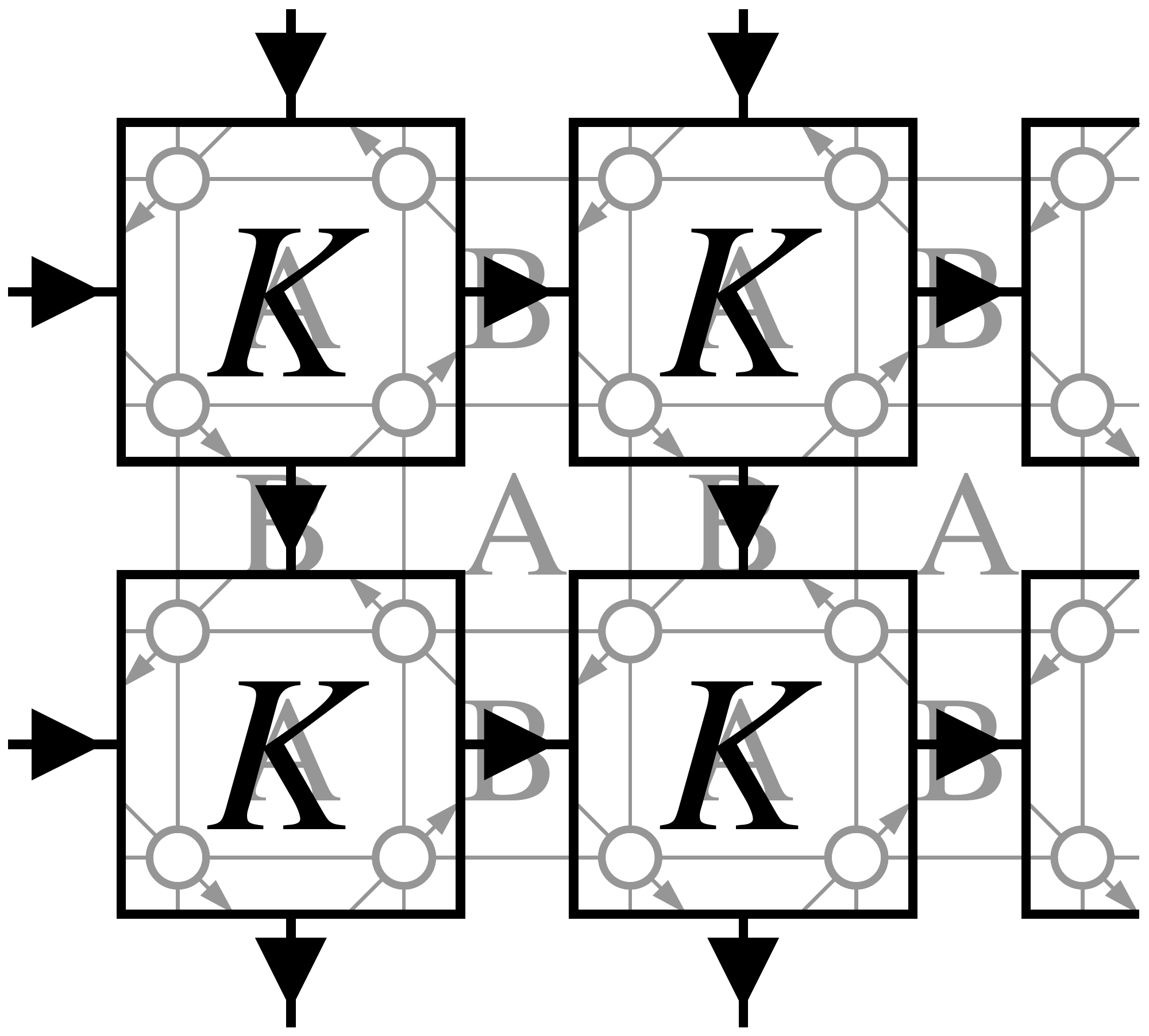}}\quad .
\]
The summation index $s$ in the PEPS construction
\eqref{eq:ex:kitaev-elemproj} and \eqref{eq:ex:kitaev-tens}, which we
argued was associated to B plaquettes, is now associated to edges,
i.e., it became an actual bond state. Going through the RG scheme (or
observing that the action of the tensor $P$,
Eq.~\eqref{eq:ex:kitaev-elemproj}, is to add $s$ to all qubits surrounding
the associated B plaquette), we find that
\begin{equation}
        \label{eq:ex:kitaev-colordiff-rep}
\begin{aligned}
K&=\sum_{pqrs}\ket{p+q,q+r,r+s,s+p}_p\otimes\ket{p,q}_v\bra{r,s}_v\\
=&\raisebox{-4em}{
    \includegraphics[height=9em]{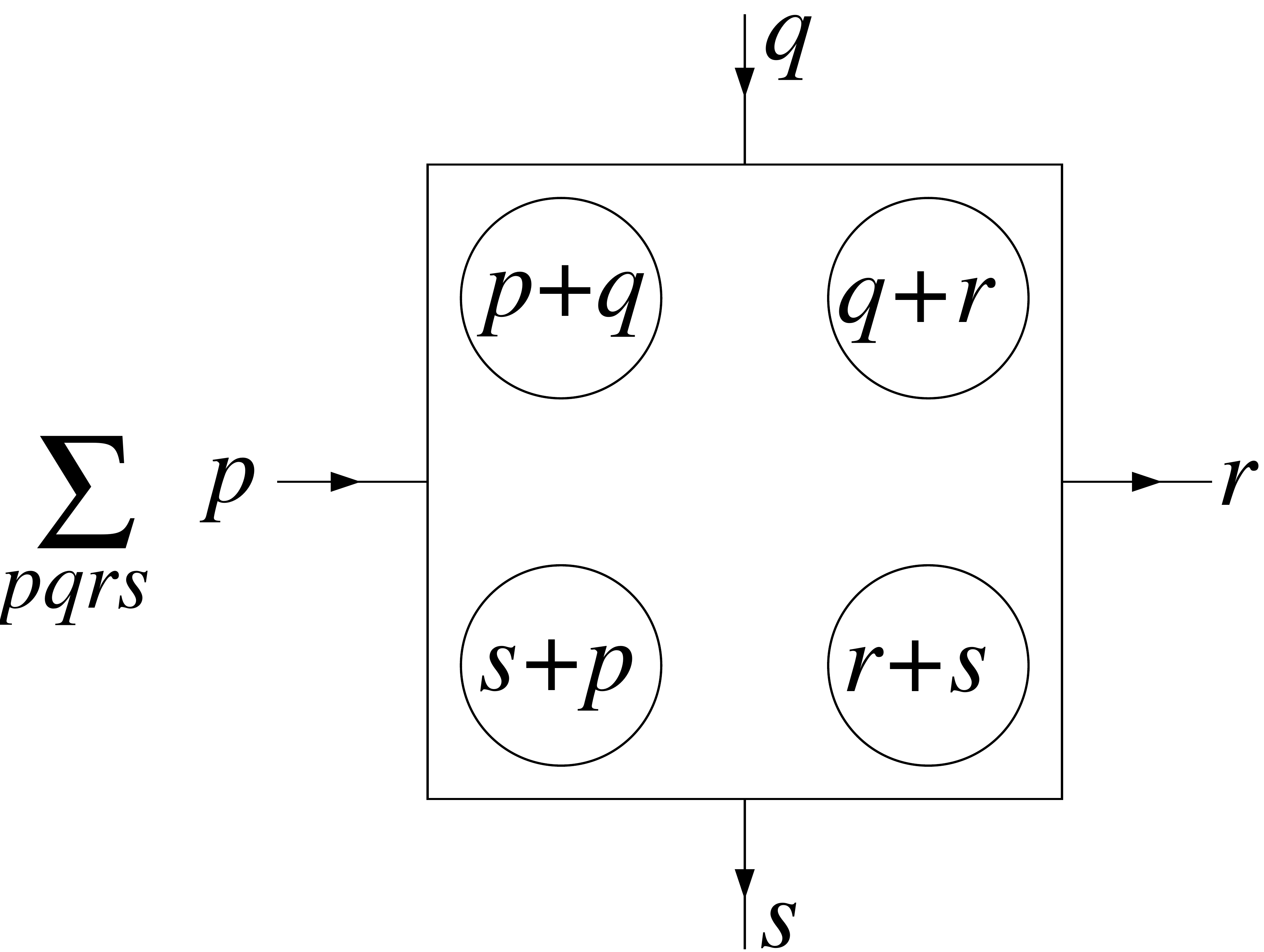}
}\quad .
\end{aligned}
\end{equation}
One possible perspective on the PEPS representation is thus to assign values
$0$ or $1$ (``colors'') to each $B$ plaquette, and let the physical states
be the difference between the colors of the adjacent plaquettes; the state
is then given by the superposition over all color patterns.  Note that
$\mathcal P(K)$ does not have full rank, as the physical subspace
in~\eqref{eq:ex:kitaev-colordiff-rep} is restricted to states with even
parity.

What is the parent Hamiltonian for the blocked tensor? As it turns out,
it can be split into three components, each of which can be associated to
a term in the original Hamiltonian: i) a local term which enforces
the local constraint; ii) a four-site term enforcing the
plaquette constraint for the A plaquette covering four tensors;
and iii) a two-site term making sure that the state is an equal weight
superposition of all edge colorings: it enforces that the state is
invariant under
adding $1$ to all four qubits surrounding the edge.

As we have seen earlier, there is also a different way to express the PEPS
tensor $K$, using Observation~\ref{obs:iso:accessible-virt}:
\[
\raisebox{-2em}{
    \includegraphics[height=5em]{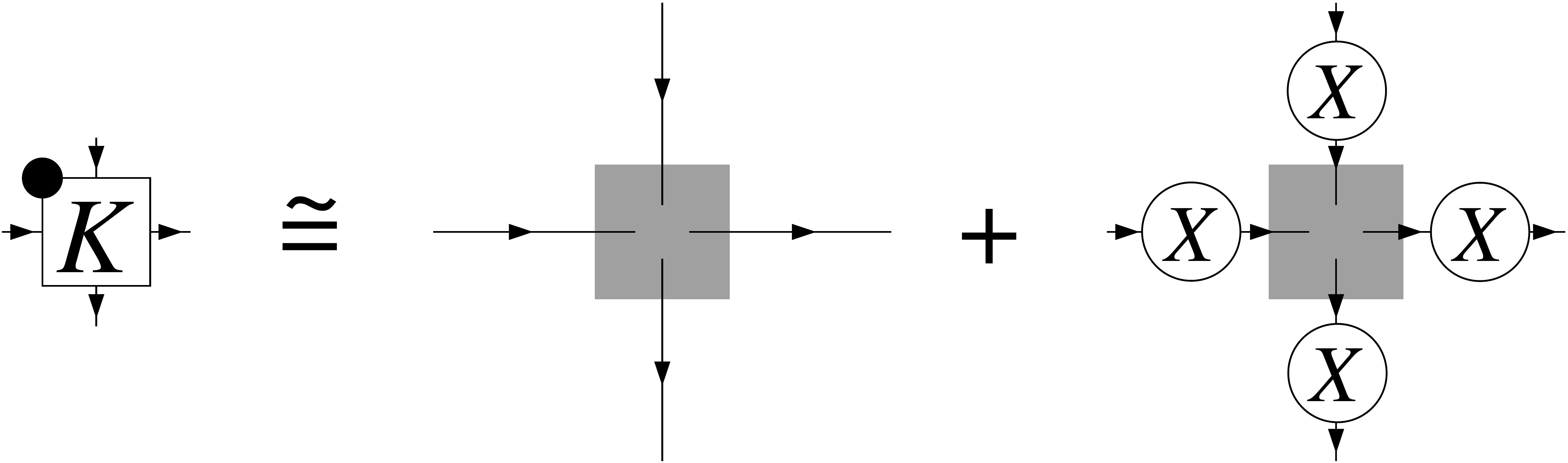}
}\quad .
\]
Written in this form, the four physical qubits are $+1$ eigenstates of
$X^{\otimes 4}$, rather than $Z^{\otimes 4}$ as in
\eqref{eq:ex:kitaev-colordiff-rep}, and the different terms in the
Hamiltonian ensure i) that this local constraint is obeyed, ii) that the
product of twirls around a plaquette vanishes, and iii) that all possible
twirls on an edge  appear with equal
weight.

\subsection{The double models}

Kitaev's code state Hamiltonian can be naturally generalized to any finite
group $G$.  To this end, identify the states at each site with group
elements $\ket{g}$, and generalize the Hamiltonian terms as follows:
\begin{equation}
h^{p_A}_A:=1-\sum_{ghkl=1}\ket{g,h,k,l}\bra{g,h,k,l}\ ,
\end{equation}
where the sum is over all  group elements whose product $ghkl$ equals the
identity $1$, and the product is taken in the order indicated by the
arrows in Fig.~\ref{fig:ex:kitaev-lattice} from left to right; and
\begin{equation}
h^{p_B}_B := 1-\sum_s \ket{U_sg,U_sh,U_sk,U_sl}\bra{g,h,k,l}\ ,
\end{equation}
where $U_s=L_s$ ($U_s=R_s$) if in Fig.~\ref{fig:ex:kitaev-lattice}, the
arrow on the spin points away from (towards) $p_B$; we keep this
convention throughout.  Here, $L_s:\ket{s}\mapsto \ket{sg}$ and
$R_s:\ket{s}\mapsto\ket{gs^{-1}}$ are the left and right multiplication,
respectively. The total Hamiltonian is again the sum of these terms for
all plaquettes.

The PEPS construction can be carried out the same way as for the $\mathbb
Z_2$ case, by starting from $\ket{1,\dots,1}$ and applying projectors
\begin{equation}
P=\sum_{s,k}U_s\ket{k}_{p_\mathrm{out}}\bra{k}_{p_\mathrm{in}}
\otimes\ket{s}_v\bra{s}_{v}\ ;
\end{equation}
applying two of these projectors to $\ket{1}$, we arrive at
the PEPS tensor
\begin{equation}
T =  \sum_{r,s}U_rU_s\ket{1}_p\ket{r,s}_{v}\bra{r,s}_{v}\ .
\end{equation}
Here, the action of $U_r$ and $U_s$ is determined by the site it acts on,
and by the plaquette associated with $r$ and $s$, respectively (which
implies the order does not matter); $U_rU_s\ket1$ is thus either
$\ket{rs^{-1}}$ or $\ket{sr^{-1}}$.

The tensor $T$ has again additional symmetries, which can be removed by
blocking around an A plaquette\footnote{
    The PEPS representation depends on the rotation direction prescribed
    by the arrows around the plaquette.
    Eq.~\eqref{eq:ex:doubles-renorm-tensor} is for clockwise rotation.
}
and renormalizing via
$\ket{a}\ket{b}\mapsto\ket{a}\ket{b^{-1}a}$;
this yields the $G$--isometric tensor for the quantum double,
\begin{equation}
        \label{eq:ex:doubles-renorm-tensor}
\begin{aligned}
K &= \sum \ket{pq^{-1},qr^{-1},rs^{-1},sp^{-1}}_p\otimes \ket{p,q}_v\bra{r,s}_v\\
=&\raisebox{-4em}{
    \includegraphics[height=9em]{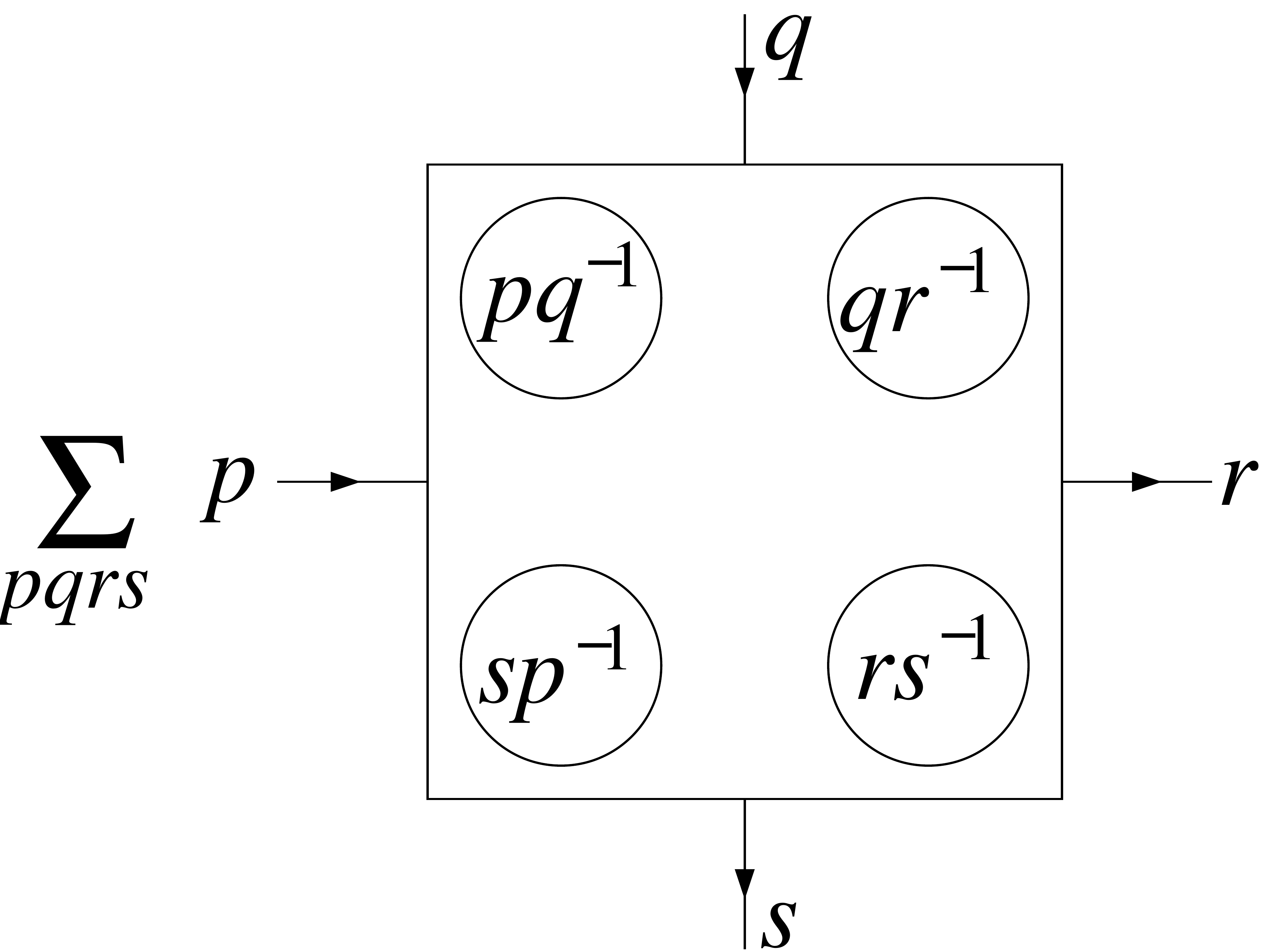}
}\quad .
\end{aligned}
\end{equation}
Note that this PEPS representation can again be interpreted by assigning
colors, i.e., group elements, to the bonds, and letting the state of each
qu-$|G|$-it be the difference between the adjacent bonds.

The Hamiltonian acts analogously to the $\mathbb Z_2$ case: i) a local
term ensures that the four sites in~\eqref{eq:ex:doubles-renorm-tensor}
multiply to the identity; ii) a plaquette term ensures the same for the
corresponding sites of four tensors, and iii) a two-body term enforces
equal weight for all bond configurations by applying $U_s$ coherently to
the four adjacent sites.

Alternatively, we can again use the representation
\begin{equation}
        \label{eq:ex:double-twirled-representation}
\raisebox{-2em}{
    \includegraphics[height=5em]{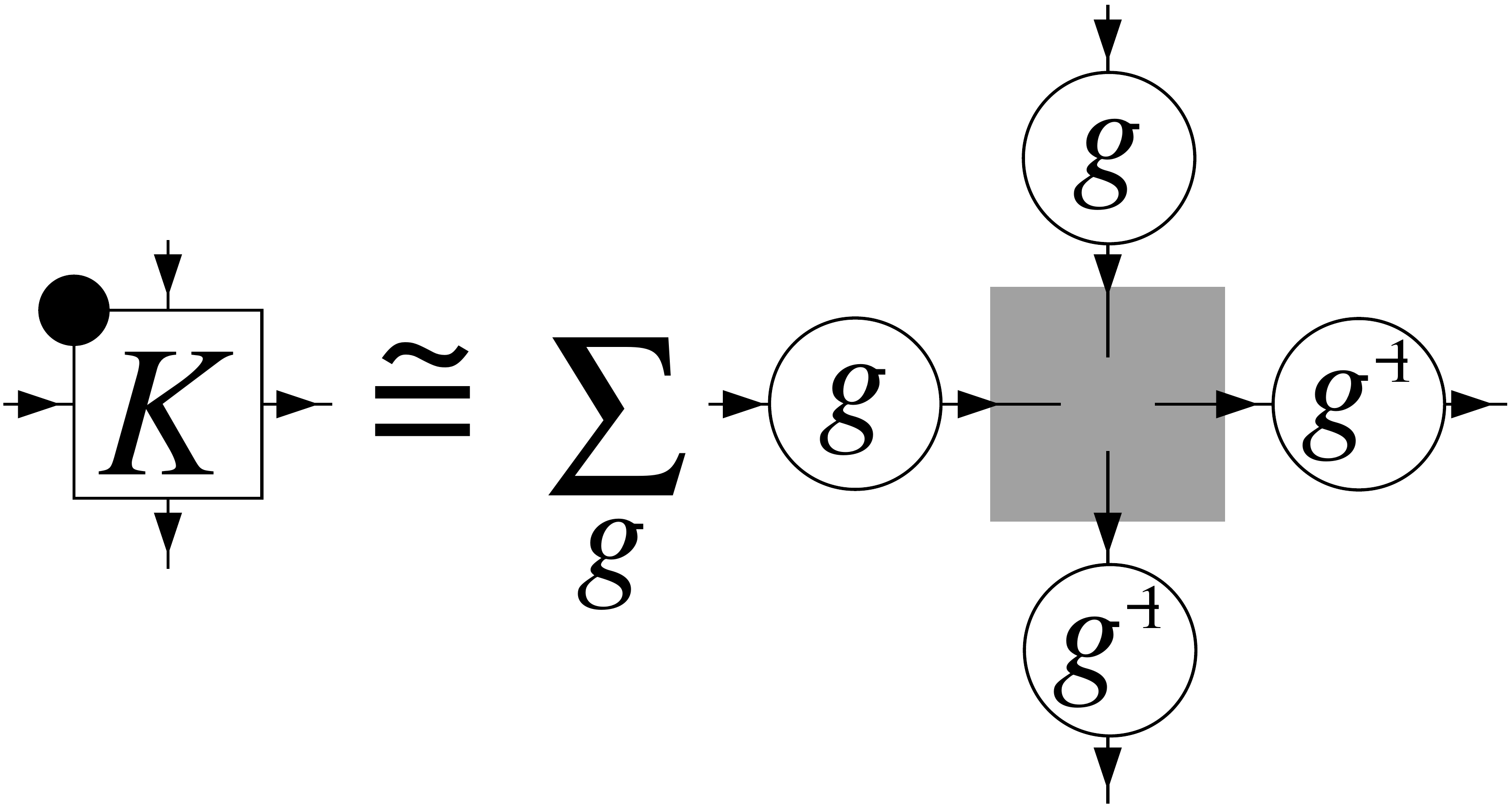}
}\quad ,
\end{equation}
where the role of the Hamiltonian terms in this representation are to
ensure that i) all $U_g$--invariant local configurations have equal
weight, ii) the product of the $g$'s around a plaquette is the identity,
and iii) for any bond, all twirls appear with equal weight.

\subsection{Reducing the bond dimension: semi-regular representations}

In the following, we would like to go beyond regular representations, and
show how a PEPS which is $G$--injective with respect to a semi-regular
representation can implement topological states which are equivalent to
the quantum doubles, but with a lower bond dimension.

To this end, consider the tensor $K$ in the
representation~\eqref{eq:ex:double-twirled-representation}. Here, $g$ is
the left-regular representation of a group, and can thus -- by a proper
relabelling of the virtual indices -- be chosen to be of the form
\[
L_g = \bigoplus_{i=1}^I D^i(g)\otimes {\openone}_{d_i}\ .
\]
Here, $D^i$, $i=1,\dots,I$, are the irreducible representations of $g$,
with dimensions $d_i$.
Let us now take the smallest semi-regular representation
\[
V_g = \bigoplus_{i=1}^I D^i(g)
\]
of $G$, which contains every irreducible representation with multiplicity
one, and define
\[
\Theta = \bigoplus_{i=1}^I \sqrt[4]{d_i}\;\openone_{d_i}\ .
\]
[Note that $\Theta^4=\Delta$, with $\Delta$
from~\eqref{eq:noninj:deltadef}.] With this, we define a tensor
\begin{equation}
        \label{eq:ex:V-Theta-double-def}
\raisebox{-4.5em}{
\includegraphics[height=10em]{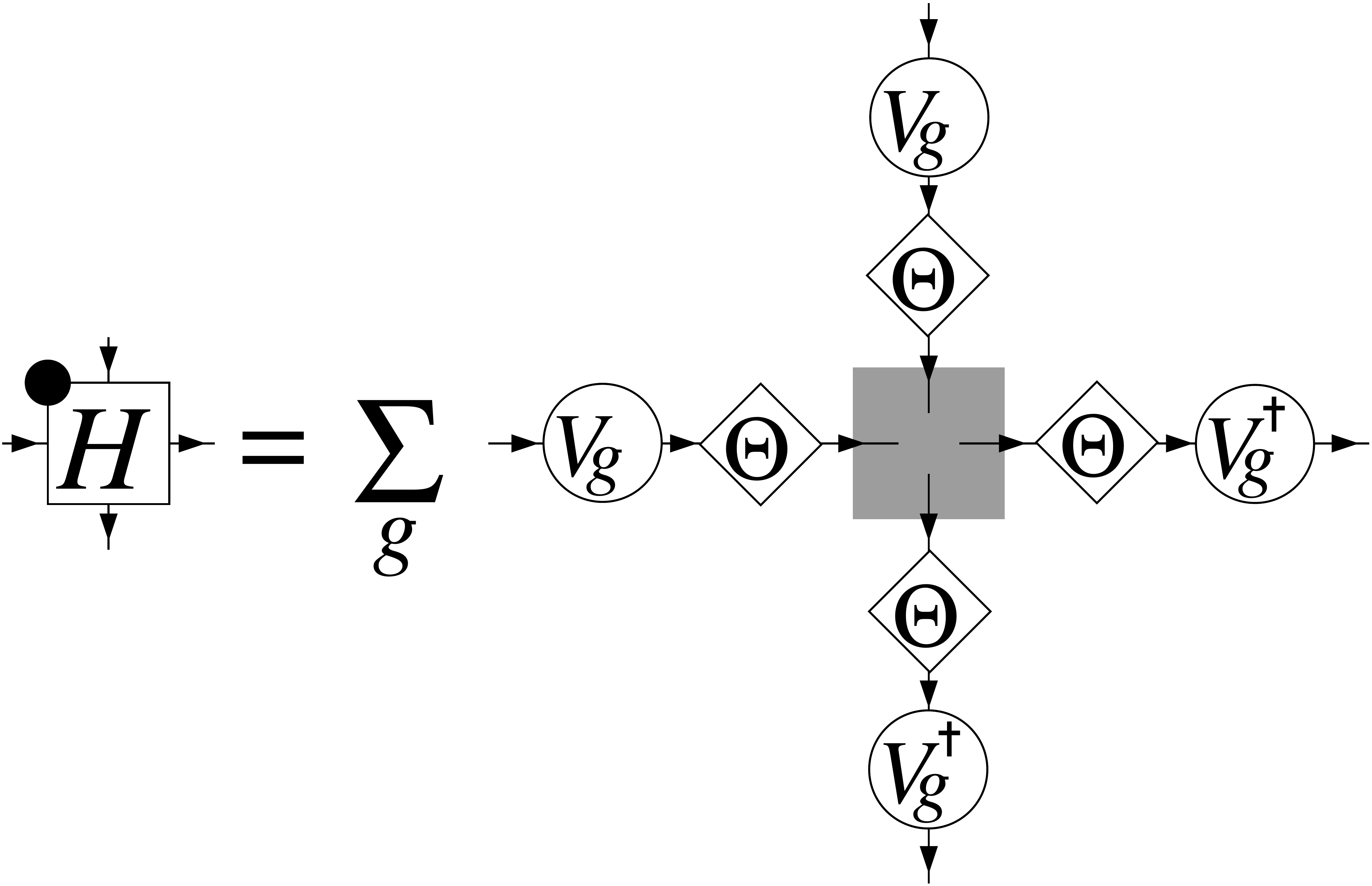}
}\ .
\end{equation}
We now claim that the state given by the tensor $H$ is the same as the
double model for the group $G$, i.e., the state described
by~\eqref{eq:ex:double-twirled-representation} with the regular
representation, up to unitaries between neighboring tensors which act on
disjoint subsystems, i.e., local operations.
This demonstrates that the double model has unneeded local degrees of
freedom, and the bond dimension needed to represent the model is $\sum_i
d_i$ rather than $\sum_i d_i^2$. Note that this is different from the bond
dimension gained by renormalization in that no blocking of tensors is
required.

In order to prove local equivalence of the two states, we will explicitly
construct an isometric transformation which maps between the two states.
(It has to be isometric as the dimension of the underlying space changes.)
The transformation will act on the two endpoints of a bond simulateously,
where it needs to implement a mapping between
\[
\raisebox{-0.5em}{
\includegraphics[height=1.8em]{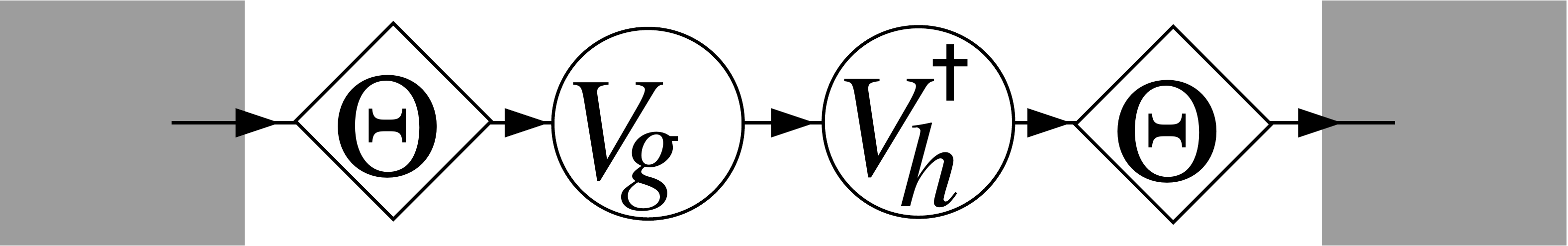}
}
\mbox{\ \ and\ \  }
\raisebox{-0.5em}{
\includegraphics[height=1.8em]{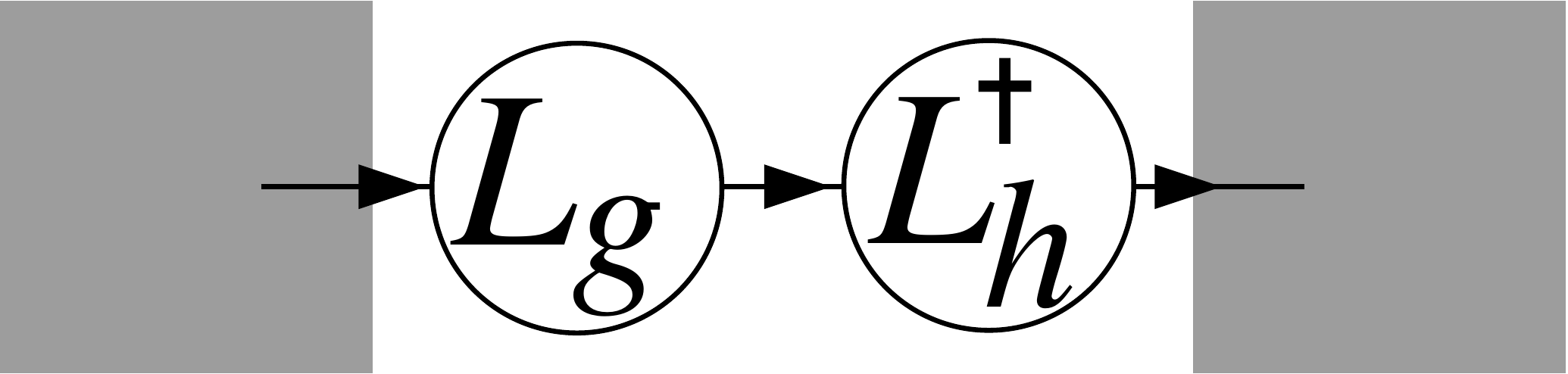}
}\ ,
\]
this is, between
\[
\sum_{i=1}^I \sum_{k_i,l_i=1}^{d_i}
    \sqrt{d_i} D^i_{k_i,l_i}(gh^{-1}) \ket{i;k_i}\bra{i;l_i}
\]
and
\[
\sum_{i=1}^I \sum_{k_i,l_i=1}^{d_i} \sum_{m_i=1}^{d_i}
    D^i_{k_i,l_i}(gh^{-1}) \ket{i;k_i;m_i}\bra{i;l_i;m_i}\ .
\]
This, however, can be accomplished by mapping between the basis elements
of the two spaces as
\[
\ket{i;k_i}\bra{i;l_i} \mapsto
    \frac{1}{\sqrt{d_i}}\sum_{m_i=1}^{d_i}
    \ket{i;k_i;m_i}\bra{i;l_i;m_i}\ ;
\]
and thus, we have shown that the PEPS with
tensor~\eqref{eq:ex:V-Theta-double-def}, using a semi-regular
representation, describes a state which is locally equivalent to the
double model for the same group.

\section{Conclusions and outlook\label{sec:conclusion}}

In this paper, we have presented a general framework for understanding
Matrix Product States and Projected Entangled Pair states, and thus a
large variety of interesting quantum states, in terms of symmetries. While this
classification allowed us to reproduce the known results in one dimension,
it also enabled us to go beyond one dimension and characterize the
properties of two-dimensional systems with non-unique ground states. Using
the symmetry-centered perspective, we could show how these states arise as
finitely degenerate ground states of local Hamiltonians, and characterize
the different ground states using symmetries at the closure. By looking
only at these symmetries, we were able to explain the topological
correction to the area law, the local indistinguishability of the ground
state subspace, the anyonic nature of excitations of these systems, and
why they are fixed points of renormalization group flows. 

We believe that the framework set forward in this article will allow to
explain a variety of things beyond what we have presented in this
paper. For instance, the results on the parent Hamiltonian and
ground state degeneracy can be directly extended to three and more spatial
dimension (replacing the pair-conjugacy classes by triplet-conjugacy
classes, etc.), as well as the statements made about the topological entropy,
local indistinguishability, RG flows, and commuting Hamiltonians. By
introducing excitations which are e.g.\ membranes of $g$'s on the bonds,
the framework of $G$--injective PEPS will also allow to explain the
excitations of e.g.\ three-dimensional models.

The results presented here likely can be generalized to models with other
symmetries. For instance, on might think of tensors for which the ordered
product of the group elements assigned to all indices is the identity for
all non-zero entries. This symmetry shares the crucial property of being
stable under concatenation, and it allows for the definition of injective
and isometric tensors, and thus should allow for comparable statements
about topological entropy, local indistinguishability, renormalization
transformations, or non-abelian excitations.  Note that Kitaev's code
state and the double models can be expressed in this form if one
constructs the PEPS by starting from a ground state of $H_B$ and
subsequently projects onto the ground state subspace of $H_A$. More
generally, similar results might hold when considering more general
symmetries,  e.g.\ tensor categories used to define string-net
models~\cite{levin-wen:string-net-models}, or different tensor-network
based ansatzes~\cite{sierra:2d-dmrg,nishino:tps}.  Note that from a
different perspective, the role of so-called gauge-like symmetries for
topological order has been studied
in~\cite{nussinov:symmetry-topology-prl,nussinov:symmetry-topology-long}.

The framework for understanding RG fixed points by mapping between
isomorphic representations suggests the extension to the case of
renormalization flows, e.g.\ by considering the tensor product of
individual bond symmetries. 
For instance, one could use the fact that the
normalized character of the tensor product of any faithful representation 
converges to the character of the regular representation, up to global
phases, to show that RG flows converge towards $G$--isometric PEPS.

\section{Acknowledgements}

We would like to thank the Perimeter Institute for Theoretical Physics in
Waterloo, Canada, where a major part of this work has been done, for its
hospitality.  We acknowledge helpful discussions with Miguel Aguado,
Oliver Buerschaper, Sofyan Iblisdir, Frank Verstraete, and Michael Wolf.
This work has been supported by the EU project QUEVADIS, the DFG
Forschergruppe 635, the Gordon and Betty Moore Foundation through
Caltech's Center for the Physics of Information, the National Science
Foundation under Grant No. PHY-0803371, and the Spanish grants I-MATH,
MTM2008-01366 and CCG08-UCM/ESP-4394.

\end{document}